\documentclass[12pt]{article}
\usepackage{amsfonts}
\usepackage{a4, amsmath, amssymb, verbatim, setspace, rotating, array, caption}
\usepackage[a4paper,left = 3cm,right = 3cm,top = 3.1cm,bottom = 3.3cm]{geometry}
\usepackage{pgfplots,tikz, verbatim,multirow}

\usepackage{xr-hyper}

\usetikzlibrary{patterns}
\usepgfplotslibrary{external}
\pgfplotsset{compat = newest} 
\usepgfplotslibrary{groupplots}
\usetikzlibrary{matrix,calc}
\pgfplotsset{compat = newest}
\usepgfplotslibrary{fillbetween}
\newenvironment{proof}[1][Proof]{\textbf{#1.} }{\ \rule{0.5em}{0.5em}}
\usepgfplotslibrary{groupplots}
\usetikzlibrary{pgfplots.groupplots}
\usetikzlibrary{plotmarks}
\usetikzlibrary{patterns}
\usepgfplotslibrary{external}
\pgfplotsset{compat = newest} 

\usepgfplotslibrary{external}
\usepgfplotslibrary{groupplots}
\newcounter{lemmacounter}
\newcounter{assumptioncounter}
\newtheorem{theorem}{Theorem}

\newtheorem{assumption}[assumptioncounter]{Assumption}

\newtheorem{lemma}[lemmacounter]{Lemma}

\newcommand*{\QED}
{\hfill\ensuremath{\blacksquare}}
\begin{document}

\title{\textbf{Marital Sorting, Household Inequality and Selection}\thanks{%
We are grateful to the editor Bo Honor\'e, an associate editor, two anonymous referees, David Card and Stella Hong for useful comments and to Richard Blundell, Bram De Rock and Costas Meghir for several important references.}}
\author{Ivan Fern\'andez-Val\thanks{
\ Boston University.} \and Aico van Vuuren\thanks{
\ University of Groningen } \and Francis Vella \thanks{
\ Georgetown University}}
\maketitle

\begin{abstract}
Using CPS data for 1976 to 2025, we explore how wage inequality has evolved for married couples with both spouses working full time full year, and its impact on household income inequality. We also investigate how marriage sorting patterns have changed over this period. To determine the factors driving income inequality we estimate a model explaining the joint distribution of couples' wages that  accounts for spouses' employment decisions. This requires a multivariate distribution regression model that incorporates a double selection rule. We find that income inequality has increased for these households and increased wage sorting has exacerbated the inequality resulting from individual wage growth. We find that positive sorting partially reflects the correlation across unobservables influencing the wages of both members of the marriage. We decompose the changes in sorting patterns over the 50-year sample into structural, composition and selection effects and find that the increase in positive sorting primarily reflects the increased skill premia for both observed and unobserved characteristics. 
\thispagestyle{empty}

\end{abstract}

\section{Introduction}

\addtocounter{page}{-1} While most empirical studies evaluate inequality via an examination of individuals' wages or earnings (see, for example, Katz and Murphy 1992, Murphy and Welch 1992, Juhn, Murphy, and Pierce 1993, Welch 2000, Autor, Katz, and Kearney 2008, Blau and Kahn 2009, Acemoglu and Autor 2011, Autor, Manning, and Smith 2016, Murphy and Topel 2016, and Fern\'{a}ndez-Val et al. 2023, 2024), for many individuals their household's income
is more determinative of their economic welfare. It is therefore useful for policy purposes to analyze income inequality by household type. One type of particular interest is married couples with both individuals working full time full year (FTFY). Although the proportion of husbands and wives in the total population aged between 24 and 65 years decreased from 76.7\% to 56.8\% from 1976 to 2025, this group represents an increasing share of married couples. In 1976 26.2\% of married couples in which both the husband and the wife were aged between 24 and 65 years had both members working FTFY, and by 2025 this share increased to 53.9\%. As the percentage of married males in FTFY employment increased from 83.8 to 84.9 while that for married females went from 31.5 to 61.7, the growth of this group reflects the increasing employment rate of married females.

Examining the patterns and determinants of income inequality across FTFY married couples is interesting for several reasons. First, studies on wage and income inequality generally do not distinguish between individuals by marital status. We investigate whether inequality for the married shows similar patterns to those of the unmarried and married individuals with non-working spouses. Second, as household income combines spouses' earnings, sorting on wages is likely to have implications for inequality across households relative to inequality across individuals. Finally, the large increase in the female participation rates, and the associated selection bias, has been shown to affect females' wages and income inequality (see, for example, Mulligan and Rubinstein 2008, Blau, Kahn, Boboshko, and Comey 2024 and Fern\'{a}ndez-Val et al. 2023). While selection bias has been largely ignored for males, the existing evidence indicates it is not economically important. However, as increasing married female participation rates increase the sample of males in FTFY married couples, female selection into FTFY employment may have implications for the wage and earnings distributions of married males with working spouses.

There is a vast and varied literature on the relationship between marriage, sorting and inequality and it provides mixed evidence on sorting behavior. Kremer (1999) documents that marital sorting on education in the United States, measured by the correlation between spouses' education levels, had declined over the period 1940-1990 and via a calibration exercise finds that marital sorting had a larger effect on intergenerational mobility than inequality. Fern\'{a}ndez and Rogerson (2001) develop and calibrate a dynamic model of intergenerational education acquisition, fertility and marital sorting based on the PSID and conclude that an increase in sorting is likely to increase the degree of income inequality. Greenwood, Guner, Kocharkov, and Santos (2015) examine US census data and document an increase in positive assortative matching on education and via the comparison of counterfactual income distributions conclude that the level of income inequality measured by the Gini coefficient has increased. Chiappori, Salani\'{e} and Weiss (2017) provide a model in which increased returns to education result in higher paid couples spending more time with their children and increasing assortative matching on education. Examining US marriages for individuals who are married and born between 1943 and 1972, they find that this occurred for white, but not black, individuals. Eika, Mogstad, and Zafar (2019) examine data for Denmark, Germany, Norway, the United Kingdom, and the United States and show that there is a considerable amount of educational assortative matching although it has changed little since the 1980s. Moreover, while there was an increase in sorting at the bottom of the educational distribution, there is a remarkable decrease in assortative matching at the top of the educational distribution. They conclude that the increases in the Gini coefficient of household income cannot be explained by changes in assortative matching. However, they also note that assortative matching on education contributes to the cross-sectional inequality in household income in each country. This supports the earlier work of Breen and Salazar (2011) for the United States that finds a very small impact of educational sorting on changes in income inequality between the late 1970s and early 2000s. Gihleb and Lang (2018) apply a range of statistical measures of assortativeness to the Current Population Survey (CPS) data for 1970 to 2010 and the 2010 American Community Survey to test for increased educational homogamy among married couples in the US and conclude that there is no evidence of increased assortative matching. However, they do not explore sorting on wages due to the non-trivial non-participation of wives in market employment. Chiappori, Costa-Dias and Meghir (2023) and Chiappori, Costa Dias, Meghir and Zhang (2025) provide alternative measures of assortativeness to illustrate the difficulty in quantifying how it varies across economies with different marginal distributions of the variable on which it is being evaluated. They conclude that the degree to which educational homogamy has changed among married couples in the US depends on the measure employed. The conclusion also varies on where in the educational distribution it is measured. Chiappori, Costa-Dias, Crossman and Meghir (2020) employed related measures in evaluating educational assortativeness in the U.K. and conclude that there are no clear patterns. However, they also conclude that the changes appear to have only slightly increased income inequality. Our empirical work provides a more detailed analysis of sorting on wages while accounting for selection although we do so by examining a more homogeneous, albeit large, group of workers. However, as our primary interest is on the role of sorting on inequality we focus on this issue rather than adopting the measures proposed by Gihleb and Lang (2018), Chiappori, Costa-Dias and Meghir (2023) and Chiappori, Costa Dias, Meghir and Zhang (2025). Finally, Chiappori, et al. (2026) extend the Choo and Siow (2006) model to examine matching behavior using a Dutch register-based dataset for the period 2013 to 2019.\label{insert_rr:chiappori} They study sorting on income in the year of marriage, thereby avoiding the issue of income changes that occur after marriage. They find that the majority of couples exhibit positive assortative matching, although a minority of ``traditional'' couples, in which the husband earns considerably more, display patterns consistent with negative assortative matching.

 Earnings appear to be a more appropriate measure than wages for evaluating the welfare implications of inequality, but examining different measures may lead to substantially different conclusions. While substantial evidence suggests female wage inequality has risen, Fern\'{a}ndez-Val, Van Vuuren, Vella and Peracchi, hereafter FVVP (2024), show that annual earnings inequality has decreased for females in the United States due to the shifts in their annual hours of work distribution. This finding is similar to Cancian and Reed (1998a, 1998b) who find household inequality decreased due to the large shifts in the female income distribution resulting from changes in their hours of work. Analyzing changes in annual income is more challenging for married couples as it requires jointly modeling the annual hours and wages of both spouses. However, one could begin by restricting attention to those with a relatively homogeneous level of hours while accounting for the accompanying selection. As FTFY individuals generally work a similar number of hours, one can then examine income inequality via comparisons of the sum of couples' wages.

We document the changes in the earnings distributions of dual FTFY
households and examine if household inequality has been exacerbated by
sorting behavior. We do so by examining how the wage distributions of husbands and wives, and the associated sorting patterns, have changed over time. We report how the probability of a male in the $j^{\mathrm{th}}$ decile of the married male wage distribution being married to a female in the $k^{\mathrm{th}}$ decile of the married female wage distribution has evolved over our sample period. We investigate the source of these changes by estimating the
conditional joint distribution of husbands' and wives' wages via bivariate distribution regression. While this provides insight into how the observed sorting patterns can be explained by observable characteristics and unobservable factors, it does not account for the selection arising from the couple's respective employment decisions. We incorporate these employment decisions by introducing an estimator for a model featuring bivariate outcomes and bivariate selection rules.

Our paper makes several important methodological and empirical contributions.\label{insert:contributions} Our primary methodological contribution is to extend the Chernozhukov, Fern\'{a}ndez-Val and Luo (2025), hereafter CFL, sample selection model to a bivariate outcome with a bivariate selection rule. This also represents an extension of the Chernozhukov et al. (2025) bivariate distribution regression estimator to a setting of bivariate selection. We show point identification of this model under the same exclusion restrictions as in CFL. This analysis
relies on a useful result for the multivariate standard normal distribution. Lemma \ref{lemma1} in the appendix characterizes the derivative of this function with respect to each element of the correlation matrix and shows that each derivative is positive. While this result is well known in the bivariate case (e.g., Sibuya 1959) we are unaware of a corresponding result for the general multivariate case. Using this result we can construct the appropriate simulation based estimator. 
A second methodological contribution is to provide an approach to identify the role of selection and sorting on observable and unobservable characteristics on the observed matching patterns. 
We highlight that our analysis is based on the Local Gaussian Representation employed by CFL who show that in the bivariate case this representation does not make any assumptions about the underlying distribution.\label{insert:normality}
% This implies that the estimation method employed here is non-parametric or at least semiparametric depending on the specification of how the covariates affect the distribution. 
While this result does not hold for higher dimensions, our approach assumes local normality which is a far weaker assumption than normality.
Finally, we provide procedures for decomposing the observed matching patterns into components due to selection, structural, and composition effects. Our empirical contribution is to apply these novel methodological tools to CPS data to analyze marriage patterns in the US for the period from 1976 to 2025 to examine how sorting behavior and the returns to observed and unobserved characteristics have affected across household income inequality.

%\textcolor{blue}{Our paper provides both methodological and empirical contributions.}
%Our approach is motivated byChernozhukov, Fern\'{a}ndez-Val, and Luo (2019) (hereafter CFL)
%which is based on the Heckman (1974, 1979) selection model although estimation is substantially more challenging. We show point identification of this model under the same exclusion restrictions as in CFL. This analysisrelies on a useful result for the multivariate standard normal distribution. Lemma \ref{lemma1} in the appendix characterizes the derivative of this function with respect to each element of the correlation matrix and shows it is positive. While this result is well known in the bivariate case (e.g., Sibuya, 1959) we have not found it for the general multivariate case. We use the model's estimates to evaluate the role of the different forcesgenerating the observed marital patterns and their implications for the couples' earnings distribution.

%\textcolor{red}{
%Our theoretical analysis is based on the Local Gaussian Representation as in CFL who show that in the bivariate case such a representation does not make any assumptions about the underlying distribution. \label{insert:normality} This implies that in such as case the estimation method is non-parametric or at least semiparametric depending on the specification of how the covariates affect the distribution. Unfortunately, this result does not hold for higher dimensions implying that even though we do not assume normality, our method is not nonparametric even in the case that there are no covariates.} 

Our primary empirical findings are the following. First, wage inequality among couples with both spouses working FTFY has increased for the period 1976-2025. Second, increasing levels of sorting on wages have contributed to increasing household inequality. Third, there is mixed evidence regarding the role of selection from work decisions on observed sorting patterns. Fourth, the primary factors behind the observed positive sorting patterns are the observed characteristics of the individuals and the correlation between the unobservables driving the wages of each spouse. Finally, positive sorting on wages has increased over the 50 years of our sample and this reflects the increasing market value of observed and unobserved individual characteristics. Consistent with Gihleb and Lang (2018), Eika, Mogstad, and Zafar (2019), Chiappori, Costa-Dias and Meghir (2023) and Chiappori, Costa Dias, Meghir and Zhang (2025), our results establish that the nature of sorting on observed characteristics, such as education, has not substantially changed. However, the prices of these characteristics have changed, pushing individuals further up, or down, their respective wage distributions.

We treat three important aspects of our model as exogenous. The first is the individual's decision to marry and the second is their choice of spouse (see, for example, Chiappori, Costa-Dias and Meghir 2019 and Killewald and Lundberg 2017).\label{insert:studies} We also do not address how household income is allocated across spouses nor the implications of this allocation for the work or marriage decisions (see, for example, Lise and Seitz 2011, Lise and Yamada 2018 and De Rock, Kovaleva and Potoms 2023). While each of these issues merits attention, our focus on understanding the sorting pattern among spouses, and its implication for inequality in the presence of endogenous employment decisions, remains important. Moreover, it is not clear that attempts to incorporate each of these considerations would produce a model which is estimable. However, we describe below how the model could be extended to incorporate the endogenous marriage decisions. 

Two other restrictions worth noting are related to our focus on married couples and our definition of work as FTFY. We limit our investigation to married couples to provide a clear definition of the household. The use of FTFY as our definition of employed is to enable comparisons with earlier work that also employs this definition (see, for example, Mulligan and Rubinstein 2008, and FVVP 2023). Moreover, if we included part-time workers in the working sample it would be necessary, as noted above, to model their hours of work decision to estimate a measure of earnings, and this would require modeling the joint distribution of hours worked. Additionally, as we use hourly wages rather than annual labor income, employing a homogeneous group of workers in terms of working hours is especially beneficial when one aims to explore inequality between households. \label{insert:only_ftfy} However, we acknowledge the potential significance of these choices and we explore, and discuss below, their impact on our primary findings. 

The following section briefly describes the CPS data and the sample-selection procedure. It also presents the time series trends in earnings and the wage distributions of FTFY households. Section \ref{sec:marital_sorting} describes the observed patterns of marital sorting and Section \ref{sec:econometric_model} provides our econometric model. The measures of sorting are discussed in Section \ref{sec:sorting}, and estimation is discussed in Section \ref{sec:estimation}. Our empirical results are presented in Section \ref{sec:empirical}. Section \ref{sec:robustness} investigates some robustness checks. Section \ref{sec:conclusion} concludes.

\section{Data}

\subsection{Data}

\label{sec:data} 
% ----------------------------------------------------------------------------------------------------------------------------

We employ the Annual Social and Economic Supplement (ASEC) of the CPS, or March CPS, for the 50~survey years from 1976 to 2025 which report annual earnings and hours worked for the previous calendar year.\footnote{The data are taken from the IPUMS-CPS website maintained by the Minnesota Population Center at the University of Minnesota (Flood, King, Ruggles, and Warren 2015).} The 1976 survey is the first for which information on weeks worked and usual hours of work per week last year are available.\footnote{We refer to the year of the survey and not the calendar year to which it refers.} To avoid issues related to retirement and ongoing educational investment we restrict attention to those aged 24--65 years in the survey year. This produces an overall sample of 2,164,205 males and 2,345,910 females. The annual sample sizes range from 30,767 males and 33,924 females in 1976 to 55,039 males and 59,622 females in 2001.

Annual hours worked are defined as the product of weeks worked and usual weekly hours of work last year. Those reporting zero hours generally report not being in the labor force (i.e., they report themselves as doing housework, unable to work, at school, or retired) in the week of the March survey. We define hourly wages as the ratio of reported annual labor earnings in the year before the survey, converted to constant 2025 prices using the consumer price index for all urban consumers, and annual hours worked. Hourly wages are unavailable for those not in the labor force. As annual earnings and hours are often poorly measured for members of the Armed Forces, self-employed, and the unpaid family workers, we exclude these groups and focus on civilian dependent employees with positive hourly wages and people out of the labor force last year. This restricted sample comprises 1,881,729 males and 2,207,407 females (respectively 86.9\% and 94.1\% of the original sample of those aged 24--65). The subsample of civilian dependent employees with positive hourly wages contains 1,623,359 males and 1,545,155 females. Married individuals aged 24-65 comprise 65\% of the total sample. Figure \ref{fig:cohabit} shows this has decreased from 77\% in 1976 to 57\% in 2025, implying they still represent a substantial fraction of the total sample. In contrast, the percentage of married couples with both spouses working FTFY increased drastically from 26.2 in 1976 to 53.9 in 2025.\label{insert::additional_numbers} This sharp increase implies that the share of individuals in dual-FTFY marriages increased from 20.1\% of the total population between 24 and 65 years in 1976 to 30.6\% in 2025.

Our analysis does not include unmarried couples who cohabit\label{insert:cohabiting} as the CPS only started reporting cohabitation in 1995. Figure \ref{fig:cohabit} indicates that cohabitation increased from 3.0\% of all couples in 1995 to 7.8\% in 2025 noting that Smock and Manning (2004) report much lower numbers for the 1970s and 1980s for US Census Bureau data. The number of married couples decreased in the period from 1995 to 2005 from 64.9\% to 57.8\%, while the number of cohabiting couples increased from 3.0\% to 7.8\%. The number of married and unmarried couples that form a household decreased by only 2.3\% over a 30-year period.

\begin{figure}
\centering
\begin{tikzpicture}
	\begin{groupplot}[
		group style={
			group size=2 by 1,
			horizontal sep=2cm
		},
		width=7cm,
		height=6cm,
		xlabel={Year},
		xticklabel style={rotate=45, anchor=east},
		every axis plot/.append style={thick}
		]
		
		% ---- Left plot ----
		\nextgroupplot[
		title={A. Married couples},
		ylabel={Percent},
		ymin=0.55, ymax=0.80,
		xmin=1976, xmax=2025,
		xtick={1976,1985,1995,2005,2015,2025}
		]
		
		\addplot[mark=none] table {
			1976 0.76757
			1977 0.76075
			1978 0.74623
			1979 0.74079
			1980 0.73091
			1981 0.71778
			1982 0.71057
			1983 0.70605
			1984 0.70176
			1985 0.69416
			1986 0.68565
			1987 0.68280
			1988 0.67672
			1989 0.67646
			1990 0.66928
			1991 0.65875
			1992 0.65432
			1993 0.65699
			1994 0.65072
			1995 0.64943
			1996 0.64639
			1997 0.63947
			1998 0.63903
			1999 0.63605
			2000 0.63431
			2001 0.65392
			2002 0.64999
			2003 0.64659
			2004 0.64748
			2005 0.64249
			2006 0.63664
			2007 0.63920
			2008 0.63033
			2009 0.62962
			2010 0.61249
			2011 0.60743
			2012 0.60122
			2013 0.59980
			2014 0.59771
			2015 0.59272
			2016 0.59156
			2017 0.59429
			2018 0.58792
			2019 0.59004
			2020 0.59184
			2021 0.57449
			2022 0.57246
			2023 0.57986
			2024 0.57417
			2025 0.56834
		};
		
% ---- Right plot ----
\nextgroupplot[
title={B. Unmarried cohabiting couples},
ymin=0.03, ymax=0.09,
xmin=1995, xmax=2025,
xtick={1995,2005,2015,2025},
ytick={0.04, 0.06, 0.08},
scaled y ticks=false,
yticklabel style={/pgf/number format/fixed, /pgf/number format/precision=2, /pgf/number format/fixed zerofill}
]

\addplot[mark=none] table {
1995   0.02973
1996   0.03288
1997   0.03503
1998   0.03652
1999   0.04003
2000   0.04453
2001   0.04444
2002   0.04615
2003   0.04955
2004   0.04926
2005   0.04974
2006   0.05173
2007   0.05178
2008   0.05494
2009   0.05640
2010   0.06443
2011   0.06491
2012   0.06878
2013   0.06958
2014   0.06920
2015   0.07271
2016   0.07126
2017   0.06775
2018   0.07457
2019   0.07424
2020   0.07233
2021   0.08039
2022   0.08053
2023   0.07824
2024   0.08070
2025   0.07750
};
		
	\end{groupplot}

\end{tikzpicture}
\caption{Married and cohabiting couples as percentage of the total population between the ages of 24 and 65.}
\label{fig:cohabit}
\end{figure}

\begin{comment}
The March CPS differs from the Outgoing Rotation Groups of the CPS, or ORG
CPS, which contains information on hourly wages in the survey week for those
paid by the hour and on weekly earnings from the primary job during the
survey week for those not paid by the hour. Lemieux (2006) and Autor et al.\
(2008) argue that the ORG CPS data are preferable because they provide a
point-in-time wage measure and workers paid by the hour (more than half of
the U.S.\ workforce) may recall their hourly wages better. However, there is
no clear evidence regarding differences in the relative reporting accuracy
of hourly wages, weekly earnings and annual earnings. In addition, many
workers paid by the hour also work overtime, so their effective hourly wage
depends on the importance of overtime work and the wage differential between
straight time and overtime. Furthermore, the failure of the March CPS to
provide a point-in-time wage measure may be an advantage as it smooths out
intra-annual variations in hourly wages.
\end{comment}

\subsection{Descriptive Statistics}

Figure \ref{fig:income} presents the time series of various quantiles of
annual household labor earnings (in 2025 dollars) for dual FTFY households.
Median (Q2) household annual income increased by 35.4\% from 113 to 154
thousand dollars from 1976 to 2025. However, growth has been more modest at
lower quantiles. For example, at the first quartile (Q1), household income increased by 24.2\% from 88.6 to 110 thousand. Moreover, household income at this quantile was virtually constant from 1976 to 2000. There were increases of 19.0\% and 17.5\% at the first decile (D1) and the 5th
percentile, respectively. Increases at higher quantiles have been notably larger. Income at Q3 grew by 55.8\% from 141.2 to 220 thousand. Increases at D9 and the 95th percentile were 106.9\% and 97.8\% respectively. These changes substantially increased inequality. The Q3/Q1 and D9/D1 ratios are shown in Figure \ref{fig:income1}. The former increased from 1.74 to 2.27 and the latter from 3.13 to 4.64.

\begin{figure}[h!]
\caption{Household income of households with both spouses working FTFY.}
\label{fig:income}\centering
% This file was created with tikzplotlib v0.10.1.
\begin{tikzpicture}

\definecolor{darkgray176}{RGB}{176,176,176}
\definecolor{lightgray204}{RGB}{204,204,204}

\begin{axis}[
legend cell align={left},
legend style={
  fill opacity=0.8,
  draw opacity=1,
  text opacity=1,
  at={(1.05,1)},
  anchor=north west,
  draw=lightgray204
},
tick align=outside,
tick pos=left,
x grid style={darkgray176},
xlabel={Years},
xmin=1976, xmax=2025,
xtick style={color=black},
y grid style={darkgray176},
ylabel={Household income in 2025\$},
ymin=35.4424132633209, ymax=430.509409589767,
ytick style={color=black}
]
\addplot [very thick, black]
table {%
1976 59.56125831604
1977 60.6553688049316
1978 61.3995094299316
1979 62.5517845153809
1980 59.6321823120117
1981 57.4152742385864
1982 55.2316398620605
1983 54.5825567245483
1984 55.4868881225586
1985 55.4182174682617
1986 55.3898315429688
1987 57.2504692077637
1988 56.0830497741699
1989 57.2682151794434
1990 55.6836128234863
1991 55.2259864807129
1992 53.9237041473389
1993 55.8851203918457
1994 54.2843704223633
1995 55.0357780456543
1996 53.5178909301758
1997 53.4000040054321
1998 54.7344627380371
1999 55.6460800170898
2000 56.4971733093262
2001 58.2749519348145
2002 58.4802284240723
2003 60.3173065185547
2004 59.6515960693359
2005 58.9690929412842
2006 59.4369049072266
2007 59.1011558532715
2008 60.5649719238281
2009 58.3079986572266
2010 58.531078338623
2011 60.4293746948242
2012 60.0056495666504
2013 58.7909622192383
2014 59.2467041015625
2015 58.9172462463379
2016 59.5762748718262
2017 61.4274940490723
2018 61.4689292907715
2019 62.4293785095215
2020 63.8493385314941
2021 67.9171295166016
2022 67.1751403808594
2023 64.2937927246094
2024 65.7244651794434
2025 70
};
\addlegendentry{p5}
\addplot [very thick, black, dash pattern=]
table {%
1976 69.9887008666992
1977 71.6836166381836
1978 71.1840896606445
1979 72.1751403808594
1980 70.016081237793
1981 68.1615753173828
1982 65.5875778198242
1983 65.0094146728516
1984 66.1638412475586
1985 66.4143142700195
1986 67.0508422851562
1987 68.7005615234375
1988 67.5408416748047
1989 67.8449157714844
1990 67.832763671875
1991 66.2711868286133
1992 65.1575698852539
1993 65.7209091186523
1994 65.141242980957
1995 66.0941604614258
1996 64.1926483154297
1997 64
1998 66.4632720947266
1999 66.0161972045898
2000 69.5141220092773
2001 71.022590637207
2002 70.885124206543
2003 71.6542663574219
2004 70.9001846313477
2005 71.4237365722656
2006 70.6817321777344
2007 70
2008 74.0906448364258
2009 71.4237365722656
2010 73.1638412475586
2011 71.9397354125977
2012 72.5649719238281
2013 71.0960464477539
2014 72.0386054992676
2015 71.5932235717773
2016 72.8154449462891
2017 74.4971694946289
2018 74.9639205932617
2019 74.9177474975586
2020 79.2836929321289
2021 82.4708099365234
2022 81.0734481811523
2023 80.3672332763672
2024 79.3201446533203
2025 83.2959976196289
};
\addlegendentry{p10}
\addplot [very thick, black, dash pattern=]
table {%
1976 88.5517406463623
1977 90.1766109466553
1978 90.6186294555664
1979 91.4218444824219
1980 89.8983001708984
1981 87.581916809082
1982 86.2994384765625
1983 84.1749992370605
1984 86.6116180419922
1985 87.546142578125
1986 87.5873622894287
1987 90.4557418823242
1988 91.1694869995117
1989 90.154426574707
1990 88.6242656707764
1991 87.641242980957
1992 87.5216598510742
1993 88.2984924316406
1994 86.8549957275391
1995 87.6338882446289
1996 86.4519729614258
1997 87.5699977874756
1998 87.9661102294922
1999 90.4595184326172
2000 92.2787170410156
2001 94.6968002319336
2002 95.6949234008789
2003 95.9133682250977
2004 95.4425659179688
2005 96.338981628418
2006 96.3841705322266
2007 94.8888931274414
2008 98.4180755615234
2009 96.2033920288086
2010 96.9713535308838
2011 97.8380355834961
2012 97.6836242675781
2013 96.3078002929688
2014 98.2956695556641
2015 98.1092376708984
2016 100.455520629883
2017 100.636528015137
2018 102.448219299316
2019 102.384185791016
2020 108.052734375
2021 111.214309692383
2022 110.028251647949
2023 107.156311035156
2024 107.759521484375
2025 110
};
\addlegendentry{p25}
\addplot [very thick, black, mark=*, mark size=1, mark options={solid}, only marks]
table {%
1976 113.343784332275
1977 115.369270324707
1978 116.162078857422
1979 116.923728942871
1980 114.966102600098
1981 111.339466094971
1982 110.463279724121
1983 110.516006469727
1984 113.423736572266
1985 114.715637207031
1986 117.159698486328
1987 120.22599029541
1988 121.559310913086
1989 119.322044372559
1990 119.140155792236
1991 117.655372619629
1992 117.463264465332
1993 118.476470947266
1994 118.061985015869
1995 119.596984863281
1996 117.327682495117
1997 118
1998 120.806777954102
1999 123.292457580566
2000 126.177017211914
2001 127.476455688477
2002 128.479293823242
2003 130.093399047852
2004 130.38134765625
2005 129.559326171875
2006 130.118637084961
2007 130.66667175293
2008 131.955924987793
2009 131.186431884766
2010 133.158203125
2011 133.807907104492
2012 135.361587524414
2013 135.355926513672
2014 134.651611328125
2015 135.231658935547
2016 139.011306762695
2017 139.845581054688
2018 142.146911621094
2019 142.338989257812
2020 149.800369262695
2021 155.239151000977
2022 152.88134765625
2023 150.018844604492
2024 149.369110107422
2025 154
};
\addlegendentry{p50}
\addplot [very thick, black, mark=*, mark size=1, mark options={solid}]
table {%
1976 141.380096435547
1977 143.367233276367
1978 144.956680297852
1979 146.756118774414
1980 143.984233856201
1981 140.892654418945
1982 139.805099487305
1983 139.770248413086
1984 144.930328369141
1985 147.922790527344
1986 151.593231201172
1987 156.665916442871
1988 157.474563598633
1989 155.895568847656
1990 156.926544189453
1991 156.073440551758
1992 154.314498901367
1993 156.47834777832
1994 158.510360717773
1995 158.757064819336
1996 157.470211029053
1997 156.855495452881
1998 162.248596191406
1999 165.523574829102
2000 169.868179321289
2001 173.003753662109
2002 173.668563842773
2003 174.387939453125
2004 175.546127319336
2005 176.067794799805
2006 176.704330444336
2007 177.333343505859
2008 179.915824890137
2009 179.365409851074
2010 184.372894287109
2011 184.165725708008
2012 188.389831542969
2013 185.943496704102
2014 188.512252807617
2015 190.915252685547
2016 195.494903564453
2017 196.045196533203
2018 203.615814208984
2019 200.772876739502
2020 212.421844482422
2021 220.730682373047
2022 220.056488037109
2023 214.312622070312
2024 211.177032470703
2025 220
};
\addlegendentry{p75}
\addplot [very thick, black, mark=triangle*, mark size=1, mark options={solid,rotate=180}]
table {%
1976 170.392161560059
1977 176.451965332031
1978 180.678161621094
1979 181.400192260742
1980 177.203384399414
1981 171.355926513672
1982 172.598876953125
1983 175.525436401367
1984 182.458447265625
1985 184.148788452148
1986 192.406784057617
1987 200.376663208008
1988 200.50905456543
1989 198.870056152344
1990 202.485855102539
1991 199.293788146973
1992 198.075317382812
1993 201.186431884766
1994 206.280609130859
1995 211.676071166992
1996 205.838027954102
1997 208
1998 215.028244018555
1999 223.261779785156
2000 226.920910644531
2001 227.636535644531
2002 230.376663208008
2003 237.167602539062
2004 238.606399536133
2005 240.847457885742
2006 240.960433959961
2007 241.111114501953
2008 242.259887695312
2009 246.339004516602
2010 253.146896362305
2011 253.227874755859
2012 253.977403259277
2013 255.674761962891
2014 257.914364624023
2015 265.160095214844
2016 268.755187988281
2017 280.998138427734
2018 288.135620117188
2019 283.429382324219
2020 306.967956542969
2021 309.265533447266
2022 310.395477294922
2023 316.646911621094
2024 309.03955078125
2025 317.060003662109
};
\addlegendentry{p90}
\addplot [very thick, black, mark=square*, mark size=1, mark options={solid}]
table {%
1976 193.323659515381
1977 200.714126586914
1978 206.604690551758
1979 206.902069091797
1980 203.135589599609
1981 197.178126525879
1982 196.762719726562
1983 204.711399841309
1984 212.392227172851
1985 217.355941772461
1986 221.559341430664
1987 231.864410400391
1988 234.830505371094
1989 233.340881347656
1990 235.389831542969
1991 235.31071472168
1992 231.620872497559
1993 234.717529296875
1994 244.740420532226
1995 254.009170532227
1996 247.005661010742
1997 250
1998 258.799194335937
1999 269.453887939453
2000 276.836151123047
2001 280.924597167968
2002 281.768371582031
2003 294.715637207031
2004 294.849304199219
2005 293.215164184569
2006 300.397358703613
2007 295.555572509766
2008 299.392349243164
2009 304.644073486328
2010 307.288146972656
2011 312.218475341797
2012 320.960479736328
2013 320.451887512206
2014 327.136044311522
2015 333.714288330078
2016 337.720391845704
2017 352.881362915039
2018 367.532958984375
2019 368.146032714844
2020 392.919006347656
2021 394.161926269531
2022 400.734454345703
2023 412.551818847656
2024 391.450103759766
2025 400.000009155273
};
\addlegendentry{p95}
\end{axis}

\end{tikzpicture}
\end{figure}

\begin{figure}[h!]
\caption{Inequality measures of households working FTFY at various
quantiles. }
\label{fig:income1}\centering
% This file was created with tikzplotlib v0.10.1.
\begin{tikzpicture}

\definecolor{darkgray176}{RGB}{176,176,176}
\definecolor{lightgray204}{RGB}{204,204,204}

\begin{axis}[
legend cell align={left},
legend style={
  fill opacity=0.8,
  draw opacity=1,
  text opacity=1,
  at={(1.05,1)},
  anchor=north west,
  draw=lightgray204
},
tick align=outside,
tick pos=left,
x grid style={darkgray176},
xlabel={Years},
xmin=1976, xmax=2025,
xtick style={color=black},
y grid style={darkgray176},
ylabel={Inequality},
ymin=1.58088546365701, ymax=5.15444725773547,
ytick style={color=black}
]
\addplot [very thick, black]
table {%
1976 1.74332009066058
1977 1.75082173254785
1978 1.80952377476982
1979 1.82442913616162
1980 1.85185185185185
1981 1.86346151125336
1982 1.87012987012987
1983 1.92000003254737
1984 1.94711549070256
1985 1.9444444125026
1986 1.9474559337987
1987 1.94285715190818
1988 1.94444439569098
1989 1.95652184165728
1990 1.9984430580739
1991 2.04040807487242
1992 2
1993 2.01355541713891
1994 2.08130098401426
1995 2.10238621421785
1996 2.11778547123021
1997 2.07476718348369
1998 2.08270568040612
1999 2.08333326176111
2000 2.08333336146673
2001 2.12165741773365
2002 2.12532708352963
2003 2.14285716607119
2004 2.11538462801057
2005 2.07879598594872
2006 2.12797504678378
2007 2.21973341204427
2008 2.125
2009 2.18181813683071
2010 2.17454542885222
2011 2.23030297864987
2012 2.18023248414434
2013 2.1954408725162
2014 2.24415592636894
2015 2.28571438342432
2016 2.28571428571429
2017 2.27656303289961
2018 2.33766243249855
2019 2.22222211630828
2020 2.34375002427152
2021 2.25
2022 2.22222217147537
2023 2.27272722364329
2024 2.24235781086575
2025 2.27272726589655
};
\addlegendentry{p75/p25}
\addplot [very thick, black, dashed]
table {%
1976 3.12624055762576
1977 3.17337317885583
1978 3.19819403011021
1979 3.26370761233631
1980 3.32110781330379
1981 3.29576093780886
1982 3.46666654407692
1983 3.60576928719149
1984 3.63206546768562
1985 3.56125363192209
1986 3.58736050454015
1987 3.71773341102296
1988 3.83333342308404
1989 3.74120778278042
1990 3.80784102467559
1991 3.88888871726907
1992 3.80952372749957
1993 3.88235286011522
1994 4.02777777425331
1995 4.10526312414831
1996 4.10740622570424
1997 4.16000008463542
1998 4.16666667723702
1999 4.35478121608975
2000 4.32900421175228
2001 4.29327557303753
2002 4.18313391585642
2003 4.26104189517831
2004 4.35000013966597
2005 4.48974956259693
2006 4.48717944455698
2007 4.65608441961545
2008 4.44377624656718
2009 4.34782633391989
2010 4.47541009684656
2011 4.565217336574
2012 4.54545420085417
2013 4.58333345423717
2014 4.66666662574537
2015 4.97571872002079
2016 4.8
2017 4.9920126307319
2018 4.96571532025959
2019 4.7310712688089
2020 4.93333339077379
2021 4.84615359413887
2022 4.70857488823789
2023 4.84837802349429
2024 4.80000035211341
2025 4.63917501739266
};
\addlegendentry{p90/p10}
\end{axis}

\end{tikzpicture}
\end{figure}
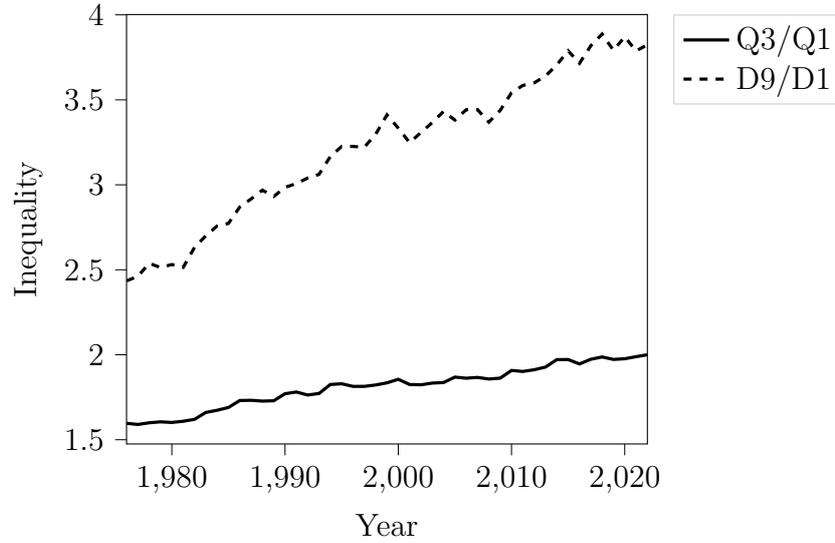

FVVP (2024) show that there is little variation in hours for either males or females working FTFY and the income variation reflects changes in wages.
Figure \ref{fig:quantiles} presents wage growth for married males and
females at D1, Q1, Q2, Q3 and D9. For 1976-1996, the real wage for married
males at D1 decreased by 17\%. It recovers somewhat and, even though there was a sharp increase in 2025, it remains marginally below its 1976 value. There are large decreases during the financial crisis and a large decrease in the 2010s. The real wage at Q1 for married males shows a similar pattern to that at D1. There is a large decrease in the 1976-1996 period but a recovery during the second half of the sample period despite the two large dips. The 2025 value is 2.5\% lower than the 1976 value. The married male median shows a large decrease in the period 1976-1996 but especially due to a sharp increase in the final years of observation it increases by 7\% over the sample period. The overall picture at Q3 is more positive, although there are several sharp decreases. However these are offset by large gains during periods of increasing real wages. The 2025 value is 27\% higher than the 1976 value. At D9, there is a similar pattern to that at Q3 but with smaller dips and faster increases. This results in a dramatic gain of 47\% over the whole sample period.

The time-series pattern for married females' wages at D1 in the earlier part of the sample period is similar to that of married males although the decreases are less dramatic. Although there is a slight increase in the late 1970s, over the period 1979-1996 it decreases by 8\%. For the remaining 29 years, it increases by 32\% and for the whole period it increases by 24\%. The pattern of real wages for married females at Q1 is dissimilar to that at D1. The trend of wage growth appears to be affected by cyclical factors but generally is steadily increasing for the time period examined. This results in an increase of 37\%. Growth at Q2 is even more drastic with an increase of 48\%. Moreover, with the exception of a dip in the early 1980s, the median real wage for this group follows a strong upward trend. At Q3, there are periods of substantial gains and these offset the small decreases which are incurred. An overall gain of 70\% is achieved. A similar story is observed at D9, but the increases amount to a very substantial 101\%.

Given the existence of a marriage premium in mean wages it is interesting to contrast this evidence on married individuals' wage distributions to that of all individuals. Extending the figures in Fern\'{a}ndez-Val et al. (2018), who provide the corresponding rates of wage changes for all males and all females, we find that median male wages declined by 5.9\% for all working males and the decrease at Q1 is 9.8\%. At Q3 there is an increase of only 12\%. This clearly suggests that there is a substantial difference between the experience of married and unmarried males. For all females, wage growth at Q1 and Q2 are 30\% and 36.9\% respectively and there are strong increases at Q3. The data indicate that married females have experienced favorable wage growth compared to unmarried females.

The primary focus of our empirical work below is to model the joint
distribution of married couples' real wage rates. To motivate what follows, we report how the distribution of the sum of the husband's and wife's wages has evolved over our sample period. This is reported in Figure \ref{fig:wages_both}. We define this as the household hourly wage as it captures the combined market value of an hour of work of each spouse. The household hourly wage at D1 decreased by 12\% over the period 1976-1996 but increases over the remainder of the sample period. By 2025 it is 15.6\% higher than in 1976. The household wage at Q1 shows some of the dramatic dips featured in the profile for married males but over the sample period there is an increase of 20.5\%. At Q2 it initially resembles the pattern of the male wage with multiple large declines during our sample. However, towards the end of the sample the large increases in wives' wages result in an increase of about 33\% over the sample period. At Q3 the large shifts in the wives' wage distribution become more important. The household wage shows a dip in the 1980s and a prolonged decrease in the 1990s but over the whole sample period it increases by 51\%. The trend at D9 is consistent with the male pattern at this decile combined with those of females at any of the higher quantiles. This produces a steadily rising profile and an increase of almost 80\%.

\begin{figure}[tbp]
\input{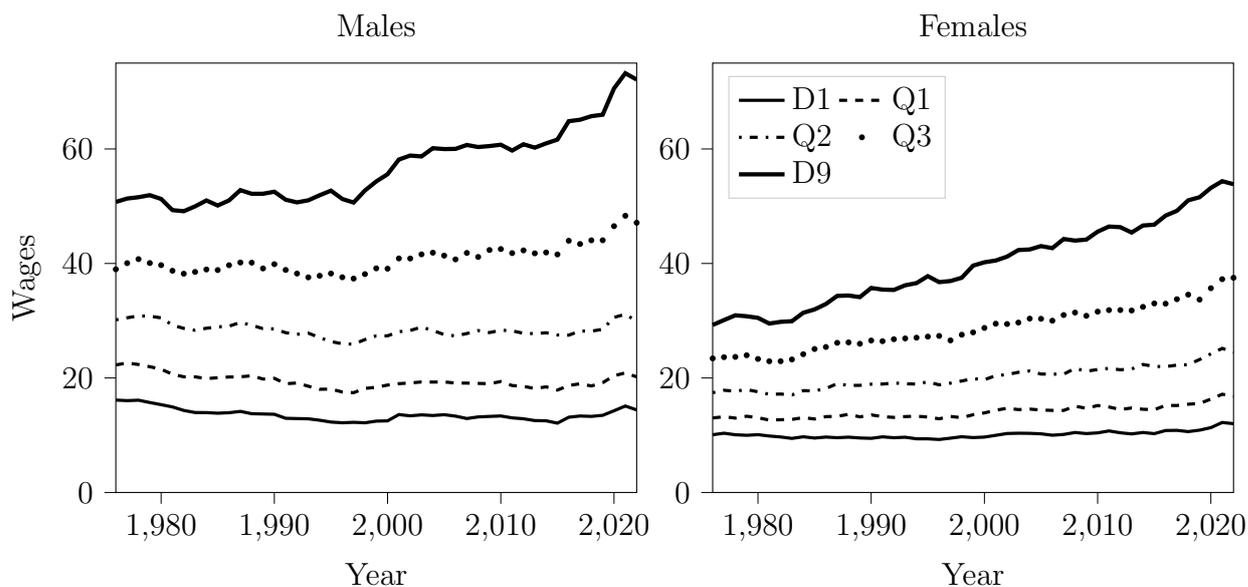}
\caption{Time series of wages of married males and females at different
quantiles.}
\label{fig:quantiles}
\end{figure}

\begin{figure}[tbp]
\centering
% This file was created with tikzplotlib v0.10.1.
\begin{tikzpicture}

\definecolor{darkgray176}{RGB}{176,176,176}
\definecolor{lightgray204}{RGB}{204,204,204}

\begin{axis}[
legend cell align={left},
legend style={
  fill opacity=0.8,
  draw opacity=1,
  text opacity=1,
  at={(0.03,0.97)},
  anchor=north west,
  draw=lightgray204
},
tick align=outside,
tick pos=left,
x grid style={darkgray176},
xlabel={Year},
xmin=1976, xmax=2025,
xtick style={color=black},
y grid style={darkgray176},
ylabel={Real dollars},
ymin=24.5934955, ymax=151.7186145,
ytick style={color=black}
]
\addplot [semithick, black]
table {%
1976 34.09943
1977 34.46328
1978 34.02898
1979 34.25084
1980 33.6422
1981 33.29907
1982 32.05611
1983 31.27589
1984 31.80954
1985 31.92996
1986 32.15602
1987 32.87869
1988 32.49242
1989 32.50761
1990 32.23333
1991 31.7611
1992 31.32552
1993 31.52494
1994 31.28591
1995 31.16858
1996 30.37191
1997 30.57692
1998 31.27684
1999 31.37684
2000 32.30479
2001 33.39284
2002 33.77267
2003 34.32797
2004 33.66055
2005 33.93905
2006 33.73446
2007 33.05949
2008 34.79574
2009 34.038
2010 34.92633
2011 34.58641
2012 34.75954
2013 33.91786
2014 34.2409
2015 33.86555
2016 34.86054
2017 35.06193
2018 35.29866
2019 35.81686
2020 37.19035
2021 39.09665
2022 38.42079
2023 37.1594
2024 37.94897
2025 39.42308
};
\addlegendentry{D1}
\addplot [semithick, black, dotted]
table {%
1976 42.6978
1977 43.08526
1978 43.11004
1979 43.42204
1980 42.57577
1981 41.74055
1982 41.09497
1983 40.17943
1984 41.06628
1985 41.31674
1986 42.04694
1987 42.5453
1988 42.71792
1989 42.20286
1990 41.87112
1991 41.43345
1992 41.22265
1993 41.37649
1994 40.88036
1995 40.70694
1996 40.24398
1997 40.38462
1998 41.35159
1999 42.41061
2000 42.75517
2001 43.77626
2002 44.36874
2003 44.93314
2004 44.90258
2005 44.71969
2006 44.57829
2007 44.12393
2008 45.13255
2009 44.92793
2010 45.4929
2011 45.65406
2012 45.62147
2013 45.35528
2014 45.78543
2015 45.48839
2016 46.8463
2017 47.12625
2018 46.8792
2019 47.74975
2020 49.8495
2021 52.13015
2022 50.98909
2023 49.45676
2024 49.90653
2025 51.44231
};
\addlegendentry{Q1}
\addplot [semithick, black, dashed]
table {%
1976 53.91617
1977 54.45971
1978 54.77275
1979 55.39408
1980 54.63935
1981 53.09105
1982 52.67889
1983 52.6862
1984 53.88927
1985 54.03033
1986 55.59539
1987 56.42474
1988 56.64842
1989 56.09155
1990 55.85397
1991 54.96798
1992 54.73695
1993 55.16865
1994 54.55493
1995 54.95436
1996 54.42833
1997 54.03846
1998 55.82763
1999 56.90733
2000 57.69864
2001 58.66018
2002 59.40104
2003 60.35397
2004 60.04327
2005 59.89244
2006 59.60321
2007 59.82906
2008 61.06644
2009 60.50848
2010 61.90787
2011 62.25554
2012 62.25243
2013 61.98024
2014 62.48042
2015 62.67806
2016 63.64986
2017 64.38841
2018 65.26148
2019 65.69759
2020 69.18379
2021 71.71881
2022 70.56756
2023 69.03339
2024 69.24966
2025 71.88871
};
\addlegendentry{Q2}
\addplot [semithick, black, dash pattern=]
table {%
1976 67.15382
1977 68.04675
1978 68.76404
1979 69.39917
1980 68.08749
1981 66.5338
1982 66.38419
1983 66.67632
1984 68.46625
1985 69.3141
1986 71.23489
1987 72.66933
1988 73.05247
1989 71.38925
1990 73.01173
1991 72.4117
1992 71.49805
1993 72.00584
1994 73.07511
1995 72.9717
1996 72.3512
1997 71.53846
1998 73.93163
1999 75.75305
2000 76.95929
2001 78.56379
2002 79.25486
2003 80.04425
2004 80.84649
2005 80.65515
2006 80.32015
2007 81.16987
2008 82.00536
2009 82.74214
2010 85.12332
2011 85.10127
2012 85.2048
2013 85.45198
2014 86.42662
2015 87.7068
2016 89.10981
2017 90.4824
2018 92.35116
2019 92.55384
2020 97.87556
2021 100.9894
2022 100.2282
2023 97.92515
2024 98.74298
2025 101.4957
};
\addlegendentry{Q3}
\addplot [semithick, black, mark=*, mark size=1, mark options={solid}]
table {%
1976 81.317
1977 82.6234
1978 83.62885
1979 85.4638
1980 83.20538
1981 80.735
1982 80.70419
1983 82.82449
1984 84.82544
1985 86.38506
1986 88.69395
1987 91.82095
1988 91.7684
1989 90.83008
1990 92.48152
1991 92.45532
1992 91.1684
1993 91.35068
1994 93.95372
1995 94.30441
1996 93.02296
1997 93.55769
1998 96.8003
1999 99.19456
2000 102.3106
2001 103.312
2002 104.9834
2003 107.3157
2004 107.2656
2005 108.6049
2006 108.5646
2007 109.188
2008 109.8387
2009 111.1441
2010 114.8373
2011 114.8269
2012 115.5552
2013 116.0129
2014 117.1728
2015 119.1946
2016 122.2077
2017 125.67
2018 129.9073
2019 129.661
2020 138.7259
2021 142.2833
2022 141.9899
2023 144.2489
2024 140.8023
2025 145.9402
};
\addlegendentry{D9}
\end{axis}

\end{tikzpicture}
\caption{Time series of the hourly household wage at different quantiles.}
\label{fig:wages_both}
\end{figure}
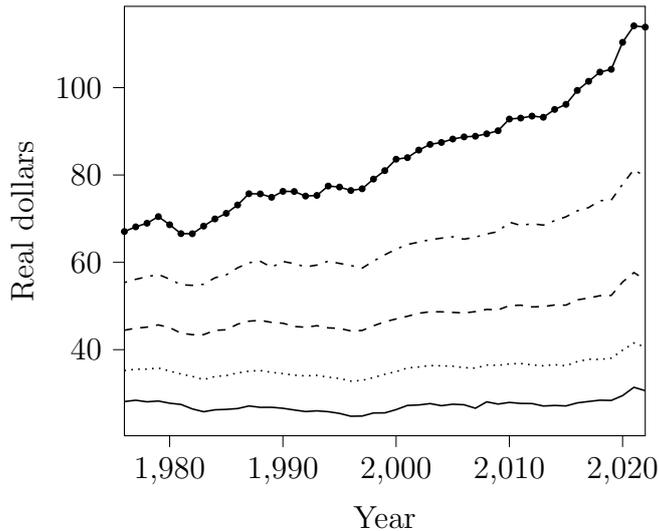
Consider the implications of these trends for inequality. For males, females and households the Q3/Q1 ratios increased from 1.80, 1.93 and 1.74 in 1976 to 2.31, 2.36 and 2.27 in 2025. The D9/D1 ratios increase from 3.33, 3.57 and 3.13 in 1976 to 4.96, 5.18 and 4.64 in 2025. While we cannot directly infer anything definitive about sorting via the relative magnitudes of these measures we report them to reflect the extent of inequality. However, the rate of growth over the time period is similar for each, which suggests positive assortative matching. Fern\'{a}ndez-Val et al. (2018) find that for the period 1976-2016 the D9/D1 ratio for all males increases from 3.6 to 5.3 and that of females increases from 3.7 to 5.1. This suggests that there is less wage inequality for married individuals. This appears to primarily reflect the relatively higher wages of married individuals at low quantiles.

\section{Marital Sorting}

\label{sec:marital_sorting}

We capture marital sorting patterns by reporting the propensity of the male in the $j^{th}$ decile of the male wage distribution to be married to a female in the $k^{th}$ decile of the female wage distribution. We consider $j $,$k=1\ldots 10$ producing 100 cells. To overcome small sample sizes we aggregate the data into five year intervals. To contrast the changes in cell sizes over our entire sample period, we compare the beginning and end periods corresponding to 1976-80 and 2021-25. Tables \ref{tab:frequencies_1} and \ref{tab:frequencies_2} report the contingency tables for these two periods. As random sorting is consistent with each cell containing 0.01 of the data, we divide the sample frequencies by this number. Deviations from 1 are suggestive of sorting behavior. Perfect positive assortative matching implies all couples should appear on the diagonal going from the top left to the bottom right with each element having a value of 10. Perfect negative assortative matching implies the elements on the diagonal going from the bottom left to the top right have value 10. Note that dj denotes that an observation is located between the D(j-1) to D(j) decile.

We acknowledge that a number of factors may generate departures from random sorting. For example, individuals are likely to sort on observed factors such as age, education, race and the individual's region of residence. As these factors are determinants of wages this may spuriously support positive sorting on wages. There may also be wage variation reflecting regional specific cost of living influences. While cost of living components of wages may also partially capture unobserved factors, we do not take a position on whether this necessarily reflects sorting.

For the 1976-80 period there is evidence of non-random sorting. The most striking features of Table \ref{tab:frequencies_1} are the cells associated with the extreme values of the joint wage distribution. The two largest observed frequencies correspond to both spouses in the bottom decile (2.19) and both in the top decile (2.38). The next two highest values correspond to the cells immediately adjacent these extremes (1.72 and 1.60). This is
consistent with positive sorting. If we define the bottom as (d1+d2) and the top as (d9+d10) we obtain 6.53\% and 6.58\%. This contrasts with the 4\% implied by random sorting. Another interesting feature of the table is revealed by the off-diagonal values. There is an almost monotonically decreasing relationship between the distance from the diagonal and the probability of marriage. Finally, the sum of the elements on the diagonal is 13.62, which appears to support positive assortative matching. However, as this ignores behavior away from the diagonal and it is not invariant to the definition of the diagonal we computed both the Kendall and Spearman rank correlation coefficients. Their values of 0.21 and 0.28 support positive assortative matching.

Table \ref{tab:frequencies_2} presents the corresponding values for the 2021-2025 period. As this period is associated with a substantial increase in wage inequality, the level of household inequality is likely to have increased even if the cell frequencies remained at the 1976-80 values.
However, it appears that increased marital sorting has exacerbated the individual level inequality. For example, the two extreme cells have increased to 3.03 and 3.24. This represents a particularly large increase in the d1/d1 cell. The frequencies in the bottom 2 and top 2 deciles have increased to 8.27 and 8.71. This reflects a large growth in the frequency in which lower (higher) paid males are married to lower (higher) paid females. The growth in positive sorting at the bottom is concerning given the manner in which wages, relative to 1976, have fallen for males and only slightly increased for females at these locations of their distributions. The other features of Table \ref{tab:frequencies_1} regarding sorting are also generally supported by Table \ref{tab:frequencies_2}. There is stronger evidence that low-wage individuals are married to other low-wage individuals, while high-wage individuals are married to other high-wage individuals. The sum of the diagonal is now 17.40 and the Kendall (Spearman) rank correlation coefficient is 0.32 (0.42). This suggests
increased positive assortative matching relative to the 1976-80 period.

\begin{table}[tbp]
\begin{tabular}{l||l|rrrrrrrrrrr}
\hline\hline
&  & \multicolumn{10}{c}{Husbands' quantiles} &  \\
 \hline\hline
&  & d1 & d2 & d3 & d4 & d5 & d6 & d7 & d8 & d9 & d10 &  \\
 \hline
\parbox[t]{2mm}{\multirow{10}{*}{\rotatebox[origin=c]{90}{Wives' quantiles}}}
	&d1	&	\textbf{2.183}	&	1.309	&	1.119	&	0.96	&	0.806	&	0.879	&	0.796	&	0.692	&	0.677	&	0.581	\\	
	&d2	&	1.723	&	\textbf{1.314}	&	1.235	&	1.114	&	0.942	&	0.862	&	0.803	&	0.75	&	0.748	&	0.508	\\	
	&d3	&	1.299	&	1.4	&	\textbf{1.079}	&	1.096	&	1.114	&	0.874	&	0.872	&	0.892	&	0.735	&	0.639	\\	
	&d4	&	1.137	&	1.415	&	1.182	&	\textbf{1.129}	&	0.917	&	0.925	&	0.889	&	0.915	&	0.806	&	0.685	\\	
	&d5	&	0.91	&	1.139	&	1.218	&	1.086	&	\textbf{1.117}	&	0.993	&	0.892	&	0.963	&	0.927	&	0.755	\\	
	&d6	&	0.74	&	1.066	&	1.071	&	1.107	&	1.066	&	\textbf{1.076}	&	1.054	&	0.995	&	0.945	&	0.879	\\	
	&d7	&	0.624	&	0.791	&	1.107	&	1.069	&	1.069	&	1.107	&	\textbf{1.046}	&	1.091	&	1.124	&	0.973	\\	
	&d8	&	0.574	&	0.574	&	0.841	&	1.127	&	1.071	&	1.099	&	1.167	&	\textbf{1.109}	&	1.251	&	1.187	\\	
	&d9	&	0.437	&	0.606	&	0.71	&	0.763	&	1.175	&	1.2	&	1.238	&	1.271	&	\textbf{1.185}	&	1.415	\\	
	&d10	&	0.376	&	0.387	&	0.437	&	0.548	&	0.723	&	0.985	&	1.243	&	1.321	&	1.602	&	\textbf{2.377}	\\	
\hline\hline
\end{tabular}
\caption{Frequencies of the combination of deciles of married couples for
the period 1976-1980.}
\label{tab:frequencies_1}
\end{table}

\begin{table}[tbp]
\begin{tabular}{l||l|rrrrrrrrrrr}
\hline\hline
&  & \multicolumn{10}{c}{Husbands' quantiles} &  \\
 \hline\hline
&  & d1 & d2 & d3 & d4 & d5 & d6 & d7 & d8 & d9 & d10 &  \\
 \hline
\parbox[t]{2mm}{\multirow{10}{*}{\rotatebox[origin=c]{90}{Wives' quantiles}}}
	&d1	&	\textbf{3.03}	&	1.477	&	1.182	&	0.934	&	0.774	&	0.694	&	0.552	&	0.57	&	0.428	&	0.36	\\	
	&d2	&	2.038	&	\textbf{1.72}	&	1.28	&	1.12	&	0.943	&	0.776	&	0.664	&	0.598	&	0.504	&	0.357	\\	
	&d3	&	1.159	&	1.862	&	\textbf{1.376}	&	1.163	&	1.06	&	0.905	&	0.811	&	0.708	&	0.561	&	0.396	\\	
	&d4	&	0.863	&	1.285	&	1.678	&	\textbf{1.294}	&	1.083	&	0.927	&	0.893	&	0.852	&	0.657	&	0.467	\\	
	&d5	&	0.744	&	0.941	&	1.31	&	1.452	&	\textbf{1.198}	&	1.049	&	1.051	&	0.948	&	0.774	&	0.531	\\	
	&d6	&	0.593	&	0.831	&	0.914	&	1.328	&	1.376	&	\textbf{1.188}	&	1.088	&	1.035	&	0.907	&	0.74	\\	
	&d7	&	0.55	&	0.634	&	0.742	&	0.916	&	1.177	&	1.504	&	\textbf{1.259}	&	1.156	&	1.117	&	0.943	\\	
	&d8	&	0.401	&	0.547	&	0.598	&	0.714	&	1.067	&	1.287	&	1.498	&	\textbf{1.342}	&	1.25	&	1.296	\\	
	&d9	&	0.353	&	0.38	&	0.543	&	0.634	&	0.74	&	1.024	&	1.278	&	1.617	&	\textbf{1.761}	&	1.672	\\	
	&d10	&	0.27	&	0.323	&	0.378	&	0.444	&	0.579	&	0.646	&	0.907	&	1.175	&	2.041	&	\textbf{3.236}	\\	
\hline\hline
\end{tabular}
\caption{Frequencies of the combination of deciles of married couples for
the period 2021-2025.}
\label{tab:frequencies_2}
\end{table}

The differences across the two tables are important given their potential implications for inequality. However, the differences might reflect a variety of factors. As employment rates have changed substantially over the sample period, the composition, in terms of observed and unobserved characteristics, of both husbands and wives may have changed and this may have affected the observed patterns of sorting. It is also possible that the prices of these observed and unobserved characteristics have changed. Finally, the nature of marital sorting, as measured by the joint distribution of the couples observed and unobserved characteristics, may have changed.

Many couples share characteristics which are determinants of wages, such as race, geographical location and age, although this does not necessarily reflect sorting on wages. To examine to what extent these characteristics can explain the observed sorting in Tables \ref{tab:frequencies_1} and \ref{tab:frequencies_2}, Tables \ref{tab:age_wives}-\ref{tab:metro_husbands1} report the average value of some measure of each of these characteristics for the same time periods. The cells represent the same wage cells as in Tables \ref{tab:frequencies_1} and \ref{tab:frequencies_2}.

Tables \ref{tab:age_wives}-\ref{tab:age_husbands1} report the ages of wives and husbands. As age is a determinant of wages and spouses generally are close in age, it is possible that the patterns in Tables \ref{tab:frequencies_1} and \ref{tab:frequencies_2} simply reflect age
differences.\label{insert:more_compelling} For the first period the average age for both husbands and
wives increases by about 2 years as one goes from the d1/d1 to d10/d10. This seems a remarkably small difference given the large wage discrepancies across these cells. There are differences across the cells, but it does not appear that age differences drive the observed sorting patterns. For the later period the average ages of the spouses increase by approximately 2.5 years as one goes from d1/d1 to d10/d10. In percentage terms this is a notable increase compared to the earlier period although it does not appear to be the driving force of increased inequality across the two periods. The age differences across cells are larger than in the earlier period but are unlikely to explain the observed wage differences.

Tables \ref{tab:non_white_wives} to \ref{tab:race_both1} unsurprisingly indicate that race has some association with an individual's location in each of the wage distribution and that the strength of this relationship has changed drastically over time. 
%Moreover, Tables \ref{tab:non_white_wives}-\ref{tab:non_white_husbands1} collectively confirm positive sorting on race.
This highlights the necessity to account for this factor in estimating the sorting models below. We also consider location of residence, as it is likely to influence wage differences by capturing variations in the cost of living. Tables \ref{tab:metro_husbands} and \ref{tab:metro_husbands1} suggest that the higher paid are living in metropolitan areas and that the fraction of individuals living in these areas has increased substantially over the sample period. As married individuals generally cohabit this may also spuriously imply sorting on wages.

Another important characteristic is education although this may reflect productivity and might reflect sorting on wages. 
%Notwithstanding these challenges with interpretation. 
Tables \ref{tab:educ_wives} to \ref{tab:educ_husbands1} report the percentage of wives and husbands with university education in the two time periods. Unsurprisingly, educational levels in the higher wage cells are higher than those in the lower. The large increase in individuals obtaining college education is reflected in the substantially higher averages for the later period.

The purpose of these tables describing the distribution of characteristics across wage cells is to explore whether observed characteristics can explain sorting on wages. It appears that while there are differences across cells, it is unlikely that they can account for the observed wage sorting patterns. In the next section we outline an approach which incorporates the role of observables in determining sorting behavior while allowing the relationship between characteristics and wages to vary across the wage distribution.

\section{Econometric Model}
\label{sec:econometric_model}

We now describe our econometric methodology. Section \ref{ss:determinants_of_sorting} employs the bivariate distribution regression estimator of Chernozhukov et al. (2025)  on the selected FTFY sample. While this ignores selection, it is important to establish sorting patterns, which do so in order that we can isolate the role of selection. Section \ref{sec:model} describes the estimation procedure that incorporates selection into FTFY employment.

\subsection{Determinants of Marital Sorting}
\label{ss:determinants_of_sorting}
We now focus on estimating the determinants of the marital sorting
frequencies in Tables \ref{tab:frequencies_1} and \ref{tab:frequencies_2}.
We do so via bivariate distribution regression (BDR) which employs a Local
Gaussian Representation (LGR) of the joint distribution of the wives' and
husbands' wages. We represent this joint distribution as: 
\begin{equation}
F_{Y_{w},Y_{h}\mid X}(y_{w},y_{h}\mid x)=\Phi _{2}(\mu _{w}(y_{w},x),\mu
_{h}(y_{h},x);\rho (y_{w},y_{h},x)),\quad (y_{w},y_{h})\in \mathbb{R}^{2},
\label{eq:bivariate}
\end{equation}%
where $Y_{w}$ and $Y_{h}$ are the observed wages of wives and husbands, $X$
is a set of observable characteristics, $\Phi _{2}(\cdot ,\cdot ;\rho )$ is
the standard bivariate normal CDF with correlation $\rho ,$ and $\mu
_{j}(y,x)$ is formally defined below. In the LGR, the marginal conditional
CDFs of $Y_{w}$ and $Y_{h}$ are represented by: 
\begin{equation*}
F_{Y_{j}\mid X}(y\mid x)=\Phi (\mu _{j}(y,x)),\quad y\in \mathbb{R},\quad
j\in \{w,h\},
\end{equation*}%
where $\Phi $ is the standard univariate normal CDF. The parameter $\rho
(y_{w},y_{h},x)$ is the local correlation between the unobservables
influencing the spouses' wages at $(y_{w},y_{h})$ and captures sorting on
unobservables. The unconditional joint distribution of the wives' and
husbands' wages can be obtained from the LGR as: 
\begin{equation*}
F_{Y_{w},Y_{h}}(y_{w},y_{h})=\int \Phi _{2}(\mu _{w}(y_{w},x),\mu
_{h}(y_{h},x);\rho (y_{w},y_{h},x))dF_{X}(x),\quad (y_{w},y_{h})\in \mathbb{R%
}^{2},
\end{equation*}%
where $F_{X}$ is the CDF of $X$, and the corresponding marginals are: 
\begin{equation*}
F_{Y_{j}}(y)=\int \Phi (\mu _{j}(y,x))dF_{X}(x),\quad y\in \mathbb{R},\quad
j\in \{w,h\}.
\end{equation*}%
CFL (2026) show that the LGR is fully non-parametric as it does not impose any
restrictions on the conditional joint distribution.

The BDR model augments the LGR with two assumptions (see also Chernozhukov et al. 2025):

\begin{assumption}[BDR]
\label{ass:bdr} (1) $\mu_j(y,x) = P_{j}(x)^{\prime }\beta _{j}(y),$ $j \in
\{w,h\},$ where $P_{w}$ and $P_{h}$ are transformations of $x$ and $\beta
_{w}(y)$ and $\beta _{h}(y)$ are vectors of coefficients; and (2) $\rho
(y_{w},y_{h},x) = \rho(y_{w},y_{h})$.
\end{assumption}

%First, we specify $\mu $ as: 
%\begin{equation*}
%\mu (y_{w},y_{h},x)=\left( 
%\begin{array}{c}
%P_{w}(x)\beta _{w}(y) \\ 
%P_{h}(x)\beta _{h}(w)%
%\end{array}%
%\right)
%\end{equation*}%
We allow different specifications for $P_{w}$ and $P_{h}.$ For example, $%
P_{w}$ includes the wife's education and age but $P_{h}$ does not. While
we assume that the marital sorting parameter does not depend on observed
characteristics, we allow it to vary by location in the joint distribution
of wages. This model is semiparametric as the parameters $y\mapsto \beta
_{j}(y),$ $j\in \{w,h\},$ and $(y_{w},y_{h})\mapsto \rho (y_{w},y_{h})$ are
function-valued.

BDR describes the joint distribution of wages conditional on the two indices
capturing the observed determinants of husbands and wives' wages. Given a
random sample of $(Y_{w},Y_{h},X)$, estimation is performed via a series of
bivariate probits of the indicators $\mathbf{1}(Y_{w}\leq y_{w})$ and $%
\mathbf{1}(Y_{h}\leq y_{h})$ on $X$ for multiple values of $y_{h}$ and $%
y_{w} $. This is done in two steps. First, estimate $\beta _{j}(y_{j})$ via
univariate probit of $\mathbf{1}(Y_{j}\leq y_{j})$ on $X$, for $j\in \{w,h\}$%
. Second, estimate $\rho (y_{w},y_{h})$ via bivariate probit plugging in the
estimates of $\beta _{w}(y_{w})$ and $\beta _{h}(y_{h})$.

%    and is
%    semi-parametric as it allows a different relationship between the $y^{\prime
%    }$s and the $X^{\prime }$s at each quantile of the wage distributions. We
%    estimate the parameters in these two indices and the parameter $\rho $.

We begin by examining our capacity to explain the variation in Tables \ref%
{tab:frequencies_1} and \ref{tab:frequencies_2} via 100 bivariate probits.
The entries correspond to estimates of: 
\begin{multline}
s(\underline{y}_{w},\overline{y}_{w},\underline{y}_{h},\overline{y}_{h}):=%
\frac{\Pr (\underline{y}_{w}<Y_{w}\leq \overline{y}_{w},\underline{y}%
_{h}<Y_{h}\leq \overline{y}_{h})}{\Pr (\underline{y}_{w}<Y_{w}\leq \overline{%
y}_{w})\Pr (\underline{y}_{h}<Y_{h}\leq \overline{y}_{h})}
\label{eq:transition} \\
=\frac{F_{Y_{w},Y_{h}}(\overline{y}_{w},\overline{y}_{h})-F_{Y_{w},Y_{h}}(%
\overline{y}_{w},\underline{y}_{h})-F_{Y_{w},Y_{h}}(\underline{y}_{w},%
\overline{y}_{h})+F_{Y_{w},Y_{h}}(\underline{y}_{w},\underline{y}_{h})}{%
[F_{Y_{w}}(\overline{y}_{w})-F_{Y_{w}}(\underline{y}_{w})][F_{Y_{h}}(%
\overline{y}_{h})-F_{Y_{h}}(\underline{y}_{h})]},
\end{multline}%
where $\underline{y}_{w}$ and $\overline{y}_{w}$ are evaluated at the sample
deciles of $Y_{w}$, and $\underline{y}_{h}$ and $\overline{y}_{h}$ at the
sample deciles of $Y_{h}$. The predicted probabilities are reported in
Tables \ref{tab:bivariate_1} for 1976-1980 and Table \ref{tab:bivariate_2}
for 2021-2025. A comparison of the predicted and empirical probabilities for each of the sample periods indicates that the estimated model reproduces the empirical probabilities across the 100 cells despite the restrictions imposed by Assumption \ref{ass:bdr}.\footnote{%
The predicted probabilities are computed by first using univariate
distribution regression models to obtain the quantiles of the marginal
distribution of wives and husbands. That is, we estimate the quantiles for
wives, $Q_{\tau }^{w};\tau \in \lbrack 0,1]$ solving the empirical analog of  $\int \Phi (x\beta (Q_{\tau }^{w}))dF_{X_{w}}(x)=\tau $, where $\beta (\cdot
)$ is estimated by univariate distribution regression. The quantiles for
husbands are estimated similarly. Based on these quantiles, we can estimate
the bivariate distribution regression model as in (\ref{eq:bivariate}) and
estimate the empirical analog of $\int \Phi _{2}(x_{w}\beta (Q_{\tau
}^{w}),x_{h}\beta (Q_{\tau }^{h}),\rho )dF_{X_{w},X_{h}}(x_{w},x_{h})$.} 
%This is not surprising given the number of
%parameters being estimated and the nature of the estimation process.
%Moreover, some of the disparities result from the treatment of predicted
%values which are bunched. 

Figure \ref{fig:kernel} presents the estimated densities of the 100
estimates of $\rho (y_{w},y_{h})$ for each sample period. Our vector of covariates, $X$, includes 3 dummy variables for education (high school, some college, college degree or higher), age, age$^{2}$, age interacted with 3 education dummy variables, $age^{2}$ interacted with 3 education dummy variables, a dummy variable for non-white, a dummy variable for Hispanic, 2 dummy variables for metropolitan area (central city, outside central city), and 8 regional dummy variables (middle Atlantic, east north central, west north central, south Atlantic, east south central, west south central, mountain, and Pacific, with reference category  New England). The average values for the two periods are $0.18$ and $0.24$ and the figure reveals a clear shift in the distribution to the right over time. Thus, despite the rich nature of $X$, there remains a strong and positive correlation between the unobservables driving wages for spouses and it has increased over time. This supports the presence of positive and increasing assortative matching on wages. This partially reflects common influences on wages, such as local adjustments related to cost of living, but it may also capture unobserved ability or other unobserved determinants of productivity.

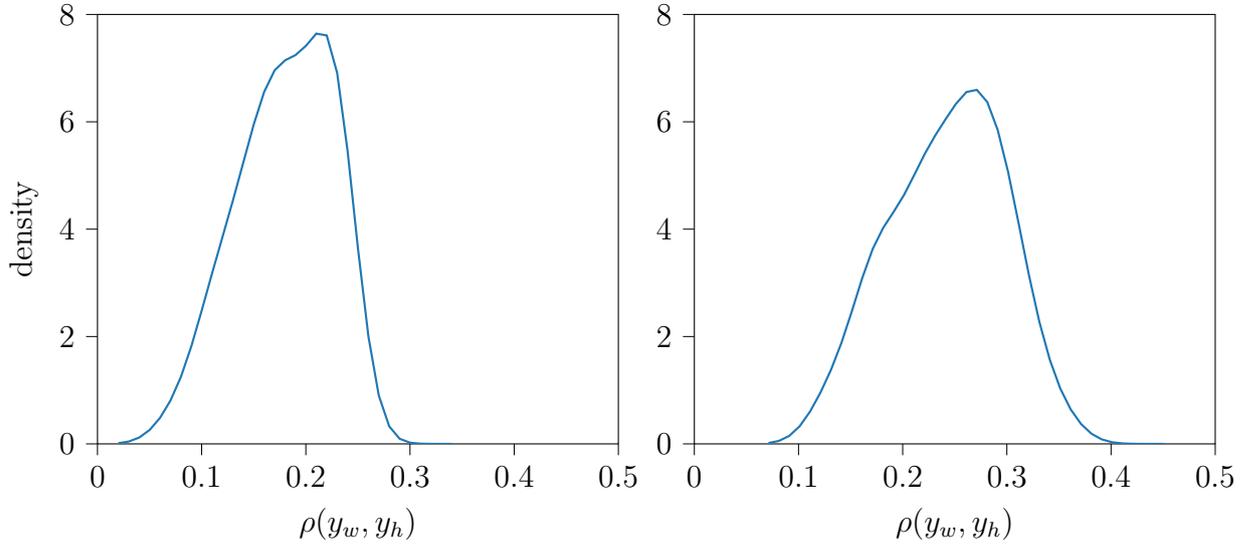
\begin{figure}[tbp]
\caption{Estimated densities of $\protect\rho (y_{w},y_{h})$}
\label{fig:kernel}\centering
% This file was created with tikzplotlib v0.10.1.
\begin{tikzpicture}

\definecolor{darkgray176}{RGB}{176,176,176}
\definecolor{steelblue31119180}{RGB}{31,119,180}

\begin{groupplot}[group style={group size=2 by 1}]
\nextgroupplot[
tick align=outside,
tick pos=left,
title={1976-1980},
x grid style={darkgray176},
xlabel={\(\displaystyle \rho (y_{w},y_{h})\)},
xmin=0, xmax=0.5,
xtick style={color=black},
y grid style={darkgray176},
ylabel={density},
ymin=0, ymax=7,
ytick style={color=black}
]
\addplot [thick, steelblue31119180]
table {%
-0.0385513 0.0251762958984574
-0.0285513 0.0883806749916158
-0.0185513 0.247418620201395
-0.0085513 0.555753299872496
0.0014487 1.01002502205032
0.0114487 1.50239168589605
0.0214487 1.85958406457785
0.0314487 1.96571020429625
0.0414487 1.85631738389304
0.0514487 1.68862200930975
0.0614487 1.62328868024806
0.0714487 1.73565412664344
0.0814487 2.01589308268785
0.0914487 2.40887595248292
0.1014487 2.8480859759638
0.1114487 3.2994234380556
0.1214487 3.78680507071084
0.1314487 4.34052514212902
0.1414487 4.92124952664441
0.1514487 5.42295779169968
0.1614487 5.74827181469679
0.1714487 5.87746248668228
0.1814487 5.8843442719842
0.1914487 5.88907252309774
0.2014487 5.96185946184184
0.2114487 6.02788189798816
0.2214487 5.85941194474303
0.2314487 5.21946959463608
0.2414487 4.08249024023645
0.2514487 2.7153531469848
0.2614487 1.50171669904887
0.2714487 0.679929037916836
0.2814487 0.249279720285898
0.2914487 0.0734153812058369
0.3014487 0.0172647286607276
0.3114487 0.00322703558083591
0.3214487 0.000477698717331689
0.3314487 5.58434531820596e-05
0.3414487 5.1436353619946e-06
0.3514487 3.72609997542232e-07
0.3614487 2.11975192064991e-08
0.3714487 9.45893914214134e-10
0.3814487 3.30752462023475e-11
0.3914487 9.05564575576549e-13
0.4014487 1.94001187190984e-14
0.4114487 3.25026834246485e-16
0.4214487 4.25661000335237e-18
0.4314487 4.35581472388786e-20
0.4414487 3.48170588038193e-22
0.4514487 2.17324464915574e-24
0.4614487 1.05904021278209e-26
0.4714487 4.02819771072939e-29
0.4814487 1.19569992089425e-31
0.4914487 2.76933070904558e-34
0.5014487 5.00385511499958e-37
0.5114487 7.05268925285306e-40
0.5214487 7.75306584851375e-43
0.5314487 6.64682758061068e-46
0.5414487 4.44360400926794e-49
0.5514487 2.31631581105996e-52
0.5614487 9.41383147771735e-56
0.5714487 2.98269539778744e-59
0.5814487 7.36707998231877e-63
0.5914487 1.41839861620705e-66
0.6014487 2.12859055340404e-70
0.6114487 2.48973551440658e-74
0.6214487 2.59049539481093e-78
0.6314487 3.64815073177619e-75
0.6414487 3.23058254092459e-71
0.6514487 2.22800062419454e-67
0.6614487 1.1966748236721e-63
0.6714487 5.00568342184133e-60
0.6814487 1.63071094008424e-56
0.6914487 4.13729958271673e-53
0.7014487 8.17491665255541e-50
0.7114487 1.2579867798366e-46
0.7214487 1.50763149101155e-43
0.7314487 1.40715099784829e-40
0.7414487 1.02285149719785e-37
0.7514487 5.79043041466116e-35
0.7614487 2.55290990410878e-32
0.7714487 8.76569809931802e-30
0.7814487 2.34403401343884e-27
0.7914487 4.88166180873979e-25
0.8014487 7.91767823328712e-23
0.8114487 1.0001252393457e-20
0.8214487 9.83869065318669e-19
0.8314487 7.53783460058256e-17
0.8414487 4.49761504613205e-15
0.8514487 2.08999062546985e-13
0.8614487 7.56367352026298e-12
0.8714487 2.13180556036094e-10
0.8814487 4.67938508203566e-09
0.8914487 7.9993806151423e-08
0.9014487 1.06500173640491e-06
0.9114487 1.10425824336597e-05
0.9214487 8.9169711931467e-05
0.9314487 0.000560777340698301
0.9414487 0.00274656530376996
0.9514487 0.010476489491636
0.9614487 0.0311220369109543
0.9714487 0.0720023423125061
0.9814487 0.129733344874579
0.9914487 0.182046773371406
1.0014487 0.198948530507424
1.0114487 0.169326474065888
1.0214487 0.112236825315262
1.0314487 0.0579391690457626
1.0414487 0.0232935419029424
1.0514487 0.00729331846748794
1.0614487 0.00177844813649241
1.0714487 0.000337740864913766
1.0814487 4.9951944316244e-05
1.0914487 5.75370462287171e-06
};

\nextgroupplot[
tick align=outside,
tick pos=left,
title={2021-2025},
x grid style={darkgray176},
xlabel={\(\displaystyle \rho (y_{w},y_{h})\)},
xmin=0, xmax=0.5,
xtick style={color=black},
y grid style={darkgray176},
ymin=0, ymax=7,
ytick style={color=black}
]
\addplot [thick, steelblue31119180]
table {%
-0.04890794 0.0116235151945003
-0.03890794 0.0396609274034304
-0.02890794 0.109415066077699
-0.01890794 0.247700194188137
-0.00890794 0.470263508994191
0.00109206000000001 0.771590505404504
0.01109206 1.13311900164203
0.02109206 1.53116890909517
0.03109206 1.91561150512756
0.04109206 2.18741191562718
0.05109206 2.24173235076741
0.06109206 2.05957114808962
0.07109206 1.73652973660033
0.08109206 1.40929664102095
0.09109206 1.17173719154314
0.10109206 1.05180352977125
0.11109206 1.02888668373048
0.12109206 1.06732810562469
0.13109206 1.1546211771523
0.14109206 1.3062946229395
0.15109206 1.5264833554291
0.16109206 1.78075834247709
0.17109206 2.02394675123102
0.18109206 2.24703235115814
0.19109206 2.48180703713974
0.20109206 2.75488664109469
0.21109206 3.05266062946639
0.22109206 3.3527782066421
0.23109206 3.66832963505604
0.24109206 4.01686591828584
0.25109206 4.36426664996231
0.26109206 4.65063984245651
0.27109206 4.86853581705275
0.28109206 5.07361846127694
0.29109206 5.28482487556392
0.30109206 5.38008562769358
0.31109206 5.14604608945785
0.32109206 4.4721665191902
0.33109206 3.48323686461624
0.34109206 2.46166659052947
0.35109206 1.64236351640769
0.36109206 1.08331865458336
0.37109206 0.709145189103915
0.38109206 0.435871094967548
0.39109206 0.234203357246275
0.40109206 0.104489392092272
0.41109206 0.0376064570848006
0.42109206 0.0107574621763297
0.43109206 0.00242702763700175
0.44109206 0.00043007602185368
0.45109206 5.97138936442118e-05
0.46109206 6.4866900620555e-06
0.47109206 5.50772844577993e-07
0.48109206 3.65293003449918e-08
0.49109206 1.89160508197386e-09
0.50109206 7.64529029358502e-11
0.51109206 2.41113271131402e-12
0.52109206 5.93232250769704e-14
0.53109206 1.13850247068636e-15
0.54109206 1.70407694305076e-17
0.55109206 1.98901905496385e-19
0.56109206 1.81024645363142e-21
0.57109206 1.28453086062326e-23
0.58109206 7.10594973651569e-26
0.59109206 3.06433425744759e-28
0.60109206 1.03004165358262e-30
0.61109206 2.69866256374584e-33
0.62109206 5.51048836209466e-36
0.63109206 8.76907468627459e-39
0.64109206 1.08746645595144e-41
0.65109206 1.05088286983579e-44
0.66109206 7.91313547510596e-48
0.67109206 4.64279737937683e-51
0.68109206 2.13475595674137e-54
0.69109206 3.15885534603471e-53
0.70109206 6.29624875182431e-50
0.71109206 9.77396522420061e-47
0.72109206 1.18164243986385e-43
0.73109206 1.11257104077157e-40
0.74109206 8.15822730959696e-38
0.75109206 4.65897340691351e-35
0.76109206 2.0721015751272e-32
0.77109206 7.17725257387142e-30
0.78109206 1.93611789072551e-27
0.79109206 4.06753919002685e-25
0.80109206 6.65515344708196e-23
0.81109206 8.48029148976368e-21
0.82109206 8.41569266189714e-19
0.83109206 6.50422165218919e-17
0.84109206 3.91495864058795e-15
0.85109206 1.83520846333896e-13
0.86109206 6.69992603014747e-12
0.87109206 1.90493864378173e-10
0.88109206 4.21811441488929e-09
0.89109206 7.27414742742352e-08
0.90109206 9.76949739561594e-07
0.91109206 1.02185370434452e-05
0.92109206 8.32399106477148e-05
0.93109206 0.000528081394510447
0.94109206 0.00260913473995777
0.95109206 0.0100396489560258
0.96109206 0.0300861715459693
0.97109206 0.0702169088693733
0.98109206 0.12762708957378
0.99109206 0.180663493526895
1.00109206 0.19917018278115
1.01109206 0.171003350369192
1.02109206 0.114343451207997
1.03109206 0.0595448701856729
1.04109206 0.0241492578581016
1.05109206 0.00762762975635277
1.06109206 0.00187629792969509
1.07109206 0.000359451560974423
1.08109206 5.36296951196969e-05
1.09109206 6.23155828854498e-06
};
\end{groupplot}

\end{tikzpicture}
\end{figure}

While the estimate of $\rho (y_{w},y_{h})$ is consistent with positive
sorting, it is not immediately clear from the coefficient how economically
important this parameter is in determining the observed sorting patterns. We employ the BDR estimates to construct counterfactual contingency tables from setting $\rho (y_{w},y_{h})=0$ for all $y_{w},y_{h}$. As the quantiles of the marginal distributions are not affected by setting $%
\rho (y_{w},y_{h})$ to zero, we can simply estimate the corresponding
distributions in \eqref{eq:bivariate} by setting $\rho (y_{w},y_{h})$ equal to zero. This corresponds to no correlation or sorting on unobservables while retaining the sorting on the observed characteristics. We implement this counterfactual in the two periods and report the estimates in Tables \ref{tab:biv2_1976} and \ref{tab:biv2_2018}. The predicted probabilities indicate that sorting on unobservables accounts for a large fraction of the observed positive marital sorting. This is remarkable given the positive sorting by race, age, educational attainments and location of residence highlighted above. 

\begin{table}[tbp]
\begin{tabular}{l||l|rrrrrrrrrr}
\hline\hline
&  & \multicolumn{10}{c}{Husbands' quantiles}   \\
 \hline\hline
&  & d1 & d2 & d3 & d4 & d5 & d6 & d7 & d8 & d9 & d10   \\
 \hline
\parbox[t]{2mm}{\multirow{10}{*}{\rotatebox[origin=c]{90}{Wives' quantiles}}}
	&d1	&	\textbf{1.31}	&	1.234	&	1.081	&	1.066	&	0.996	&	0.936	&	0.919	&	0.878	&	0.827	&	0.703	\\	
	&d2	&	1.258	&	\textbf{1.218}	&	1.077	&	1.078	&	1.01	&	0.954	&	0.942	&	0.898	&	0.845	&	0.716	\\	
	&d3	&	1.159	&	1.159	&	\textbf{1.047}	&	1.064	&	1.01	&	0.959	&	0.952	&	0.919	&	0.873	&	0.763	\\	
	&d4	&	1.077	&	1.122	&	1.033	&	\textbf{1.067}	&	1.025	&	0.979	&	0.973	&	0.955	&	0.914	&	0.815	\\	
	&d5	&	1.055	&	1.118	&	1.056	&	1.114	&	\textbf{1.092}	&	1.055	&	1.052	&	1.048	&	1.012	&	0.938	\\	
	&d6	&	0.944	&	1.008	&	0.963	&	1.026	&	1.015	&	\textbf{0.989}	&	0.996	&	1.003	&	0.976	&	0.936	\\	
	&d7	&	0.88	&	0.942	&	0.926	&	0.995	&	1.004	&	0.992	&	\textbf{1.01}	&	1.021	&	1.017	&	1.001	\\	
	&d8	&	0.818	&	0.895	&	0.889	&	0.984	&	1.006	&	1.016	&	1.033	&	\textbf{1.077}	&	1.085	&	1.137	\\	
	&d9	&	0.773	&	0.846	&	0.858	&	0.954	&	1.001	&	1.039	&	1.072	&	1.128	&	\textbf{1.158}	&	1.275	\\	
	&d10	&	0.658	&	0.716	&	0.745	&	0.845	&	0.912	&	0.978	&	1.057	&	1.153	&	1.268	&	\textbf{1.632}	\\	
\hline\hline
\end{tabular}
\caption{Results of estimated sorting measures in 1976-1980 at different
quantiles when $\protect\rho(y_w, y_h) = 0$.}
\label{tab:biv2_1976}
\end{table}

\begin{table}[tbp]
\begin{tabular}{l||l|rrrrrrrrrrr}
\hline\hline
&  & \multicolumn{10}{c}{Husbands' quantiles} &  \\
 \hline\hline
&  & d1 & d2 & d3 & d4 & d5 & d6 & d7 & d8 & d9 & d10 &  \\
 \hline
\parbox[t]{2mm}{\multirow{10}{*}{\rotatebox[origin=c]{90}{Wives' quantiles}}}
	&d1	&	\textbf{1.63}	&	1.382	&	1.181	&	1.142	&	1.027	&	0.935	&	0.825	&	0.71	&	0.622	&	0.546	\\	
	&d2	&	1.436	&	\textbf{1.303}	&	1.174	&	1.144	&	1.048	&	0.976	&	0.867	&	0.767	&	0.682	&	0.603	\\	
	&d3	&	1.272	&	1.215	&	\textbf{1.141}	&	1.129	&	1.06	&	1.009	&	0.917	&	0.829	&	0.754	&	0.674	\\	
	&d4	&	1.122	&	1.115	&	1.089	&	\textbf{1.092}	&	1.052	&	1.032	&	0.971	&	0.897	&	0.849	&	0.782	\\	
	&d5	&	1.013	&	1.037	&	1.044	&	1.054	&	\textbf{1.047}	&	1.053	&	1.005	&	0.958	&	0.919	&	0.87	\\	
	&d6	&	0.908	&	0.951	&	0.986	&	1.012	&	1.032	&	\textbf{1.062}	&	1.045	&	1.018	&	1.014	&	0.972	\\	
	&d7	&	0.793	&	0.848	&	0.913	&	0.95	&	1.005	&	1.057	&	\textbf{1.089}	&	1.089	&	1.127	&	1.13	\\	
	&d8	&	0.723	&	0.775	&	0.855	&	0.89	&	0.98	&	1.065	&	1.097	&	\textbf{1.142}	&	1.207	&	1.265	\\	
	&d9	&	0.63	&	0.681	&	0.771	&	0.821	&	0.944	&	1.041	&	1.12	&	1.205	&	\textbf{1.338}	&	1.451	\\	
	&d10	&	0.555	&	0.597	&	0.696	&	0.744	&	0.886	&	1.004	&	1.13	&	1.246	&	1.45	&	\textbf{1.692}	\\	
\hline\hline
\end{tabular}
\caption{Results of estimated sorting measures in 2021-2025 at different
quantiles when $\protect\rho(y_w, y_h) = 0$.}
\label{tab:biv2_2018}
\end{table}

%The drastic reduction in the D1/D1 and D10/D10 cells, for example,
%indicate that a substantial component of the positive sorting in Tables \ref%
%{tab:frequencies_1} and \ref{tab:frequencies_2} is due to matching on the
%basis of unobservable characteristics.

%This extension requires that the model
%incorporates the possibility of selection arising from endogenous work
%decisions, which we address in the next section.

\subsection{Marital Sorting with Employment Selection}

\label{sec:model}

The evidence above is based on the subpopulation of working married couples as the BDR model above does not account for selection into employment. 
%the BDR model has a very restrictive treatment of the role of the unobservables. 
As our sample period witnessed drastic increases in the market participation of married females it is possible that both the composition of the subpopulations of working married females and the working married males whose spouses also work have changed. We incorporate each spouse's participation decision, while allowing for a relationship between these decisions, by extending BDR to incorporate endogenous employment decisions. We begin with some preliminaries.

\subsection{The Sample Selection Model}

Consider the vector of random variables $(D_{w}^{\ast },D_{h}^{\ast
},Y_{w}^{\ast },Y_{h}^{\ast })$, where $D_{w}^{\ast }$ and $D_{h}^{\ast }$ are latent variables that determine the employment decision of the wife and husband, and $Y_{w}^{\ast }$ and $Y_{h}^{\ast }$ are the offered wages to the wife and husband. $Z$ is a vector of observed characteristics of which $X$ is a subset. We make the following assumptions regarding the joint CDF of $(D_{w}^{\ast },D_{h}^{\ast },Y_{w}^{\ast },Y_{h}^{\ast })$ conditional on $Z $:

\begin{assumption}[LGR, Relevance and Exclusion]
\label{ass:lgr} (1) For $(d_{w},d_{h},y_{w},y_{h}) \in \mathbb{R}^4,$ 
\begin{equation*}
F_{D_{w}^{\ast },D_{h}^{\ast },Y_{w}^{\ast },Y_{h}^{\ast } \mid
Z}(d_{w},d_{h},y_{w},y_{h} \mid z) = \Phi _{4}(\boldsymbol{\mu}
(d_{w},d_{h},y_{w},y_{h},z); \boldsymbol{\Sigma}
(d_{w},d_{h},y_{w},y_{h},z)),
\end{equation*}
where $\Phi _{4}(\cdot; \boldsymbol{\Sigma})$ is the standard tetravariate normal CDF with correlation matrix $\boldsymbol{\Sigma}$, 
\begin{equation*}
\boldsymbol{\mu} (d_{w},d_{h},y_{w},y_{h},z)=\left( 
\begin{array}{c}
\mu _{D_{w}^{\ast}}(d_{w},z) \\ 
\mu _{D_{h}^{\ast}}(d_{h},z) \\ 
\mu _{Y_{w}^{\ast}}(y_{w},z) \\ 
\mu _{Y_{h}^{\ast}}(y_{h},z)%
\end{array}%
\right)
\end{equation*}%
and 
\begin{multline*}
\boldsymbol{\Sigma} (d_{w},d_{h},y_{w},y_{h},z)= \\
\left( 
\begin{array}{cccc}
1 & \rho _{D_{w}^{\ast },D_{h}^{\ast }}(d_{w},d_{h},z) & - \rho
_{D_{w}^{\ast },Y_{w}^{\ast }}(d_{w},y_{w},z) & - \rho _{D_{w}^{\ast
},Y_{h}^{\ast }}(d_{w},y_{h},z) \\ 
\rho _{D_{w}^{\ast },D_{h}^{\ast }}(d_{w},d_{h},z) & 1 & - \rho
_{D_{h}^{\ast },Y_{w}^{\ast }}(d_{h},y_{w},z) & - \rho _{D_{h}^{\ast
},Y_{h}^{\ast }}(d_{h},y_{h},z) \\ 
- \rho _{D_{w}^{\ast },Y_{w}^{\ast }}(d_{w},y_{w},z) & - \rho _{D_{h}^{\ast
},Y_{w}^{\ast }}(d_{h},y_{w},z) & 1 & \rho _{Y_{w}^{\ast },Y_{h}^{\ast
}}(y_{w},y_{h},z) \\ 
- \rho _{D_{w}^{\ast },Y_{h}^{\ast }}(d_{w},y_{h},z) & - \rho _{D_{h}^{\ast
},Y_{h}^{\ast }}(d_{h},y_{h},z) & \rho _{Y_{w}^{\ast },Y_{h}^{\ast
}}(y_{w},y_{h},z) & 1%
\end{array}%
\right)
\end{multline*}%
is non-singular almost everywhere in $z$; (2) $\mu _{D_{w}^{\ast}}(d_{w},z)
\neq \mu _{D_{w}^{\ast}}(d_{w},z^{\prime }) $ and $\mu
_{D_{h}^{\ast}}(d_{h},z^{\prime \prime }) \neq \mu
_{D_{h}^{\ast}}(d_{h},z^{\prime \prime \prime })$ for some $z = (z_1,x)$ , $%
z^{\prime } = (z_1^{\prime },x)$, $z^{\prime \prime }= (z_1^{\prime \prime
},x)$, and $z^{\prime \prime \prime} = (z_1^{\prime \prime \prime},x)$; and
(3) $\mu _{Y_{w}^{\ast }}(y_{w},z) = \mu _{Y_{w}^{\ast }}(y_{w},x)$, $\mu
_{Y_{h}^{\ast }}(y_{h},z) = \mu _{Y_{h}^{\ast }}(y_{h},x)$, $%
\rho_{D_{w}^{\ast },Y_{w}^{\ast }}(d_{w},y_{w},z) = \rho _{D_{w}^{\ast
},Y_{w}^{\ast }}(d_{w},y_{w},x)$, 
%$\rho _{D_{w}^{\ast},Y_{h}^{\ast }}(d_{w},y_{h},z) = \rho _{D_{w}^{\ast},Y_{h}^{\ast }}(d_{w},y_{h},x)$, 
%$\rho_{D_{h}^{\ast },Y_{w}^{\ast }}(d_{h},y_{w},z) = \rho_{D_{h}^{\ast },Y_{w}^{\ast }}(d_{h},y_{w},x)$, 
and $\rho _{D_{h}^{\ast},Y_{h}^{\ast }}(d_{h},y_{h},z) = \rho _{D_{h}^{\ast
},Y_{h}^{\ast }}(d_{h},y_{h},x)$, 
%, and $\rho _{Y_{w}^{\ast },Y_{h}^{\ast}}(y_{w},y_{h},z) = \rho _{Y_{w}^{\ast },Y_{h}^{\ast}}(y_{w},y_{h},x)$, 
for $z = (z_1,x)$. %and $\boldsymbol{%
%\Sigma} (d_{w},d_{h},y_{w},y_{h},z)= \boldsymbol{\Sigma}
%(d_{w},d_{h},y_{w},y_{h},x)$ for some $x \subset z$.
\end{assumption}

Assumption \ref{ass:lgr}(1) is similar to the LGR of a joint CDF. In
contrast to the bivariate case, this representation restricts some features of the joint tetravariate distribution. While the univariate and bivariate marginals remain unrestricted, some restrictions are imposed on the trivariate marginals and joint tetravariate distributions. To highlight this, note that the local dependence between any pair of random variables, as measured by the corresponding component of the matrix $\boldsymbol{\Sigma}(d_{w},d_{h},y_{w},y_{h},z)$, does not depend on the value of the other components. For example, local pairwise independence of all the components, that is $\boldsymbol{\Sigma }(d_{w},d_{h},y_{w},y_{h},z)$ equal to the identity matrix, implies joint local independence of all the components.

Assumption \ref{ass:lgr}(3) embodies exclusion restrictions on the marginal distributions of $Y_w^*$ and $Y_h^*$, and the local dependence matrix $\boldsymbol{\Sigma} (d_{w},d_{h},y_{w},y_{h},z)$. Thus, $Y_w^*$ and $Y_h^*$ are independent of the components of $Z$ not included in $X$. Moreover, these components do not affect the local dependence between all the components in $(D_{w}^{\ast },D_{h}^{\ast },Y_{w}^{\ast },Y_{h}^{\ast })$. Assumption \ref{ass:lgr}(2) is a relevance condition of these excluded components of $Z$ on the expectations of $D_w^*$ and $D_h^*$.

The observed variables $(D_{w},D_{h},Y_{w},Y_{h})$ are related to the latent variables as: 
\begin{equation*}
\begin{split}
D_{w}& =\mathbf{1}(D_{w}^{\ast }\leq 0), \\
D_{h}& =\mathbf{1}(D_{h}^{\ast }\leq 0),
\end{split}%
\end{equation*}%
where $D_{w}$ ($D_{h}$) equals 1 when the wife (husband) is working FTFY and 0 otherwise. Moreover, 
\begin{equation*}
\begin{split}
Y_{w}& =Y_{w}^{\ast }~\mbox{ if}~D_{w}=1~\mbox{and}~D_{h}=1 \\
Y_{h}& =Y_{h}^{\ast }~\mbox{ if}~D_{w}=1~\mbox{and}~D_{h}=1
\end{split}%
\end{equation*}%
where $Y_{w}$ and $Y_{h}$ are the FTFY hourly wages of the wife and husband. These are only observed for FTFY working couples.

We show in Appendix \ref{app:id} that all the parameters of the LGR of the latent variables are identified from the distribution of the observed variables.

\begin{theorem}[Identification Under Employment Selection]
\label{thm:id} Under Assumption \ref{ass:lgr}, $\boldsymbol{\mu}
(0,0,y_{w},y_{h},z)$ and $\boldsymbol{\Sigma} (0,0,y_{w},y_{h},z)$ are
identified from the joint distribution of $(D_{w},D_{h},Y_{w},Y_{h},Z)$.
\end{theorem}

\section{Measures of Sorting in the Presence of Selection}

\label{sec:sorting}

We consider how the expressions for the contingency tables entries change when we account for endogenous sample selection. They are now written: 
\begin{equation}
\begin{split}
& s(\underline{y}_{w},\overline{y}_{w},\underline{y}_{h},\overline{y}%
_{h}\mid D_{w}=1,D_{h}=1)= \\
& \frac{\Pr (\underline{y}_{w}<Y_{w}\leq \overline{y}_{w},\underline{y}%
_{h}<Y_{h}\leq \overline{y}_{h}\mid D_{w}=1,D_{h}=1)}{\Pr (\underline{y}%
_{w}<Y_{w}\leq \overline{y}_{w}\mid D_{w}=1,D_{h}=1)\Pr (\underline{y}%
_{h}<Y_{h}\leq \overline{y}_{h}\mid D_{w}=1,D_{h}=1)},
\end{split}
\label{eq:sorting}
\end{equation}%
which corresponds to the ratio of the joint distribution of $(Y_{w},Y_{h})$ to the product of the marginals in the selected population. It is equal to 1 when the wages of the wives and husbands are independent from each other in the selected population.

The values of $s$ are identified from the wage data among couples in which both partners work and does not depend on the identification of the sample selection model. Nevertheless, it is interesting to understand how changes in the sorting measure reflect changes in the model parameters. This can be conducted via counterfactuals. Note that the numerator of this sorting measure can be written as: 
\begin{equation}
\int_{\mathcal{Z}}\Pr (\underline{y}_{w}<Y_{w}\leq \overline{y}_{w},%
\underline{y}_{h}<Y_{h}\leq \overline{y}_{h}\mid
D_{w}=1,D_{h}=1,Z=z)dF_{Z\mid D_{w},D_{h}}(z\mid 1,1)  \label{eq:integral}
\end{equation}%
where the integrand equals: 
\begin{equation*}
\begin{split}
\Pr (\underline{y}_{w}& <Y_{w}\leq \overline{y}_{w},\underline{y}%
_{h}<Y_{h}\leq \overline{y}_{h}\mid D_{w}=1,D_{h}=1,Z=z)= \\
& \frac{\Pr (\underline{y}_{w}<Y_{w}\leq \overline{y}_{w},\underline{y}%
_{h}<Y_{h}\leq \overline{y}_{h},D_{w}=1,D_{h}=1\mid Z=z)}{\Pr
(D_{w}=1,D_{h}=1\mid Z=z)}= \\
& \frac{\Pr (\underline{y}_{w}<Y_{w}^{\ast }\leq \overline{y}_{w},\underline{%
y}_{h}<Y_{h}^{\ast }\leq \overline{y}_{h},D_{w}^{\ast }\leq 0,D_{h}^{\ast
}\leq 0\mid Z=z)}{\Phi _{2}(\mu _{D_{w}^{\ast }}(0,z),\mu _{D_{h}^{\ast
}}(0,z),\rho _{D_{w}^{\ast },D_{h}^{\ast }}(0,0,x))},
\end{split}%
\end{equation*}%
with 
\begin{multline*}
\Pr (\underline{y}_{w}<Y_{w}^{\ast }\leq \overline{y}_{w},\underline{y}%
_{h}<Y_{h}^{\ast }\leq \overline{y}_{h},D_{w}^{\ast }\leq 0,D_{h}^{\ast
}\leq 0\mid Z=z)= \\
\Phi _{4}(\boldsymbol{\mu }(0,0,\overline{y}_{w},\overline{y}_{h},z);%
\boldsymbol{\Sigma }(0,0,\overline{y}_{w},\overline{y}_{h},x))-\Phi _{4}(%
\boldsymbol{\mu }(0,0,\underline{y}_{w},\overline{y}_{h},z);\boldsymbol{%
\Sigma }(0,0,\underline{y}_{w},\overline{y}_{h},x)) \\
-\Phi _{4}(\boldsymbol{\mu }(0,0,\overline{y}_{w},\underline{y}_{h},z);%
\boldsymbol{\Sigma }(0,0,\overline{y}_{w},\underline{y}_{h},x))+\Phi _{4}(%
\boldsymbol{\mu }(0,0,\underline{y}_{w},\underline{y}_{h},z);\boldsymbol{%
\Sigma }(0,0,\underline{y}_{w},\underline{y}_{h},x)).
\end{multline*}%
%
%where 
%\begin{equation}
%\begin{split}
%G^{s}(\overline{y}_{w},\overline{y}_{h},z)& :=\mathbb{P}(Y_{w}\leq
%y_{w};Y_{h}\leq y_{h},D_{w}=1,D_{h}=1|Z=z) \\
%& =\int_{-\infty }^{\nu _{D_{w}^{\ast }}(z_{j})}\int_{-\infty }^{\nu
%_{D_{h}^{\ast }}(z_{j})}\int_{-\infty }^{\mu _{Y_{w}^{\ast
%}}(y_{w},x)}\int_{-\infty }^{\mu _{Y_{h}^{\ast }}(y_{h},x)}d\Phi
%_{4}(u_{1},u_{2},u_{3},u_{4},\Sigma (y_{w},y_{h},z_{j}))
%\end{split}%
%\end{equation}%
%Comparable formulae can be derived for the numerator of the definition of
%the sorting measure $s$ as defined in (\ref{eq:sorting}).

\subsection{Counterfactuals Based on a Specific Time Period}

For each time period counterfactuals can be obtained by replacing $%
\boldsymbol{\Sigma }(d_{w},d_{h},y_{w},y_{h},x)$ with an alternative
positive definite $\widetilde{\boldsymbol{\Sigma }}%
(d_{w},d_{h},y_{w},y_{h},x)$ reflecting different model parameters. For
example, setting $\rho _{D_{w}^{\ast },D_{h}^{\ast }}(0,0,x)$ to zero
examines the impact of making the employment decisions for husbands and
wives conditionally independent. Setting $\rho _{D_{w}^{\ast },Y_{w}^{\ast
}}(0,y_{w},x)$, $\rho _{D_{h}^{\ast },Y_{h}^{\ast }}(0,y_{h},x)$, $\rho
_{D_{w}^{\ast },Y_{h}^{\ast }}(0,y_{h},x)$, and $\rho _{D_{h}^{\ast
},Y_{w}^{\ast }}(0,y_{w},x)$ to zero eliminates the role of selection while
setting $\rho _{Y_{w}^{\ast },Y_{h}^{\ast }}(y_{w},y_{h},x)$ to zero
eliminates sorting on the unobservables correlated with wages. Sorting may
still arise here on the basis of observed characteristics. Another potential
counterfactual distribution integrates (\ref{eq:integral}) over the
distribution of $Z$ for the whole population of married couples in our
sample rather than the population of FTFY working couples.

\subsection{Counterfactuals Based on Different Periods of Time}

To investigate intertemporal changes in sorting we examine counterfactuals
over time. We start by rewriting the elements of the contingency matrix in %
\eqref{eq:sorting} as: 
\begin{multline*}
s(\underline{y}_{w},\overline{y}_{w},\underline{y}_{h},\overline{y}_{h}\mid
D_{w}=1,D_{h}=1) \\
=\scriptstyle{\ \frac{F_{Y_{w},Y_{h}\mid D_{w},D_{h}}(\overline{y}_{w},%
\overline{y}_{h}\mid 1,1)-F_{Y_{w},Y_{h}\mid D_{w},D_{h}}(\overline{y}_{w},%
\underline{y}_{h}\mid 1,1)-F_{Y_{w},Y_{h}\mid D_{w},D_{h}}(\underline{y}_{w},%
\overline{y}_{h}\mid 1,1)+F_{Y_{w},Y_{h}\mid D_{w},D_{h}}(\underline{y}_{w},%
\underline{y}_{h}\mid 1,1)}{[F_{Y_{w}\mid D_{w},D_{h}}(\overline{y}_{w}\mid
1,1)-F_{Y_{w}\mid D_{w},D_{h}}(\underline{y}_{w}\mid 1,1)][F_{Y_{h}\mid
D_{w},D_{h}}(\overline{y}_{h}\mid 1,1)-F_{Y_{h}\mid D_{w},D_{h}}(\underline{y%
}_{h}\mid 1,1)]}},
\end{multline*}%
where 
\begin{equation*}
F_{Y_{w},Y_{h}\mid D_{w},D_{h}}(y_{w},y_{h}\mid 1,1)=\Pr (Y_{w}\leq
y_{w},Y_{h}\leq y_{h}\mid D_{w}=1,D_{h}=1),
\end{equation*}%
and 
\begin{equation*}
F_{Y_{j}\mid D_{w},D_{h}}(y_{j}\mid 1,1)=\Pr (Y_{j}\leq y_{j}\mid
D_{w}=1,D_{h}=1),\quad j=w,h.
\end{equation*}%

%Consider a simplification of (\ref{eq:sorting}) WHAT IS THE
%SIMPLIFICATION: 
%\begin{equation*}
%\frac{\Pr (Y_{w}\leq \overline{y}_{w},Y_{h}\leq \overline{y}_{h}\mid
%D_{w}=1,D_{h}=1)}{\Pr (Y_{w}\leq \overline{y}_{w}\mid D_{w}=1,D_{h}=1)\Pr
%(Y_{h}\leq \overline{y}_{h}\mid D_{w}=1,D_{h}=1)}.
%\end{equation*}%

We calculate counterfactual joint distributions assuming employment
selection is as in year $q$, the wage structure is as in year $r,$ and the
composition of the workforce is as in year $s$ by: 
\begin{multline}
F_{Y_{w},Y_{h}\mid D_{w},D_{h}}^{q,r,s}(y_{w},y_{h}\mid 1,1)
\label{eq:cdist} \\
:=\int {\Pr }^{q,r}(Y_{w}\leq y_{w},Y_{h}\leq y_{h}\mid
D_{w}=1,D_{h}=1,Z=z)dF_{Z}^{s}(z)
\end{multline}%
where $F_{Z}^{s}$ is the distribution of $Z$ in year $s$, and 
\begin{multline*}
{\Pr }^{qr}(Y_{w}\leq y_{w},Y_{h}\leq y_{h}\mid D_{w}=1,D_{h}=1,Z=z) \\
=\frac{\Phi _{4}(\boldsymbol{\mu }^{q,r}(0,0,y_{w},y_{h},z);\boldsymbol{%
\Sigma }^{q,r}(0,0,y_{w},y_{h},x))}{\Phi _{2}(\mu _{D_{w}^{\ast
}}^{q}(0,z),\mu _{D_{h}^{\ast }}^{q}(0,z),\rho _{D_{w}^{\ast },D_{h}^{\ast
}}^{q}(0,0,x))},
\end{multline*}%
with 
\begin{equation*}
\boldsymbol{\mu }^{q,r}(0,0,y_{w},y_{h},z)=\left( 
\begin{array}{c}
\mu _{D_{w}^{\ast }}^{q}(0,z) \\ 
\mu _{D_{h}^{\ast }}^{q}(0,z) \\ 
\mu _{Y_{w}^{\ast }}^{r}(y_{w},x) \\ 
\mu _{Y_{h}^{\ast }}^{r}(y_{h},x)%
\end{array}%
\right)
\end{equation*}%
and\footnote{%
A practical problem is that there is no guarantee that $\boldsymbol{\Sigma }%
^{q,r}(0,0,y_{w},y_{h},x)$ is positive definite. Nevertheless, we did not
encounter this problem in our empirical analysis.} 
\begin{multline*}
\boldsymbol{\Sigma }^{q,r}(0,0,y_{w},y_{h},x)= \\
\left( 
\begin{array}{cccc}
1 & \rho _{D_{w}^{\ast },D_{h}^{\ast }}^{q}(0,0,x) & -\rho _{D_{w}^{\ast
},Y_{h}^{\ast }}^{q}(0,y_{h},x) & -\rho _{D_{w}^{\ast },Y_{h}^{\ast
}}^{q}(0,y_{h},x) \\ 
\rho _{D_{w}^{\ast },D_{h}^{\ast }}^{q}(0,0,x) & 1 & -\rho _{D_{h}^{\ast
},Y_{w}^{\ast }}^{q}(0,y_{w},x) & -\rho _{D_{h}^{\ast },Y_{h}^{\ast
}}^{q}(0,y_{h},x) \\ 
-\rho _{D_{w}^{\ast },Y_{w}^{\ast }}^{q}(0,y_{1},x) & -\rho _{D_{h}^{\ast
},Y_{w}^{\ast }}^{q}(0,y_{w},x) & 1 & \rho _{Y_{w}^{\ast },Y_{h}^{\ast
}}^{r}(y_{w},y_{h},x) \\ 
-\rho _{D_{w}^{\ast },Y_{h}^{\ast }}^{q}(0,y_{h},x) & -\rho _{D_{h}^{\ast
},Y_{h}^{\ast }}^{q}(0,y_{h},x) & \rho _{Y_{w}^{\ast },Y_{h}^{\ast
}}^{r}(y_{w},y_{h},x) & 1%
\end{array}%
\right) .
\end{multline*}

%
%and 
%\begin{equation*}
%\begin{split}
%G^{qr}({y}_{w};{y}_{h},z)& :=\mathbb{P}^{qr}(Y_{w}\leq {y}_{w};Y_{h}\leq {y}%
%_{h},D_{w}=1,D_{h}=1|Z=z) \\
%& \int_{-\infty }^{\nu _{D_{w}^{\ast }}^{q}(z)}\int_{-\infty }^{\nu
%_{D_{h}^{\ast }}^{q}(z)}\int_{-\infty }^{\mu _{Y_{w}^{\ast
%}}^{r}(y_{w},z)}\int_{-\infty }^{\mu _{Y_{h}^{\ast }}^{r}(y_{h},z)}d\Phi
%_{4}(u_{1},u_{2},u_{3},u_{4},\Sigma ^{r}(y_{w},y_{h},z))
%\end{split}%
%\end{equation*}%
Comparable counterfactual marginal distributions can be computed for the
denominator.

Comparisons should be based on wage levels of husbands and wives which
incorporate the changes in the wage distributions. Since the real wages of males decreased at the bottom of the distribution while the real wages increased over the whole distribution for females, using fixed wage levels is not appropriate. For example, we are more likely to find a husband earning more than 20 dollars an hour with a wife earning more than 15 dollars an hour in 2020 than in 1976. This, however, does not imply a change in sorting behavior. Rather, it reflects the changes in the likelihoods from the respective marginal distributions. As we use quantiles rather than fixed wages to make these comparisons, we need the counterfactual quantiles of the corresponding marginal wage distributions. For example, the $\tau $-th quantile of the marginal wage distribution of the wives is: 
\begin{multline}
Q_{Y_{w}}^{q,r,s}(\tau ) \\
=\inf \left\{ y\in \mathbb{R}:\int {\Pr }^{q,r}(Y_{w}\leq y,Y_{h}\leq \infty
\mid D_{w}=1,D_{h}=1,Z=z)dF_{Z}^{s}(z)\geq \tau \right\} .
\label{eq:quantile_w}
\end{multline}%
%
%
%with 
%\begin{equation*}
%\begin{split}
%\mathbb{P}^{qr}(Y_{w}& \leq \overline{y}_{w};Y_{h}\leq \overline{y}%
%_{h}|D_{w}=1,D_{h}=1,X_{s}=x,Z_{s}=z)= \\
%& \frac{\int_{-\infty }^{\nu _{D_{w}^{\ast ^{q}}}(x,z)}\int_{-\infty }^{\nu
%_{D_{h}^{\ast ^{q}}}(x,z)}\int_{\mu _{Y_{w}^{\ast ^{r}}}(\underline{y}%
%_{w},x)}^{\mu _{Y_{w}^{\ast ^{r}}}(\overline{y}_{w},x)}\int_{\mu
%_{Y_{w}^{\ast ^{r}}}(\underline{y}_{h},x)}^{\mu _{Y_{w}^{\ast ^{r}}}(%
%\overline{y}_{h},x)}d\Phi _{4}(u_{1},u_{2},u_{3},u_{4},\Sigma
%^{q,r}(y_{w},y_{h}|z))}{\int_{-\infty }^{\nu _{D_{w}^{\ast
%^{q}}}(x,z)}\int_{-\infty }^{\nu _{D_{h}^{\ast ^{q}}}(x,z)}d\Phi
%_{2}(u_{1},u_{2},\rho _{D_{w}^{\ast },D_{h}^{\ast }}^{q})}.
%\end{split}%
%\end{equation*}%
%where $Q_{w}^{qrs}(\tau )$ is the $\tau $-th quantile for the case that
%selection is as in year $q$, the wage structure is as in year $r$ and the
%composition is as in year $s$.

\section{Estimation}

\label{sec:estimation}

We consider a semiparametric BDR model with selection that imposes:

\begin{assumption}[BDR with Selection]
\label{ass:bdr-selection} (1) $\mu_{Y^*_j}(y,x) = P_{j}(x)^{\prime }\beta
_{j}(y),$ $\mu_{D^*_j}(0,z) = Q_{j}(z)^{\prime }\gamma_j,$ where $P_{j}$ and 
$Q_{j}$ are transformations of $x$ and $z$, and $\beta _{j}(y)$ and $\gamma
_{j}$ are vectors of coefficients, $j \in \{w,h\}$; and (2) $\boldsymbol{%
\Sigma} (d_{w},d_{h},y_{w},y_{h},x) = \boldsymbol{\Sigma}
(d_{w},d_{h},y_{w},y_{h})$.
\end{assumption}

\subsection{Estimation of the Local Model Parameters}

Using Assumption \ref{ass:bdr-selection} with $P_{j}(x)=x$ and $Q_{j}(z)=z$
for $j\in \{w,h\}$ we estimate the parameters in 2 steps:

%We use the following parameterization for $\nu _{D_{w}^{\ast }}(x,z)$ and $%
%\nu _{D_{h}^{\ast }}(x,z)$ 
%\begin{equation*}
%\begin{split}
%\nu _{D_{w}^{\ast }}(x,z)& =-\gamma _{w}^{T}z \\
%\nu _{D_{h}^{\ast }}(x,z)& =-\gamma _{h}^{T}z \\
%\mu _{Y_{w}^{\ast }}(y,x)& =-\beta _{w}^{T}(y)x \\
%\mu _{Y_{h}^{\ast }}(y,x)& =-\beta _{h}^{T}(y)x
%\end{split}%
%\end{equation*}%

\begin{enumerate}
\item Bivariate probit to obtain $\gamma _{w}$, $\gamma _{h}$ and $\rho
_{D_{w}^{\ast },D_{h}^{\ast }}:=\rho _{D_{w}^{\ast },D_{h}^{\ast }}(0,0)$
using: %
%
%first focus on the estimation of the parameters in $\gamma _{w}$ and $%
%\gamma _{h}$. From the definition of the model the probability that both the
%wife and the husband work is given by: 
\begin{equation*}
\mathbb{P}(D_{w}=1,D_{h}=1\mid Z=z)=\Phi _{2}(z^{\prime }\gamma
_{w},z^{\prime }\gamma _{h};\rho _{D_{w}^{\ast },D_{h}^{\ast }}).
\end{equation*}%

%Given the normality assumption the above is a bivariate probit model and $%
%\gamma _{w}^{T}z,$ $\gamma _{h}^{T}z$ and $\rho _{32}(x,z)$ can be estimated
%by maximum likelihood. 
We denote the estimators as $\widehat{\gamma }_{w},$ $\widehat{\gamma }_{h}$
and $\widehat{\rho }_{D_{w}^{\ast },D_{h}^{\ast }}$.

\item Multivariate probit with sample selection correction to estimate the
remaining parameters using: 
\begin{multline*}
\mathbb{P}(Y_{w}\leq y_{w};Y_{h}\leq y_{h}|D_{w}=1;D_{h}=1,X=x,Z=z)\propto \\
\Phi _{4}(z^{\prime }\gamma _{w},z^{\prime }\gamma _{h},x^{\prime }\beta
_{w}(y_{w}),x^{\prime }\beta _{h}(y_{h});\boldsymbol{\Sigma }%
(0,0,y_{w},y_{h})),
\end{multline*}%
which can be estimated by a small adaption of the multivariate probit model
after plugging in the first-stage estimators $\widehat{\gamma }_{w},$ $%
\widehat{\gamma }_{h}$ and $\widehat{\rho }_{D_{w}^{\ast },D_{h}^{\ast }}$.
\end{enumerate}

The first step is standard. The second step is straightforward, but the
calculation of higher-order integrals in the multivariate probit is both
computationally intensive and imprecise. The imprecision is especially
unfortunate when combined with a numerical optimization method that assumes
smoothness of the first and second order derivative of the criterion
function. Therefore, we employ the GHK importance sampling simulator of
Geweke (1991), Hajivassiliou and McFadden (1990) and Keane (1990) to
simulate these probabilities. This importance sampling simulator uses the
result that multivariate normal distributions, conditional on realizations
of one or more elements of the outcome vector, are also normally distributed
but with a lower dimension.

\subsection{Extending the Model to Endogenous Marriage Decisions}

\label{ss:endogenous_marriage}
Our model and estimator abstract from the decision to marry and the choice of partner. Although we do not incorporate these into the model we estimate, we outline how one could account for the marriage decision. Including only married individuals in the sample potentially produces a selection bias similar to estimating the model for workers only without accounting for the work decision. Accordingly, we could account for the selection bias from the marriage decisions in the corresponding manner. This requires expanding the dataset to include non-married individuals. We would first estimate the decision to marry using a likelihood function that consists of a mixture of two one-dimensional log-normal components--corresponding to single men and single women--and one bivariate component, given by the log of a bivariate normal distribution, corresponding to married couples. We would then estimate the decision to work using only observations of married couples, based on a likelihood function characterized by a four-dimensional normal distribution. Specifically, two dimensions capture the marriage decision, while the other two capture the work decision. The final step is analogous to the final step in our baseline model without marriage decisions, but is now conditioned on both the marriage and work decisions. This approach does not introduce any methodological complications but it does make estimation challenging as it is necessary to estimate a higher order integral via multinomial probit when estimating the individual's location in the wage distributions while accounting for selection bias from both the work FTFY and the marriage decisions. As one can estimate the marriage decisions in isolation this step is not the issue. However, one needs to estimate the additional associated correlations, noting that the unobservables associated with an individual's marriage decision are potentially correlated with the unobservables related to the other decisions and wages, and this is challenging. Moreover, one needs additional exclusion restrictions for identification and there are no obvious candidates in this data set.\footnote{McConnel and Valladares-Esteban (2025) use the marriage rate of individuals who have the same observed characteristics in terms of state, education and gender but who are 6 to 15 years older as a proxy for social norms.} For these reasons, we do not attempt to endogenize the marriage decision here and leave it for future research.
While the marriage decision can be handled in a straightforward manner, modeling the choice of spouse is far more challenging. One would have to augment the likelihood function to incorporate the probability that male $m$ is married to female $f$ but given the large number of observations this does not seem feasible even if we restricted attention to the married population. A potential solution could be to introduce types of wives and husbands in a similar manner as Choo and Siow (2006). However, the dimension of integrals would increase drastically for any reasonable number of types.

\subsection{Estimation of the Measures of Sorting}

Estimation of the sorting measures requires estimates of the quantiles of
the marginal distributions of wives' and husbands' wage distributions. These
are obtained via application of the generalized inverse or rearrangement
operator to the plug-in estimator of the distribution. For example: 
\begin{equation}
\widehat{Q}_{Y_{w}}^{q,r,s}(\tau )=\int_{0}^{\infty }\boldsymbol{1}\left\{ 
\frac{1}{n_{s}}\sum_{i=1}^{n_{s}}\widehat{\Pr }^{q,r}(Y_{w}\leq y\mid
D_{w}=1,D_{h}=1,Z^{s}=Z_{i})\leq \tau \right\} dy,  \label{eq:quant1}
\end{equation}%
with: 
\begin{multline*}
\widehat{\Pr }^{q,r}(Y_{w}\leq y_{w}\mid D_{w}=1,D_{h}=1,Z^{s}=z)= \\
\frac{\Phi _{4}(z^{\prime }\widehat{\gamma }_{w}^{q},z^{\prime }\widehat{%
\gamma }_{h}^{q},x^{\prime }\widehat{\beta }_{w}^{r}(y_{w}),\infty ;\widehat{%
\boldsymbol{\Sigma }}^{q,r}(0,0,y_{w},y_{h}))}{\Phi _{2}(z^{\prime }\widehat{%
\gamma }_{w}^{q},z^{\prime }\widehat{\gamma }_{h}^{q};\widehat{\rho }%
_{D_{w}^{\ast },D_{h}^{\ast }}^{q})},
\end{multline*}%
for any $y_{h}$, where $\widehat{\boldsymbol{\Sigma }}%
^{q,r}(0,0,y_{w},y_{h}) $ is the plug-in estimator of $\boldsymbol{\Sigma }%
^{q,r}(0,0,y_{w},y_{h})$, and the integrals are calculated using the GHK
importance sampling method.

Equation \eqref{eq:quant1} is solved using bisection implying that we need
to calculate the term on the right-hand side of that equation for different
trial values of the quantile. The counterfactual quantiles for husbands can
be estimated in a similar way.

The estimator of the measures of sorting is: \vspace{-0.5cm} 
\begin{multline*}
\widehat{s}^{q,r,s}(\underline{y}_{w},\overline{y}_{w},\underline{y}_{h},%
\overline{y}_{h}\mid D_{w}=1,D_{h}=1) \\
=\scriptstyle{\ \frac{\widehat{F}_{Y_{w},Y_{h}\mid D_{w},D_{h}}^{q,r,s}(%
\overline{y}_{w},\overline{y}_{h}\mid 1,1)-\widehat{F}_{Y_{w},Y_{h}\mid
D_{w},D_{h}}^{q,r,s}(\overline{y}_{w},\underline{y}_{h}\mid 1,1)-\widehat{F}%
_{Y_{w},Y_{h}\mid D_{w},D_{h}}^{q,r,s}(\underline{y}_{w},\overline{y}%
_{h}\mid 1,1)+\widehat{F}_{Y_{w},Y_{h}\mid D_{w},D_{h}}^{q,r,s}(\underline{y}%
_{w},\underline{y}_{h}\mid 1,1)}{[\widehat{F}_{Y_{w}\mid
D_{w},D_{h}}^{q,r,s}(\overline{y}_{w}\mid 1,1)-\widehat{F}_{Y_{w}\mid
D_{w},D_{h}}^{q,r,s}(\underline{y}_{w}\mid 1,1)][\widehat{F}_{Y_{h}\mid
D_{w},D_{h}}^{q,r,s}(\overline{y}_{h}\mid 1,1)-\widehat{F}_{Y_{h}\mid
D_{w},D_{h}}^{q,r,s}(\underline{y}_{h}\mid 1,1)]}},
\end{multline*}%
where $\widehat{F}_{Y_{w},Y_{h}\mid D_{w},D_{h}}^{q,r,s}$, $\widehat{F}%
_{Y_{w}\mid D_{w},D_{h}}^{q,r,s}$ and $\widehat{F}_{Y_{h}\mid
D_{w},D_{h}}^{q,r,s}$ are plug-in estimators of $F_{Y_{w},Y_{h}\mid
D_{w},D_{h}}^{q,r,s}$, $F_{Y_{w}\mid D_{w},D_{h}}^{q,r,s}$ and $F_{Y_{h}\mid
D_{w},D_{h}}^{q,r,s}$ in \eqref{eq:cdist}, respectively; and $\underline{y}%
_{j}$ and $\overline{y}_{j}$ evaluated at $\widehat{Q}%
_{Y_{j}}^{q,r,s}((i-1)/10)$ and $\widehat{Q}_{Y_{k}}^{q,r,s}(i/10)$, $%
i=1,\ldots ,10$ and $j\in \{w,h\}$.

%With the counterfactual quantiles, we estimate our measure of sorting by: 
%\begin{equation*}
%\begin{split}
%& \widehat{s}^{qrs}(\underline{w}_{w},\overline{w}_{w},\underline{w}_{w},%
%\overline{w}_{w})= \\
%& \frac{\widehat{\mathbb{P}}(W_{w}^{qrs}\leq \widehat{Q}_{w}^{qrs}(\tau
%);W_{h}^{qrs}\leq \widehat{Q}_{h}^{qrs}(\tau )|D_{w}=1,D_{h}=1)}{\tau ^{2}}
%\end{split}%
%\end{equation*}%
%and with 
%\begin{equation*}
%\int \mathbb{P}^{qr}(W_{w}\leq \overline{w}_{w};W_{h}\leq \overline{w}%
%_{h}|D_{w}=1,D_{h}=1,X_{s}=x,Z_{s}=z)dF_{X_{s},Z_{s}}(x,z)
%\end{equation*}%
%where $\sigma $%
%\begin{equation*}
%\begin{split}
%\mathbb{P}^{qr}(W_{w}& \leq \overline{w}_{w};W_{h}\leq \overline{w}%
%_{h}|D_{w}=1,D_{h}=1,X_{s}=x,Z_{s}=z)= \\
%& \frac{\int_{-\infty }^{-\beta _{w}^{r}(w_{w})x}\int_{-\infty }^{-\beta
%_{h}^{r}(w_{h})x}\int_{-\gamma _{w}^{q}(x,z)}^{\infty }\int_{-\gamma
%_{h}^{q}(x,z)}^{\infty }\varphi _{4}(u_{1},u_{2},u_{3},u_{4},\Sigma
%_{4}^{q,r}(w_{w},w_{h},x,z))du_{4}du_{3}du_{2}du_{1}}{\int_{-\gamma
%_{w}^{q^{T}}(x,z)}^{\infty }\int_{-\gamma _{h}^{q^{T}}(x,z)}^{\infty
%}\varphi _{2}(u_{1},u_{2},\rho _{32}^{q}(x,z))du_{2}du_{1}}
%\end{split}%
%\end{equation*}%
%Again, we can calculate the integrals using the GHK importance sampling
%method.

\section{Empirical Results}

\label{sec:empirical}

\subsection{Parameter Estimates}

We estimate the model at each decile for both of the males and females wage
distribution for 5 year periods starting from 1976 and ending at 2025. We estimate one hundred different sets of coefficients for each time period corresponding to the different combinations of deciles for each spouse. As our primary focus is on the role of selection and sorting, we do not report the 100 sets of parameter estimates but Figures \ref{fig:covariance_q1}-\ref{fig:covariance_q3} present the estimates of the different $\rho $'s at the same quartile of both distributions along with their bootstrapped standard errors. The parameters, along with their confidence intervals, are estimated over five‑year periods, with the first being 1976–1980 and the last 2021–2025. Results for the other quantiles are available from the authors. We employ the 100 sets of parameter estimates in conducting the counterfactuals reported below and the discussion of the $\rho $'s which follows is only to provide some insight into the behavior of these selection and sorting parameters.

We discuss our model specification before proceeding to the results.
We listed the variables in $X$ employed in the BDR
specification in Section \ref{ss:determinants_of_sorting}. We continue to use these variables but identification now
requires some variables in $Z$ which are excluded from $X.$ We follow
Mulligan and Rubinstein (2008) and employ a dummy variable for children at
home, family size, and a dummy variable for children at home under 5 years
of age in this role. These variables are seen as contentious and we discuss
the implications of these choices below.

Consider the correlation between the unobservables in the spouses'
employment equations, $\rho _{D_{w}^{\ast },D_{h}^{\ast }},$ recalling it is
invariant to the location in the wage distributions. The positive value
indicates that the unobservables affecting the spouses' work decisions are
positively correlated. The coefficient is small in magnitude but is
statistically significantly different from zero. It generally increases over
time, except for a dip in the middle of the sample period, and reveals that
the unobserved factors which make spouses both work FTFY are increasingly
correlated over time. We leave a discussion of the marginal impact of
changing this, and other parameters, to the following section.

The estimated selection coefficients for wives ($\rho _{D_{w}^{\ast
},Y_{w}^{\ast }}$) and husbands ($\rho _{D_{h}^{\ast },Y_{h}^{\ast }}$)
capture the correlation between the unobservables affecting the individual's
respective work decision and their own wages. While in conventional
selection models they capture the relationship between the selection
decision and the mean wage, we evaluate this correlation at different
quantiles of the wage distribution. The estimate for wives is particularly
interesting and consistent with earlier evidence using similar
identification approaches (see, for example, Mulligan and Rubinstein, 2008,
and FVVP, 2023) for FTFY workers. Similar to these earlier papers we find
evidence of negative selection changing to positive selection. At each
quantile there is negative selection in the earlier period with an estimate
around -0.3 which is statistically different from zero. For the latter
sample period the estimate has increased. At Q1/Q1 the estimate has
increased to zero, while at Q2/Q2 and Q3/Q3 the estimates are approximately
0.2 and 0.3. As discussed in FVVP (2023), the sign of $\rho _{D_{w}^{\ast
},Y_{w}^{\ast }}$ is contentious given its interpretation and its
implication for selection. They note that the sign of the selection terms
reflects the impact of the variables used to explain participation
which are excluded from the wage equation. The change of sign captures the changing impact of these variables on the participation decision.%
\footnote{%
FVVP (2023b) show that this negative coefficient for the earlier period is
likely to be due to the exclusion restrictions. The reader is referred to
that paper for a detailed discussion of how the exclusion restrictions may
be generating the selection related results. One could reproduce the
counterfactual that follows using the FVVP (2023b) approach which employs an
alternative identification strategy but that is not feasible without making
a number of additional assumptions.} While the sign of $\rho _{D_{w}^{\ast
},Y_{w}^{\ast }}$ is controversial we explore the impact of changing its
value in counterfactual exercises below.

It is not typical to account for selection when estimating male wage
equations. However, it is important to do so here as we do not know apriori
the impact of selection in this model. Moreover, as we account for the role
of male selection on the female wage it is necessary to model the male work
decision. The parameter $\rho _{D_{h}^{\ast },Y_{h}^{\ast }}$ is very poorly
estimated at each of the quantiles reported. At Q1/Q1 and Q2/Q2 it is
negative and large but very imprecisely estimated. At Q3/Q3 it is both
negative and positive but generally imprecisely estimated. We attribute this
result to our inability to identify this parameter from these data.

The parameters $\rho _{D_{w}^{\ast },Y_{h}^{\ast }}$ and $\rho _{D_{h}^{\ast
},Y_{w}^{\ast }}$ also vary by location in the respective wage distributions
and capture the correlation between the unobservables affecting the work
decision of the individual with the unobservables affecting the wage of
their spouse at a certain quantile in the spouses wage distribution. First
consider $\rho _{D_{w}^{\ast },Y_{h}^{\ast }}$. Figures \ref%
{fig:covariance_q1}-\ref{fig:covariance_q3} suggest that the estimate is
both negative and positive at each quantile depending on the time period.
However, the confidence bands suggest that we cannot reject that the
estimate is zero for many of the periods with the exception of some middle
years at the first quartile. $\rho _{D_{h}^{\ast },Y_{w}^{\ast }}$ captures
how the unobservables driving the husband's participation decision affects
the wife's wage. Given the evidence above on male selection it is
unsurprising that this parameter is imprecisely estimated although it
is generally negative and statistically different from zero.

\begin{figure}[tbp]
\caption{Results for covariance matrix at Q1.}
\label{fig:covariance_q1}%
\resizebox{14cm}{!}{
		\input{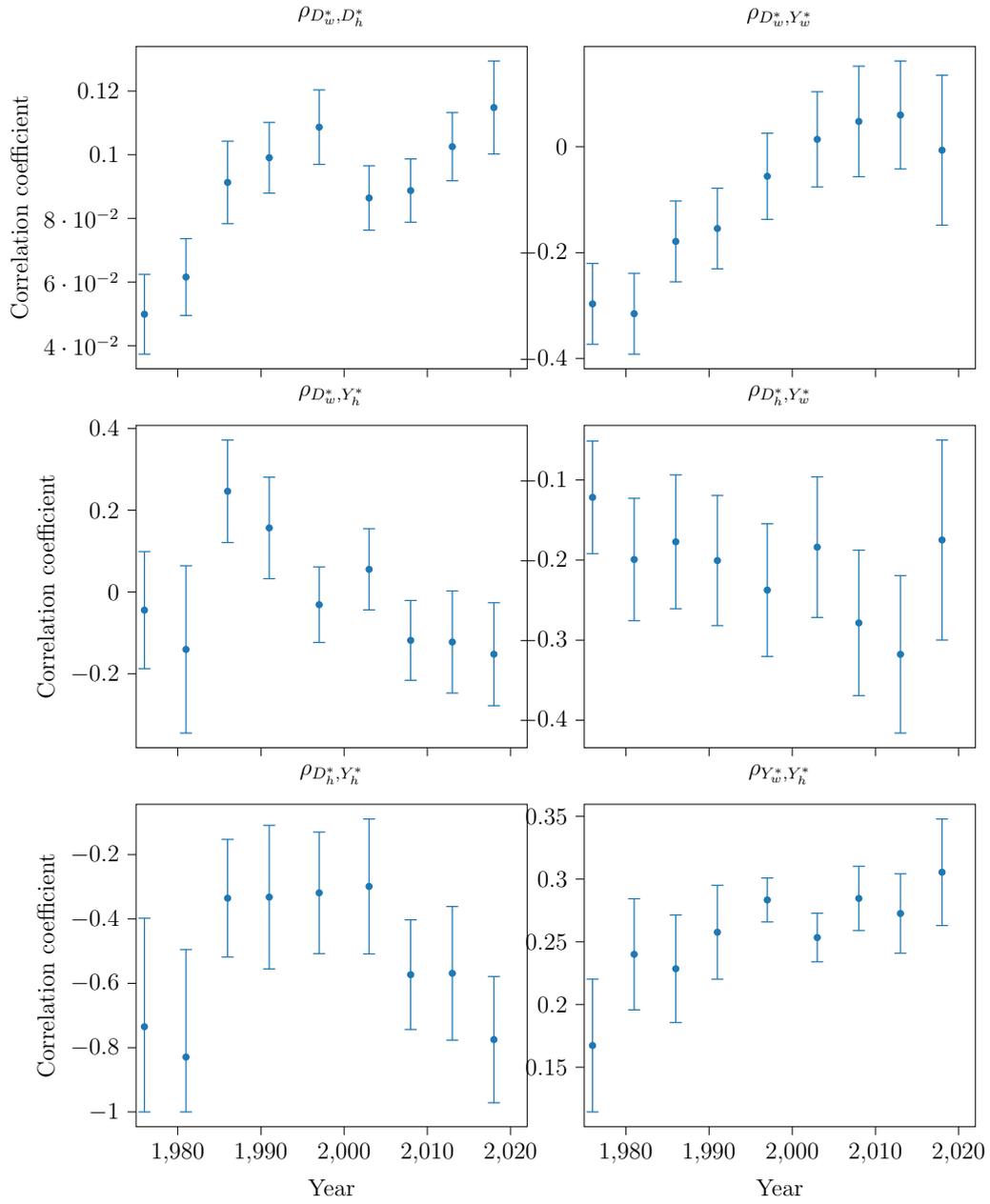}
	}
\end{figure}

\begin{figure}[tbp]
\caption{Results for covariance matrix at Q2.}
\label{fig:covariance_q2}%
\resizebox{14cm}{!}{
		\input{inputs/results_bootstraps_50_2026.tex}
	}
\end{figure}

\begin{figure}[tbp]
\caption{Results for covariance matrix at Q3.}
\label{fig:covariance_q3}%
\resizebox{14cm}{!}{
		\input{inputs/results_bootstraps_75_2026.tex}
	}
\end{figure}

\begin{figure}[tbp]
\caption{Kernel density estimates of $\protect\rho (y_{w},y_{h})$.}
\label{fig:kernelModelComplete}\centering
% This file was created with tikzplotlib v0.10.1.
\begin{tikzpicture}

\definecolor{darkgray176}{RGB}{176,176,176}
\definecolor{steelblue31119180}{RGB}{31,119,180}
33
\begin{groupplot}[group style={group size=2 by 1}]
\nextgroupplot[
tick align=outside,
tick pos=left,
title={1976-1980},
x grid style={darkgray176},
xlabel={$\rho_{Y^*_{w},Y^*_{h}}$},
xmin=0, xmax=0.5,
xtick style={color=black},
y grid style={darkgray176},
ylabel={density},
ymin=0, ymax=9,
ytick style={color=black}
]
\addplot [thick, steelblue31119180]
table {%
0.0295096931384328 0.0245320271287022
0.0395096931384328 0.0866418705369432
0.0495096931384328 0.24617412667695
0.0595096931384328 0.567460928834484
0.0695096931384328 1.07430458191764
0.0795096931384328 1.70189182923709
0.0895096931384328 2.32021811818542
0.0995096931384328 2.82809901766932
0.109509693138433 3.22001425629461
0.119509693138433 3.56740815576991
0.129509693138433 3.96073182522171
0.139509693138433 4.44596375087894
0.149509693138433 4.96427301294013
0.159509693138433 5.37860766592894
0.169509693138433 5.62207177177008
0.179509693138433 5.79549446085388
0.189509693138433 6.05438769595535
0.199509693138433 6.40842957104477
0.209509693138433 6.66931166304799
0.219509693138433 6.58869127697568
0.229509693138433 6.04359746554019
0.239509693138433 5.1223524255872
0.249509693138433 4.0673097308106
0.259509693138433 3.137470866124
0.269509693138433 2.4783184517364
0.279509693138433 2.06268621167905
0.289509693138433 1.74761437815636
0.299509693138433 1.41049821823194
0.309509693138433 1.0341013687928
0.319509693138433 0.674176321073502
0.329509693138433 0.385774263416689
0.339509693138433 0.190125968677695
0.349509693138433 0.0786992821889005
0.359509693138433 0.0266934637177538
0.369509693138433 0.00727779895035892
0.379509693138433 0.0015745330088707
0.389509693138433 0.000268170303620294
0.399509693138433 3.57893597659996e-05
0.409509693138433 3.73271107165951e-06
};

\nextgroupplot[
tick align=outside,
tick pos=left,
title={2021-2025},
x grid style={darkgray176},
xlabel={$\rho_{Y^*_{w},Y^*_{h}}$},
xmin=0, xmax=0.5,
xtick style={color=black},
y grid style={darkgray176},
ymin=0, ymax=9,
ytick style={color=black}
]
\addplot [thick, steelblue31119180]
table {%
	0.0164169 0.0175460837606707
	0.0264169 0.059333471062033
	0.0364169 0.159384063899142
	0.0464169 0.343128272185818
	0.0564169 0.60193244662882
	0.0664169 0.886723488255902
	0.0764169 1.15188737214371
	0.0864169 1.40666006238532
	0.0964169 1.70912821154231
	0.1064169 2.11100092394681
	0.1164169 2.60697631334594
	0.1264169 3.1231509935091
	0.1364169 3.56560454797027
	0.1464169 3.90421265495984
	0.1564169 4.20163109347358
	0.1664169 4.55160438491231
	0.1764169 5.00619423170213
	0.1864169 5.55324849140295
	0.1964169 6.11108269169856
	0.2064169 6.53768080152981
	0.2164169 6.70844218860071
	0.2264169 6.6026871513277
	0.2364169 6.27136157928288
	0.2464169 5.73865696960602
	0.2564169 4.99539178978778
	0.2664169 4.07826898637024
	0.2764169 3.10907925036128
	0.2864169 2.2474149225415
	0.2964169 1.60851868215538
	0.3064169 1.21401632129444
	0.3164169 1.00142096906169
	0.3264169 0.869324649489276
	0.3364169 0.729878749361931
	0.3464169 0.549527311370347
	0.3564169 0.353390400585733
	0.3664169 0.188595508679342
	0.3764169 0.082036809073504
	0.3864169 0.0287305299791581
	0.3964169 0.00802870238830816
	0.4064169 0.00177837297781081
	0.4164169 0.000310696823420192
	0.4264169 4.26602546600696e-05
	0.4364169 4.59148338250004e-06
};
\end{groupplot}

\end{tikzpicture}
\end{figure}

$\rho _{Y_{w}^{\ast },Y_{h}^{\ast }}$ is an economically important parameter as it captures sorting via the correlation between the unobservables generating the husband's and wife's wages. If individuals are sorting with respect to these unobservables we expect this parameter to be positive. Moreover, as this parameter also captures the prices of these unobservables it will capture similar structural effects associated with the implicit prices of observed characteristics. As noted above, this parameter will capture factors related to influences, such as cost of living or wage premia, which are not captured by the conditioning variables but which are
shared by spouses. For example, while the data measure various aspects of
the location in which the spouses live we are unable to distinguish between
those living in costly urban areas and those who are not living in such
areas. It will also capture other factors specific to unobserved issues
shared by the spouses. For example, it captures that both may work for the
same firm and or went to the same college and share the costs and/or
benefits associated with their wages. The estimate of this parameter is
positive and generally ranges between 0.2 and 0.3. Figure \ref%
{fig:kernelModelComplete} plots the estimated densities for this parameter
for the starting and ending time periods for all 100 models while Figures %
\ref{fig:covariance_q1}-\ref{fig:covariance_q3} present the estimates at the
various quantiles we consider. There is some evidence that it is increasing
at these different quantiles over time but the confidence intervals do not
appear to reject that the impact is constant.

With respect to sorting the only clear and consistent evidence is associated
with $\rho _{Y_{w}^{\ast },Y_{h}^{\ast }}.$ The evidence from the
counterfactuals based on the BDR estimates revealed that setting $\rho
_{Y_{w}^{\ast },Y_{h}^{\ast }}$ to zero drastically reduced the presence of
positive sorting. However, this may now be offset by allowing for selection
even though the estimates of $\rho _{D_{w}^{\ast },D_{h}^{\ast }},$ $\rho
_{D_{w}^{\ast },Y_{w}^{\ast }}$, $\rho _{D_{h}^{\ast },Y_{h}^{\ast }}$, $%
\rho _{D_{w}^{\ast },Y_{h}^{\ast }}$ and $\rho _{D_{h}^{\ast },Y_{w}^{\ast
}} $ do not individually or collectively present a clear story regarding
sorting.

The estimates of the parameters related to the husband's selection decision
warrant further discussion. As male participation is typically treated as
exogenous there is little empirical work on the impact of male selection on
wages. Moreover, as methods to analyze the role of selection at different
quantiles have only recently been developed there is little existing
evidence on how it varies across the wage distribution. However, the
parameters are not precisely estimated in this setting and in some instances
may be unreasonable in magnitude. This is likely to reflect the form of
identification employed. Alternatively, husbands' work decisions may be
exogenous. Estimating the model while treating husband employment as exogenous is a restricted and special case
of the richer model in which both spouse's employment decisions are endogenous. Accordingly, we primarily focus on the results for the unrestricted 
model. However, we also estimated the restricted model treating the husbands' work decisions
as exogenous and conducted the corresponding counterfactuals. We comment on
those results below. However, to anticipate the main finding, the treatment
of husbands' work decisions as exogenous did not alter the substantive
results. We report the estimates of $\rho _{D_{w}^{\ast },Y_{w}^{\ast }}$, $%
\rho _{D_{w}^{\ast },Y_{h}^{\ast }}$ and $\rho _{Y_{w}^{\ast },Y_{h}^{\ast
}} $ from this restricted model in Figures \ref{fig:results_bootstraps_25_1}-%
\ref{fig:results_bootstraps_75_1} of Section \ref{web_appendix::restricted} in our Web Appendix.

\subsection{Counterfactual Sorting Patterns}

\label{ss:counterfactual_results}

The estimated model's capacity to reproduce the frequencies in Tables \ref%
{tab:frequencies_1} and \ref{tab:frequencies_2} is reflected in Tables \ref%
{tab:results_original_1976} and \ref{tab:results_original_2018} and,
although the large number of cells makes it difficult to draw a conclusion
by a visual inspection, the estimated cells appear close to the true values.
We now examine the role of various model parameters in generating the
observed data by changing selected model parameters and examining the
predicted allocation across cells. We compare these counterfactuals to
Tables \ref{tab:results_original_1976} and \ref{tab:results_original_2018}
as they represent the predictions of our estimated model.

\begin{table}[h!]
\resizebox{15cm}{!}{
	\begin{tabular}{l||l|rrrrrrrrrr}
		\hline\hline
		&  & \multicolumn{10}{c}{Husbands' quantiles}   \\ \hline\hline
		&  & d1 & d2 & d3 & d4 & d5 & d6 & d7 & d8 & d9 & d10   \\ \hline
		\parbox[t]{2mm}{\multirow{10}{*}{\rotatebox[origin=c]{90}{Wives' quantiles}}}
	&d1	&	\textbf{1.998}	&	1.257	&	1.035	&	0.935	&	0.861	&	0.873	&	0.835	&	0.721	&	0.792	&	0.692	\\	
	&d2	&	1.648	&	\textbf{1.349}	&	1.204	&	1.108	&	0.996	&	0.851	&	0.871	&	0.782	&	0.836	&	0.308	\\	
	&d3	&	1.239	&	1.386	&	\textbf{1.083}	&	1.091	&	1.131	&	0.841	&	0.898	&	0.862	&	0.831	&	0.691	\\	
	&d4	&	1.194	&	1.384	&	1.159	&	\textbf{1.122}	&	0.926	&	0.895	&	0.935	&	0.912	&	0.819	&	0.651	\\	
	&d5	&	0.959	&	1.155	&	1.187	&	1.144	&	\textbf{1.084}	&	0.934	&	0.979	&	0.955	&	0.946	&	0.663	\\	
	&d6	&	0.795	&	1.061	&	1.062	&	1.125	&	1.071	&	\textbf{1.079}	&	1.029	&	0.983	&	0.965	&	0.87	\\	
	&d7	&	0.66	&	0.812	&	1.084	&	1.091	&	1.09	&	1.07	&	\textbf{1.088}	&	1.115	&	1.162	&	0.784	\\	
	&d8	&	0.644	&	0.645	&	0.808	&	1.15	&	1.042	&	1.046	&	1.154	&	\textbf{1.084}	&	1.179	&	1.252	\\	
	&d9	&	0.49	&	0.621	&	0.686	&	0.744	&	1.136	&	1.111	&	1.235	&	1.214	&	\textbf{1.155}	&	1.612	\\	
	&d10	&	0.373	&	0.331	&	0.693	&	0.489	&	0.663	&	1.301	&	0.978	&	1.372	&	1.315	&	\textbf{2.476}		   \\ \hline\hline
	\end{tabular}}
\caption{Results of estimated sorting measures in 1976-1980 at different
quantiles when $\protect\rho_{D^*_w, D^*_h} = 0$.}
\label{tab:results_1_1976}
\end{table}

\begin{table}[h!]
\resizebox{15cm}{!}{
	\begin{tabular}{l||l|rrrrrrrrrr}
		\hline\hline
		&  & \multicolumn{10}{c}{Husbands' quantiles}   \\ \hline\hline
		&  & d1 & d2 & d3 & d4 & d5 & d6 & d7 & d8 & d9 & d10   \\ \hline
		\parbox[t]{2mm}{\multirow{10}{*}{\rotatebox[origin=c]{90}{Wives' quantiles}}}
	&d1	&	\textbf{2.03}	&	1.16	&	1.018	&	0.872	&	0.948	&	0.788	&	0.942	&	0.746	&	0.654	&	0.835	\\	
	&d2	&	1.588	&	\textbf{1.482}	&	1.189	&	1.18	&	0.857	&	0.941	&	0.83	&	0.852	&	0.62	&	0.457	\\	
	&d3	&	1.269	&	1.57	&	\textbf{1.281}	&	1.027	&	1.039	&	0.892	&	0.949	&	0.812	&	0.704	&	0.461	\\	
	&d4	&	0.958	&	1.397	&	1.323	&	\textbf{1.194}	&	1.083	&	1.006	&	0.935	&	0.879	&	0.76	&	0.401	\\	
	&d5	&	0.793	&	1.128	&	1.324	&	1.208	&	\textbf{1.224}	&	1.037	&	1.011	&	0.961	&	0.868	&	0.52	\\	
	&d6	&	0.819	&	0.756	&	1.25	&	1.16	&	1.076	&	\textbf{1.128}	&	1.026	&	1.135	&	0.945	&	0.71	\\	
	&d7	&	0.643	&	0.911	&	0.865	&	1.288	&	1.201	&	1.145	&	\textbf{1.089}	&	0.994	&	1.103	&	0.77	\\	
	&d8	&	0.76	&	0.671	&	0.672	&	0.882	&	1.265	&	1.209	&	1.165	&	\textbf{1.114}	&	1.115	&	1.18	\\	
	&d9	&	0.526	&	0.53	&	0.655	&	0.718	&	0.742	&	1.131	&	1.121	&	1.282	&	\textbf{1.607}	&	1.641	\\	
	&d10	&	0.616	&	0.395	&	0.423	&	0.47	&	0.565	&	0.724	&	0.933	&	1.224	&	1.623	&	\textbf{3.025}		  \\ \hline\hline
	\end{tabular}}
\caption{Results of estimated sorting measures in 1976-1980 at different
quantiles when $\protect\rho_{D^*_w, D^*_h} = 0, \protect\rho_{D^*_w Y^*_w}
= 0, \protect\rho_{D^*_w, Y^*_h} = 0, \protect\rho_{D^*_h, Y^*_w} = 0, 
\protect\rho_{D^*_h, Y^*_h} = 0$.}
\label{tab:results_3_1976}
\end{table}

\begin{table}[h!]
\resizebox{15cm}{!}{
	\begin{tabular}{l||l|rrrrrrrrrr}
		\hline\hline
		&  & \multicolumn{10}{c}{Husbands' quantiles}  \\ \hline\hline
		&  & d1 & d2 & d3 & d4 & d5 & d6 & d7 & d8 & d9 & d10   \\ \hline
		\parbox[t]{2mm}{\multirow{10}{*}{\rotatebox[origin=c]{90}{Wives' quantiles}}}
	&d1	&	\textbf{2.1}	&	1.113	&	1.018	&	0.867	&	0.904	&	0.795	&	0.914	&	0.753	&	0.665	&	0.877	\\	
	&d2	&	1.696	&	\textbf{1.506}	&	1.221	&	1.117	&	0.891	&	0.889	&	0.844	&	0.798	&	0.617	&	0.417	\\	
	&d3	&	1.096	&	1.71	&	\textbf{1.244}	&	1.1	&	1.057	&	0.917	&	0.915	&	0.843	&	0.649	&	0.465	\\	
	&d4	&	0.947	&	1.295	&	1.436	&	\textbf{1.234}	&	1.059	&	1.003	&	0.958	&	0.906	&	0.767	&	0.403	\\	
	&d5	&	0.788	&	1.141	&	1.276	&	1.159	&	\textbf{1.227}	&	0.99	&	1.081	&	0.941	&	0.848	&	0.546	\\	
	&d6	&	0.887	&	0.75	&	1.231	&	1.146	&	1.052	&	\textbf{1.139}	&	0.959	&	1.15	&	1.017	&	0.678	\\	
	&d7	&	0.566	&	0.875	&	0.916	&	1.265	&	1.271	&	1.157	&	\textbf{1.103}	&	0.985	&	1.126	&	0.757	\\	
	&d8	&	0.776	&	0.663	&	0.651	&	0.864	&	1.238	&	1.258	&	1.158	&	\textbf{1.082}	&	1.102	&	1.177	\\	
	&d9	&	0.5	&	0.52	&	0.574	&	0.762	&	0.752	&	1.126	&	1.116	&	1.362	&	\textbf{1.602}	&	1.71	\\	
	&d10	&	0.644	&	0.428	&	0.434	&	0.487	&	0.548	&	0.726	&	0.952	&	1.179	&	1.607	&	\textbf{2.971}		  \\ \hline\hline
	\end{tabular}}
\caption{Results of estimated sorting measures in 1976-1980 at different
quantiles when $\protect\rho_{D^*_w, D^*_h} = 0, \protect\rho_{D^*_w, Y^*_w}
= 0, \protect\rho_{D^*_w, Y^*_h} = 0, \protect\rho_{D^*_h, Y^*_w} = 0, 
\protect\rho_{D^*_h, Y^*_h} = 0$ and distribution of $Z$ as in total
population.}
\label{tab:results_4_1976}
\end{table}

\begin{table}[h!]
\resizebox{15cm}{!}{
	\begin{tabular}{l||l|rrrrrrrrrr}
		\hline\hline
		&  & \multicolumn{10}{c}{Husbands' quantiles}   \\ \hline\hline
		&  & d1 & d2 & d3 & d4 & d5 & d6 & d7 & d8 & d9 & d10   \\ \hline
		\parbox[t]{2mm}{\multirow{10}{*}{\rotatebox[origin=c]{90}{Wives' quantiles}}}
	&d1	&	\textbf{1.34}	&	1.191	&	1.073	&	1.013	&	0.963	&	0.949	&	0.931	&	0.893	&	0.779	&	0.875	\\	
	&d2	&	1.244	&	\textbf{1.175}	&	1.087	&	1.024	&	0.98	&	0.962	&	0.941	&	0.898	&	0.819	&	0.861	\\	
	&d3	&	1.14	&	1.108	&	\textbf{1.066}	&	1.036	&	0.993	&	0.982	&	0.97	&	0.944	&	0.882	&	0.876	\\	
	&d4	&	1.02	&	1.049	&	1.063	&	\textbf{1.05}	&	1.016	&	0.997	&	0.979	&	0.977	&	0.948	&	0.905	\\	
	&d5	&	0.971	&	1.005	&	1.021	&	1.021	&	\textbf{1.005}	&	1.0	&	0.999	&	1.003	&	0.996	&	0.972	\\	
	&d6	&	0.962	&	0.996	&	1.018	&	1.01	&	1.025	&	\textbf{1.02}	&	1.005	&	0.984	&	1.018	&	0.965	\\	
	&d7	&	0.863	&	0.934	&	0.988	&	1.018	&	1.036	&	1.028	&	\textbf{1.026}	&	1.038	&	1.07	&	1.018	\\	
	&d8	&	0.851	&	0.871	&	0.907	&	0.965	&	0.983	&	1.01	&	1.052	&	\textbf{1.114}	&	1.113	&	1.104	\\	
	&d9	&	0.754	&	0.82	&	0.903	&	0.947	&	1.013	&	1.044	&	1.062	&	1.09	&	\textbf{1.246}	&	1.149	\\	
	&d10	&	0.854	&	0.851	&	0.875	&	0.916	&	0.985	&	1.009	&	1.035	&	1.06	&	1.13	&	\textbf{1.275}		 \\ \hline\hline
	\end{tabular}}
\caption{Results of estimated sorting measures in 1976-1980 at different
quantiles when all correlations are zero and distribution of $Z$ as in total
population.}
\label{tab:results_5_1976}
\end{table}

We begin by setting $\rho _{D_{w}^{\ast },D_{h}^{\ast }}$ to zero. 
%This parameter is interesting if one interprets the
%unobservables as capturing ability. If so, the positive correlation would
%suggest that the working talented are couples. 
Table \ref{tab:results_1_1976} presents the results \ for 1976-1980 and
Table \ref{tab:results_1_2018} those for 2021-2025. They appear similar to
Table \ref{tab:results_original_1976} and Table \ref%
{tab:results_original_2018} respectively and we conclude that this parameter
is unimportant for determining the joint wage distribution.

\begin{table}[h!]
\resizebox{15cm}{!}{
	\begin{tabular}{l||l|rrrrrrrrrrr}
		\hline\hline
		&  & \multicolumn{10}{c}{Husbands' quantiles}   \\ \hline\hline
		&  & d1 & d2 & d3 & d4 & d5 & d6 & d7 & d8 & d9 & d10   \\ \hline
		\parbox[t]{2mm}{\multirow{10}{*}{\rotatebox[origin=c]{90}{Wives' quantiles}}}
	&d1	&	\textbf{2.778}	&	1.626	&	1.185	&	0.94	&	0.876	&	0.772	&	0.589	&	0.673	&	0.532	&	0.028	\\	
	&d2	&	1.714	&	\textbf{1.461}	&	1.21	&	1.205	&	0.888	&	0.828	&	0.833	&	0.607	&	0.503	&	0.749	\\	
	&d3	&	1.169	&	1.666	&	\textbf{1.328}	&	1.136	&	1.163	&	0.906	&	0.862	&	0.787	&	0.518	&	0.466	\\	
	&d4	&	0.939	&	1.339	&	1.62	&	\textbf{1.191}	&	0.944	&	1.01	&	0.888	&	0.847	&	0.787	&	0.436	\\	
	&d5	&	0.793	&	0.952	&	1.157	&	1.34	&	\textbf{1.009}	&	0.979	&	0.983	&	1.041	&	0.896	&	0.849	\\	
	&d6	&	0.748	&	0.714	&	1.139	&	1.313	&	1.415	&	\textbf{1.278}	&	1.089	&	0.892	&	0.973	&	0.439	\\	
	&d7	&	0.681	&	0.685	&	0.772	&	0.85	&	0.95	&	1.507	&	\textbf{1.182}	&	1.095	&	0.996	&	1.282	\\	
	&d8	&	0.567	&	0.638	&	0.662	&	0.806	&	0.886	&	1.286	&	1.456	&	\textbf{1.371}	&	1.471	&	0.855	\\	
	&d9	&	0.577	&	0.526	&	0.641	&	0.654	&	0.74	&	0.977	&	1.225	&	1.544	&	\textbf{1.469}	&	1.647	\\	
	&d10	&	0.046	&	0.401	&	0.292	&	0.57	&	1.126	&	0.461	&	0.893	&	1.141	&	1.845	&	\textbf{3.225}	\\	 \hline\hline
	\end{tabular}}
\caption{Results of estimated sorting measures in 2021-2025 at different
quantiles when $\protect\rho_{D^*_w, D^*_h} = 0$.}
\label{tab:results_1_2018}
\end{table}

\begin{table}[h!]
\resizebox{15cm}{!}{
	\begin{tabular}{l||l|rrrrrrrrrrr}
		\hline\hline
		&  & \multicolumn{10}{c}{Husbands' quantiles}   \\ \hline\hline
		&  & d1 & d2 & d3 & d4 & d5 & d6 & d7 & d8 & d9 & d10   \\ \hline
		\parbox[t]{2mm}{\multirow{10}{*}{\rotatebox[origin=c]{90}{Wives' quantiles}}}
	&d1	&	\textbf{2.522}	&	1.404	&	1.119	&	1.077	&	0.88	&	0.811	&	0.76	&	0.62	&	0.496	&	0.311	\\	
	&d2	&	1.522	&	\textbf{1.663}	&	1.423	&	1.33	&	0.983	&	0.98	&	0.729	&	0.524	&	0.391	&	0.454	\\	
	&d3	&	0.979	&	1.9	&	\textbf{1.371}	&	1.084	&	0.878	&	1.003	&	0.867	&	0.803	&	0.695	&	0.42	\\	
	&d4	&	0.928	&	0.954	&	1.515	&	\textbf{1.164}	&	1.193	&	1.087	&	1.052	&	0.891	&	0.656	&	0.559	\\	
	&d5	&	0.747	&	0.939	&	1.24	&	1.328	&	\textbf{1.211}	&	1.032	&	1.085	&	0.967	&	0.699	&	0.752	\\	
	&d6	&	0.871	&	0.797	&	0.951	&	1.151	&	1.383	&	\textbf{0.875}	&	1.243	&	1.17	&	0.955	&	0.604	\\	
	&d7	&	0.663	&	0.656	&	0.588	&	1.046	&	1.427	&	1.306	&	\textbf{0.975}	&	1.188	&	1.151	&	1.0	\\	
	&d8	&	0.768	&	0.652	&	0.7	&	0.833	&	0.735	&	1.256	&	1.144	&	\textbf{1.241}	&	1.408	&	1.261	\\	
	&d9	&	0.665	&	0.465	&	0.644	&	0.549	&	0.707	&	0.997	&	1.314	&	1.523	&	\textbf{1.801}	&	1.335	\\	
	&d10	&	0.337	&	0.571	&	0.449	&	0.436	&	0.601	&	0.65	&	0.829	&	1.072	&	1.747	&	\textbf{3.308}	\\ \hline\hline
	\end{tabular}}
\caption{Results of estimated sorting measures in 2021-2025 at different
quantiles when $\protect\rho_{D^*_w, D^*_h} = 0, \protect\rho_{D^*_w Y^*_w}
= 0, \protect\rho_{D^*_w, Y^*_h} = 0, \protect\rho_{D^*_h, Y^*_w} = 0, 
\protect\rho_{D^*_h, Y^*_h} = 0$.}
\label{tab:results_3_2018}
\end{table}

\begin{table}[h!]
\resizebox{15cm}{!}{
	\begin{tabular}{l||l|rrrrrrrrrr}
		\hline\hline
		&  & \multicolumn{10}{c}{Husbands' quantiles}   \\ \hline\hline
		&  & d1 & d2 & d3 & d4 & d5 & d6 & d7 & d8 & d9 & d10   \\ \hline
		\parbox[t]{2mm}{\multirow{10}{*}{\rotatebox[origin=c]{90}{Wives' quantiles}}}
	&d1	&	\textbf{2.668}	&	1.398	&	1.141	&	1.024	&	0.859	&	0.737	&	0.742	&	0.6	&	0.467	&	0.365	\\	
	&d2	&	1.417	&	\textbf{1.757}	&	1.355	&	1.395	&	0.902	&	1.064	&	0.759	&	0.447	&	0.451	&	0.454	\\	
	&d3	&	1.005	&	1.829	&	\textbf{1.348}	&	1.235	&	1.0	&	0.86	&	0.797	&	0.872	&	0.593	&	0.461	\\	
	&d4	&	0.779	&	1.09	&	1.445	&	\textbf{1.093}	&	0.94	&	1.304	&	1.137	&	0.885	&	0.743	&	0.584	\\	
	&d5	&	0.812	&	0.827	&	1.295	&	1.45	&	\textbf{1.281}	&	0.984	&	1.006	&	0.974	&	0.681	&	0.692	\\	
	&d6	&	0.859	&	0.686	&	0.975	&	1.095	&	1.495	&	\textbf{0.95}	&	1.239	&	1.046	&	0.953	&	0.702	\\	
	&d7	&	0.626	&	0.799	&	0.581	&	0.968	&	1.293	&	1.308	&	\textbf{1.01}	&	1.284	&	1.167	&	0.964	\\	
	&d8	&	0.776	&	0.541	&	0.819	&	0.751	&	0.808	&	1.235	&	1.2	&	\textbf{1.25}	&	1.387	&	1.234	\\	
	&d9	&	0.626	&	0.489	&	0.598	&	0.575	&	0.775	&	0.942	&	1.297	&	1.596	&	\textbf{1.793}	&	1.309	\\	
	&d10	&	0.433	&	0.583	&	0.446	&	0.414	&	0.645	&	0.618	&	0.814	&	1.048	&	1.764	&	\textbf{3.235}	\\	
    \hline\hline
	\end{tabular}}
\caption{Results of estimated sorting measures in 2021-2025 at different
quantiles when $\protect\rho_{D^*_w, D^*_h} = 0, \protect\rho_{D^*_w, Y^*_w}
= 0, \protect\rho_{D^*_w, Y^*_h} = 0, \protect\rho_{D^*_h, Y^*_w} = 0, 
\protect\rho_{D^*_h, Y^*_h} = 0$ and distribution of $Z$ as in total
population.}
\label{tab:results_4_2018}
\end{table}

\begin{table}[h!]
\resizebox{15cm}{!}{
	\begin{tabular}{l||l|rrrrrrrrrr}
		\hline\hline
		&  & \multicolumn{10}{c}{Husbands' quantiles}   \\ \hline\hline
		&  & d1 & d2 & d3 & d4 & d5 & d6 & d7 & d8 & d9 & d10   \\ \hline
		\parbox[t]{2mm}{\multirow{10}{*}{\rotatebox[origin=c]{90}{Wives' quantiles}}}
	&d1	&	\textbf{1.471}	&	1.212	&	1.102	&	1.09	&	0.913	&	0.944	&	0.865	&	0.803	&	0.803	&	0.798	\\	
	&d2	&	1.145	&	\textbf{1.123}	&	1.089	&	1.054	&	0.992	&	0.977	&	0.949	&	0.913	&	0.893	&	0.866	\\	
	&d3	&	1.129	&	1.117	&	\textbf{1.103}	&	1.066	&	1.004	&	1.001	&	0.963	&	0.927	&	0.86	&	0.83	\\	
	&d4	&	1.042	&	1.066	&	1.074	&	\textbf{1.041}	&	1.013	&	1.015	&	0.981	&	0.96	&	0.927	&	0.879	\\	
	&d5	&	0.952	&	1.004	&	1.031	&	1.02	&	\textbf{1.032}	&	1.031	&	1.024	&	1.015	&	0.957	&	0.934	\\	
	&d6	&	0.927	&	0.982	&	1.0	&	1.008	&	1.03	&	\textbf{1.029}	&	1.031	&	1.019	&	0.993	&	0.981	\\	
	&d7	&	0.87	&	0.912	&	0.935	&	0.972	&	1.01	&	1.016	&	\textbf{1.047}	&	1.072	&	1.085	&	1.082	\\	
	&d8	&	0.883	&	0.918	&	0.941	&	0.967	&	1.001	&	1.019	&	1.042	&	\textbf{1.061}	&	1.099	&	1.069	\\	
	&d9	&	0.744	&	0.824	&	0.872	&	0.904	&	1.029	&	0.994	&	1.075	&	1.137	&	\textbf{1.17}	&	1.252	\\	
	&d10	&	0.838	&	0.843	&	0.853	&	0.878	&	0.976	&	0.972	&	1.024	&	1.09	&	1.214	&	\textbf{1.311}	\\	\hline\hline
	\end{tabular}}
\caption{Results of estimated sorting measures in 2021-2025 at different
quantiles when all correlations are zero and distribution of $Z$ as in total
population.}
\label{tab:results_5_2018}
\end{table}

Tables \ref{tab:results_3_1976} and \ref{tab:results_3_2018} present the
results when $\rho _{D_{w}^{\ast },Y_{w}^{\ast }},$ $\rho _{D_{h}^{\ast
},Y_{h}^{\ast }}$, $\rho _{D_{w}^{\ast },Y_{h}^{\ast }}$ and $\rho
_{D_{h}^{\ast },Y_{w}^{\ast }}$ are also set to zero. We characterize these
parameters as collectively capturing the selection process. For the earlier
period the d10/d10 cell increases from 2.48 to 3.03 while the sum in the
(d9+d10)/(d9+d10) cells increases from 6.56 to 7.90. This suggests that
negative selection bias operating in both labor markets is reducing
household inequality. Negative selection is removing the relatively higher
paid males and females from FTFY employment. Setting these parameters to
zero inserts more higher paid workers into the market. This results in a
higher propensity of the higher paid males and females to marry. This is an
interesting result although it may reflect the identification strategy.
Interestingly the frequencies in Table \ref{tab:results_3_2018} provide a
somewhat comparable story although the changes are less drastic. Setting the
other parameters to zero has little effect on the lower cells.
However, while the d10/d10 cell increases marginally the sum in the
(d9+d10)/(d9+d10) cells stay identical at 8.19.

We also examine the counterfactual in which the above parameters are set to
zero and we use the $Z^{\prime }$s of the whole sample rather than only
those who are working. The results are in Tables \ref{tab:results_4_1976}
and \ref{tab:results_4_2018} and a comparison with Tables \ref%
{tab:results_3_1976} and \ref{tab:results_3_2018} indicates they do not
affect the observed sorting patterns. Although previous studies (see, for
example, CFL 2019) have analyzed the impact of selection on the distribution
of wages and have found it to have some impact, our results indicate it
would not change the sorting patterns even if it is changing the
distribution of both, or either of, the male and female wages.

Finally we set $\rho _{Y_{w}^{\ast },Y_{h}^{\ast }}$ to zero. For the BDR
model the corresponding counterfactual produced a substantial reduction in
sorting. Table \ref{tab:results_5_1976} reports the results for the earlier
time period. Given the large number of cells it is useful to focus on the
extreme cases as these are the outcomes more closely associated with
inequality. First consider the d1/d1 and d10/d10 cells. The former decreases
from 2.1 to 1.34 while the latter decreases even more dramatically from
2.97 to 1.28. Expanding the 4 lower and upper combinations to include d2 and
d9 respectively produces reductions from 6.42 and 7.89 to 4.95 and 4.82
respectively. It is interesting that the more substantial reductions appear
to occur at the cells for the higher wage deciles. Turning to the 2021-2025
time period we see a similar pattern. The d1/d1 cell decreases from 2.88 to
1.54 and d10/d10 decreased from 2.67 to 1.47. Expanding to the four lowest
and highest we see reductions of 7.24 to 4.95 and 8.10 to 4.95 respectively.
This clearly suggests that the correlation in the unexplained components of
wages is largely driving the observed sorting patterns. That is, sorting on
unobservables is an important factor in driving inequality.

\begin{table}[tbp]
\caption{Kendall and Spearman rank correlation coefficient.}
\label{tab:kendall}\centering
\begin{tabular}{lrr}
	\hline\hline
	& 1976-1980 & 2021-2025 \\
	\hline\hline
	Kendall's tau \\ \hline
	 Data	&	 0.2067	&	 0.3246	\\
	 Estimated	&	 0.1924	&	 0.2954	\\
	 Participation ($\rho_{D^*_w, D^*_h} = 0$)	&	 0.1910	&	 0.2852	\\
	 Selection ($\rho_{D^*_w, Y^*_w}= \rho_{D^*_w, Y^*_h}= \rho_{D^*_h, Y^*_h} = 0$)		&	 0.2136	&	 0.2672	\\
	 $Z$ as in whole population	&	 0.2158	&	 0.2659	\\
	 $\rho_{Y^*_w, Y^*_h} = 0$	&	 0.0687	&	 0.0770	\\
	\hline
	Spearman's rank correlation index \\ \hline
	 Data	&	 0.2767	&	 0.4244	\\
	 Estimated	&	 0.2579	&	 0.3896	\\
	 Participation ($\rho_{D^*_w, D^*_h} = 0$)	&	 0.2561	&	 0.3772	\\
	 Selection ($\rho_{D^*_w, Y^*_w}= \rho_{D^*_w, Y^*_h}= \rho_{D^*_h, Y^*_h} = 0$)	&	 0.2826	&	 0.3525	\\
	 $Z$ as in whole population	&	 0.2849	&	 0.3498	\\
	 $\rho_{Y^*_w, Y^*_h} = 0$	&	 0.0933	&	 0.1045	\\
	\hline\hline
\end{tabular}
\end{table}

To capture the sorting behavior which includes the off diagonals we compute
the Kendall's tau and the Spearman rank correlation coefficient for these counterfactuals. Note that the latter corresponds to the coefficient from regressing the rank of the husband on the rank of the wife. These are reported in Table \ref{tab:kendall} and confirm the patterns described above. Consider the different estimates of Kendall's tau corresponding to different scenarios. For the 1976-1980 period neither the selection parameters nor the use of the working or total population composition of $Z^{\prime }$s affect the rank correlation coefficient. For each of these experiments the rank correlation coefficient does not differ greatly from its value for the data of 0.207. However, when we additionally set the value of $\rho _{Y_{w}^{\ast},Y_{h}^{\ast }}$ to zero the value falls to 0.069. For the 2021-2025 period the results are similar despite the rank correlation coefficient of 0.325 being 50\% higher than the 1976-1980 value. For this period the only counterfactual which generates a different pattern of sorting is that corresponding to $\rho _{Y_{w}^{\ast },Y_{h}^{\ast }}=0.$ For this latter period the reduced rank correlation coefficient is 0.077. This is consistent with Tables \ref{tab:results_5_1976} and \ref{tab:results_5_2018}. We obtain similar conclusions when using the Spearman's rank correlation index.

We noted above that we also estimated the model treating the male employment decision
as exogenous. Although we do not report the results here, we conducted the same
counterfactuals for this restricted model. The general
flavor of the simulations was similar although there was stronger evidence
of selection affecting the observed sorting for the earlier period. However,
the result that setting $\rho _{Y_{w}^{\ast },Y_{h}^{\ast }}$ to zero
reduced the observed level of positive sorting was also clearly supported.

\subsection{Decomposing Changes in the Sorting Patterns}

We now focus on the changes in the sorting patterns in the earlier tables
and examine how the probability of both spouses being in a specific quantile
of their wage distribution changes over our sample period. We decompose the
total change into structural, selection and composition effects. Note that
unlike the previous section, we now also examine the impact of the changes
in the $\beta ^{\prime }$s over time. As we employ a local multivariate normal
approximation for each point of the four dimensional partition of the data
we can use the model parameter estimates to decompose the observed changes
into the various components. While the composition effects are defined as
the changes in the conditioning variables, the allocation of the other
parameters are less straightforward. 
%While it is clear that the slope coefficients should be
%interpreted as structural parameters the designation of the parameters
%capturing the covariances is less straightforward. 
The changes in the $\beta ^{\prime }$s can be assigned the interpretation of
conventional structural effects. We assign $\rho _{D_{w}^{\ast },Y_{w}^{\ast
}}$, $\rho _{D_{h}^{\ast },Y_{h}^{\ast }},$ $\rho _{D_{w}^{\ast
},Y_{h}^{\ast }}$ and $\rho _{D_{h}^{\ast },Y_{w}^{\ast }}$ to the selection
effects. We isolate the impact of the $\rho _{Y_{w}^{\ast },Y_{h}^{\ast }}$
parameter as it captures the correlation between the unobservables which
impact the spouses' wages. While this is a type of structural effect we
separate it from the structural effects operating through the $\beta
^{\prime }$s. Although we can evaluate the sorting behavior observed in 100
cells it appears that the interesting behavior occurs in the top and bottom
cells. Accordingly we conduct the decompositions for both members of the
couple being located in the bottom and top deciles and quartiles. These
plots are presented in Figure \ref{fig:decomposition_sorting} noting that the estimates are based on five‑year periods, with the first  being 1976–1980 and the last 2021–2025. The
decompositions represent how the change in the measures can be explained by
changes in the various components. If a specific component appears to be
close to zero then this indicates that its contribution has remained
unchanged and not that it is zero.

\begin{figure}[tbp]
\caption{Decompositions of changes in the sorting measures.}
\label{fig:decomposition_sorting}\centering
\resizebox{\textwidth}{!}{
% This file was created with tikzplotlib v0.10.1.
\begin{tikzpicture}

\definecolor{crimson2143940}{RGB}{214,39,40}
\definecolor{darkgray176}{RGB}{176,176,176}
\definecolor{darkorange25512714}{RGB}{255,127,14}
\definecolor{forestgreen4416044}{RGB}{44,160,44}
\definecolor{mediumpurple148103189}{RGB}{148,103,189}
\definecolor{steelblue31119180}{RGB}{31,119,180}

\begin{groupplot}[group style={group size=2 by 2, vertical sep=2cm}]
\nextgroupplot[
tick align=outside,
tick pos=left,
title={Both spouses below D1},
xmin=1976, xmax=2021,
y grid style={darkgray176},
ylabel={Difference in probability (\(\displaystyle 	\times\) 100)},
ymin=-0.2781965, ymax=0.8010465,
ytick style={color=black}
]
\addplot [very thick, steelblue31119180]
table {%
1976 0
1981 0.29295
1986 0.44895
1991 0.43977
1996 0.54798
2001 0.4104
2006 0.61473
2011 0.41222
2016 0.53267
2021 0.69957
};
\addplot [very thick, darkorange25512714]
table {%
1976 0
1981 0.00398999999999994
1986 -0.12545
1991 0.0118799999999997
1996 0.00273000000000012
2001 0.0412600000000003
2006 0.0582099999999999
2011 0.0150999999999999
2016 0.0486800000000001
2021 0.000920000000000254
};
\addplot [very thick, forestgreen4416044]
table {%
1976 0
1981 0.14266
1986 0.48214
1991 0.55351
1996 0.60639
2001 0.44209
2006 0.47414
2011 0.65369
2016 0.52622
2021 0.75199
};
\addplot [very thick, crimson2143940]
table {%
1976 0
1981 0.16227
1986 0.16079
1991 -0.0129600000000001
1996 0.0688
2001 0.0858099999999999
2006 0.27904
2011 -0.0274299999999998
2016 0.17424
2021 0.17238
};
\addplot [very thick, mediumpurple148103189]
table {%
1976 0
1981 -0.0159699999999998
1986 -0.06853
1991 -0.11266
1996 -0.12994
2001 -0.15876
2006 -0.19666
2011 -0.22914
2016 -0.21647
2021 -0.22572
};

\nextgroupplot[
tick align=outside,
tick pos=left,
title={Both spouses below Q1},
xmin=1976, xmax=2021,
y grid style={darkgray176},
ymin=-1.437496875, ymax=2.954184375,
ytick style={color=black}
]
\addplot [very thick, steelblue31119180]
table {%
1976 0
1981 0.5798125
1986 1.170375
1991 1.3493125
1996 1.777
2001 1.1805625
2006 1.3188125
2011 1.35625
2016 1.3774375
2021 1.809875
};
\addplot [very thick, darkorange25512714]
table {%
1976 0
1981 0.0330625000000001
1986 0.0479375000000004
1991 -0.0917500000000002
1996 0.0369999999999995
2001 -0.0164374999999994
2006 0.00118749999999945
2011 0.0838749999999999
2016 0.0841874999999992
2021 -0.0444374999999997
};
\addplot [very thick, forestgreen4416044]
table {%
1976 0
1981 0.478687500000001
1986 0.773
1991 1.6885625
1996 2.3389375
2001 1.548125
2006 1.4768125
2011 2.7545625
2016 2.287625
2021 2.5016875
};
\addplot [very thick, crimson2143940]
table {%
1976 0
1981 0.087124999999999
1986 0.397312499999999
1991 -0.0907500000000006
1996 -0.3725625
2001 -0.140625
2006 0.0965624999999998
2011 -1.237875
2016 -0.79675
2021 -0.4544375
};
\addplot [very thick, mediumpurple148103189]
table {%
1976 0
1981 -0.0190625
1986 -0.047875
1991 -0.15675
1996 -0.226375000000001
2001 -0.210500000000001
2006 -0.25575
2011 -0.2443125
2016 -0.197625
2021 -0.1929375
};

\nextgroupplot[
tick align=outside,
tick pos=left,
title={Both spouses above Q3},
x grid style={darkgray176},
xlabel={years},
xmin=1976, xmax=2021,
xtick style={color=black},
y grid style={darkgray176},
ylabel={Difference in probability (\(\displaystyle 	\times\) 100)},
ymin=-0.68030625, ymax=2.08193125,
ytick style={color=black}
]
\addplot [very thick, steelblue31119180]
table {%
1976 0
1981 0.0921249999999998
1986 0.54675
1991 1.213625
1996 1.13075
2001 0.943187499999999
2006 1.1308125
2011 1.2893125
2016 1.7130625
2021 1.956375
};
\addplot [very thick, darkorange25512714]
table {%
1976 0
1981 -0.0625000000000001
1986 -0.1004375
1991 -0.00850000000000017
1996 -0.0911250000000002
2001 -0.0262500000000013
2006 -0.040312499999999
2011 -0.0569374999999997
2016 -0.0593750000000004
2021 -0.1415
};
\addplot [very thick, forestgreen4416044]
table {%
1976 0
1981 0.567624999999999
1986 1.064125
1991 1.857625
1996 1.809
2001 1.3956875
2006 1.5534375
2011 1.674125
2016 1.083
2021 0.827750000000001
};
\addplot [very thick, crimson2143940]
table {%
1976 0
1981 -0.366187499999999
1986 -0.3315625
1991 -0.5299375
1996 -0.55475
2001 -0.427187499999999
2006 -0.393625
2011 -0.401375
2016 0.5393125
2021 0.941
};
\addplot [very thick, mediumpurple148103189]
table {%
1976 0
1981 -0.0468125
1986 -0.085375
1991 -0.1055625
1996 -0.0323749999999998
2001 0.000937499999999203
2006 0.0113125000000006
2011 0.0735
2016 0.150125
2021 0.329125
};

\nextgroupplot[
legend cell align={left},
legend columns=5,
legend style={fill opacity=0.8, draw opacity=1, text opacity=1, at={(0.5,-0.4)}, anchor=north},
tick align=outside,
tick pos=left,
title={Both spouses above D9},
x grid style={darkgray176},
xlabel={years},
xmin=1976, xmax=2021,
xtick style={color=black},
y grid style={darkgray176},
ymin=-0.2385795, ymax=0.7586695,
ytick style={color=black}
]
\addplot [very thick, steelblue31119180]
table {%
1976 0
1981 -0.0905399999999998
1986 0.13197
1991 0.16309
1996 0.20118
2001 0.29859
2006 0.43003
2011 0.44493
2016 0.54828
2021 0.47113
};
\addlegendentry{total}
\addplot [very thick, darkorange25512714]
table {%
1976 0
1981 -0.0339399999999999
1986 -0.14397
1991 -0.09335
1996 -0.19325
2001 0.0288300000000001
2006 -0.0171599999999996
2011 -0.1355
2016 -0.13822
2021 -0.1917
};
\addlegendentry{selection}
\addplot [very thick, forestgreen4416044]
table {%
1976 0
1981 0.12503
1986 0.30487
1991 0.20908
1996 0.36502
2001 0.0581499999999999
2006 0.17183
2011 0.43083
2016 0.17994
2021 -0.0824199999999999
};
\addlegendentry{structural}
\addplot [very thick, crimson2143940]
table {%
1976 0
1981 -0.1493
1986 0.00524000000000013
1991 0.07152
1996 0.0484799999999996
2001 0.20746
2006 0.23821
2011 0.0975699999999997
2016 0.45303
2021 0.71334
};
\addlegendentry{rho23}
\addplot [very thick, mediumpurple148103189]
table {%
1976 0
1981 -0.03233
1986 -0.03417
1991 -0.0241599999999997
1996 -0.0190699999999997
2001 0.0041500000000001
2006 0.03715
2011 0.0520300000000002
2016 0.0535300000000003
2021 0.0319100000000003
};
\addlegendentry{composition}
\end{groupplot}

\end{tikzpicture}}
\end{figure}

We start with the probability of both spouses being located in the bottom decile and quartile. Over the 50 years of our sample the former increases by almost 70 percent and the latter by almost 40 percent. For the probability of both partners below D1, \emph{i.e.} d1, the large positive change is driven by changes in the structural effect and $\rho_{Y_{w}^{\ast },Y_{h}^{\ast }}.$ Changes in the selection effect are unimportant and the changes in the composition effect are negative. For the probability of both partners below Q1 the results are essentially the same although the changes in the contribution of $\rho _{Y_{w}^{\ast },Y_{h}^{\ast }}$ are generally unimportant except for a decrease towards the end of the sample. At the top of the distribution there is a similar finding. The probability of both in the top decile and quartile increases by around 50\% and 30\% respectively. The change in the probability of both spouses being above D9 is again driven
by changes in the structural effect and $\rho _{Y_{w}^{\ast },Y_{h}^{\ast }}$ although the change in the structural effect vanishes towards the end of the sample. The change in the probability of both spouses being above Q3 is due to changes in the structural effects. Changes in the composition and selection effects appear unimportant for either probabilities of both spouses above Q3 or D9. The results on the composition effects are consistent with the evidence as noted above in Breen and Salazar (2011), Gihleb and Lang (2018), Eika, Mogstad, and Zafar (2019), Chiappori, Costa-Dias and Meghir (2023) and Chiappori, Costa Dias, Meghir and Zhang (2025). Education is a component of the composition effect and there is no impact at either Q1 or Q3. At D1 and D9 it is less clear as both probabilities decline.

The increasing presence of married females into the FTFY in the 1970s and 1980s had little impact on the sorting process in the marriage market. Overall the selection effects are small. Second, composition effects are generally unimportant and have reduced positive sorting. This reflects that
the increased level of acquired education has increased everyone's
probability of marrying a relatively highly educated spouse. This is
consistent with the findings of others noted above regarding the impact of
composition effects on positive sorting as measured by education. The
clearest evidence is associated with the structural effects. These are the
most important contributors to the total effect\ and their interpretation is
also quite clear. Sorting behavior has not greatly changed in terms of the
observed characteristics of the spouses. However, the higher skill premia,
and the heterogeneity of the skill premia, have resulted in the highly
(lowly) paid being even more highly (lowly) paid. This is consistent with empirical evidence on decompositions of changes in wage inequality. The evidence on the role of the $\rho _{Y_{w}^{\ast },Y_{h}^{\ast }}$ provides a similar story. While this parameter captures the impact of unobservables it also reflects how they are priced. Our evidence suggests that the increased observed level of positive assortative sorting reflects the increases in prices of observables and unobservables. Individuals appear to marry the \textquotedblleft same" spouses but the value of their characteristics has
changed due to the increased skill premia. It is also possible that $\rho_{Y_{w}^{\ast },Y_{h}^{\ast }}$ partially captures how an individual's wage is influenced by the value of their spouse's characteristics. For example, individuals with highly successful spouses may benefit from exposure to their spouse's network and this may produce some relationship between their wage and, for example, their spouse's education level. Some preliminary investigations of the data suggest that this relationship exists, and that
it changes over time, but we do not pursue it here.

In light of this result highlighting the role of unobservables, it is useful to revisit what they, and their price, capture. It was noted above that one component may capture both spouses residing in the same urban area associated with a wage premium or both spouses attending the same elite educational institutions. It may also capture correlations across college major type or occupation of the spouses. In these instances the unobservables capture aspects of the data which are not measured and they are reflected in the unexplained component of wages. For example, Eika, Mogstad, and Zafar (2019) show for Norwegian data that the choice of college major is an important source of educational sorting. As college major is not measured in our data, this is the type of unobserved factor that would be captured in $\rho _{Y_{w}^{\ast },Y_{h}^{\ast }}$. The unobservables may also reflect that both spouses are high ability types or highly motivated.
%It is important to note that this parameter is a specific type of structural effect as it includes the price of unobservables. The increased price is consistent with a large body of evidence that returns to skills has increased and this is what is largely driving wage inequality. 
The evidence above is that there is not increased sorting on observables so there is no reason to conclude that sorting on unobservables has increased.\label{insert:section7} A more likely explanation is that the increased role of $\rho _{Y_{w}^{\ast },Y_{h}^{\ast }}$ captures an increased premium to unobserved skills. It is important to recall that this parameter is a specific type of structural effect as it includes the price of unobservables. The increased price is consistent with a large body of evidence (see, for example Juhn, Murphy and Pierce 1993, Card and Lemieux 1996, Acemoglu and Autor 2011, Acemoglu and Restrepo 2022 and Acemoglu and L\"obbing, 2024) that the returns to skills has increased and this has largely driven wage inequality. This is also consistent with our results reported here that the structural effects, capturing the prices of observable characteristics, have substantially contributed to the observed increased inequality. 

\subsection{Household Income Inequality}

We now explore the implication of the model's estimates for household income inequality. We assume that all individuals comprising the FTFY couples are working the same number of hours and we employ the household wage, defined as the sum of the spouses' hourly wages, as our measure of family income. We examine how inequality has evolved by examining the D8/D2 ratio under different counterfactuals. We focus on this ratio as the manner in which its components have changed is clear from the tables above. The results are shown in Table \ref{tab:inequality}. Note that we do not isolate the role of structural effects operating through the $\beta ^{\prime }$s nor the composition effects reflecting the changes in the $X^{\prime }$s over time. As noted in our introduction, there is a large existing literature which has established the impact of these factors on wage inequality. We repeat the same exercise as in the counterfactual sorting exercises of Section \ref{ss:counterfactual_results} in which we incrementally set the parameters capturing different features of the model to zero. We also explore the impact of employing the population characteristics rather than those of the working sample.

\begin{table}[tbp]
\caption{Decomposition of changes in inequality.}
\label{tab:inequality}\centering
\begin{small}
\begin{tabular}{l|rr|rr}
\hline\hline
& \multicolumn{2}{|c|}{1976-1980} & \multicolumn{2}{c}{2021-2025} \\ \hline
& D8/D2 & \multicolumn{1}{p{1.75cm}|}{Percentage change} & D8/D2 & \multicolumn{1}{p{1.75cm}}{Percentage change} \\ \hline\hline
	Data	&	1.7849	&	--	&	2.3028	&	--\\
	Estimated	&	1.7901	&	0.0	&	2.2554	&	0.0\\
	Participation ($\rho_{D^*_w, D^*_h} = 0$)	&	1.7901	&	-0.0	&	2.2634	&	0.35\\
	Selection of wife ($\rho_{D^*_w, Y^*_w} = 0$)	&	1.776	&	-0.79	&	2.2477	&	-0.35\\
	Selection of husband ($\rho_{D^*_h, Y^*_h} = \rho_{D^*_w, Y^*_h} = 0$)	&	1.786	&	-0.23	&	2.2563	&	0.04\\
	$Z$ as in whole population	&	1.7775	&	-0.7	&	2.2491	&	-0.28\\
	$\rho_{Y^*_w, Y^*_h}=0$	&	1.7284	&	-3.45	&	2.1917	&	-2.83\\
	Random sorting	&	1.6789	&	-6.21	&	2.0695	&	-8.24\\
\hline\hline
\end{tabular}
\end{small}
\end{table}

The first row of Table \ref{tab:inequality} presents the D8/D2 ratio for our two extreme sample periods. There is a large increase of 29.6\%, from 1.78 to 2.30, noting that it is similar to the growth in comparable measures of inequality presented in the discussion of the data. The second row presents the comparable estimates from the predictions of our model and indicates that the level of inequality observed in the data are generally maintained by our model's estimates although for the latter period the model underestimates the level of inequality as measured by the D8/D2 ratio. 
To examine the role of sorting we first estimate the corresponding measure in a setting
of unconditional random sorting in which husbands are assigned to wives in a random manner. This is shown in the table's bottom row. The results represent the D8/D2 ratio based on 100,000 draws. The estimates are 1.68 and 2.07 and reflect an increase of 22\% over the sample period. This large impact reflects the large increase in wage inequality which occurs for both husbands and wives. More interestingly, comparing the top and bottom rows of
this table indicates that sorting increases inequality by 7.0\% in the
earlier period and 13.1\% in the latter when compared to the actual data. This suggests that sorting has
substantially increased household inequality.

To investigate the role of the various model parameters, rows 3 to 7 repeat
the counterfactuals conducted above. Each of the
rows reflects the incremental change in the D8/D2 ratio as the
parameters associated with some feature of the model are set to zero. Note that the ``percentage change" reflects changes relative to the D8/D2 measure estimated for the full model reported in the second row. 
For the earlier period the impact of the work decisions, captured by selection, have slightly reduced inequality although we repeat
our concerns regarding this result as it is reliant on the potentially contentious form
of identification. For the latter period the role of selection and participation are also very small. The impact of using the population $%
Z^{\prime }$s rather than the working sample $Z^{\prime }$s confirms our
earlier results that this does not affect the results regarding
sorting. For both periods the effects are very small.

Row 7 highlights our recurring result that positive sorting is
operating through the correlation of the spouses' unobservables affecting wages. Setting $\rho _{Y_{w}^{\ast },Y_{h}^{\ast
}}=0$ decreases the D8/D2 ratio for the earlier period by 16\% from 2.06 to
1.73 and for the later period by 7\% from 2.26 to 2.12. This is a striking result as it supports the results above that the unobservables are driving the sorting patterns. Moreover, they are substantially contributing to the increases in inequality across households.

While there is an important role for $\rho _{Y_{w}^{\ast },Y_{h}^{\ast }}=0,$ the remaining large increase in inequality in the household wage appears due to changes in composition and structural effects. However, the existing evidence, here and in the existing literature, suggests the sorting on observables has not changed for this time period and composition effects are unimportant. This supports the results above that the increase in wage inequality across households is due to structural effects capturing changes in the skill premia.

\section{Robustness Checks}

\label{sec:robustness}

We now investigate the robustness of our more notable empirical results to various choices in our analysis. We do so by examining their impact on the associated values of Kendall's tau and Spearman rank correlation index and report them in Table \ref{tab:kendall_robust1}. Additional results are provided in the Web Appendix.

\subsection{Investigation of the Impact of Decisions During Marriage}

\label{ss:robust_decisions}

Chiappori, Fiorio, Galichon, and Verzillo (2026) note the observed assortative matching patterns may reflect decisions made during marriage. To investigate this, they focus on individuals in their first year of marriage. They also examine how their results change when they use income observations from individuals prior to marriage. We cannot conduct a similar analysis because the CPS does not report the year of marriage and, as it is not a panel dataset, it does not provide data on the same individuals before marriage. Nevertheless, we can partially address the impact of within-marriage decisions by focusing on relatively young individuals, who are less likely to have been married for a long period. Specifically, we restrict our sample to individuals aged 25 to 35, \emph{i.e.}, couples in which neither partner is older than 35. The results, reported in column I of Table \ref{tab:kendall_robust1} and in Web Appendix Section \ref{web_appendix::robust_newly_wed}, show no meaningful differences in the decomposition exercises compared to our main findings.

A potential criticism of this approach is that young married couples may have previously made  decisions such as, for example, having children, which affect wages. We  investigate how the results change when we also exclude couples with no dependent children. These results are reported in the second section of Table \ref{tab:kendall_robust1} and in the Web Appendix, Section \ref{web_appendix::robust_newly_wed} and are also similar. 

\begin{table}[tbp]
\caption{Robustness checks for the Kendall and Spearman's rank correlation coefficient.}
\label{tab:kendall_robust1}
\centering
\begin{turn}{90}
\begin{scriptsize}
\begin{tabular}{l|r|r|r|r|r|r|r|r|r|r|}
\hline \hline
& \multicolumn{2}{>{\centering\arraybackslash}p{2cm}|}{Young}& \multicolumn{2}{>{\centering\arraybackslash}p{2cm}|}{Young and no children}& \multicolumn{2}{>{\centering\arraybackslash}p{2cm}|}{ACS}& \multicolumn{2}{>{\centering\arraybackslash}p{2cm}|}{Less than two years in marriage}& \multicolumn{2}{>{\centering\arraybackslash}p{2cm}|}{Including two earner couples}	\\
& \multicolumn{2}{>{\centering\arraybackslash}p{2cm}|}{(I)}& \multicolumn{2}{>{\centering\arraybackslash}p{2cm}|}{(II)}& \multicolumn{2}{>{\centering\arraybackslash}p{2cm}|}{(III)}& \multicolumn{2}{>{\centering\arraybackslash}p{2cm}|}{(IV)}& \multicolumn{2}{>{\centering\arraybackslash}p{2cm}|}{(V)}	\\
\hline 
& \multicolumn{1}{p{1cm}|}{1976-1981}& \multicolumn{1}{p{1cm}|}{2021-2026}& \multicolumn{1}{p{1cm}|}{1976-1981}& \multicolumn{1}{p{1cm}|}{2021-2026}& \multicolumn{1}{p{1cm}|}{1980-1985}& \multicolumn{1}{p{1cm}|}{2021-2026}& \multicolumn{1}{p{1cm}|}{1980-1985}& \multicolumn{1}{p{1cm}|}{2021-2026}& \multicolumn{1}{p{1cm}|}{1976-1981}& \multicolumn{1}{p{1cm}|}{2021-2026}	\\
\hline 
\multicolumn{11}{l}{Kendall's tau}\\
\hline 
	%\multicolumn{1}{p{2.3cm}|}{Data}	&	0.207	&	0.325	&	0.207	&	0.325	&	0.205	&	0.304	&	0.205	&	0.304	&	0.152	&	0.287	\\
	\multicolumn{1}{p{2.3cm}|}{Estimated}	&	0.216	&	0.301	&	0.209	&	0.295	&	0.216	&	0.3	&	0.266	&	0.325	&	0.139	&	0.268	\\
	\multicolumn{1}{p{2.3cm}|}{Participation}	&	0.216	&	0.3	&	0.206	&	0.283	&	0.218	&	0.303	&	0.246	&	0.344	&	0.138	&	0.263	\\
	\multicolumn{1}{p{2.3cm}|}{Selection}	&	0.205	&	0.267	&	0.186	&	0.224	&	0.189	&	0.293	&	0.261	&	0.338	&	0.155	&	0.239	\\
	\multicolumn{1}{p{2.3cm}|}{$Z$ as in whole population}	&	0.21	&	0.276	&	0.185	&	0.223	&	0.191	&	0.294	&	0.263	&	0.334	&	0.156	&	0.239	\\
	\multicolumn{1}{p{2.3cm}|}{$\rho_{Y^*_w, Y^*_h} = 0$}	&	0.087	&	0.136	&	0.079	&	0.122	&	0.088	&	0.139	&	0.08	&	0.136	&	0.055	&	0.062	\\
\hline 
\multicolumn{11}{l}{Spearman's rank correlation coefficient}\\
\hline 
	%\multicolumn{1}{p{2.3cm}|}{Data}	&	0.277	&	0.424	&	0.277	&	0.424	&	0.275	&	0.398	&	0.275	&	0.398	&	0.204	&	0.376	\\
	\multicolumn{1}{p{2.3cm}|}{Estimated}	&	0.288	&	0.395	&	0.28	&	0.392	&	0.288	&	0.391	&	0.351	&	0.422	&	0.187	&	0.355	\\
	\multicolumn{1}{p{2.3cm}|}{Participation}	&	0.288	&	0.394	&	0.278	&	0.377	&	0.292	&	0.396	&	0.329	&	0.45	&	0.186	&	0.35	\\
	\multicolumn{1}{p{2.3cm}|}{Selection}	&	0.273	&	0.35	&	0.249	&	0.299	&	0.251	&	0.382	&	0.342	&	0.44	&	0.206	&	0.316	\\
	\multicolumn{1}{p{2.3cm}|}{$Z$ as in whole population}	&	0.28	&	0.361	&	0.248	&	0.296	&	0.254	&	0.384	&	0.344	&	0.435	&	0.208	&	0.315	\\
	\multicolumn{1}{p{2.3cm}|}{$\rho_{Y^*_w, Y^*_h} = 0$}	&	0.119	&	0.184	&	0.107	&	0.167	&	0.119	&	0.187	&	0.108	&	0.184	&	0.075	&	0.084	\\
\hline 
\hline 
\end{tabular}
\end{scriptsize}\end{turn}
\end{table}

\begin{table}[tbp]\ContinuedFloat
\caption{Robustness checks for the Kendall and Spearman's rank correlation coefficient (\emph{cntd.}).}
\label{tab:kendall_robust2}
\centering
\begin{turn}{90}
\begin{scriptsize}
\begin{tabular}{l|r|r|r|r|r|r|}
\hline \hline
& \multicolumn{2}{>{\centering\arraybackslash}p{2cm}|}{Including couples who cohabit}& \multicolumn{2}{>{\centering\arraybackslash}p{2cm}|}{Using only children $<$ 5 y/o as instrument}& \multicolumn{1}{>{\centering\arraybackslash}p{1.2cm}|}{Using merged dataset with NDCP data}& \multicolumn{1}{>{\centering\arraybackslash}p{1.2cm}|}{Using childcare prices as instrument}	\\
& \multicolumn{2}{>{\centering\arraybackslash}p{2cm}|}{(VI)}& \multicolumn{2}{>{\centering\arraybackslash}p{2cm}|}{(VII)}& \multicolumn{1}{>{\centering\arraybackslash}p{1.2cm}|}{(VIII)}& \multicolumn{1}{>{\centering\arraybackslash}p{1.2cm}|}{(IX)}	\\
\hline 
& \multicolumn{1}{p{1cm}|}{1976-1981}& \multicolumn{1}{p{1cm}|}{2021-2026}& \multicolumn{1}{p{1cm}|}{1976-1981}& \multicolumn{1}{p{1cm}|}{2021-2026}& \multicolumn{1}{p{1cm}|}{2018-2023}& \multicolumn{1}{p{1cm}|}{2018-2023}	\\
\hline 
\multicolumn{7}{l}{Kendall's tau}\\
\hline 
	%\multicolumn{1}{p{2.3cm}|}{Data}	&	0.207	&	0.336	&	0.207	&	0.325	&	0.333	&	0.333	\\
	\multicolumn{1}{p{2.3cm}|}{Estimated}	&	0.192	&	0.296	&	0.196	&	0.283	&	0.257	&	0.26	\\
	\multicolumn{1}{p{2.3cm}|}{Participation}	&	0.191	&	0.288	&	0.195	&	0.276	&	0.256	&	0.257	\\
	\multicolumn{1}{p{2.3cm}|}{Selection}	&	0.214	&	0.247	&	0.207	&	0.217	&	0.21	&	0.21	\\
	\multicolumn{1}{p{2.3cm}|}{$Z$ as in whole population}	&	0.216	&	0.251	&	0.205	&	0.224	&	0.217	&	0.218	\\
	\multicolumn{1}{p{2.3cm}|}{$\rho_{Y^*_w, Y^*_h} = 0$}	&	0.069	&	0.126	&	0.077	&	0.126	&	0.112	&	0.112	\\
\hline 
\multicolumn{7}{l}{Spearman's rank correlation coefficient}\\
\hline 
	%\multicolumn{1}{p{2.3cm}|}{Data}	&	0.277	&	0.438	&	0.277	&	0.424	&	0.434	&	0.434	\\
	\multicolumn{1}{p{2.3cm}|}{Estimated}	&	0.258	&	0.389	&	0.263	&	0.371	&	0.337	&	0.341	\\
	\multicolumn{1}{p{2.3cm}|}{Participation}	&	0.256	&	0.379	&	0.262	&	0.363	&	0.338	&	0.339	\\
	\multicolumn{1}{p{2.3cm}|}{Selection}	&	0.283	&	0.326	&	0.274	&	0.287	&	0.278	&	0.278	\\
	\multicolumn{1}{p{2.3cm}|}{$Z$ as in whole population}	&	0.285	&	0.33	&	0.271	&	0.295	&	0.287	&	0.287	\\
	\multicolumn{1}{p{2.3cm}|}{$\rho_{Y^*_w, Y^*_h} = 0$}	&	0.093	&	0.17	&	0.104	&	0.17	&	0.152	&	0.152	\\
\hline 
\hline 
\end{tabular}
\end{scriptsize}\end{turn}
\end{table}

As the CPS does not contain information on the year of marriage, we also use the American Community Survey (ACS) to further explore the impact of working decisions during marriage. The ACS is an ongoing survey conducted by the United States Census Bureau that collects detailed annual data on the U.S. population. While the ACS has been conducted on a nationwide basis since 2005, data prior to 2000 are derived from the long form of the decennial census. The lower frequency before 2000 is a drawback for our analysis compared to the CPS. As part of our robustness check, we use data from 1980 and 2021. These years contain 2,267,320 and 3,252,599 observations, respectively, of which we employ 170,378 and 222,250 couples based on criteria similar to those applied to the CPS. Of these, 93,429 and 164,705 couples have both partners working full-time, full-year (FTFY). The results in column III of Table \ref{tab:kendall_robust1} and in Web Appendix \ref{web_appendix::robust_newly_wed} show that the decomposition for this sample is remarkably similar to that of our primary analysis for the periods 1976–1980 and 2021–2025, respectively. We also examine a restricted sample of individuals in their first marriage who have been married for less than two years. Although the values of Kendall's tau and the Spearman rank correlation coefficient—reported in column IV of Table \ref{tab:kendall_robust1} and in Web Appendix \ref{web_appendix::robust_newly_wed}—are somewhat different, the overall conclusions remain unchanged.
%Specifically, the observed assortative matching patterns can be explained almost entirely by observed and unobserved characteristics.

\subsection{Investigation of the Impact of Including only FTFY Workers}

\label{ss:robust_ftfy}

Our primary analysis is based on couples in which both individuals work FTFY. We motivate this approach by arguing that it yields a more homogeneous group in terms of working hours, for which the position in the hourly wage distribution is a reasonable proxy for the position in the annual income distribution. It also avoids the need to assume that hourly wages are independent of working hours. Nevertheless, this approach restricts our focus to a subsample of all working individuals. As noted in Section \ref{sec:data}, depending on the year, only between 20 and 31 percent of the total population aged 25 to 65 consists of married couples in which both partners work FTFY. We repeat the analysis using all couples in which both partners work a positive number of hours during the year. The results are presented in column V of Table \ref{tab:kendall_robust1} and in Web Appendix \ref{web_appendix:robust_ftfy}. We find somewhat lower levels of assortative matching at the very top and bottom of the distributions, which also leads to lower values of Kendall's tau and the Spearman rank correlation coefficient. Nevertheless, the correlation between unobserved and observed characteristics continues to explain most of the observed sorting pattern.

\subsection{Investigation of the Impact of Including only Married Couples}

\label{ss:robust_cohab}

Section \ref{sec:data} highlights the percentage of married couples in the total population declined from 77\% to 57\%. In contrast, the share of cohabiting couples increased from virtually zero to 8\%. It is possible that the sorting patterns of cohabiting couples differ from those of married couples and this may affect our results regarding the evolution of household inequality. Section \ref{sec:data} noted information on cohabitation is only available from 1995 onward and our analysis of this issue is restricted to the period since that year. We perform a robustness check including couples who are either married or cohabiting and find that our results are not affected by this alternative definition of households. Details are reported in column VI of Table \ref{tab:kendall_robust1} and Web Appendix \ref{web_appendix:robust_cohab}.

\subsection{Investigation of the Impact of Choice of Exclusion Restrictions}

\label{ss:robust_instruments}

As the employed exclusion restrictions for the selection equations of both wives and husbands are potentially contentious, we examine the impact of these exclusion restrictions. We first exclude two of the three exclusion restrictions and retain only the indicator variable for children below the age of 5. The results are reported in column VII of Table \ref{tab:kendall_robust1} and in Web Appendix \ref{web_appendix:instruments} and do not differ from the conclusions of our main model. We then examine data from the National Database of Childcare Prices (NDCP) for the period 2008–2022. We introduce an additional exclusion restriction that interacts county-level childcare costs for infants up to 24 months old with the indicator variable for children below the age of 5. Since not all observations in the CPS contain county-level information, only about half of the observations for the period 2008–2022 can be matched to the NDCP. We find that the results for this subsample are unchanged when using the standard exclusion restrictions for the years 2018–2022. Moreover, the results do not change substantially when we include the additional exclusion restriction. These results are reported in columns VIII and IX of Table \ref{tab:kendall_robust1} and in Web Appendix \ref{web_appendix:instruments}.

\section{Conclusion}

\label{sec:discusion}

We examine the role of marital sorting on household income inequality in a period of substantial and increasing wage inequality. We do so by examining a sample of married couples from the CPS for the years 1976 to 2025 in which each of the spouses works full time/full year. To investigate the determinants of marital sorting and its impact on inequality we estimate a model explaining the spouse's employment decisions and their location in the gender-specific wage distributions. This is done via a multivariate distribution regression approach which incorporates selection. We provide a number of important empirical findings.

First, we confirm that wage inequality has increased for both spouses in dual FTFY couples. Changes in inequality for these groups do not precisely correspond to those for all males and females but the trends are similar.

Second, we find clear evidence of positive sorting defined as higher
incidences of high (low) wage males marrying high (low) wage females than expected under random sorting. However, we acknowledge that random sorting is not a realistic alternative to positive sorting as couples are frequently similar in terms of age, race and location of residence and each of these are known to determine wages. However, the disproportionate fraction of males in the extremes of the wage distribution marrying females in the corresponding extremes of their distribution is supportive of positive sorting. Moreover, this positive sorting has increased over our time period. Combining this increasing sorting with increasing spouse specific inequality has increased inequality across couples.

Third, a series of counterfactual experiments do not provide any evidence that unobservables related to either of the spouses' work decisions have any implications for the observed sorting behavior. However, there is compelling evidence that the correlation between the unobservables influencing the wages of the spouses is substantially contributing to the observed sorting behavior. While some of this correlation may be due to factors, such as cost of living premia, which are incurred by both members of the couple it is also possible that this captures factors corresponding to unobserved ability.

Finally, we find that positive sorting associated with the increased share of extreme cells has increased over the 50 years of our sample. A decomposition of the changes in these cells shares over the period examined reveals that their growth is almost entirely due to the prices of observable and unobservable characteristics. This suggests that the increase in observed sorting patterns is not due to different behavior nor the increase of more females in the labor market. Rather as the skill premia has increased this has increased the wages of both members of some couples while couples without these skills appear to be both negatively affected.

% --------------------------------------------------------------------------------------------------

\label{sec:conclusion}

\section*{\thinspace\ References}

% \begin{thebibliography}{99}

\begin{description}

\item \textsc{Acemoglu, D. and D. Autor} (2011), \textquotedblleft Skills, tasks and technologies: Implications for employment and earnings\textquotedblright , in \emph{Handbook of Labor Economics}, Vol. 4B, edited by O. Ashenfelter and D. Card, Elsevier, 1043--1171.

\item \textsc{Acemoglu, D. and P. Restrepo} (2022), \textquotedblleft Tasks, automation, and the rise in US wage inequality\textquotedblright , \emph{Econometrica} \textbf{90}, 1973--2016.

\item \textsc{Acemoglu, D. and J. L\"{o}bbing} (2024), \textquotedblleft Automation and polarization\textquotedblright , \emph{Journal of Political Economy} \textbf{132}, 1--49.

\item \textsc{Autor, D.H., L.F. Katz, and M. Kearney} (2008), \textquotedblleft Trends in U.S. wage inequality: Revising the
revisionists\textquotedblright, \emph{Review of Economics and Statistics} \textbf{90}, 300--23.

\item \textsc{Autor, D.H., A. Manning, and C.L. Smith} (2016),
\textquotedblleft The contribution of the minimum wage to US wage inequality over three decades: A reassessment\textquotedblright, \emph{American Economic Journal: Applied Economics} \textbf{8}, 58--99.

\item \textsc{Blau, F. and L.M. Kahn} (2009), ``Inequality and earnings distribution'', In: W. Salverda, B. Nolan, and T.M. Smeeding, in: \emph{Oxford Handbook on Economic Inequality}, Oxford University Press, 177--203.

\item \textsc{Blau, F.D., L.M. Kahn, N. Boboshko, and M.L. Comey} (2024), \textquotedblleft The impact of selection into the labor force on the gender wage gap\textquotedblright, \emph{Journal of Labor Economics} \textbf{42}, 1093--1133.

\item \textsc{Breen, R. and L. Salazar} (2011), \textquotedblleft
Educational assortative mating and earnings inequality in the United
States\textquotedblright, \emph{American Journal of Sociology} \textbf{117}, 808--43.

\item \textsc{Cancian, M. and D. Reed} (1998a), \textquotedblleft Assessing the effects of wives' earnings on family income inequality\textquotedblright, \emph{The Review of Economics and Statistics} \textbf{80}, 73--79.

\item \textsc{Cancian, M. and D. Reed } (1998b), \textquotedblleft The
impact of wives' earnings on income inequality: issues and
estimates\textquotedblright, \emph{Demography} \textbf{36}, 173--84.

\item \textsc{Card, D. and T. Lemieux} (1996), \textquotedblleft Wage dispersion, returns to skill, and black-white wage differentials\textquotedblright, \emph{Journal of Econometrics} \textbf{74}, 319--61.

\item \textsc{Chernozhukov, V., I. Fern\'{a}ndez-Val, and S. Luo} (2025), \textquotedblleft Distribution regression with sample selection, with an application to wage decompositions in the UK", \emph{Journal of Political Economy } \textbf{133}, 3952--92.

\item \textsc{Chernozhukov, V., I. Fern\'andez-Val, J. Meier, A. van Vuuren and F. Vella} (2025), \textquotedblleft Bivariate distribution regression; theory, estimation and an application to intergenerational mobility \textquotedblright, working paper, Massachusetts Institute of Technology, Cambridge.

\item \textsc{Chiappori, P.A., B. Salani\'e and Y.Weiss}, (2017), ``Partner choice, investment in children, and the marital college premium", \emph{American Economic Review} \textbf{107}, 2109--67.

\item \textsc{Chiappori, P.A., M. Costa-Dias and C. Meghir}, (2019), ``The marriage market, labor supply, and education choice", \emph{Journal of Political Economy} \textbf{126}, s26--s72.

\item \textsc{Chiappori, P.A., M. Costa-Dias, S. Crossman and C. Meghir} (2020), ``Changes in assortative matching and inequality in income: Evidence for the UK", \emph{Fiscal Studies} \textbf{41}, 39--63.

\item \textsc{Chiappori, P.A., M. Costa-Dias and C.Meghir,} (2023), ``The measuring of assortativeness in marriage", working paper, Yale University.

\item \textsc{Chiappori, P.A., M. Costa-Dias, C. Meghir and H. Zhang}, (2025).
``Changes in assortative matching: Theory and evidence for the US'', \emph{Journal of Political Economy} \textbf{133}, 3045--77.

\item \textsc{Chiappori, P.A., C.V. Fiorio, A. Galichon, and S. 
Verzillo} (2026), ``Assortative matching on income'', \emph{Econometrica} \textbf{94}, 957--89.

\item \textsc{Choo, E. and A. Siow} (2006), ``Who marries whom and why''. \emph{Journal of Political Economy} \textbf{114}, 175--201.

%\item \textsc{Chernozhukov, V., I. Fern\'{a}ndez-Val, and S. Luo}(2025), \textquotedblleft Distribution regression with sample selection, with an application to wage decompositions in the UK", \emph{Journal of Political Economy } \textbf{133}, 3952--3992.

\item \textsc{De Rock, B., M. Kovaleva, and T. Potoms} (2023),
\textquotedblleft A spouse and a house are all we need? Housing demand, labor supply and divorce over the lifecycle", working paper, Free University of Brussels.

\item \textsc{Eika, L., M. Mogstad, and B. Zafar} (2019), \textquotedblleft Educational assortative mating and household income
inequality\textquotedblright , \emph{Journal of Political Economy} \textbf{127}, 2795--2835.

\item \textsc{Fern\'{a}ndez-Val I., A. van Vuuren, and F. Vella} (2018), \textquotedblleft Decomposing real wage changes in the United
States\textquotedblright , \emph{IZA Discussion paper} \textbf{12044}, Bonn.

\item \textsc{Fern\'{a}ndez-Val I., A. van Vuuren, F. Vella and F. Peracchi} (2023) \textquotedblleft Selection and the distribution of female real hourly wages in the US \textquotedblright,  \emph{Quantitative Economics} \textbf{14}, 571--607.

\item \textsc{Fern\'{a}ndez-Val I., A. van Vuuren, F. Vella and F. Peracchi} (2024), \textquotedblleft Hours worked and the U.S distribution of real annual earnings 1976--2019\textquotedblright , \emph{Journal of Applied Econometrics} \textbf{39}, 659--78

\item \textsc{Fern\'{a}ndez, R. and R. Rogerson} (2001), \textquotedblleft Sorting and long-run inequality\textquotedblright , \emph{Quarterly Journal of Economics} \textbf{116}, 1305--41.

\item \textsc{Flood, S., M. King, S. Ruggles, and J.R. Warren} (2015),
\textquotedblleft Integrated public use microdata series, Current Population Survey: Version 4.0 [Machine-readable database]\textquotedblright , working paper, University of Minnesota.

\item \textsc{Geweke, J.} (1991), \textquotedblleft Efficient simulation from the multivariate normal and student-t distributions subject to linear constraints\textquotedblright , \emph{Computing Science and Statistics} \textbf{23}, 571--78.

\item \textsc{Gihleb, R. and K. Lang} (2020), \textquotedblleft Educational homogamy and assortative mating have not increased\textquotedblright , \emph{Research in Labor Economics} \textbf{48}, 1--26.

\item \textsc{Greenwood, J., N. Guner, G. Kocharkov and C. Santos} (2014), \textquotedblleft Marry your like: assortative mating and income inequality\textquotedblright , \emph{American Economic Review} \textbf{104}, 348--54.

\item \textsc{Hajivassiliou, V.} (1990) \textquotedblleft Smooth simulation estimation of panel data LDV models\textquotedblright , working paper, Yale University.

\item \textsc{Heckman J.J.} (1974), \textquotedblleft Shadow prices, market wages and labor supply\textquotedblright ,\ \emph{Econometrica} \textbf{42}, 679--94.

\item \textsc{Heckman J.J.} (1979), \textquotedblleft Sample selection bias as a specification error\textquotedblright , \emph{Econometrica} \textbf{47}, 153--61.

\item \textsc{Juhn C., K.M. Murphy, and B. Pierce} (1993), \textquotedblleft Wage inequality and the rise in returns to skill\textquotedblright , \emph{Journal of Political Economy} \textbf{101}, 410--42.

\item \textsc{Katz L.~F., and K. Murphy} (1992), \textquotedblleft Changes in relative wages, 1963--1987: Supply and demand factors\textquotedblright, \emph{Quarterly Journal of Economics} \textbf{107}, 35--78.

\item \textsc{Keane, M.} (1994), \textquotedblleft A computationally practical simulation estimator for panel data\textquotedblright , \emph{Econometrica} \textbf{62}, 95--116.

\item \textsc{Killewald, A., and I. Lundberg} (2017), \textquotedblleft New evidence against a causal marriage wage premium\textquotedblright, \emph{Demography} \textbf{54}, 1007--28. 

\item \textsc{Kremer, M.} (1997), \textquotedblleft How much does sorting increase inequality?\textquotedblright, \emph{Quarterly Journal of Economics} \textbf{112}, 115--39.

\item \textsc{Lise, J. and S. Seitz} (2011), \textquotedblleft Consumption inequality and intra-household allocations", \emph{Review of Economic Studies} \textbf{78}, 328--55.

\item \textsc{Lise, J. and K. Yamada (2018)},\textquotedblleft Household sharing and commitment: Evidence from panel data on individual expenditures and time use", \emph{Review of Economic Studies }\textbf{86}, 2184-2219.

\item \textsc{Mulligan, C. and Y. Rubinstein} (2008), \textquotedblleft Selection, investment, and women's relative wages over time\textquotedblright , \emph{Quarterly Journal of Economics} \textbf{123}, 1061--1110.

\item \textsc{Murphy, K.M. and F. Welch} (1992), \textquotedblleft The
structure of wages\textquotedblright , \emph{Quarterly Journal of Economics}  \textbf{107}, 285--326.

\item \textsc{Sibuya, M.} (1959), \textquotedblleft Bivariate extreme
statistics, I\textquotedblright, \emph{Annals of the Institute of
Statistical Mathematics} \textbf{11}, 195--210.

\item \textsc{Smock, P. J. and W. D. Manning} (2004), \textquotedblleft Living together unmarried in the United States: Demographic perspectives and implications for family policy", \emph{Law and Policy} \textbf{26}, 87--117.

\item \textsc{Welch F.} (2000), \textquotedblleft Growth in women's relative
wages and inequality among men: One phenomenon or two?\textquotedblright\ 
\emph{American Economic Review Papers \& Proceedings} \textbf{90}, 444--49.
\end{description}

\appendix

\renewcommand{\baselinestretch}{1.05}
\section{Proof of Theorem \protect\ref{thm:id}}

\label{app:id} We start with a useful result about the properties of the
multivariate normal CDF.

\begin{lemma}
\label{lemma1} Let $\Phi_N(\cdot; \boldsymbol{\Sigma})$ denote the CDF of
the multivariate standard normal distribution of dimension $N \geq 3$ with
nonsingular correlation matrix $\boldsymbol{\Sigma}$ and $\rho_{12}$ denote
the $(1,2)$-element of $\boldsymbol{\Sigma}$. Then, 
\begin{equation*}
\frac{\partial \Phi_N(\boldsymbol{x}; \boldsymbol{\Sigma})}{\partial
\rho_{12}} = \frac{\mathrm{E}[\phi_2(z_{1:3}, z_{2:3}; \rho_{12:3}) \mid 
\boldsymbol{X}_3 \leq \boldsymbol{x}_3] \Phi_{N-2}(\boldsymbol{x}_3; 
\boldsymbol{\Sigma}_{33}) }{\sigma_{1:3} \sigma_{2:3}} > 0,
\end{equation*}
where $\boldsymbol{X} = (X_1,X_2,\boldsymbol{X}_3^{\prime })^{\prime }$ is a
multivariate standard normal random variable with correlation matrix 
\begin{equation*}
\boldsymbol{\Sigma} = \left(%
\begin{array}{ccc}
1 & \rho_{12} & \boldsymbol{\Sigma}_{13}^{\prime } \\ 
\rho_{12} & 1 & \boldsymbol{\Sigma}_{23}^{\prime } \\ 
\boldsymbol{\Sigma}_{13} & \boldsymbol{\Sigma}_{23} & \boldsymbol{\Sigma}%
_{33}%
\end{array}%
\right),
\end{equation*}
$\boldsymbol{x} = (x_1,x_2,\boldsymbol{x}_3^{\prime })^{\prime }$, $%
\phi_2(\cdot; \rho)$ is the PDF of the bivariate standard normal
distribution with correlation coefficient $\rho$, $z_{j:3} := (x_j - 
\boldsymbol{\Sigma}_{33}^{-1} \boldsymbol{\Sigma}_{j3}^{\prime }\boldsymbol{X%
}_3)/\sigma_{j:3}$ and $\sigma_{j:3}^2 := 1 - \boldsymbol{\Sigma}%
_{j3}^{\prime }\boldsymbol{\Sigma}_{33}^{-1} \boldsymbol{\Sigma}_{j3}$ for $%
j=1,2$, and 
\begin{equation*}
\rho_{12:3} := \frac{\sigma_{12:3}}{\sigma_{1:3} \sigma_{2:3}} := \frac{%
\rho_{12} - \boldsymbol{\Sigma}_{13}^{\prime }\boldsymbol{\Sigma}_{33}^{-1} 
\boldsymbol{\Sigma}_{23}}{\sigma_{1:3} \sigma_{2:3}}.
\end{equation*}
Hence, $\rho_{12} \mapsto \Phi_N(\boldsymbol{x}; \boldsymbol{\Sigma})$ is
strictly increasing on $(-1,1)$.
\end{lemma}

\begin{proof}
By the definition of conditional probability and iterated expectations 
\begin{multline}  \label{eq:mn}
\Phi_N(\boldsymbol{x}; \boldsymbol{\Sigma}) = \mathrm{E}[\Pr(X_1 \leq x_1,
X_2 \leq x_2 \mid \boldsymbol{X}_3) \mid \boldsymbol{X}_3 \leq \boldsymbol{x}%
_3] \Pr(\boldsymbol{X}_3 \leq \boldsymbol{x}_3) \\
= \mathrm{E}[\Phi_2(z_{1:3}, z_{2:3}; \rho_{12:3}) \mid \boldsymbol{X}_3
\leq \boldsymbol{x}_3] \Phi_{N-2}(\boldsymbol{x}_3; \boldsymbol{\Sigma}%
_{33}),
\end{multline}
where the second equality uses that 
\begin{equation*}
\left(%
\begin{array}{c}
X_1 \\ 
X_2%
\end{array}%
\right) \mid \boldsymbol{X}_3 = \boldsymbol{x}_3 \sim \mathcal{N}_2 \left(
\left(%
\begin{array}{c}
\boldsymbol{\Sigma}_{33}^{-1} \boldsymbol{\Sigma}_{13}^{\prime }\boldsymbol{x%
}_3 \\ 
\boldsymbol{\Sigma}_{33}^{-1} \boldsymbol{\Sigma}_{23}^{\prime }\boldsymbol{x%
}_3%
\end{array}%
\right), \left(%
\begin{array}{cc}
\sigma_{1:3}^2 & \sigma_{12:3} \\ 
\sigma_{12:3} & \sigma_{2:3}^2%
\end{array}%
\right)\right),
\end{equation*}
by the properties of the multivariate normal distribution. Note that the
variance-covariance matrix is non-singular because $\boldsymbol{\Sigma}$ is
non-singular.

The result then follows from noting that the RHS of \eqref{eq:mn} only
depends on $\rho_{12}$ through $\rho_{12:3}$ in $\Phi_2(z_{1:3}, z_{2:3};
\rho_{12:3})$, $\partial \Phi_2(\cdot; \rho)/\partial \rho =
\phi_2(\cdot;\rho)$ (e.g., Sibuya, 1959), $\partial \rho_{12:3}/\partial
\rho_{12} = 1/(\sigma_{1:3} \sigma_{2:3})$, and the chain rule. \hfill
\end{proof}

We now proceed with the proof of Theorem \ref{thm:id}. The argument is
conditional on all the covariates $Z$ in steps (1), (3) and (4), and the
common covariates $X$ in step (2) below. We consider the worst case where
the covariates $Z_{1}$ that satisfy the exclusion restriction only include a
binary variable. To lighten the notation, we drop the arguments of a
function when they are $0$. For example, $\mu _{D_{w}^{\ast }}(z):=\mu
_{D_{w}^{\ast }}(0,z)$. The proof shows identification of all the parameters
sequentially:

\begin{enumerate}
\item $\mu _{D_{w}^{\ast}}(z)$, $\mu _{D_{h}^{\ast}}(z)$ and $\rho
_{D_{w}^{\ast },D_{h}^{\ast }}(z) $: these parameters are identified from
the distribution of $(D_{w},D_{h})$ conditional on $Z=z$. In particular, 
\begin{equation*}
\mu _{D_{j}^{\ast}}(z) = \Phi^{-1} \left(\Pr(D_{j} = 1 \mid Z=z) \right),
\quad j = w,h,
\end{equation*}
and $\rho _{D_{w}^{\ast },D_{h}^{\ast }}(z) $ is identified as the solution
in $\rho$ to the equation 
\begin{equation*}
\Pr(D_{w} = 1, D_{h} = 1 \mid Z=z) = \Phi_2\left( \mu _{D_{w}^{\ast}}(z),\mu
_{D_{h}^{\ast}}(z); \rho \right).
\end{equation*}
The solution exists and is unique by the same argument as in Lemma 1 of CFL.

\item $\mu _{Y_{j}^{\ast}}(y_{j},x) $ and $\rho _{D_{j}^{\ast },Y_{j}^{\ast
}}(y_{j},x)$, $j = w,h$: these parameters are identified by Theorem 1 of CFL
using the relevance conditions and exclusion restrictions in Assumption \ref%
{ass:lgr}(2)--(3).

\item $\rho _{D_{w}^{\ast },Y_{h}^{\ast }}(y_{h},z)$: this parameter is
identified from the distribution of $(D_{w},D_{h},Y_{h})$ conditional on $%
Z=z $ as the solution in $\rho_{12}$ to the equation 
\begin{equation*}
\Pr(Y_{h} \leq y_h, D_{w} =1, D_{h} =1 \mid Z = z) = \Phi_3(\mu
_{Y_{h}^{\ast}}(y_{h},x),\mu _{D_{w}^{\ast}}(z), \mu _{D_{h}^{\ast}}(z); 
\boldsymbol{\Sigma}(y_h,z)),
\end{equation*}
where 
\begin{equation*}
\boldsymbol{\Sigma}(y_h,z) = \left(%
\begin{array}{ccc}
1 & \rho_{12} & \rho _{D_{h}^{\ast },Y_{h}^{\ast }}(y_{h},x) \\ 
\rho_{12} & 1 & \rho _{D_{w}^{\ast },D_{h}^{\ast }}(z) \\ 
\rho _{D_{h}^{\ast },Y_{h}^{\ast }}(y_{h},x) & \rho _{D_{w}^{\ast
},D_{h}^{\ast }}(z) & 1%
\end{array}%
\right).
\end{equation*}
The solution exists by Assumption \ref{ass:lgr}(1) and is unique by Lemma %
\ref{lemma1}. A similar argument shows that $\rho _{D_{h}^{\ast
},Y_{w}^{\ast }}(y_{w},x)$ is identified.

\item $\rho _{Y_{w}^{\ast },Y_{h}^{\ast }}(y_{w},y_{h},z)$: this parameter
is identified from the distribution of $(D_{w},D_{h},Y_{w},Y_{h})$
conditional on $Z=z$ as the solution in $\rho _{12}$ to the equation 
\begin{multline*}
\Pr (Y_{w}\leq y_{w},Y_{h}\leq y_{h},D_{w}=1,D_{h}=1\mid Z=z) \\
=\Phi _{4}(\mu _{Y_{w}^{\ast }}(y_{w},x),\mu _{Y_{h}^{\ast }}(y_{h},x),\mu
_{D_{w}^{\ast }}(z),\mu _{D_{h}^{\ast }}(z);\boldsymbol{\Sigma }%
(y_{w},y_{h},z)),
\end{multline*}%
where 
\begin{equation*}
\boldsymbol{\Sigma }(y_{w},y_{h},z)=\left( 
\begin{array}{cccc}
1 & \rho _{12} & \rho _{D_{w}^{\ast },Y_{w}^{\ast }}(y_{w},x) & \rho
_{D_{h}^{\ast },Y_{w}^{\ast }}(y_{w},x) \\ 
\rho _{12} & 1 & \rho _{D_{w}^{\ast },Y_{h}^{\ast }}(y_{h},x) & \rho
_{D_{h}^{\ast },Y_{h}^{\ast }}(y_{h},x) \\ 
\rho _{D_{w}^{\ast },Y_{w}^{\ast }}(y_{w},x) & \rho _{D_{w}^{\ast
},Y_{h}^{\ast }}(y_{h},x) & 1 & \rho _{D_{w}^{\ast },D_{h}^{\ast }}(z) \\ 
\rho _{D_{h}^{\ast },Y_{w}^{\ast }}(y_{w},x) & \rho _{D_{h}^{\ast
},Y_{h}^{\ast }}(y_{h},x) & \rho _{D_{w}^{\ast },D_{h}^{\ast }}(z) & 1%
\end{array}%
\right) .
\end{equation*}%
The solution exists by Assumption \ref{ass:lgr}(1) and is unique by Lemma %
\ref{lemma1}.
\end{enumerate}

\QED

\newpage

\section{Tables}

\setcounter{table}{0} \renewcommand{\thetable}{B\arabic{table}}

\begin{table}[h]
\begin{tabular}{l||l|rrrrrrrrrr}
\hline\hline
&  & \multicolumn{10}{c}{Husbands' quantiles} \\ \hline\hline
&  & d1 & d2 & d3 & d4 & d5 & d6 & d7 & d8 & d9 & d10 \\ \hline
\parbox[t]{2mm}{\multirow{10}{*}{\rotatebox[origin=c]{90}{Wives' quantiles}}}
& d1 & \textbf{39.2} & 38.2 & 39.4 & 38.9 & 37.5 & 39.0 & 40.6 & 40.0 & 41.3
& 42.8 \\ 
& d2 & 39.9 & \textbf{38.9} & 37.9 & 38.6 & 39.3 & 39.6 & 40.8 & 41.0 & 41.4
& 43.3 \\ 
& d3 & 38.8 & 38.2 & \textbf{38.2} & 37.9 & 38.9 & 39.1 & 40.1 & 41.2 & 41.2
& 42.1 \\ 
& d4 & 38.5 & 37.8 & 38.3 & \textbf{37.0} & 38.3 & 39.8 & 40.1 & 39.5 & 41.9
& 42.1 \\ 
& d5 & 38.2 & 36.9 & 36.9 & 37.4 & \textbf{36.8} & 38.5 & 39.0 & 38.7 & 40.4
& 42.7 \\ 
& d6 & 37.5 & 36.4 & 37.8 & 37.2 & 36.6 & \textbf{37.8} & 38.5 & 39.8 & 40.3
& 42.3 \\ 
& d7 & 36.9 & 39.0 & 36.1 & 37.2 & 37.1 & 38.4 & \textbf{40.2} & 40.2 & 40.8
& 42.4 \\ 
& d8 & 38.7 & 37.2 & 36.9 & 37.8 & 37.7 & 38.6 & 38.7 & \textbf{39.7} & 40.1
& 42.1 \\ 
& d9 & 37.0 & 38.0 & 37.7 & 37.5 & 37.5 & 38.2 & 39.0 & 39.7 & \textbf{40.2}
& 41.8 \\ 
& d10 & 41.9 & 39.2 & 40.2 & 39.3 & 39.5 & 39.4 & 40.1 & 38.6 & 40.9 & 
\textbf{41.2} \\ \hline\hline
\end{tabular}%
\caption{Average age of the wives for the years 1976-1980.}
\label{tab:age_wives}
\end{table}

\begin{table}[h]
\begin{tabular}{l||l|rrrrrrrrrr}
\hline\hline
&  & \multicolumn{10}{c}{Husbands' quantiles} \\ \hline\hline
&  & d1 & d2 & d3 & d4 & d5 & d6 & d7 & d8 & d9 & d10 \\ \hline
\parbox[t]{2mm}{\multirow{10}{*}{\rotatebox[origin=c]{90}{Wives' quantiles}}}
& d1 & \textbf{42.3} & 41.0 & 42.0 & 41.7 & 40.5 & 41.7 & 43.4 & 42.7 & 43.9
& 46.0 \\ 
& d2 & 42.9 & \textbf{41.7} & 40.7 & 41.6 & 41.7 & 42.5 & 43.7 & 43.9 & 44.0
& 46.2 \\ 
& d3 & 41.3 & 40.8 & \textbf{40.6} & 40.1 & 41.4 & 41.7 & 42.5 & 44.4 & 43.7
& 44.9 \\ 
& d4 & 41.6 & 40.1 & 40.8 & \textbf{39.8} & 40.7 & 42.6 & 42.4 & 42.0 & 44.3
& 44.9 \\ 
& d5 & 40.7 & 39.6 & 39.3 & 39.7 & \textbf{39.4} & 40.9 & 41.9 & 41.4 & 43.2
& 45.6 \\ 
& d6 & 39.9 & 38.8 & 40.1 & 39.9 & 39.0 & \textbf{40.3} & 41.0 & 42.3 & 43.0
& 45.3 \\ 
& d7 & 39.0 & 41.2 & 38.4 & 40.0 & 39.8 & 40.9 & \textbf{42.9} & 42.7 & 43.5
& 45.3 \\ 
& d8 & 41.3 & 39.9 & 38.8 & 39.7 & 40.1 & 41.0 & 41.2 & \textbf{42.2} & 42.7
& 44.9 \\ 
& d9 & 39.4 & 40.1 & 40.2 & 40.2 & 40.1 & 40.1 & 41.3 & 42.1 & \textbf{42.9}
& 44.5 \\ 
& d10 & 44.3 & 41.3 & 42.3 & 40.9 & 41.4 & 41.3 & 42.5 & 40.6 & 43.4 & 
\textbf{44.4} \\ \hline\hline
\end{tabular}%
\caption{Average age of the husbands for the years 1976-1980.}
\label{tab:age_husbands}
\end{table}

\begin{table}[h]
\begin{tabular}{l||l|rrrrrrrrrr}
\hline\hline
&  & \multicolumn{10}{c}{Husbands' quantiles} \\ \hline\hline
&  & d1 & d2 & d3 & d4 & d5 & d6 & d7 & d8 & d9 & d10 \\ \hline
\parbox[t]{2mm}{\multirow{10}{*}{\rotatebox[origin=c]{90}{Wives' quantiles}}}
	&d1	&	\textbf{42.5}	&	43.0	&	42.7	&	43.6	&	43.7	&	43.6	&	44.7	&	45.0	&	43.9	&	45.5	\\	
	&d2	&	42.0	&	\textbf{42.3}	&	42.6	&	42.8	&	44.0	&	43.6	&	43.7	&	44.1	&	45.9	&	46.1	\\	
	&d3	&	42.5	&	41.7	&	\textbf{41.8}	&	42.4	&	42.6	&	43.0	&	43.5	&	44.2	&	45.8	&	46.1	\\	
	&d4	&	42.3	&	41.7	&	42.5	&	\textbf{42.3}	&	42.7	&	42.4	&	42.6	&	43.3	&	44.9	&	46.2	\\	
	&d5	&	42.3	&	42.2	&	42.2	&	42.2	&	\textbf{43.0}	&	42.6	&	43.4	&	43.2	&	43.5	&	44.5	\\	
	&d6	&	43.3	&	41.2	&	42.1	&	42.6	&	42.5	&	\textbf{42.8}	&	43.7	&	43.2	&	43.4	&	45.0	\\	
	&d7	&	42.5	&	43.1	&	42.7	&	42.9	&	43.2	&	43.0	&	\textbf{42.8}	&	43.4	&	43.8	&	44.8	\\	
	&d8	&	43.9	&	44.1	&	43.7	&	43.2	&	43.6	&	43.7	&	43.0	&	\textbf{43.8}	&	44.4	&	43.8	\\	
	&d9	&	44.9	&	44.3	&	44.8	&	43.6	&	44.3	&	43.3	&	43.8	&	42.9	&	\textbf{43.6}	&	44.5	\\	
	&d10	&	44.9	&	43.6	&	44.8	&	43.0	&	44.4	&	43.8	&	42.9	&	43.7	&	43.4	&	\textbf{43.5}	\\	 \hline\hline
\end{tabular}%
\caption{Average age of the wives for the years 2021-2025.}
\label{tab:age_wives1}
\end{table}

\begin{table}[h]
\begin{tabular}{l||l|rrrrrrrrrr}
\hline\hline
&  & \multicolumn{10}{c}{Husbands' quantiles} \\ \hline\hline
&  & d1 & d2 & d3 & d4 & d5 & d6 & d7 & d8 & d9 & d10 \\ \hline
\parbox[t]{2mm}{\multirow{10}{*}{\rotatebox[origin=c]{90}{Wives' quantiles}}}
	&d1	&	\textbf{44.3}	&	45.1	&	44.6	&	45.4	&	46.3	&	45.8	&	46.8	&	46.7	&	46.0	&	47.8	\\	
	&d2	&	43.5	&	\textbf{43.9}	&	44.8	&	45.5	&	45.5	&	45.8	&	45.6	&	46.3	&	48.0	&	48.8	\\	
	&d3	&	44.1	&	43.7	&	\textbf{43.9}	&	44.3	&	43.9	&	44.6	&	45.2	&	45.9	&	47.2	&	48.6	\\	
	&d4	&	44.2	&	43.3	&	44.3	&	\textbf{43.8}	&	44.6	&	44.2	&	44.5	&	45.2	&	46.7	&	48.1	\\	
	&d5	&	44.0	&	43.8	&	43.9	&	44.0	&	\textbf{44.7}	&	44.3	&	44.9	&	45.1	&	45.2	&	46.7	\\	
	&d6	&	45.0	&	43.0	&	44.0	&	44.1	&	44.1	&	\textbf{44.4}	&	45.3	&	44.9	&	45.1	&	46.7	\\	
	&d7	&	44.4	&	44.6	&	44.2	&	44.5	&	44.8	&	44.4	&	\textbf{44.4}	&	45.1	&	45.7	&	46.7	\\	
	&d8	&	45.1	&	45.5	&	45.2	&	44.7	&	45.0	&	45.2	&	44.8	&	\textbf{45.6}	&	46.0	&	45.8	\\	
	&d9	&	46.2	&	45.7	&	45.9	&	45.9	&	45.7	&	45.2	&	45.6	&	44.7	&	\textbf{45.3}	&	46.1	\\	
	&d10	&	46.5	&	46.0	&	46.5	&	44.8	&	46.3	&	45.2	&	44.4	&	45.2	&	45.0	&	\textbf{45.2}	\\	  \hline\hline
\end{tabular}%
\caption{Average age of the husbands for the years 2021-2025.}
\label{tab:age_husbands1}
\end{table}

\begin{table}[h]
\resizebox{15cm}{!}{
\begin{tabular}{l||l|rrrrrrrrrr}
\hline\hline
&  & \multicolumn{10}{c}{Husbands' quantiles} \\ \hline\hline
&  & d1 & d2 & d3 & d4 & d5 & d6 & d7 & d8 & d9 & d10 \\ \hline
\parbox[t]{2mm}{\multirow{10}{*}{\rotatebox[origin=c]{90}{Wives' quantiles}}}
& d1 & \textbf{0.222} & 0.144 & 0.152 & 0.111 & 0.093 & 0.12 & 0.124 & 0.087
& 0.068 & 0.056 \\ 
& d2 & 0.184 & \textbf{0.159} & 0.141 & 0.089 & 0.09 & 0.103 & 0.115 & 0.089
& 0.055 & 0.051 \\ 
& d3 & 0.209 & 0.126 & \textbf{0.092} & 0.092 & 0.109 & 0.134 & 0.086 & 0.087
& 0.102 & 0.041 \\ 
& d4 & 0.166 & 0.128 & 0.095 & \textbf{0.12} & 0.074 & 0.072 & 0.083 & 0.083
& 0.073 & 0.03 \\ 
& d5 & 0.14 & 0.12 & 0.112 & 0.107 & \textbf{0.081} & 0.086 & 0.1 & 0.06 & 
0.066 & 0.063 \\ 
& d6 & 0.169 & 0.15 & 0.14 & 0.107 & 0.098 & \textbf{0.089} & 0.089 & 0.099
& 0.109 & 0.052 \\ 
& d7 & 0.151 & 0.137 & 0.108 & 0.133 & 0.096 & 0.092 & \textbf{0.112} & 0.053
& 0.081 & 0.071 \\ 
& d8 & 0.119 & 0.113 & 0.097 & 0.137 & 0.136 & 0.102 & 0.132 & \textbf{0.095}
& 0.098 & 0.077 \\ 
& d9 & 0.151 & 0.169 & 0.132 & 0.126 & 0.12 & 0.132 & 0.143 & 0.107 & 
\textbf{0.102} & 0.078 \\ 
& d10 & 0.142 & 0.156 & 0.167 & 0.18 & 0.113 & 0.134 & 0.175 & 0.146 & 0.144
& \textbf{0.104} \\ \hline\hline
\end{tabular}}
\caption{Percentage of non-white wives 1976-1980.}
\label{tab:non_white_wives}
\end{table}

\begin{table}[h]
\resizebox{15cm}{!}{
\begin{tabular}{l||l|rrrrrrrrrr}
\hline\hline
&  & \multicolumn{10}{c}{Husbands' quantiles} \\ \hline\hline
&  & d1 & d2 & d3 & d4 & d5 & d6 & d7 & d8 & d9 & d10 \\ \hline
\parbox[t]{2mm}{\multirow{10}{*}{\rotatebox[origin=c]{90}{Wives' quantiles}}}
& d1 & \textbf{0.223} & 0.139 & 0.155 & 0.118 & 0.097 & 0.123 & 0.124 & 0.087
& 0.063 & 0.046 \\ 
& d2 & 0.185 & \textbf{0.161} & 0.141 & 0.084 & 0.093 & 0.1 & 0.115 & 0.086
& 0.059 & 0.045 \\ 
& d3 & 0.197 & 0.122 & \textbf{0.092} & 0.084 & 0.112 & 0.13 & 0.079 & 0.083
& 0.098 & 0.032 \\ 
& d4 & 0.163 & 0.126 & 0.095 & \textbf{0.115} & 0.058 & 0.072 & 0.079 & 0.083
& 0.073 & 0.026 \\ 
& d5 & 0.15 & 0.131 & 0.105 & 0.099 & \textbf{0.081} & 0.086 & 0.1 & 0.067 & 
0.069 & 0.071 \\ 
& d6 & 0.153 & 0.159 & 0.143 & 0.113 & 0.104 & \textbf{0.089} & 0.089 & 0.087
& 0.116 & 0.059 \\ 
& d7 & 0.166 & 0.149 & 0.108 & 0.138 & 0.107 & 0.095 & \textbf{0.112} & 0.056
& 0.081 & 0.068 \\ 
& d8 & 0.134 & 0.113 & 0.108 & 0.134 & 0.136 & 0.097 & 0.14 & \textbf{0.098}
& 0.088 & 0.071 \\ 
& d9 & 0.144 & 0.198 & 0.14 & 0.122 & 0.13 & 0.134 & 0.141 & 0.1 & \textbf{0.098} & 0.073 \\ 
& d10 & 0.15 & 0.163 & 0.167 & 0.186 & 0.109 & 0.131 & 0.165 & 0.146 & 0.144
& \textbf{0.102} \\ \hline\hline
\end{tabular}}
\caption{Percentage of non-white husbands 1976-1980.}
\label{tab:non_white_husbands}
\end{table}

\begin{table}[tbp]
\resizebox{15cm}{!}{
\begin{tabular}{l||l|rrrrrrrrrr}
\hline\hline
&  & \multicolumn{10}{c}{Husbands' quantiles} \\ \hline\hline
&  & d1 & d2 & d3 & d4 & d5 & d6 & d7 & d8 & d9 & d10 \\ \hline
\parbox[t]{2mm}{\multirow{10}{*}{\rotatebox[origin=c]{90}{Wives' quantiles}}}
& d1 & \textbf{0.217} & 0.137 & 0.147 & 0.105 & 0.09 & 0.12 & 0.116 & 0.087
& 0.059 & 0.046 \\ 
& d2 & 0.175 & \textbf{0.155} & 0.141 & 0.084 & 0.084 & 0.093 & 0.112 & 0.082
& 0.047 & 0.045 \\ 
& d3 & 0.197 & 0.116 & \textbf{0.087} & 0.073 & 0.107 & 0.127 & 0.079 & 0.083
& 0.091 & 0.027 \\ 
& d4 & 0.158 & 0.124 & 0.088 & \textbf{0.115} & 0.055 & 0.068 & 0.079 & 0.08
& 0.069 & 0.026 \\ 
& d5 & 0.133 & 0.117 & 0.1 & 0.091 & \textbf{0.078} & 0.08 & 0.094 & 0.06 & 
0.066 & 0.063 \\ 
& d6 & 0.149 & 0.144 & 0.132 & 0.105 & 0.096 & \textbf{0.089} & 0.086 & 0.087
& 0.106 & 0.052 \\ 
& d7 & 0.141 & 0.137 & 0.103 & 0.127 & 0.093 & 0.09 & \textbf{0.106} & 0.051
& 0.076 & 0.065 \\ 
& d8 & 0.119 & 0.103 & 0.097 & 0.124 & 0.127 & 0.089 & 0.13 & \textbf{0.093}
& 0.083 & 0.069 \\ 
& d9 & 0.13 & 0.164 & 0.128 & 0.122 & 0.115 & 0.129 & 0.133 & 0.098 & 
\textbf{0.095} & 0.071 \\ 
& d10 & 0.142 & 0.148 & 0.167 & 0.165 & 0.104 & 0.128 & 0.163 & 0.14 & 0.139
& \textbf{0.091} \\ \hline\hline
\end{tabular}}
\caption{Percentage of non-white husbands and wives for the years 1976-1980.}
\label{tab:race_both}
\end{table}

\begin{table}[h]
\resizebox{15cm}{!}{
\begin{tabular}{l||l|rrrrrrrrrr}
\hline\hline
&  & \multicolumn{10}{c}{Husbands' quantiles} \\ \hline\hline
&  & d1 & d2 & d3 & d4 & d5 & d6 & d7 & d8 & d9 & d10 \\ \hline
\parbox[t]{2mm}{\multirow{10}{*}{\rotatebox[origin=c]{90}{Wives' quantiles}}}
	&d1	&	\textbf{0.212}	&	0.166	&	0.18	&	0.167	&	0.157	&	0.139	&	0.174	&	0.181	&	0.139	&	0.172	\\	
	&d2	&	0.217	&	\textbf{0.205}	&	0.179	&	0.155	&	0.163	&	0.162	&	0.172	&	0.13	&	0.205	&	0.179	\\	
	&d3	&	0.225	&	0.192	&	\textbf{0.193}	&	0.14	&	0.123	&	0.137	&	0.133	&	0.162	&	0.159	&	0.225	\\	
	&d4	&	0.199	&	0.217	&	0.186	&	\textbf{0.177}	&	0.188	&	0.151	&	0.154	&	0.148	&	0.171	&	0.186	\\	
	&d5	&	0.228	&	0.19	&	0.168	&	0.17	&	\textbf{0.197}	&	0.131	&	0.15	&	0.169	&	0.151	&	0.211	\\	
	&d6	&	0.189	&	0.193	&	0.193	&	0.171	&	0.151	&	\textbf{0.152}	&	0.171	&	0.173	&	0.202	&	0.164	\\	
	&d7	&	0.217	&	0.184	&	0.167	&	0.192	&	0.177	&	0.177	&	\textbf{0.125}	&	0.196	&	0.197	&	0.194	\\	
	&d8	&	0.16	&	0.192	&	0.157	&	0.131	&	0.187	&	0.169	&	0.194	&	\textbf{0.193}	&	0.211	&	0.233	\\	
	&d9	&	0.195	&	0.187	&	0.194	&	0.17	&	0.195	&	0.179	&	0.199	&	0.289	&	\textbf{0.247}	&	0.285	\\	
	&d10	&	0.237	&	0.255	&	0.182	&	0.206	&	0.213	&	0.177	&	0.242	&	0.259	&	0.277	&	\textbf{0.333}	\\ \hline\hline
\end{tabular}}
\caption{Percentage of non-white wives 2021-2025.}
\label{tab:non_white_wives1}
\end{table}

\begin{table}[h]
\resizebox{15cm}{!}{
\begin{tabular}{l||l|rrrrrrrrrr}
\hline\hline
&  & \multicolumn{10}{c}{Husbands' quantiles} \\ \hline\hline
&  & d1 & d2 & d3 & d4 & d5 & d6 & d7 & d8 & d9 & d10 \\ \hline
\parbox[t]{2mm}{\multirow{10}{*}{\rotatebox[origin=c]{90}{Wives' quantiles}}}
	&d1	&	\textbf{0.215}	&	0.171	&	0.157	&	0.162	&	0.157	&	0.142	&	0.178	&	0.165	&	0.107	&	0.172	\\	
	&d2	&	0.225	&	\textbf{0.204}	&	0.172	&	0.164	&	0.143	&	0.165	&	0.152	&	0.142	&	0.182	&	0.154	\\	
	&d3	&	0.233	&	0.199	&	\textbf{0.183}	&	0.136	&	0.123	&	0.139	&	0.141	&	0.175	&	0.151	&	0.197	\\	
	&d4	&	0.196	&	0.217	&	0.201	&	\textbf{0.173}	&	0.188	&	0.141	&	0.138	&	0.137	&	0.16	&	0.147	\\	
	&d5	&	0.231	&	0.207	&	0.18	&	0.166	&	\textbf{0.195}	&	0.118	&	0.146	&	0.15	&	0.157	&	0.172	\\	
	&d6	&	0.189	&	0.185	&	0.17	&	0.171	&	0.158	&	\textbf{0.154}	&	0.158	&	0.108	&	0.167	&	0.149	\\	
	&d7	&	0.233	&	0.209	&	0.188	&	0.202	&	0.165	&	0.177	&	\textbf{0.118}	&	0.178	&	0.195	&	0.182	\\	
	&d8	&	0.194	&	0.205	&	0.172	&	0.135	&	0.178	&	0.171	&	0.18	&	\textbf{0.195}	&	0.227	&	0.221	\\	
	&d9	&	0.208	&	0.163	&	0.181	&	0.17	&	0.176	&	0.19	&	0.179	&	0.251	&	\textbf{0.235}	&	0.256	\\	
	&d10	&	0.254	&	0.255	&	0.206	&	0.211	&	0.194	&	0.17	&	0.217	&	0.218	&	0.246	&	\textbf{0.317}	\\	 \hline\hline
\end{tabular}}
\caption{Percentage of non-white husbands 2021-2025.}
\label{tab:non_white_husbands1}
\end{table}

\begin{table}[tbp]
\resizebox{15cm}{!}{
\begin{tabular}{l||l|rrrrrrrrrr}
\hline\hline
&  & \multicolumn{10}{c}{Husbands' quantiles} \\ \hline\hline
&  & d1 & d2 & d3 & d4 & d5 & d6 & d7 & d8 & d9 & d10 \\ \hline
\parbox[t]{2mm}{\multirow{10}{*}{\rotatebox[origin=c]{90}{Wives' quantiles}}}
	&d1	&	\textbf{0.24}	&	0.192	&	0.196	&	0.189	&	0.18	&	0.175	&	0.195	&	0.193	&	0.15	&	0.204	\\	
	&d2	&	0.24	&	\textbf{0.233}	&	0.208	&	0.188	&	0.184	&	0.198	&	0.203	&	0.161	&	0.236	&	0.186	\\	
	&d3	&	0.261	&	0.226	&	\textbf{0.216}	&	0.165	&	0.153	&	0.162	&	0.167	&	0.207	&	0.18	&	0.237	\\	
	&d4	&	0.236	&	0.253	&	0.226	&	\textbf{0.195}	&	0.222	&	0.178	&	0.172	&	0.177	&	0.202	&	0.211	\\	
	&d5	&	0.271	&	0.226	&	0.199	&	0.2	&	\textbf{0.222}	&	0.155	&	0.176	&	0.205	&	0.198	&	0.237	\\	
	&d6	&	0.232	&	0.237	&	0.213	&	0.202	&	0.186	&	\textbf{0.183}	&	0.2	&	0.179	&	0.217	&	0.195	\\	
	&d7	&	0.262	&	0.231	&	0.213	&	0.23	&	0.214	&	0.216	&	\textbf{0.153}	&	0.228	&	0.232	&	0.228	\\	
	&d8	&	0.223	&	0.234	&	0.195	&	0.167	&	0.212	&	0.201	&	0.232	&	\textbf{0.232}	&	0.267	&	0.263	\\	
	&d9	&	0.24	&	0.205	&	0.219	&	0.209	&	0.226	&	0.226	&	0.222	&	0.31	&	\textbf{0.286}	&	0.316	\\	
	&d10	&	0.305	&	0.284	&	0.236	&	0.247	&	0.241	&	0.199	&	0.273	&	0.287	&	0.312	&	\textbf{0.377}	\\	 \hline\hline
\end{tabular}}
\caption{Percentage of non-white husbands and wives for the years 2021-2025.}
\label{tab:race_both1}
\end{table}

\begin{table}[h]
\resizebox{15cm}{!}{
\begin{tabular}{l||l|rrrrrrrrrr}
\hline\hline
&  & \multicolumn{10}{c}{Husbands' quantiles} \\ \hline\hline
&  & d1 & d2 & d3 & d4 & d5 & d6 & d7 & d8 & d9 & d10 \\ \hline
\parbox[t]{2mm}{\multirow{10}{*}{\rotatebox[origin=c]{90}{Wives' quantiles}}}
& d1 & \textbf{0.391} & 0.394 & 0.407 & 0.408 & 0.502 & 0.472 & 0.496 & 0.559
& 0.534 & 0.646 \\ 
& d2 & 0.354 & \textbf{0.413} & 0.432 & 0.453 & 0.474 & 0.548 & 0.535 & 0.49
& 0.542 & 0.612 \\ 
& d3 & 0.437 & 0.422 & \textbf{0.38} & 0.48 & 0.485 & 0.562 & 0.572 & 0.562
& 0.567 & 0.639 \\ 
& d4 & 0.406 & 0.446 & 0.435 & \textbf{0.457} & 0.532 & 0.56 & 0.558 & 0.638
& 0.598 & 0.612 \\ 
& d5 & 0.498 & 0.437 & 0.524 & 0.525 & \textbf{0.554} & 0.598 & 0.644 & 0.632
& 0.597 & 0.692 \\ 
& d6 & 0.492 & 0.435 & 0.542 & 0.573 & 0.593 & \textbf{0.621} & 0.631 & 0.618
& 0.669 & 0.693 \\ 
& d7 & 0.493 & 0.531 & 0.52 & 0.571 & 0.569 & 0.601 & \textbf{0.653} & 0.657
& 0.664 & 0.671 \\ 
& d8 & 0.479 & 0.564 & 0.574 & 0.576 & 0.62 & 0.642 & 0.692 & \textbf{0.721}
& 0.721 & 0.681 \\ 
& d9 & 0.514 & 0.618 & 0.634 & 0.63 & 0.609 & 0.64 & 0.693 & 0.712 & \textbf{0.732} & 0.737 \\ 
& d10 & 0.583 & 0.622 & 0.639 & 0.686 & 0.657 & 0.696 & 0.745 & 0.755 & 0.753
& \textbf{0.752} \\ \hline\hline
\end{tabular}}
\caption{Percentage in metro 1976-1980.}
\label{tab:metro_husbands}
\end{table}

\begin{table}[h]
\resizebox{15cm}{!}{
\begin{tabular}{l||l|rrrrrrrrrr}
\hline\hline
&  & \multicolumn{10}{c}{Husbands' quantiles} \\ \hline\hline
&  & d1 & d2 & d3 & d4 & d5 & d6 & d7 & d8 & d9 & d10 \\ \hline
\parbox[t]{2mm}{\multirow{10}{*}{\rotatebox[origin=c]{90}{Wives' quantiles}}}
	&d1	&	\textbf{0.581}	&	0.58	&	0.539	&	0.561	&	0.553	&	0.535	&	0.627	&	0.606	&	0.642	&	0.65	\\	
	&d2	&	0.579	&	\textbf{0.557}	&	0.564	&	0.556	&	0.556	&	0.587	&	0.638	&	0.598	&	0.664	&	0.712	\\	
	&d3	&	0.579	&	0.563	&	\textbf{0.561}	&	0.561	&	0.581	&	0.6	&	0.605	&	0.628	&	0.71	&	0.746	\\	
	&d4	&	0.599	&	0.572	&	0.565	&	\textbf{0.575}	&	0.539	&	0.573	&	0.59	&	0.659	&	0.683	&	0.686	\\	
	&d5	&	0.609	&	0.55	&	0.572	&	0.599	&	\textbf{0.648}	&	0.574	&	0.673	&	0.708	&	0.669	&	0.703	\\	
	&d6	&	0.575	&	0.579	&	0.589	&	0.647	&	0.647	&	\textbf{0.636}	&	0.659	&	0.686	&	0.745	&	0.777	\\	
	&d7	&	0.588	&	0.628	&	0.593	&	0.618	&	0.683	&	0.659	&	\textbf{0.725}	&	0.743	&	0.789	&	0.799	\\	
	&d8	&	0.623	&	0.661	&	0.613	&	0.651	&	0.7	&	0.687	&	0.749	&	\textbf{0.739}	&	0.826	&	0.82	\\	
	&d9	&	0.695	&	0.633	&	0.696	&	0.708	&	0.7	&	0.761	&	0.769	&	0.829	&	\textbf{0.848}	&	0.845	\\	
	&d10	&	0.729	&	0.652	&	0.709	&	0.706	&	0.767	&	0.77	&	0.785	&	0.85	&	0.872	&	\textbf{0.907}	\\	 \hline\hline
\end{tabular}}
\caption{Percentage in metro 2021-2025.}
\label{tab:metro_husbands1}
\end{table}

\begin{table}[h]
\resizebox{15cm}{!}{
\begin{tabular}{l||l|rrrrrrrrrr}
\hline\hline
&  & \multicolumn{10}{c}{Husbands' quantiles}   \\ \hline\hline
&  & d1 & d2 & d3 & d4 & d5 & d6 & d7 & d8 & d9 & d10   \\ \hline
\parbox[t]{2mm}{\multirow{10}{*}{\rotatebox[origin=c]{90}{Wives' quantiles}}}
& d1 & \textbf{0.064} & 0.053 & 0.049 & 0.045 & 0.05 & 0.056 & 0.058 & 0.048
& 0.063 & 0.149   \\ 
& d2 & 0.05 & \textbf{0.04} & 0.057 & 0.053 & 0.062 & 0.05 & 0.041 & 0.062 & 
0.095 & 0.079   \\ 
& d3 & 0.066 & 0.069 & \textbf{0.073} & 0.094 & 0.064 & 0.075 & 0.083 & 0.059
& 0.059 & 0.078   \\ 
& d4 & 0.082 & 0.109 & 0.09 & \textbf{0.115} & 0.106 & 0.098 & 0.106 & 0.073
& 0.077 & 0.129   \\ 
& d5 & 0.143 & 0.149 & 0.124 & 0.096 & \textbf{0.151} & 0.117 & 0.128 & 0.124
& 0.118 & 0.19   \\ 
& d6 & 0.194 & 0.193 & 0.177 & 0.173 & 0.149 & \textbf{0.169} & 0.17 & 0.155
& 0.125 & 0.186   \\ 
& d7 & 0.21 & 0.21 & 0.26 & 0.205 & 0.198 & 0.198 & \textbf{0.221} & 0.194 & 
0.224 & 0.243   \\ 
& d8 & 0.361 & 0.25 & 0.249 & 0.295 & 0.28 & 0.224 & 0.277 & \textbf{0.241}
& 0.267 & 0.332   \\ 
& d9 & 0.384 & 0.309 & 0.379 & 0.323 & 0.357 & 0.35 & 0.302 & 0.331 & 
\textbf{0.346} & 0.403   \\ 
& d10 & 0.402 & 0.378 & 0.389 & 0.407 & 0.413 & 0.435 & 0.44 & 0.484 & 0.507
& \textbf{0.552}   \\ \hline\hline
\end{tabular}}
\caption{Percentage of university education among wives for the years
1976-1980.}
\label{tab:educ_wives}
\end{table}

\begin{table}[h]
\resizebox{15cm}{!}{
\begin{tabular}{l||l|rrrrrrrrrr}
\hline\hline
&  & \multicolumn{10}{c}{Husbands' quantiles}   \\ \hline\hline
&  & d1 & d2 & d3 & d4 & d5 & d6 & d7 & d8 & d9 & d10   \\ \hline
\parbox[t]{2mm}{\multirow{10}{*}{\rotatebox[origin=c]{90}{Wives' quantiles}}}
& d1 & \textbf{0.069} & 0.084 & 0.093 & 0.096 & 0.136 & 0.12 & 0.124 & 0.135
& 0.149 & 0.385   \\ 
& d2 & 0.095 & \textbf{0.073} & 0.102 & 0.103 & 0.118 & 0.139 & 0.104 & 0.113
& 0.202 & 0.36   \\ 
& d3 & 0.089 & 0.092 & \textbf{0.133} & 0.148 & 0.12 & 0.154 & 0.166 & 0.153
& 0.189 & 0.352   \\ 
& d4 & 0.12 & 0.115 & 0.108 & \textbf{0.15} & 0.19 & 0.173 & 0.152 & 0.206 & 
0.172 & 0.397   \\ 
& d5 & 0.173 & 0.157 & 0.188 & 0.179 & \textbf{0.205} & 0.199 & 0.159 & 0.244
& 0.266 & 0.435   \\ 
& d6 & 0.177 & 0.21 & 0.216 & 0.207 & 0.171 & \textbf{0.231} & 0.262 & 0.259
& 0.269 & 0.434   \\ 
& d7 & 0.205 & 0.191 & 0.255 & 0.207 & 0.201 & 0.242 & \textbf{0.215} & 0.261
& 0.32 & 0.446   \\ 
& d8 & 0.273 & 0.211 & 0.242 & 0.263 & 0.255 & 0.251 & 0.29 & \textbf{0.284}
& 0.316 & 0.526   \\ 
& d9 & 0.281 & 0.275 & 0.23 & 0.28 & 0.299 & 0.328 & 0.281 & 0.379 & \textbf{\ 0.4} & 0.547   \\ 
& d10 & 0.283 & 0.296 & 0.285 & 0.304 & 0.283 & 0.389 & 0.348 & 0.473 & 0.496
& \textbf{0.64}   \\ \hline\hline
\end{tabular}}
\caption{Percentage of university education among husbands for the years
1976-1980.}
\label{tab:educ_husbands}
\end{table}

\begin{comment}

\begin{table}[tbp]
\resizebox{15cm}{!}{
\begin{tabular}{l||l|rrrrrrrrrrr}
\hline\hline
&  & \multicolumn{10}{c}{Husbands' quantiles} \\ \hline\hline
&  & d1 & d2 & d3 & d4 & d5 & d6 & d7 & d8 & d9 & d10 \\ \hline
\parbox[t]{2mm}{\multirow{10}{*}{\rotatebox[origin=c]{90}{Wives' quantiles}}}
& d1 & \textbf{0.034} & 0.032 & 0.034 & 0.025 & 0.039 & 0.033 & 0.027 & 0.035
& 0.041 & 0.133 \\ 
& d2 & 0.038 & \textbf{0.026} & 0.032 & 0.031 & 0.04 & 0.028 & 0.026 & 0.027
& 0.071 & 0.067 \\ 
& d3 & 0.04 & 0.043 & \textbf{0.046} & 0.067 & 0.037 & 0.048 & 0.062 & 0.045
& 0.043 & 0.073 \\ 
& d4 & 0.046 & 0.066 & 0.052 & \textbf{0.061} & 0.081 & 0.062 & 0.069 & 0.063
& 0.061 & 0.116 \\ 
& d5 & 0.1 & 0.096 & 0.093 & 0.064 & \textbf{0.092} & 0.08 & 0.075 & 0.083 & 
0.085 & 0.15 \\ 
& d6 & 0.121 & 0.115 & 0.124 & 0.11 & 0.093 & \textbf{0.113} & 0.128 & 0.108
& 0.084 & 0.134 \\ 
& d7 & 0.117 & 0.107 & 0.168 & 0.121 & 0.124 & 0.13 & \textbf{0.135} & 0.13
& 0.172 & 0.191 \\ 
& d8 & 0.216 & 0.127 & 0.152 & 0.187 & 0.169 & 0.137 & 0.161 & \textbf{0.167}
& 0.186 & 0.258 \\ 
& d9 & 0.205 & 0.174 & 0.17 & 0.189 & 0.227 & 0.236 & 0.204 & 0.24 & \textbf{0.268} & 0.325 \\ 
& d10 & 0.197 & 0.207 & 0.229 & 0.222 & 0.209 & 0.286 & 0.275 & 0.365 & 0.408
& \textbf{0.467} \\ \hline\hline
\end{tabular}}
\caption{Percentage of households with both university education for the
years 1976-1980.}
\label{tab:educ_both}
\end{table}
\end{comment}

\begin{table}[h]
\resizebox{15cm}{!}{
\begin{tabular}{l||l|rrrrrrrrrr}
\hline\hline
&  & \multicolumn{10}{c}{Husbands' quantiles} \\ \hline\hline
&  & d1 & d2 & d3 & d4 & d5 & d6 & d7 & d8 & d9 & d10 \\ \hline
\parbox[t]{2mm}{\multirow{10}{*}{\rotatebox[origin=c]{90}{Wives' quantiles}}}
	&d1	&	\textbf{0.15}	&	0.144	&	0.194	&	0.174	&	0.189	&	0.238	&	0.303	&	0.41	&	0.465	&	0.529	\\	
	&d2	&	0.179	&	\textbf{0.198}	&	0.231	&	0.217	&	0.286	&	0.277	&	0.345	&	0.387	&	0.495	&	0.622	\\	
	&d3	&	0.312	&	0.292	&	\textbf{0.314}	&	0.333	&	0.391	&	0.415	&	0.424	&	0.434	&	0.482	&	0.636	\\	
	&d4	&	0.342	&	0.385	&	0.398	&	\textbf{0.396}	&	0.444	&	0.509	&	0.49	&	0.532	&	0.603	&	0.657	\\	
	&d5	&	0.526	&	0.482	&	0.538	&	0.514	&	\textbf{0.535}	&	0.55	&	0.623	&	0.647	&	0.734	&	0.75	\\	
	&d6	&	0.544	&	0.565	&	0.531	&	0.56	&	0.631	&	\textbf{0.617}	&	0.657	&	0.73	&	0.755	&	0.802	\\	
	&d7	&	0.629	&	0.585	&	0.651	&	0.62	&	0.669	&	0.735	&	\textbf{0.724}	&	0.788	&	0.799	&	0.859	\\	
	&d8	&	0.669	&	0.736	&	0.686	&	0.724	&	0.773	&	0.788	&	0.829	&	\textbf{0.824}	&	0.842	&	0.869	\\	
	&d9	&	0.74	&	0.807	&	0.781	&	0.783	&	0.774	&	0.846	&	0.869	&	0.884	&	\textbf{0.896}	&	0.936	\\	
	&d10	&	0.847	&	0.787	&	0.879	&	0.887	&	0.83	&	0.89	&	0.909	&	0.94	&	0.941	&	\textbf{0.957}	\\
    \hline\hline
\end{tabular}}
\caption{Percentage of university education among wives for the years
2021-2025.}
\label{tab:educ_wives1}
\end{table}

\begin{table}[h]
\resizebox{15cm}{!}{
\begin{tabular}{l||l|rrrrrrrrrr}
\hline\hline
&  & \multicolumn{10}{c}{Husbands' quantiles} \\ \hline\hline
&  & d1 & d2 & d3 & d4 & d5 & d6 & d7 & d8 & d9 & d10 \\ \hline
\parbox[t]{2mm}{\multirow{10}{*}{\rotatebox[origin=c]{90}{Wives' quantiles}}}
	&d1	&	\textbf{0.115}	&	0.121	&	0.184	&	0.186	&	0.213	&	0.29	&	0.378	&	0.438	&	0.583	&	0.701	\\	
	&d2	&	0.126	&	\textbf{0.15}	&	0.2	&	0.196	&	0.25	&	0.304	&	0.376	&	0.544	&	0.582	&	0.788	\\	
	&d3	&	0.164	&	0.162	&	\textbf{0.24}	&	0.246	&	0.315	&	0.329	&	0.424	&	0.505	&	0.596	&	0.711	\\	
	&d4	&	0.188	&	0.242	&	0.277	&	\textbf{0.276}	&	0.366	&	0.442	&	0.428	&	0.594	&	0.62	&	0.775	\\	
	&d5	&	0.237	&	0.248	&	0.309	&	0.342	&	\textbf{0.38}	&	0.424	&	0.505	&	0.582	&	0.74	&	0.802	\\	
	&d6	&	0.228	&	0.262	&	0.321	&	0.333	&	0.426	&	\textbf{0.487}	&	0.566	&	0.644	&	0.732	&	0.858	\\	
	&d7	&	0.296	&	0.307	&	0.355	&	0.398	&	0.49	&	0.572	&	\textbf{0.616}	&	0.671	&	0.77	&	0.835	\\	
	&d8	&	0.309	&	0.356	&	0.391	&	0.429	&	0.479	&	0.555	&	0.691	&	\textbf{0.718}	&	0.795	&	0.869	\\	
	&d9	&	0.377	&	0.416	&	0.481	&	0.451	&	0.517	&	0.629	&	0.708	&	0.779	&	\textbf{0.85}	&	0.903	\\	
	&d10	&	0.551	&	0.404	&	0.473	&	0.541	&	0.593	&	0.674	&	0.732	&	0.809	&	0.868	&	\textbf{0.946}	\\ \hline\hline
\end{tabular}}
\caption{Percentage of university education among husbands for the years
2021-2025.}
\label{tab:educ_husbands1}
\end{table}

\begin{comment}
\begin{table}[tbp]
\resizebox{15cm}{!}{
\begin{tabular}{l||l|rrrrrrrrrr}
\hline\hline
&  & \multicolumn{10}{c}{Husbands' quantiles} \\ \hline\hline
&  & d1 & d2 & d3 & d4 & d5 & d6 & d7 & d8 & d9 & d10 \\ \hline
\parbox[t]{2mm}{\multirow{10}{*}{\rotatebox[origin=c]{90}{Wives' quantiles}}}
	&d1	&	\textbf{0.067}	&	0.062	&	0.107	&	0.091	&	0.109	&	0.149	&	0.207	&	0.281	&	0.38	&	0.49	\\	
	&d2	&	0.066	&	\textbf{0.096}	&	0.113	&	0.092	&	0.148	&	0.186	&	0.245	&	0.318	&	0.405	&	0.571	\\	
	&d3	&	0.113	&	0.111	&	\textbf{0.163}	&	0.171	&	0.22	&	0.223	&	0.277	&	0.333	&	0.384	&	0.549	\\	
	&d4	&	0.135	&	0.189	&	0.213	&	\textbf{0.212}	&	0.235	&	0.333	&	0.315	&	0.435	&	0.495	&	0.583	\\	
	&d5	&	0.197	&	0.204	&	0.248	&	0.278	&	\textbf{0.317}	&	0.33	&	0.41	&	0.486	&	0.627	&	0.668	\\	
	&d6	&	0.189	&	0.223	&	0.263	&	0.283	&	0.359	&	\textbf{0.418}	&	0.486	&	0.553	&	0.634	&	0.746	\\	
	&d7	&	0.25	&	0.267	&	0.296	&	0.332	&	0.436	&	0.525	&	\textbf{0.536}	&	0.616	&	0.678	&	0.757	\\	
	&d8	&	0.274	&	0.326	&	0.326	&	0.378	&	0.436	&	0.52	&	0.641	&	\textbf{0.657}	&	0.72	&	0.802	\\	
	&d9	&	0.344	&	0.392	&	0.43	&	0.426	&	0.467	&	0.586	&	0.674	&	0.746	&	\textbf{0.795}	&	0.866	\\	
	&d10	&	0.551	&	0.404	&	0.473	&	0.531	&	0.565	&	0.631	&	0.712	&	0.791	&	0.842	&	\textbf{0.924}	\\	 \hline\hline
\end{tabular}}
\caption{Percentage of university education among husbands and wives for the
years 2021-2025.}
\label{tab:educ_both1}
\end{table}
\end{comment}
\begin{table}[tbp]
\resizebox{15cm}{!}{
\begin{tabular}{l||l|rrrrrrrrrr}
\hline\hline
&  & \multicolumn{10}{c}{Husbands' quantiles}   \\ \hline\hline
&  & d1 & d2 & d3 & d4 & d5 & d6 & d7 & d8 & d9 & d10   \\ \hline
\parbox[t]{2mm}{\multirow{10}{*}{\rotatebox[origin=c]{90}{Wives' quantiles}}}
& d1 & \textbf{2.127} & 1.302 & 1.095 & 0.978 & 0.858 & 0.886 & 0.789 & 0.686
& 0.683 & 0.56   \\ 
& d2 & 1.712 & \textbf{1.374} & 1.196 & 1.115 & 0.983 & 0.868 & 0.828 & 0.779
& 0.761 & 0.549   \\ 
& d3 & 1.202 & 1.386 & \textbf{1.033} & 1.083 & 1.11 & 0.851 & 0.851 & 0.844
& 0.773 & 0.594   \\ 
& d4 & 1.188 & 1.454 & 1.137 & \textbf{1.144} & 0.964 & 0.948 & 0.928 & 0.915
& 0.784 & 0.676   \\ 
& d5 & 0.927 & 1.182 & 1.179 & 1.173 & \textbf{1.153} & 0.996 & 0.974 & 0.988
& 0.949 & 0.767   \\ 
& d6 & 0.714 & 1.0 & 1.003 & 1.094 & 1.016 & \textbf{1.041} & 0.972 & 0.993
& 0.936 & 0.814   \\ 
& d7 & 0.627 & 0.834 & 1.055 & 1.082 & 1.116 & 1.101 & \textbf{1.063} & 1.139
& 1.153 & 0.99   \\ 
& d8 & 0.577 & 0.622 & 0.784 & 1.146 & 1.065 & 1.082 & 1.106 & \textbf{1.126}
& 1.2 & 1.153   \\ 
& d9 & 0.484 & 0.657 & 0.688 & 0.773 & 1.172 & 1.208 & 1.296 & 1.301 & 
\textbf{1.228} & 1.402   \\ 
& d10 & 0.386 & 0.412 & 0.419 & 0.583 & 0.716 & 0.946 & 1.205 & 1.31 & 1.563
& \textbf{2.364}   \\ \hline\hline
\end{tabular}}
\caption{Estimated sorting measures for the years 1976-1980 at different
quantiles}
\label{tab:bivariate_1}
\end{table}

\begin{table}[tbp]
\resizebox{15cm}{!}{
\begin{tabular}{l||l|rrrrrrrrrrr}
\hline\hline
&  & \multicolumn{10}{c}{Husbands' quantiles}   \\ \hline\hline
&  & d1 & d2 & d3 & d4 & d5 & d6 & d7 & d8 & d9 & d10   \\ \hline
\parbox[t]{2mm}{\multirow{10}{*}{\rotatebox[origin=c]{90}{Wives' quantiles}}}
& d1 & \textbf{2.799} & 1.546 & 1.067 & 0.901 & 0.895 & 0.799 & 0.548 & 0.604
& 0.472 & 0.368   \\ 
& d2 & 1.789 & \textbf{1.745} & 1.353 & 1.109 & 0.94 & 0.833 & 0.656 & 0.681
& 0.541 & 0.352   \\ 
& d3 & 1.21 & 1.683 & \textbf{1.429} & 1.195 & 0.977 & 0.877 & 0.875 & 0.636
& 0.626 & 0.493   \\ 
& d4 & 0.861 & 1.204 & 1.594 & \textbf{1.244} & 1.127 & 1.003 & 0.935 & 0.776
& 0.708 & 0.548   \\ 
& d5 & 0.749 & 0.964 & 1.244 & 1.346 & \textbf{1.087} & 1.075 & 1.049 & 1.022
& 0.85 & 0.614   \\ 
& d6 & 0.662 & 0.693 & 0.883 & 1.336 & 1.408 & \textbf{1.234} & 1.139 & 0.999
& 0.93 & 0.715   \\ 
& d7 & 0.476 & 0.72 & 0.982 & 1.056 & 1.141 & 1.369 & \textbf{1.142} & 1.158
& 1.119 & 0.836   \\ 
& d8 & 0.48 & 0.55 & 0.695 & 0.745 & 0.961 & 1.264 & 1.512 & \textbf{1.22} & 
1.499 & 1.072   \\ 
& d9 & 0.421 & 0.551 & 0.566 & 0.687 & 0.755 & 0.971 & 1.229 & 1.435 & 
\textbf{1.801} & 1.583   \\ 
& d10 & 0.366 & 0.397 & 0.388 & 0.553 & 0.544 & 0.673 & 0.882 & 1.27 & 1.829
& \textbf{3.098}   \\ \hline\hline
\end{tabular}}
\caption{Estimated sorting measures for the years 2018-2025 at different
quantiles}
\label{tab:bivariate_2}
\end{table}

\begin{table}[h]
\resizebox{15cm}{!}{
\begin{tabular}{l||l|rrrrrrrrrrr}
\hline\hline
&  & \multicolumn{10}{c}{Husbands' quantiles}   \\ \hline\hline
&  & d1 & d2 & d3 & d4 & d5 & d6 & d7 & d8 & d9 & d10   \\ \hline
\parbox[t]{2mm}{\multirow{10}{*}{\rotatebox[origin=c]{90}{Wives' quantiles}}}
	&d1	&	\textbf{2.0}	&	1.251	&	1.034	&	0.931	&	0.858	&	0.92	&	0.798	&	0.735	&	0.78	&	0.695	\\	
	&d2	&	1.634	&	\textbf{1.346}	&	1.21	&	1.104	&	1.003	&	0.881	&	0.828	&	0.796	&	0.837	&	0.366	\\	
	&d3	&	1.251	&	1.389	&	\textbf{1.063}	&	1.098	&	1.107	&	0.893	&	0.882	&	0.852	&	0.823	&	0.646	\\	
	&d4	&	1.203	&	1.409	&	1.156	&	\textbf{1.129}	&	0.939	&	0.926	&	0.899	&	0.908	&	0.796	&	0.632	\\	
	&d5	&	0.945	&	1.126	&	1.202	&	1.135	&	\textbf{1.109}	&	0.985	&	0.935	&	0.946	&	0.947	&	0.673	\\	
	&d6	&	0.8	&	1.071	&	1.069	&	1.137	&	1.053	&	\textbf{1.093}	&	1.011	&	0.997	&	0.957	&	0.812	\\	
	&d7	&	0.667	&	0.829	&	1.069	&	1.065	&	1.105	&	1.124	&	\textbf{1.043}	&	1.11	&	1.174	&	0.82	\\	
	&d8	&	0.645	&	0.622	&	0.818	&	1.14	&	1.031	&	1.078	&	1.118	&	\textbf{1.1}	&	1.143	&	1.302	\\	
	&d9	&	0.481	&	0.641	&	0.667	&	0.751	&	1.138	&	1.139	&	1.168	&	1.185	&	\textbf{1.161}	&	1.666	\\	
	&d10	&	0.374	&	0.316	&	0.712	&	0.51	&	0.656	&	0.961	&	1.318	&	1.37	&	1.382	&	\textbf{2.388}	\\
    \hline\hline
\end{tabular}}
\caption{Results of estimated sorting measures for the years 1976-1980 at
different quantiles}
\label{tab:results_original_1976}
\end{table}

\begin{table}[h!]
\resizebox{15cm}{!}{
\begin{tabular}{l||l|rrrrrrrrrr}
\hline\hline
&  & \multicolumn{10}{c}{Husbands' quantiles}   \\ \hline\hline
&  & d1 & d2 & d3 & d4 & d5 & d6 & d7 & d8 & d9 & d10   \\ \hline
\parbox[t]{2mm}{\multirow{10}{*}{\rotatebox[origin=c]{90}{Wives' quantiles}}}
	&d1	&	\textbf{2.772}	&	1.593	&	1.251	&	0.932	&	0.898	&	0.694	&	0.632	&	0.653	&	0.515	&	0.061	\\	
	&d2	&	1.742	&	\textbf{1.479}	&	1.179	&	1.208	&	0.982	&	0.71	&	0.826	&	0.65	&	0.461	&	0.763	\\	
	&d3	&	1.175	&	1.713	&	\textbf{1.356}	&	1.155	&	1.248	&	0.888	&	0.818	&	0.714	&	0.556	&	0.377	\\	
	&d4	&	0.945	&	1.27	&	1.537	&	\textbf{1.179}	&	1.044	&	0.979	&	0.908	&	0.919	&	0.827	&	0.392	\\	
	&d5	&	0.81	&	0.978	&	1.193	&	1.356	&	\textbf{1.146}	&	0.944	&	0.976	&	1.016	&	0.812	&	0.77	\\	
	&d6	&	0.71	&	0.702	&	1.146	&	1.369	&	1.313	&	\textbf{1.263}	&	1.11	&	1.027	&	1.023	&	0.338	\\	
	&d7	&	0.691	&	0.721	&	0.737	&	0.929	&	1.073	&	1.428	&	\textbf{1.1}	&	1.064	&	0.991	&	1.266	\\	
	&d8	&	0.549	&	0.605	&	0.677	&	0.744	&	0.965	&	1.303	&	1.481	&	\textbf{1.401}	&	1.428	&	0.846	\\	
	&d9	&	0.528	&	0.522	&	0.663	&	0.616	&	0.822	&	0.94	&	1.244	&	1.455	&	\textbf{1.468}	&	1.742	\\	
	&d10	&	0.08	&	0.417	&	0.262	&	0.512	&	0.509	&	0.853	&	0.906	&	1.102	&	1.919	&	\textbf{3.441}	\\
    \hline\hline
\end{tabular}}
\caption{Results of estimated sorting measures for the years 2021-2025 at
different quantiles}
\label{tab:results_original_2018}
\end{table}

\end{document}

\title{\textbf{Web appendix to: Marital Sorting, Household Inequality and Selection}}
\author{Ivan Fern\'andez-Val\thanks{
\ Boston University.} \and Aico van Vuuren\thanks{
\ University of Groningen } \and Francis Vella \thanks{
\ Georgetown University}}
\maketitle

\section{Results for the restricted model (exogenous selection of husbands)}

\label{web_appendix::restricted}

Figures \ref{fig:results_bootstraps_25_1}-\ref{fig:results_bootstraps_75_1} report the results for our restricted model with exogenous selection of husbands.

\pgfplotsset{/pgfplots/group/.cd,
	horizontal sep=2cm,
	vertical sep=1.5cm
} 
\begin{figure}[tbp]
\caption{Results of covariance matrix at Q1 (exogenous selection of
husbands).}
\label{fig:results_bootstraps_25_1}%
\resizebox{14cm}{!}{
		% This file was created with tikzplotlib v0.10.1.
\begin{tikzpicture}

\definecolor{darkgray176}{RGB}{176,176,176}
\definecolor{steelblue31119180}{RGB}{31,119,180}

\begin{groupplot}[group style={group size=2 by 2}]
\nextgroupplot[
scaled x ticks=manual:{}{\pgfmathparse{#1}},
tick align=outside,
tick pos=left,
title={\(\displaystyle \rho_{D^*_w, Y^*_w}\)},
x grid style={darkgray176},
xmin=1975, xmax=2022,
xtick style={color=black},
xticklabels={},
y grid style={darkgray176},
ylabel={Correlation coefficient},
ymin=-0.382366313801214, ymax=0.203518169306898,
ytick style={color=black}
]
\path [draw=steelblue31119180, semithick]
(axis cs:1976,-0.335327578432484)
--(axis cs:1976,-0.189552874192768);

\path [draw=steelblue31119180, semithick]
(axis cs:1981,-0.355735200932663)
--(axis cs:1981,-0.214439878053186);

\path [draw=steelblue31119180, semithick]
(axis cs:1986,-0.235094670631829)
--(axis cs:1986,-0.0930466471488747);

\path [draw=steelblue31119180, semithick]
(axis cs:1991,-0.275788471671401)
--(axis cs:1991,-0.0851151211111225);

\path [draw=steelblue31119180, semithick]
(axis cs:1996,-0.167444007883782)
--(axis cs:1996,0.00584749067795118);

\path [draw=steelblue31119180, semithick]
(axis cs:2001,-0.097963667630263)
--(axis cs:2001,0.0507569549631014);

\path [draw=steelblue31119180, semithick]
(axis cs:2006,-0.0829445303672927)
--(axis cs:2006,0.0898411203192203);

\path [draw=steelblue31119180, semithick]
(axis cs:2011,-0.0358637652963815)
--(axis cs:2011,0.176887056438347);

\path [draw=steelblue31119180, semithick]
(axis cs:2016,-0.0465051582272872)
--(axis cs:2016,0.176594154345881);

\path [draw=steelblue31119180, semithick]
(axis cs:2021,-0.272950435372052)
--(axis cs:2021,-0.0355644464238534);

\addplot [semithick, steelblue31119180, mark=-, mark size=3, mark options={solid}, only marks]
table {%
1976 -0.335327578432484
1981 -0.355735200932663
1986 -0.235094670631829
1991 -0.275788471671401
1996 -0.167444007883782
2001 -0.097963667630263
2006 -0.0829445303672927
2011 -0.0358637652963815
2016 -0.0465051582272872
2021 -0.272950435372052
};
\addplot [semithick, steelblue31119180, mark=-, mark size=3, mark options={solid}, only marks]
table {%
1976 -0.189552874192768
1981 -0.214439878053186
1986 -0.0930466471488747
1991 -0.0851151211111225
1996 0.00584749067795118
2001 0.0507569549631014
2006 0.0898411203192203
2011 0.176887056438347
2016 0.176594154345881
2021 -0.0355644464238534
};
\addplot [semithick, steelblue31119180, mark=*, mark size=1.5, mark options={solid}, only marks]
table {%
1976 -0.262440226312626
1981 -0.285087539492925
1986 -0.164070658890352
1991 -0.180451796391262
1996 -0.0807982586029153
2001 -0.0236033563335808
2006 0.00344829497596377
2011 0.0705116455709829
2016 0.0650444980592969
2021 -0.154257440897953
};

\nextgroupplot[
tick align=outside,
tick pos=left,
title={\(\displaystyle \rho_{D^*_w, Y^*_h}\)},
x grid style={darkgray176},
xlabel={Years},
xmin=1975, xmax=2022,
xtick style={color=black},
y grid style={darkgray176},
ymin=-0.450566869670156, ymax=-0.0104398564890532,
ytick style={color=black}
]
\path [draw=steelblue31119180, semithick]
(axis cs:1976,-0.158333165540237)
--(axis cs:1976,-0.030445629815467);

\path [draw=steelblue31119180, semithick]
(axis cs:1981,-0.238445997701672)
--(axis cs:1981,-0.111972674630358);

\path [draw=steelblue31119180, semithick]
(axis cs:1986,-0.207300441212492)
--(axis cs:1986,-0.0732343831362034);

\path [draw=steelblue31119180, semithick]
(axis cs:1991,-0.269020812150876)
--(axis cs:1991,-0.116490267141391);

\path [draw=steelblue31119180, semithick]
(axis cs:1996,-0.271358104801924)
--(axis cs:1996,-0.113643901372111);

\path [draw=steelblue31119180, semithick]
(axis cs:2001,-0.293714916946702)
--(axis cs:2001,-0.144747902452527);

\path [draw=steelblue31119180, semithick]
(axis cs:2006,-0.280245517358501)
--(axis cs:2006,-0.118395951410379);

\path [draw=steelblue31119180, semithick]
(axis cs:2011,-0.430561096343742)
--(axis cs:2011,-0.267974190081793);

\path [draw=steelblue31119180, semithick]
(axis cs:2016,-0.316369612599349)
--(axis cs:2016,-0.141395098718611);

\path [draw=steelblue31119180, semithick]
(axis cs:2021,-0.379857962263666)
--(axis cs:2021,-0.188883978562752);

\addplot [semithick, steelblue31119180, mark=-, mark size=3, mark options={solid}, only marks]
table {%
1976 -0.158333165540237
1981 -0.238445997701672
1986 -0.207300441212492
1991 -0.269020812150876
1996 -0.271358104801924
2001 -0.293714916946702
2006 -0.280245517358501
2011 -0.430561096343742
2016 -0.316369612599349
2021 -0.379857962263666
};
\addplot [semithick, steelblue31119180, mark=-, mark size=3, mark options={solid}, only marks]
table {%
1976 -0.030445629815467
1981 -0.111972674630358
1986 -0.0732343831362034
1991 -0.116490267141391
1996 -0.113643901372111
2001 -0.144747902452527
2006 -0.118395951410379
2011 -0.267974190081793
2016 -0.141395098718611
2021 -0.188883978562752
};
\addplot [semithick, steelblue31119180, mark=*, mark size=1.5, mark options={solid}, only marks]
table {%
1976 -0.094389397677852
1981 -0.175209336166015
1986 -0.140267412174348
1991 -0.192755539646134
1996 -0.192501003087018
2001 -0.219231409699615
2006 -0.19932073438444
2011 -0.349267643212767
2016 -0.22888235565898
2021 -0.284370970413209
};

\nextgroupplot[
tick align=outside,
tick pos=left,
title={\(\displaystyle \rho_{Y^*_w, Y^*_h}\)},
x grid style={darkgray176},
xlabel={Years},
xmin=1975, xmax=2022,
xtick style={color=black},
y grid style={darkgray176},
ylabel={Correlation coefficient},
ymin=0.181343171950055, ymax=0.346580040433017,
ytick style={color=black}
]
\path [draw=steelblue31119180, semithick]
(axis cs:1976,0.18885393869928)
--(axis cs:1976,0.235680191109481);

\path [draw=steelblue31119180, semithick]
(axis cs:1981,0.252751054160192)
--(axis cs:1981,0.295059337844088);

\path [draw=steelblue31119180, semithick]
(axis cs:1986,0.264013008468089)
--(axis cs:1986,0.304717004933605);

\path [draw=steelblue31119180, semithick]
(axis cs:1991,0.27670754907361)
--(axis cs:1991,0.319577813950284);

\path [draw=steelblue31119180, semithick]
(axis cs:1996,0.274093995949278)
--(axis cs:1996,0.313166909590689);

\path [draw=steelblue31119180, semithick]
(axis cs:2001,0.230732793611896)
--(axis cs:2001,0.268166070359487);

\path [draw=steelblue31119180, semithick]
(axis cs:2006,0.252906827530295)
--(axis cs:2006,0.294287536920959);

\path [draw=steelblue31119180, semithick]
(axis cs:2011,0.213439256599096)
--(axis cs:2011,0.27495520658702);

\path [draw=steelblue31119180, semithick]
(axis cs:2016,0.240875193002638)
--(axis cs:2016,0.293926123300152);

\path [draw=steelblue31119180, semithick]
(axis cs:2021,0.296069068188741)
--(axis cs:2021,0.339069273683791);

\addplot [semithick, steelblue31119180, mark=-, mark size=3, mark options={solid}, only marks]
table {%
1976 0.18885393869928
1981 0.252751054160192
1986 0.264013008468089
1991 0.27670754907361
1996 0.274093995949278
2001 0.230732793611896
2006 0.252906827530295
2011 0.213439256599096
2016 0.240875193002638
2021 0.296069068188741
};
\addplot [semithick, steelblue31119180, mark=-, mark size=3, mark options={solid}, only marks]
table {%
1976 0.235680191109481
1981 0.295059337844088
1986 0.304717004933605
1991 0.319577813950284
1996 0.313166909590689
2001 0.268166070359487
2006 0.294287536920959
2011 0.27495520658702
2016 0.293926123300152
2021 0.339069273683791
};
\addplot [semithick, steelblue31119180, mark=*, mark size=1.5, mark options={solid}, only marks]
table {%
1976 0.212267064904381
1981 0.27390519600214
1986 0.284365006700847
1991 0.298142681511947
1996 0.293630452769984
2001 0.249449431985691
2006 0.273597182225627
2011 0.244197231593058
2016 0.267400658151395
2021 0.317569170936266
};
\end{groupplot}

\end{tikzpicture}
	}
\end{figure}
\begin{figure}[tbp]
\caption{Results of covariance matrix at Q2 (exogenous selection of
husbands).}
\label{fig:results_bootstraps_50_1}%
\resizebox{14cm}{!}{
		% This file was created with tikzplotlib v0.10.1.
\begin{tikzpicture}

\definecolor{darkgray176}{RGB}{176,176,176}
\definecolor{steelblue31119180}{RGB}{31,119,180}

\begin{groupplot}[group style={group size=2 by 2}]
\nextgroupplot[
scaled x ticks=manual:{}{\pgfmathparse{#1}},
tick align=outside,
tick pos=left,
title={\(\displaystyle \rho_{D^*_w, Y^*_w}\)},
x grid style={darkgray176},
xmin=1975, xmax=2022,
xtick style={color=black},
xticklabels={},
y grid style={darkgray176},
ylabel={Correlation coefficient},
ymin=-0.3415827815273, ymax=0.453007669324621,
ytick style={color=black}
]
\path [draw=steelblue31119180, semithick]
(axis cs:1976,-0.305465033761304)
--(axis cs:1976,-0.174253266603757);

\path [draw=steelblue31119180, semithick]
(axis cs:1981,-0.299889559962985)
--(axis cs:1981,-0.16774001184665);

\path [draw=steelblue31119180, semithick]
(axis cs:1986,-0.211768913821239)
--(axis cs:1986,-0.076535596527536);

\path [draw=steelblue31119180, semithick]
(axis cs:1991,-0.164581808127037)
--(axis cs:1991,-0.00511326264105506);

\path [draw=steelblue31119180, semithick]
(axis cs:1996,0.00944129878609075)
--(axis cs:1996,0.175802836101989);

\path [draw=steelblue31119180, semithick]
(axis cs:2001,0.0885396840588397)
--(axis cs:2001,0.221638342990784);

\path [draw=steelblue31119180, semithick]
(axis cs:2006,0.088131743812423)
--(axis cs:2006,0.289121410886252);

\path [draw=steelblue31119180, semithick]
(axis cs:2011,0.186909136616056)
--(axis cs:2011,0.371065072169811);

\path [draw=steelblue31119180, semithick]
(axis cs:2016,0.246265693788439)
--(axis cs:2016,0.416889921558624);

\path [draw=steelblue31119180, semithick]
(axis cs:2021,-0.112684040311992)
--(axis cs:2021,0.10103872230989);

\addplot [semithick, steelblue31119180, mark=-, mark size=3, mark options={solid}, only marks]
table {%
1976 -0.305465033761304
1981 -0.299889559962985
1986 -0.211768913821239
1991 -0.164581808127037
1996 0.00944129878609075
2001 0.0885396840588397
2006 0.088131743812423
2011 0.186909136616056
2016 0.246265693788439
2021 -0.112684040311992
};
\addplot [semithick, steelblue31119180, mark=-, mark size=3, mark options={solid}, only marks]
table {%
1976 -0.174253266603757
1981 -0.16774001184665
1986 -0.076535596527536
1991 -0.00511326264105506
1996 0.175802836101989
2001 0.221638342990784
2006 0.289121410886252
2011 0.371065072169811
2016 0.416889921558624
2021 0.10103872230989
};
\addplot [semithick, steelblue31119180, mark=*, mark size=1.5, mark options={solid}, only marks]
table {%
1976 -0.239859150182531
1981 -0.233814785904818
1986 -0.144152255174388
1991 -0.0848475353840458
1996 0.0926220674440401
2001 0.155089013524812
2006 0.188626577349338
2011 0.278987104392934
2016 0.331577807673532
2021 -0.00582265900105128
};

\nextgroupplot[
tick align=outside,
tick pos=left,
title={\(\displaystyle \rho_{D^*_w, Y^*_h}\)},
x grid style={darkgray176},
xlabel={Years},
xmin=1975, xmax=2022,
xtick style={color=black},
y grid style={darkgray176},
ymin=-0.473070315758984, ymax=-0.0511501843152443,
ytick style={color=black}
]
\path [draw=steelblue31119180, semithick]
(axis cs:1976,-0.18023019195282)
--(axis cs:1976,-0.0703283721081415);

\path [draw=steelblue31119180, semithick]
(axis cs:1981,-0.225149450831304)
--(axis cs:1981,-0.104454308927797);

\path [draw=steelblue31119180, semithick]
(axis cs:1986,-0.234783748615162)
--(axis cs:1986,-0.114791992259512);

\path [draw=steelblue31119180, semithick]
(axis cs:1991,-0.259104018388012)
--(axis cs:1991,-0.114600232792729);

\path [draw=steelblue31119180, semithick]
(axis cs:1996,-0.270280550956129)
--(axis cs:1996,-0.126704482895907);

\path [draw=steelblue31119180, semithick]
(axis cs:2001,-0.266376532416877)
--(axis cs:2001,-0.150881033799521);

\path [draw=steelblue31119180, semithick]
(axis cs:2006,-0.339875818767401)
--(axis cs:2006,-0.184999069462008);

\path [draw=steelblue31119180, semithick]
(axis cs:2011,-0.366600401519154)
--(axis cs:2011,-0.233570816036831);

\path [draw=steelblue31119180, semithick]
(axis cs:2016,-0.315161037991573)
--(axis cs:2016,-0.144238018351449);

\path [draw=steelblue31119180, semithick]
(axis cs:2021,-0.453892127966087)
--(axis cs:2021,-0.244235818442047);

\addplot [semithick, steelblue31119180, mark=-, mark size=3, mark options={solid}, only marks]
table {%
1976 -0.18023019195282
1981 -0.225149450831304
1986 -0.234783748615162
1991 -0.259104018388012
1996 -0.270280550956129
2001 -0.266376532416877
2006 -0.339875818767401
2011 -0.366600401519154
2016 -0.315161037991573
2021 -0.453892127966087
};
\addplot [semithick, steelblue31119180, mark=-, mark size=3, mark options={solid}, only marks]
table {%
1976 -0.0703283721081415
1981 -0.104454308927797
1986 -0.114791992259512
1991 -0.114600232792729
1996 -0.126704482895907
2001 -0.150881033799521
2006 -0.184999069462008
2011 -0.233570816036831
2016 -0.144238018351449
2021 -0.244235818442047
};
\addplot [semithick, steelblue31119180, mark=*, mark size=1.5, mark options={solid}, only marks]
table {%
1976 -0.125279282030481
1981 -0.164801879879551
1986 -0.174787870437337
1991 -0.186852125590371
1996 -0.198492516926018
2001 -0.208628783108199
2006 -0.262437444114704
2011 -0.300085608777993
2016 -0.229699528171511
2021 -0.349063973204067
};

\nextgroupplot[
tick align=outside,
tick pos=left,
title={\(\displaystyle \rho_{Y^*_w, Y^*_h}\)},
x grid style={darkgray176},
xlabel={Years},
xmin=1975, xmax=2022,
xtick style={color=black},
y grid style={darkgray176},
ylabel={Correlation coefficient},
ymin=0.142843643911394, ymax=0.332664450381295,
ytick style={color=black}
]
\path [draw=steelblue31119180, semithick]
(axis cs:1976,0.198721232978846)
--(axis cs:1976,0.23752248521003);

\path [draw=steelblue31119180, semithick]
(axis cs:1981,0.211462088542752)
--(axis cs:1981,0.253833321877258);

\path [draw=steelblue31119180, semithick]
(axis cs:1986,0.237307048880619)
--(axis cs:1986,0.271944823457665);

\path [draw=steelblue31119180, semithick]
(axis cs:1991,0.2371520792742)
--(axis cs:1991,0.273552912783725);

\path [draw=steelblue31119180, semithick]
(axis cs:1996,0.199950852109797)
--(axis cs:1996,0.23676407971198);

\path [draw=steelblue31119180, semithick]
(axis cs:2001,0.197730937250178)
--(axis cs:2001,0.231826294452191);

\path [draw=steelblue31119180, semithick]
(axis cs:2006,0.209709616800329)
--(axis cs:2006,0.259812620587168);

\path [draw=steelblue31119180, semithick]
(axis cs:2011,0.151471862387299)
--(axis cs:2011,0.207373510146215);

\path [draw=steelblue31119180, semithick]
(axis cs:2016,0.184400087535799)
--(axis cs:2016,0.232031435201219);

\path [draw=steelblue31119180, semithick]
(axis cs:2021,0.276192127496422)
--(axis cs:2021,0.32403623190539);

\addplot [semithick, steelblue31119180, mark=-, mark size=3, mark options={solid}, only marks]
table {%
1976 0.198721232978846
1981 0.211462088542752
1986 0.237307048880619
1991 0.2371520792742
1996 0.199950852109797
2001 0.197730937250178
2006 0.209709616800329
2011 0.151471862387299
2016 0.184400087535799
2021 0.276192127496422
};
\addplot [semithick, steelblue31119180, mark=-, mark size=3, mark options={solid}, only marks]
table {%
1976 0.23752248521003
1981 0.253833321877258
1986 0.271944823457665
1991 0.273552912783725
1996 0.23676407971198
2001 0.231826294452191
2006 0.259812620587168
2011 0.207373510146215
2016 0.232031435201219
2021 0.32403623190539
};
\addplot [semithick, steelblue31119180, mark=*, mark size=1.5, mark options={solid}, only marks]
table {%
1976 0.218121859094438
1981 0.232647705210005
1986 0.254625936169142
1991 0.255352496028962
1996 0.218357465910889
2001 0.214778615851185
2006 0.234761118693748
2011 0.179422686266757
2016 0.208215761368509
2021 0.300114179700906
};
\end{groupplot}

\end{tikzpicture}
	}
\end{figure}
\begin{figure}[tbp]
\caption{Results of covariance matrix at Q3 (exogenous selection of
husbands).}
\label{fig:results_bootstraps_75_1}%
\resizebox{14cm}{!}{
		% This file was created with tikzplotlib v0.10.1.
\begin{tikzpicture}

\definecolor{darkgray176}{RGB}{176,176,176}
\definecolor{steelblue31119180}{RGB}{31,119,180}

\begin{groupplot}[group style={group size=2 by 2}]
\nextgroupplot[
scaled x ticks=manual:{}{\pgfmathparse{#1}},
tick align=outside,
tick pos=left,
title={\(\displaystyle \rho_{D^*_w, Y^*_w}\)},
x grid style={darkgray176},
xmin=1975, xmax=2022,
xtick style={color=black},
xticklabels={},
y grid style={darkgray176},
ylabel={Correlation coefficient},
ymin=-0.294120815918366, ymax=0.543304053324524,
ytick style={color=black}
]
\path [draw=steelblue31119180, semithick]
(axis cs:1976,-0.252636788017483)
--(axis cs:1976,-0.125947925193145);

\path [draw=steelblue31119180, semithick]
(axis cs:1981,-0.256056049134598)
--(axis cs:1981,-0.132337124363596);

\path [draw=steelblue31119180, semithick]
(axis cs:1986,-0.211661804691858)
--(axis cs:1986,-0.0637335603997455);

\path [draw=steelblue31119180, semithick]
(axis cs:1991,0.00782049497299231)
--(axis cs:1991,0.173337701970441);

\path [draw=steelblue31119180, semithick]
(axis cs:1996,0.182299888764599)
--(axis cs:1996,0.363705968401697);

\path [draw=steelblue31119180, semithick]
(axis cs:2001,0.233400936558304)
--(axis cs:2001,0.379702179562079);

\path [draw=steelblue31119180, semithick]
(axis cs:2006,0.241020109648757)
--(axis cs:2006,0.415357760226989);

\path [draw=steelblue31119180, semithick]
(axis cs:2011,0.240101471308857)
--(axis cs:2011,0.448245366120455);

\path [draw=steelblue31119180, semithick]
(axis cs:2016,0.314032423195537)
--(axis cs:2016,0.505239286540756);

\path [draw=steelblue31119180, semithick]
(axis cs:2021,-0.0655181952224434)
--(axis cs:2021,0.169011913376479);

\addplot [semithick, steelblue31119180, mark=-, mark size=3, mark options={solid}, only marks]
table {%
1976 -0.252636788017483
1981 -0.256056049134598
1986 -0.211661804691858
1991 0.00782049497299231
1996 0.182299888764599
2001 0.233400936558304
2006 0.241020109648757
2011 0.240101471308857
2016 0.314032423195537
2021 -0.0655181952224434
};
\addplot [semithick, steelblue31119180, mark=-, mark size=3, mark options={solid}, only marks]
table {%
1976 -0.125947925193145
1981 -0.132337124363596
1986 -0.0637335603997455
1991 0.173337701970441
1996 0.363705968401697
2001 0.379702179562079
2006 0.415357760226989
2011 0.448245366120455
2016 0.505239286540756
2021 0.169011913376479
};
\addplot [semithick, steelblue31119180, mark=*, mark size=1.5, mark options={solid}, only marks]
table {%
1976 -0.189292356605314
1981 -0.194196586749097
1986 -0.137697682545802
1991 0.0905790984717167
1996 0.273002928583148
2001 0.306551558060192
2006 0.328188934937873
2011 0.344173418714656
2016 0.409635854868147
2021 0.0517468590770176
};

\nextgroupplot[
tick align=outside,
tick pos=left,
title={\(\displaystyle \rho_{D^*_w, Y^*_h}\)},
x grid style={darkgray176},
xlabel={Years},
xmin=1975, xmax=2022,
xtick style={color=black},
y grid style={darkgray176},
ymin=-0.441761632946379, ymax=-0.00456566094514929,
ytick style={color=black}
]
\path [draw=steelblue31119180, semithick]
(axis cs:1976,-0.173152213141182)
--(axis cs:1976,-0.0315755906272698);

\path [draw=steelblue31119180, semithick]
(axis cs:1981,-0.238368236012869)
--(axis cs:1981,-0.100711234057405);

\path [draw=steelblue31119180, semithick]
(axis cs:1986,-0.23228564028636)
--(axis cs:1986,-0.109658450546842);

\path [draw=steelblue31119180, semithick]
(axis cs:1991,-0.248129060200329)
--(axis cs:1991,-0.112125513265045);

\path [draw=steelblue31119180, semithick]
(axis cs:1996,-0.190500621985557)
--(axis cs:1996,-0.0244382051270234);

\path [draw=steelblue31119180, semithick]
(axis cs:2001,-0.201512070047711)
--(axis cs:2001,-0.0455137632891217);

\path [draw=steelblue31119180, semithick]
(axis cs:2006,-0.303330752387927)
--(axis cs:2006,-0.152639212375155);

\path [draw=steelblue31119180, semithick]
(axis cs:2011,-0.389429031473449)
--(axis cs:2011,-0.220810309233897);

\path [draw=steelblue31119180, semithick]
(axis cs:2016,-0.282337670096633)
--(axis cs:2016,-0.108886244329377);

\path [draw=steelblue31119180, semithick]
(axis cs:2021,-0.421889088764505)
--(axis cs:2021,-0.219733449419992);

\addplot [semithick, steelblue31119180, mark=-, mark size=3, mark options={solid}, only marks]
table {%
1976 -0.173152213141182
1981 -0.238368236012869
1986 -0.23228564028636
1991 -0.248129060200329
1996 -0.190500621985557
2001 -0.201512070047711
2006 -0.303330752387927
2011 -0.389429031473449
2016 -0.282337670096633
2021 -0.421889088764505
};
\addplot [semithick, steelblue31119180, mark=-, mark size=3, mark options={solid}, only marks]
table {%
1976 -0.0315755906272698
1981 -0.100711234057405
1986 -0.109658450546842
1991 -0.112125513265045
1996 -0.0244382051270234
2001 -0.0455137632891217
2006 -0.152639212375155
2011 -0.220810309233897
2016 -0.108886244329377
2021 -0.219733449419992
};
\addplot [semithick, steelblue31119180, mark=*, mark size=1.5, mark options={solid}, only marks]
table {%
1976 -0.102363901884226
1981 -0.169539735035137
1986 -0.170972045416601
1991 -0.180127286732687
1996 -0.10746941355629
2001 -0.123512916668416
2006 -0.227984982381541
2011 -0.305119670353673
2016 -0.195611957213005
2021 -0.320811269092248
};

\nextgroupplot[
tick align=outside,
tick pos=left,
title={\(\displaystyle \rho_{Y^*_w, Y^*_h}\)},
x grid style={darkgray176},
xlabel={Years},
xmin=1975, xmax=2022,
xtick style={color=black},
y grid style={darkgray176},
ylabel={Correlation coefficient},
ymin=0.143005607760663, ymax=0.33572008828261,
ytick style={color=black}
]
\path [draw=steelblue31119180, semithick]
(axis cs:1976,0.164122558164908)
--(axis cs:1976,0.204832753857966);

\path [draw=steelblue31119180, semithick]
(axis cs:1981,0.19598484102107)
--(axis cs:1981,0.242336150126219);

\path [draw=steelblue31119180, semithick]
(axis cs:1986,0.21633062542162)
--(axis cs:1986,0.256205462178187);

\path [draw=steelblue31119180, semithick]
(axis cs:1991,0.210508790293773)
--(axis cs:1991,0.246778198031691);

\path [draw=steelblue31119180, semithick]
(axis cs:1996,0.191676661969896)
--(axis cs:1996,0.243808882222331);

\path [draw=steelblue31119180, semithick]
(axis cs:2001,0.175843813668412)
--(axis cs:2001,0.219583975369305);

\path [draw=steelblue31119180, semithick]
(axis cs:2006,0.197984665464826)
--(axis cs:2006,0.254731073927942);

\path [draw=steelblue31119180, semithick]
(axis cs:2011,0.151765356875297)
--(axis cs:2011,0.218924757025193);

\path [draw=steelblue31119180, semithick]
(axis cs:2016,0.209573170400529)
--(axis cs:2016,0.276962851350622);

\path [draw=steelblue31119180, semithick]
(axis cs:2021,0.271518953521615)
--(axis cs:2021,0.326960339167976);

\addplot [semithick, steelblue31119180, mark=-, mark size=3, mark options={solid}, only marks]
table {%
1976 0.164122558164908
1981 0.19598484102107
1986 0.21633062542162
1991 0.210508790293773
1996 0.191676661969896
2001 0.175843813668412
2006 0.197984665464826
2011 0.151765356875297
2016 0.209573170400529
2021 0.271518953521615
};
\addplot [semithick, steelblue31119180, mark=-, mark size=3, mark options={solid}, only marks]
table {%
1976 0.204832753857966
1981 0.242336150126219
1986 0.256205462178187
1991 0.246778198031691
1996 0.243808882222331
2001 0.219583975369305
2006 0.254731073927942
2011 0.218924757025193
2016 0.276962851350622
2021 0.326960339167976
};
\addplot [semithick, steelblue31119180, mark=*, mark size=1.5, mark options={solid}, only marks]
table {%
1976 0.184477656011437
1981 0.219160495573644
1986 0.236268043799904
1991 0.228643494162732
1996 0.217742772096113
2001 0.197713894518859
2006 0.226357869696384
2011 0.185345056950245
2016 0.243268010875576
2021 0.299239646344795
};
\end{groupplot}

\end{tikzpicture}
	}
\end{figure}
\section{Robustness check using younger individuals and couples that are married at most two years.}

\label{web_appendix::robust_newly_wed}

As discussed in the main text, our results can be driven by the decisions made during marriage. In order to investigate this, we present a robustness check for individuals who are in the age range 25-35 years. This implies that any couple that has at least one person who is older than this age range will be excluded from the sample. As these individuals are young, they are likely to not have been married for a long period of time. Therefore, Tables \ref{tab:results_1_1976_young}-\ref{tab:results_4_1976_young} redo our main results of the selection model for the year 1976-1980 using only individuals who are aged between 25 and 35 years. These results should be compared with the results in Tables \ref{tab:results_1_1976} to \ref{tab:results_5_1976} in the main text. We do not find large quantitative differences between these Tables. Similarly, we reproduce 
our results in Tables \ref{tab:results_1_2018}-\ref{tab:results_5_2018} in the main text for the years 2021-2025 in Tables \ref{tab:results_1_2021_young}-\ref{tab:results_5_2021_young}.

\begin{table}[tbp]
\begin{tabular}{l||l|rrrrrrrrrrr}
\hline\hline
&  & \multicolumn{10}{c}{Husbands' quantiles} &  \\
 \hline\hline
&  & d1 & d2 & d3 & d4 & d5 & d6 & d7 & d8 & d9 & d10 &  \\
 \hline
\parbox[t]{2mm}{\multirow{10}{*}{\rotatebox[origin=c]{90}{Wives' quantiles}}}
	& d1	&	\textbf{2.117}	&	1.334	&	1.043	&	1.059	&	0.887	&	0.938	&	0.827	&	0.88	&	0.728	&	0.244	\\	
	& d2	&	1.643	&	\textbf{1.434}	&	1.304	&	1.1	&	0.958	&	0.877	&	0.697	&	0.573	&	0.583	&	0.77	\\	
	& d3	&	1.191	&	1.481	&	\textbf{1.11}	&	1.131	&	1.182	&	0.91	&	0.92	&	0.895	&	0.602	&	0.582	\\	
	& d4	&	1.039	&	1.267	&	1.272	&	\textbf{1.032}	&	1.128	&	0.997	&	0.909	&	0.907	&	0.847	&	0.593	\\	
	& d5	&	0.805	&	0.984	&	1.222	&	1.205	&	\textbf{1.114}	&	1.171	&	0.865	&	0.897	&	0.976	&	0.771	\\	
	& d6	&	0.799	&	1.014	&	0.973	&	1.132	&	1.227	&	\textbf{1.069}	&	0.999	&	1.131	&	0.898	&	0.742	\\	
	& d7	&	0.778	&	0.709	&	1.052	&	1.199	&	1.015	&	1.23	&	\textbf{1.13}	&	1.114	&	0.986	&	0.801	\\	
	& d8	&	0.637	&	0.571	&	0.812	&	1.045	&	1.273	&	1.165	&	1.133	&	\textbf{1.344}	&	1.132	&	0.87	\\	
	& d9	&	0.554	&	0.607	&	0.672	&	0.603	&	0.726	&	1.125	&	1.166	&	1.309	&	\textbf{1.259}	&	1.997	\\	
	& d10	&	0.436	&	0.599	&	0.539	&	0.492	&	0.49	&	0.518	&	1.353	&	0.949	&	1.99	&	\textbf{2.63}	\\	
\hline\hline
\end{tabular}
\caption{Results of estimated sorting measures for individuals 25-35 years old in 1976-1980 at different quantiles when $\protect\sigma_{D^*_w, Y^*_w} = 0$.}
\label{tab:results_1_1976_young}
\end{table}

\begin{table}[tbp]
\begin{tabular}{l||l|rrrrrrrrrrr}
\hline\hline
&  & \multicolumn{10}{c}{Husbands' quantiles} &  \\
 \hline\hline
&  & d1 & d2 & d3 & d4 & d5 & d6 & d7 & d8 & d9 & d10 &  \\
 \hline
\parbox[t]{2mm}{\multirow{10}{*}{\rotatebox[origin=c]{90}{Wives' quantiles}}}
	& d1	&	\textbf{2.131}	&	1.394	&	1.042	&	1.079	&	0.822	&	0.825	&	0.737	&	0.693	&	0.704	&	0.565	\\	
	& d2	&	1.446	&	\textbf{1.515}	&	1.297	&	1.148	&	1.117	&	0.818	&	0.86	&	0.681	&	0.592	&	0.531	\\	
	& d3	&	1.065	&	1.651	&	\textbf{1.262}	&	1.09	&	1.003	&	0.891	&	0.798	&	0.866	&	0.818	&	0.58	\\	
	& d4	&	0.926	&	1.075	&	1.259	&	\textbf{1.26}	&	1.114	&	1.123	&	0.807	&	0.903	&	0.834	&	0.7	\\	
	& d5	&	0.922	&	1.088	&	1.096	&	1.119	&	\textbf{1.157}	&	0.932	&	1.057	&	0.956	&	0.909	&	0.712	\\	
	& d6	&	0.87	&	0.865	&	1.171	&	1.209	&	0.984	&	\textbf{1.047}	&	1.122	&	0.989	&	0.968	&	0.817	\\	
	& d7	&	0.742	&	0.728	&	0.953	&	1.133	&	1.202	&	1.136	&	\textbf{0.987}	&	1.029	&	1.191	&	0.926	\\	
	& d8	&	0.624	&	0.669	&	0.722	&	0.871	&	1.138	&	1.163	&	1.277	&	\textbf{1.141}	&	1.138	&	1.169	\\	
	& d9	&	0.83	&	0.48	&	0.69	&	0.607	&	0.78	&	1.215	&	1.359	&	1.401	&	\textbf{1.284}	&	1.361	\\	
	& d10	&	0.444	&	0.535	&	0.508	&	0.484	&	0.685	&	0.849	&	0.998	&	1.341	&	1.563	&	\textbf{2.64}	\\	
\hline\hline
\end{tabular}
\caption{Results of estimated sorting measures for individuals 25-35 years old in 1976-1980 at different quantiles when $\protect\sigma_{D^*_w, Y^*_w} = 0, \protect\sigma_{D^*_w, Y^*_h} = 0.$}
\label{tab:results_2_1976_young}
\end{table}

\begin{table}[tbp]
\begin{tabular}{l||l|rrrrrrrrrrr}
\hline\hline
&  & \multicolumn{10}{c}{Husbands' quantiles} &  \\
 \hline\hline
&  & d1 & d2 & d3 & d4 & d5 & d6 & d7 & d8 & d9 & d10 &  \\
 \hline
\parbox[t]{2mm}{\multirow{10}{*}{\rotatebox[origin=c]{90}{Wives' quantiles}}}
	& d1	&	\textbf{2.142}	&	1.374	&	1.038	&	1.12	&	0.785	&	0.837	&	0.757	&	0.733	&	0.64	&	0.541	\\	
	& d2	&	1.455	&	\textbf{1.567}	&	1.297	&	1.075	&	1.11	&	0.844	&	0.799	&	0.703	&	0.644	&	0.535	\\	
	& d3	&	1.165	&	1.623	&	\textbf{1.253}	&	1.102	&	0.984	&	0.852	&	0.825	&	0.794	&	0.791	&	0.603	\\	
	& d4	&	1.016	&	1.027	&	1.252	&	\textbf{1.362}	&	1.039	&	1.034	&	0.875	&	0.924	&	0.839	&	0.644	\\	
	& d5	&	0.812	&	1.189	&	1.125	&	1.012	&	\textbf{1.307}	&	0.965	&	0.986	&	0.986	&	0.885	&	0.711	\\	
	& d6	&	0.876	&	0.764	&	1.222	&	1.275	&	1.0	&	\textbf{1.081}	&	1.076	&	0.954	&	0.998	&	0.766	\\	
	& d7	&	0.782	&	0.73	&	0.926	&	1.005	&	1.154	&	1.146	&	\textbf{1.108}	&	0.994	&	1.229	&	0.933	\\	
	& d8	&	0.587	&	0.673	&	0.687	&	0.923	&	1.1	&	1.173	&	1.286	&	\textbf{1.202}	&	1.12	&	1.244	\\	
	& d9	&	0.767	&	0.543	&	0.718	&	0.648	&	0.835	&	1.231	&	1.295	&	1.276	&	\textbf{1.316}	&	1.37	\\	
	& d10	&	0.397	&	0.51	&	0.481	&	0.478	&	0.685	&	0.839	&	0.995	&	1.433	&	1.539	&	\textbf{2.654}	\\	
\hline\hline
\end{tabular}
\caption{Results of estimated sorting measures for individuals 25-35 years old in 1976-1980 at different quantiles when $\protect\protect\sigma_{D^*_w, Y^*_w} = 0, \protect\sigma_{D^*_w, Y^*_h} = 0$, and the distribution of $X$ as in total population.}
\label{tab:results_3_1976_young}
\end{table}

\begin{table}[tbp]
\begin{tabular}{l||l|rrrrrrrrrrr}
\hline\hline
&  & \multicolumn{10}{c}{Husbands' quantiles} &  \\
 \hline\hline
&  & d1 & d2 & d3 & d4 & d5 & d6 & d7 & d8 & d9 & d10 &  \\
 \hline
\parbox[t]{2mm}{\multirow{10}{*}{\rotatebox[origin=c]{90}{Wives' quantiles}}}
	& d1	&	\textbf{1.313}	&	1.224	&	1.112	&	1.071	&	0.979	&	0.921	&	0.888	&	0.845	&	0.829	&	0.783	\\	
	& d2	&	1.232	&	\textbf{1.221}	&	1.141	&	1.096	&	1.01	&	0.963	&	0.906	&	0.865	&	0.835	&	0.763	\\	
	& d3	&	1.139	&	1.136	&	\textbf{1.096}	&	1.063	&	1.006	&	0.986	&	0.936	&	0.912	&	0.884	&	0.838	\\	
	& d4	&	1.03	&	1.073	&	1.052	&	\textbf{1.047}	&	1.024	&	1.008	&	0.981	&	0.956	&	0.946	&	0.894	\\	
	& d5	&	1.017	&	1.033	&	1.025	&	1.029	&	\textbf{1.014}	&	1.0	&	0.988	&	0.976	&	0.963	&	0.934	\\	
	& d6	&	0.957	&	0.974	&	0.993	&	1.005	&	1.011	&	\textbf{1.017}	&	1.017	&	1.019	&	1.009	&	1.005	\\	
	& d7	&	0.921	&	0.954	&	0.979	&	0.995	&	1.017	&	1.028	&	\textbf{1.032}	&	1.038	&	1.033	&	1.008	\\	
	& d8	&	0.861	&	0.88	&	0.924	&	0.955	&	1.0	&	1.03	&	1.062	&	\textbf{1.077}	&	1.093	&	1.114	\\	
	& d9	&	0.791	&	0.8	&	0.878	&	0.906	&	0.993	&	1.038	&	1.084	&	1.133	&	\textbf{1.157}	&	1.217	\\	
	& d10	&	0.739	&	0.706	&	0.802	&	0.833	&	0.946	&	1.009	&	1.106	&	1.178	&	1.252	&	\textbf{1.444}	\\	
\hline\hline
\end{tabular}
\caption{Results of estimated sorting measures for individuals 25-35 years old in 1976-1980 at different quantiles when all correlations are zero and distribution of $X$ as in total population.}
\label{tab:results_4_1976_young}
\end{table}

\begin{table}[tbp]
\begin{tabular}{l||l|rrrrrrrrrrr}
\hline\hline
&  & \multicolumn{10}{c}{Husbands' quantiles} &  \\
 \hline\hline
&  & d1 & d2 & d3 & d4 & d5 & d6 & d7 & d8 & d9 & d10 &  \\
 \hline
\parbox[t]{2mm}{\multirow{10}{*}{\rotatebox[origin=c]{90}{Wives' quantiles}}}
	&d1	&	\textbf{2.971}	&	1.492	&	1.251	&	0.716	&	1.245	&	0.724	&	0.473	&	0.443	&	0.585	&	0.1	\\	
	&d2	&	1.829	&	\textbf{1.609}	&	0.905	&	1.1	&	0.867	&	0.991	&	0.929	&	0.662	&	0.581	&	0.527	\\	
	&d3	&	1.302	&	1.645	&	\textbf{1.577}	&	1.119	&	0.738	&	0.928	&	0.95	&	0.684	&	0.54	&	0.518	\\	
	&d4	&	0.84	&	0.962	&	1.62	&	\textbf{1.178}	&	1.282	&	1.081	&	1.031	&	0.77	&	0.736	&	0.498	\\	
	&d5	&	0.627	&	0.986	&	1.47	&	1.725	&	\textbf{1.208}	&	0.87	&	0.763	&	1.096	&	0.955	&	0.301	\\	
	&d6	&	0.82	&	0.83	&	0.79	&	1.24	&	1.265	&	\textbf{1.405}	&	0.935	&	0.952	&	0.891	&	0.872	\\	
	&d7	&	0.308	&	0.732	&	0.83	&	0.985	&	1.206	&	1.605	&	\textbf{1.194}	&	1.415	&	1.116	&	0.609	\\	
	&d8	&	0.653	&	0.854	&	0.702	&	0.588	&	1.038	&	0.946	&	1.419	&	\textbf{1.17}	&	1.629	&	1.001	\\	
	&d9	&	0.418	&	0.413	&	0.358	&	0.787	&	1.005	&	0.938	&	1.202	&	1.587	&	\textbf{1.362}	&	1.93	\\	
	&d10	&	0.23	&	0.466	&	0.489	&	0.571	&	0.141	&	0.513	&	1.104	&	1.197	&	1.67	&	\textbf{3.619}	\\	
\hline\hline
\end{tabular}
\caption{Results of estimated sorting measures for individuals 25-35 years old for the years 2021-2025 at different quantiles}
\label{tab:results_1_2021_young}
\end{table}

\begin{table}[tbp]
\begin{tabular}{l||l|rrrrrrrrrrr}
\hline\hline
&  & \multicolumn{10}{c}{Husbands' quantiles} &  \\
 \hline\hline
&  & d1 & d2 & d3 & d4 & d5 & d6 & d7 & d8 & d9 & d10 &  \\
 \hline
\parbox[t]{2mm}{\multirow{10}{*}{\rotatebox[origin=c]{90}{Wives' quantiles}}}
	&d1	&	\textbf{2.986}	&	1.495	&	1.18	&	0.792	&	1.205	&	0.767	&	0.471	&	0.452	&	0.573	&	0.08	\\	
	&d2	&	1.897	&	\textbf{1.579}	&	0.911	&	1.105	&	0.871	&	0.925	&	1.02	&	0.651	&	0.541	&	0.499	\\	
	&d3	&	1.296	&	1.671	&	\textbf{1.569}	&	1.112	&	0.722	&	0.84	&	0.983	&	0.72	&	0.512	&	0.574	\\	
	&d4	&	0.804	&	0.974	&	1.626	&	\textbf{1.187}	&	1.238	&	1.093	&	1.034	&	0.677	&	0.824	&	0.543	\\	
	&d5	&	0.672	&	0.934	&	1.514	&	1.72	&	\textbf{1.255}	&	0.901	&	0.763	&	1.025	&	0.889	&	0.327	\\	
	&d6	&	0.806	&	0.813	&	0.702	&	1.186	&	1.312	&	\textbf{1.463}	&	0.893	&	1.054	&	0.845	&	0.926	\\	
	&d7	&	0.354	&	0.739	&	0.91	&	0.921	&	1.128	&	1.579	&	\textbf{1.19}	&	1.408	&	1.137	&	0.633	\\	
	&d8	&	0.674	&	0.875	&	0.692	&	0.624	&	1.038	&	0.839	&	1.385	&	\textbf{1.244}	&	1.585	&	1.045	\\	
	&d9	&	0.417	&	0.39	&	0.344	&	0.761	&	0.983	&	0.96	&	1.211	&	1.614	&	\textbf{1.389}	&	1.933	\\	
	&d10	&	0.092	&	0.526	&	0.566	&	0.587	&	0.227	&	0.634	&	1.063	&	1.14	&	1.759	&	\textbf{3.405}	\\	
\hline\hline
\end{tabular}
\caption{Results of estimated sorting measures for individuals 25-35 years old in 2021-2025 at different quantiles when $\protect\sigma_{D^*_w, Y^*_w} = 0$.}
\label{tab:results_2_2021_young}
\end{table}

\begin{table}[tbp]
\begin{tabular}{l||l|rrrrrrrrrrr}
\hline\hline
&  & \multicolumn{10}{c}{Husbands' quantiles} &  \\
 \hline\hline
&  & d1 & d2 & d3 & d4 & d5 & d6 & d7 & d8 & d9 & d10 &  \\
 \hline
\parbox[t]{2mm}{\multirow{10}{*}{\rotatebox[origin=c]{90}{Wives' quantiles}}}
	&d1	&	\textbf{2.712}	&	1.491	&	1.05	&	0.957	&	1.038	&	0.664	&	0.755	&	0.804	&	0.228	&	0.301	\\	
	&d2	&	2.088	&	\textbf{1.696}	&	0.797	&	1.105	&	1.002	&	0.963	&	0.667	&	0.493	&	0.495	&	0.694	\\	
	&d3	&	1.297	&	1.653	&	\textbf{1.435}	&	1.085	&	0.782	&	1.222	&	1.105	&	0.545	&	0.681	&	0.195	\\	
	&d4	&	0.57	&	0.703	&	2.026	&	\textbf{1.457}	&	1.082	&	1.277	&	0.602	&	0.705	&	0.838	&	0.74	\\	
	&d5	&	0.414	&	1.102	&	1.439	&	1.498	&	\textbf{1.197}	&	0.654	&	1.154	&	1.192	&	0.973	&	0.378	\\	
	&d6	&	0.566	&	0.605	&	0.66	&	1.216	&	1.588	&	\textbf{1.091}	&	1.381	&	0.984	&	1.121	&	0.788	\\	
	&d7	&	0.364	&	1.099	&	0.718	&	0.664	&	0.943	&	1.439	&	\textbf{1.056}	&	1.144	&	1.428	&	1.144	\\	
	&d8	&	0.852	&	0.45	&	0.719	&	0.548	&	1.024	&	0.984	&	1.68	&	\textbf{1.326}	&	1.324	&	1.093	\\	
	&d9	&	0.661	&	0.622	&	0.672	&	1.057	&	0.938	&	0.781	&	0.922	&	1.495	&	\textbf{1.369}	&	1.483	\\	
	&d10	&	0.494	&	0.528	&	0.468	&	0.468	&	0.344	&	0.892	&	0.791	&	1.217	&	1.605	&	\textbf{3.193}	\\	
\hline\hline
\end{tabular}
\caption{Results of estimated sorting measures for individuals 25-35 years old in 2021-2025 at different quantiles when $\protect\sigma_{D^*_w, Y^*_w} = 0, \protect\sigma_{D^*_w, Y^*_h} = 0.$}
\label{tab:results_3_2021_young}
\end{table}

\begin{table}[tbp]
\begin{tabular}{l||l|rrrrrrrrrrr}
\hline\hline
&  & \multicolumn{10}{c}{Husbands' quantiles} &  \\
 \hline\hline
&  & d1 & d2 & d3 & d4 & d5 & d6 & d7 & d8 & d9 & d10 &  \\
 \hline
\parbox[t]{2mm}{\multirow{10}{*}{\rotatebox[origin=c]{90}{Wives' quantiles}}}
	&d1	&	\textbf{2.681}	&	1.491	&	1.177	&	0.921	&	1.178	&	0.313	&	0.642	&	0.827	&	0.412	&	0.359	\\	
	&d2	&	2.067	&	\textbf{1.665}	&	0.869	&	0.875	&	0.903	&	1.143	&	0.617	&	0.526	&	0.73	&	0.604	\\	
	&d3	&	1.526	&	1.555	&	\textbf{1.582}	&	1.169	&	0.854	&	1.015	&	1.077	&	0.351	&	0.537	&	0.334	\\	
	&d4	&	0.606	&	1.016	&	1.571	&	\textbf{1.245}	&	1.296	&	1.246	&	1.319	&	0.782	&	0.625	&	0.294	\\	
	&d5	&	0.346	&	1.047	&	1.714	&	1.505	&	\textbf{1.181}	&	0.925	&	0.428	&	1.303	&	0.634	&	0.917	\\	
	&d6	&	0.57	&	0.488	&	0.817	&	1.434	&	1.318	&	\textbf{1.179}	&	1.254	&	1.064	&	1.15	&	0.728	\\	
	&d7	&	0.374	&	0.77	&	0.868	&	0.781	&	1.069	&	1.477	&	\textbf{1.161}	&	1.379	&	1.231	&	0.89	\\	
	&d8	&	0.813	&	0.79	&	0.485	&	0.429	&	0.909	&	0.785	&	1.704	&	\textbf{1.225}	&	1.538	&	1.322	\\	
	&d9	&	0.528	&	0.776	&	0.428	&	0.864	&	1.043	&	1.036	&	1.171	&	1.489	&	\textbf{1.306}	&	1.358	\\	
	&d10	&	0.406	&	0.423	&	0.587	&	0.583	&	0.409	&	0.771	&	0.744	&	1.019	&	1.968	&	\textbf{3.09}	\\	
\hline\hline
\end{tabular}
\caption{Results of estimated sorting measures for individuals 25-35 years old in 2021-2025 at different quantiles when $\protect\protect\sigma_{D^*_w, Y^*_w} = 0, \protect\sigma_{D^*_w, Y^*_h} = 0$, and the distribution of $X$ as in total population.}
\label{tab:results_4_2021_young}
\end{table}

\begin{table}[tbp]
\begin{tabular}{l||l|rrrrrrrrrrr}
\hline\hline
&  & \multicolumn{10}{c}{Husbands' quantiles} &  \\
 \hline\hline
&  & d1 & d2 & d3 & d4 & d5 & d6 & d7 & d8 & d9 & d10 &  \\
 \hline
\parbox[t]{2mm}{\multirow{10}{*}{\rotatebox[origin=c]{90}{Wives' quantiles}}}
	&d1	&	\textbf{1.589}	&	1.194	&	1.194	&	1.073	&	1.145	&	0.891	&	0.843	&	0.806	&	0.7	&	0.564	\\	
	&d2	&	1.317	&	\textbf{1.284}	&	1.159	&	1.134	&	1.095	&	0.945	&	0.883	&	0.839	&	0.785	&	0.559	\\	
	&d3	&	1.236	&	1.122	&	\textbf{1.11}	&	1.048	&	1.084	&	0.962	&	0.984	&	0.842	&	0.854	&	0.758	\\	
	&d4	&	1.248	&	1.136	&	1.153	&	\textbf{1.09}	&	1.053	&	0.977	&	0.923	&	0.91	&	0.851	&	0.659	\\	
	&d5	&	0.859	&	1.025	&	1.063	&	1.029	&	\textbf{1.04}	&	1.025	&	0.997	&	1.035	&	1.019	&	0.909	\\	
	&d6	&	0.864	&	0.948	&	0.994	&	0.974	&	1.034	&	\textbf{1.001}	&	1.037	&	1.059	&	1.07	&	1.019	\\	
	&d7	&	0.894	&	0.988	&	1.011	&	0.986	&	1.0	&	1.016	&	\textbf{1.014}	&	1.043	&	1.056	&	0.994	\\	
	&d8	&	0.685	&	0.814	&	0.861	&	0.859	&	0.972	&	1.027	&	1.113	&	\textbf{1.136}	&	1.23	&	1.303	\\	
	&d9	&	0.663	&	0.785	&	0.809	&	0.837	&	0.924	&	1.023	&	1.132	&	1.144	&	\textbf{1.25}	&	1.432	\\	
	&d10	&	0.577	&	0.725	&	0.752	&	0.778	&	0.826	&	0.998	&	1.19	&	1.17	&	1.302	&	\textbf{1.684}	\\	
\hline\hline
\end{tabular}
\caption{Results of estimated sorting measures for individuals 25-35 years old in 2021-2025 at different quantiles when all correlations are zero and distribution of $X$ as in total population.}
\label{tab:results_5_2021_young}
\end{table}

One additional problem is that decisions made during marriage can be especially important in the presence of children. Therefore, we have also done a robustness check excluding children in the household (in addition to restricting the sample to individuals in the age range 25-35 years old). These results are reported in Tables \ref{tab:results_1_1976_young_nochild}-\ref{tab:results_4_1976_young_nochild} for the years 1976-1980 and in Tables \ref{tab:results_1_2021_young_nochild}-\ref{tab:results_5_2021_young_nochild} for the years 2021-2025.

\begin{table}[tbp]
\begin{tabular}{l||l|rrrrrrrrrrr}
\hline\hline
&  & \multicolumn{10}{c}{Husbands' quantiles} &  \\
 \hline\hline
&  & d1 & d2 & d3 & d4 & d5 & d6 & d7 & d8 & d9 & d10 &  \\
 \hline
\parbox[t]{2mm}{\multirow{10}{*}{\rotatebox[origin=c]{90}{Wives' quantiles}}}
	& d1	&	\textbf{2.347}	&	1.468	&	1.374	&	1.074	&	0.994	&	0.661	&	0.691	&	0.752	&	0.517	&	0.13	\\	
	& d2	&	1.282	&	\textbf{1.616}	&	0.993	&	1.084	&	1.065	&	0.89	&	0.893	&	0.579	&	0.634	&	0.95	\\	
	& d3	&	1.066	&	1.371	&	\textbf{1.184}	&	1.084	&	0.992	&	0.945	&	0.908	&	1.012	&	0.906	&	0.53	\\	
	& d4	&	0.878	&	1.195	&	1.324	&	\textbf{1.18}	&	0.835	&	1.149	&	0.881	&	0.788	&	0.998	&	0.775	\\	
	& d5	&	0.938	&	1.098	&	1.141	&	1.025	&	\textbf{1.337}	&	1.019	&	1.026	&	1.002	&	0.711	&	0.706	\\	
	& d6	&	0.887	&	0.753	&	1.262	&	1.154	&	1.215	&	\textbf{0.924}	&	0.901	&	1.018	&	1.002	&	0.877	\\	
	& d7	&	0.812	&	0.83	&	0.897	&	1.102	&	1.008	&	1.08	&	\textbf{0.91}	&	1.163	&	0.928	&	1.289	\\	
	& d8	&	0.7	&	0.787	&	0.642	&	0.827	&	0.805	&	1.257	&	1.254	&	\textbf{1.33}	&	1.092	&	1.302	\\	
	& d9	&	0.852	&	0.648	&	0.704	&	0.692	&	1.012	&	1.25	&	1.059	&	1.209	&	\textbf{1.272}	&	1.297	\\	
	& d10	&	0.238	&	0.236	&	0.478	&	0.778	&	0.735	&	0.826	&	1.479	&	1.148	&	1.939	&	\textbf{2.145}	\\	
\hline\hline
\end{tabular}
\caption{Results of estimated sorting measures for individuals 25-35 years old and without children in 1976-1980 at different quantiles when $\protect\sigma_{D^*_w, Y^*_w} = 0$.}
\label{tab:results_1_1976_young_nochild}
\end{table}

\begin{table}[tbp]
\begin{tabular}{l||l|rrrrrrrrrrr}
\hline\hline
&  & \multicolumn{10}{c}{Husbands' quantiles} &  \\
 \hline\hline
&  & d1 & d2 & d3 & d4 & d5 & d6 & d7 & d8 & d9 & d10 &  \\
 \hline
\parbox[t]{2mm}{\multirow{10}{*}{\rotatebox[origin=c]{90}{Wives' quantiles}}}
	& d1	&	\textbf{2.316}	&	1.203	&	1.301	&	1.003	&	1.008	&	0.742	&	0.774	&	0.729	&	0.383	&	0.544	\\	
	& d2	&	1.236	&	\textbf{1.637}	&	1.078	&	1.122	&	1.111	&	0.825	&	0.749	&	0.574	&	0.983	&	0.516	\\	
	& d3	&	0.724	&	1.682	&	\textbf{1.042}	&	1.192	&	1.064	&	1.07	&	1.004	&	0.961	&	0.925	&	0.496	\\	
	& d4	&	0.842	&	1.02	&	1.478	&	\textbf{1.117}	&	0.818	&	1.023	&	1.143	&	0.862	&	1.074	&	0.665	\\	
	& d5	&	0.993	&	1.08	&	1.166	&	1.147	&	\textbf{1.25}	&	1.043	&	0.985	&	0.756	&	0.818	&	0.756	\\	
	& d6	&	0.668	&	0.629	&	1.259	&	1.177	&	1.041	&	\textbf{1.253}	&	0.888	&	1.319	&	1.008	&	0.742	\\	
	& d7	&	0.988	&	0.889	&	0.864	&	1.283	&	1.042	&	0.796	&	\textbf{0.94}	&	1.033	&	0.971	&	1.204	\\	
	& d8	&	0.462	&	0.779	&	0.664	&	0.669	&	0.9	&	1.499	&	1.401	&	\textbf{1.413}	&	0.918	&	1.307	\\	
	& d9	&	1.278	&	0.549	&	0.667	&	0.739	&	1.024	&	1.024	&	1.268	&	1.039	&	\textbf{1.236}	&	1.159	\\	
	& d10	&	0.493	&	0.533	&	0.48	&	0.551	&	0.743	&	0.725	&	0.848	&	1.314	&	1.684	&	\textbf{2.612}	\\	
\hline\hline
\end{tabular}
\caption{Results of estimated sorting measures for individuals 25-35 years old and without children in 1976-1980 at different quantiles when $\protect\sigma_{D^*_w, Y^*_w} = 0, \protect\sigma_{D^*_w, Y^*_h} = 0.$}
\label{tab:results_2_1976_young_nochild}
\end{table}

\begin{table}[tbp]
\begin{tabular}{l||l|rrrrrrrrrrr}
\hline\hline
&  & \multicolumn{10}{c}{Husbands' quantiles} &  \\
 \hline\hline
&  & d1 & d2 & d3 & d4 & d5 & d6 & d7 & d8 & d9 & d10 &  \\
 \hline
\parbox[t]{2mm}{\multirow{10}{*}{\rotatebox[origin=c]{90}{Wives' quantiles}}}
	& d1	&	\textbf{2.32}	&	1.253	&	1.119	&	1.161	&	1.079	&	0.703	&	0.666	&	0.721	&	0.407	&	0.545	\\	
	& d2	&	1.315	&	\textbf{1.556}	&	1.226	&	1.225	&	1.005	&	0.83	&	0.754	&	0.621	&	1.019	&	0.489	\\	
	& d3	&	0.744	&	1.703	&	\textbf{1.104}	&	0.824	&	1.168	&	0.967	&	1.121	&	0.998	&	0.797	&	0.574	\\	
	& d4	&	0.803	&	1.004	&	1.223	&	\textbf{1.335}	&	0.886	&	1.173	&	1.173	&	0.779	&	1.026	&	0.593	\\	
	& d5	&	0.898	&	0.958	&	1.324	&	1.064	&	\textbf{1.191}	&	1.006	&	0.93	&	0.841	&	0.819	&	0.929	\\	
	& d6	&	0.767	&	0.875	&	1.248	&	1.105	&	1.155	&	\textbf{0.993}	&	0.961	&	1.121	&	1.021	&	0.78	\\	
	& d7	&	0.881	&	0.846	&	0.89	&	1.244	&	0.981	&	1.091	&	\textbf{0.88}	&	1.123	&	0.976	&	1.071	\\	
	& d8	&	0.438	&	0.667	&	0.703	&	0.679	&	0.856	&	1.534	&	1.336	&	\textbf{1.334}	&	1.199	&	1.271	\\	
	& d9	&	1.337	&	0.646	&	0.66	&	0.953	&	0.808	&	0.987	&	1.183	&	1.167	&	\textbf{1.059}	&	1.204	\\	
	& d10	&	0.495	&	0.491	&	0.504	&	0.409	&	0.869	&	0.716	&	0.996	&	1.295	&	1.676	&	\textbf{2.543}	\\	
\hline\hline
\end{tabular}
\caption{Results of estimated sorting measures for individuals 25-35 years old and without children in 1976-1980 at different quantiles when $\protect\protect\sigma_{D^*_w, Y^*_w} = 0, \protect\sigma_{D^*_w, Y^*_h} = 0$, and the distribution of $X$ as in total population.}
\label{tab:results_3_1976_young_nochild}
\end{table}

\begin{table}[tbp]
\begin{tabular}{l||l|rrrrrrrrrrr}
\hline\hline
&  & \multicolumn{10}{c}{Husbands' quantiles} &  \\
 \hline\hline
&  & d1 & d2 & d3 & d4 & d5 & d6 & d7 & d8 & d9 & d10 &  \\
 \hline
\parbox[t]{2mm}{\multirow{10}{*}{\rotatebox[origin=c]{90}{Wives' quantiles}}}
	& d1	&	\textbf{1.372}	&	1.217	&	1.105	&	1.03	&	1.001	&	0.91	&	0.867	&	0.861	&	0.817	&	0.796	\\	
	& d2	&	1.206	&	\textbf{1.247}	&	1.128	&	1.075	&	1.016	&	0.968	&	0.885	&	0.879	&	0.857	&	0.774	\\	
	& d3	&	1.126	&	1.117	&	\textbf{1.071}	&	1.064	&	1.01	&	0.992	&	0.943	&	0.914	&	0.904	&	0.853	\\	
	& d4	&	1.005	&	1.045	&	1.035	&	\textbf{1.045}	&	1.019	&	1.012	&	0.972	&	0.978	&	0.968	&	0.918	\\	
	& d5	&	0.938	&	0.994	&	1.0	&	1.024	&	\textbf{1.011}	&	1.03	&	1.008	&	1.013	&	0.995	&	0.949	\\	
	& d6	&	0.951	&	0.993	&	1.0	&	1.025	&	1.003	&	\textbf{1.015}	&	1.025	&	1.027	&	1.001	&	0.982	\\	
	& d7	&	0.885	&	0.928	&	0.977	&	0.992	&	0.996	&	1.029	&	\textbf{1.044}	&	1.051	&	1.045	&	1.04	\\	
	& d8	&	0.853	&	0.841	&	0.928	&	0.941	&	0.999	&	1.06	&	1.063	&	\textbf{1.086}	&	1.107	&	1.14	\\	
	& d9	&	0.887	&	0.904	&	0.936	&	0.955	&	1.0	&	1.02	&	1.049	&	1.082	&	\textbf{1.075}	&	1.102	\\	
	& d10	&	0.777	&	0.714	&	0.819	&	0.849	&	0.945	&	0.964	&	1.144	&	1.109	&	1.23	&	\textbf{1.445}	\\	
\hline\hline
\end{tabular}
\caption{Results of estimated sorting measures for individuals 25-35 years old and without children in 1976-1980 at different quantiles when all correlations are zero and distribution of $X$ as in total population.}
\label{tab:results_4_1976_young_nochild}
\end{table}

\begin{table}[tbp]
\begin{tabular}{l||l|rrrrrrrrrrr}
\hline\hline
&  & \multicolumn{10}{c}{Husbands' quantiles} &  \\
 \hline\hline
&  & d1 & d2 & d3 & d4 & d5 & d6 & d7 & d8 & d9 & d10 &  \\
 \hline
\parbox[t]{2mm}{\multirow{10}{*}{\rotatebox[origin=c]{90}{Wives' quantiles}}}
	&d1	&	\textbf{2.98}	&	1.275	&	0.818	&	0.81	&	1.207	&	1.015	&	0.496	&	0.555	&	0.468	&	0.377	\\	
	&d2	&	1.448	&	\textbf{1.755}	&	1.365	&	1.56	&	1.147	&	0.457	&	0.82	&	0.559	&	0.45	&	0.44	\\	
	&d3	&	0.935	&	1.447	&	\textbf{1.377}	&	1.337	&	1.636	&	0.863	&	0.352	&	0.642	&	0.798	&	0.614	\\	
	&d4	&	1.028	&	1.177	&	1.645	&	\textbf{1.01}	&	0.861	&	1.263	&	1.266	&	0.335	&	0.512	&	0.903	\\	
	&d5	&	1.089	&	0.79	&	0.422	&	1.635	&	\textbf{1.204}	&	0.967	&	0.983	&	1.027	&	0.7	&	1.182	\\	
	&d6	&	0.789	&	0.938	&	1.197	&	0.684	&	1.137	&	\textbf{1.655}	&	1.137	&	0.917	&	1.257	&	0.289	\\	
	&d7	&	0.495	&	0.775	&	0.604	&	0.855	&	1.0	&	1.011	&	\textbf{1.231}	&	1.226	&	1.409	&	1.392	\\	
	&d8	&	0.765	&	0.507	&	1.011	&	1.019	&	0.574	&	0.928	&	1.206	&	\textbf{1.67}	&	0.786	&	1.535	\\	
	&d9	&	0.476	&	0.467	&	0.591	&	0.696	&	0.651	&	0.758	&	1.588	&	1.939	&	\textbf{1.799}	&	1.035	\\	
	&d10	&	0.007	&	0.857	&	0.958	&	0.41	&	0.581	&	1.072	&	0.921	&	1.15	&	1.808	&	\textbf{2.237}	\\	
\hline\hline
\end{tabular}
\caption{Results of estimated sorting measures for individuals 25-35 years old and without children for the years 2021-2025 at different quantiles}
\label{tab:results_1_2021_young_nochild}
\end{table}

\begin{table}[tbp]
\begin{tabular}{l||l|rrrrrrrrrrr}
\hline\hline
&  & \multicolumn{10}{c}{Husbands' quantiles} &  \\
 \hline\hline
&  & d1 & d2 & d3 & d4 & d5 & d6 & d7 & d8 & d9 & d10 &  \\
 \hline
\parbox[t]{2mm}{\multirow{10}{*}{\rotatebox[origin=c]{90}{Wives' quantiles}}}
	&d1	&	\textbf{2.881}	&	1.377	&	0.798	&	0.858	&	1.143	&	1.153	&	0.457	&	0.558	&	0.476	&	0.3	\\	
	&d2	&	1.485	&	\textbf{1.769}	&	1.391	&	1.671	&	1.126	&	0.471	&	0.873	&	0.601	&	0.484	&	0.129	\\	
	&d3	&	0.824	&	1.4	&	\textbf{1.31}	&	1.34	&	1.608	&	0.732	&	0.496	&	0.562	&	0.77	&	0.957	\\	
	&d4	&	0.971	&	1.12	&	1.697	&	\textbf{0.941}	&	0.817	&	1.319	&	1.178	&	0.261	&	0.544	&	1.15	\\	
	&d5	&	0.996	&	0.903	&	0.395	&	1.651	&	\textbf{1.199}	&	1.015	&	0.932	&	1.107	&	0.611	&	1.191	\\	
	&d6	&	0.895	&	0.953	&	1.197	&	0.58	&	0.777	&	\textbf{2.114}	&	1.137	&	1.051	&	1.292	&	0.003	\\	
	&d7	&	0.449	&	0.885	&	0.594	&	0.92	&	1.521	&	0.47	&	\textbf{1.159}	&	1.188	&	1.485	&	1.33	\\	
	&d8	&	0.758	&	0.547	&	0.956	&	1.095	&	0.737	&	0.782	&	1.092	&	\textbf{1.511}	&	0.596	&	1.926	\\	
	&d9	&	0.449	&	0.53	&	0.523	&	0.711	&	0.644	&	0.603	&	1.377	&	2.203	&	\textbf{1.722}	&	1.239	\\	
	&d10	&	0.377	&	0.57	&	1.144	&	0.301	&	0.47	&	1.324	&	1.267	&	0.956	&	1.914	&	\textbf{1.676}	\\	
\hline\hline
\end{tabular}
\caption{Results of estimated sorting measures for individuals 25-35 years old and without children in 2021-2025 at different quantiles when $\protect\sigma_{D^*_w, Y^*_w} = 0$.}
\label{tab:results_2_2021_young_nochild}
\end{table}

\begin{table}[tbp]
\begin{tabular}{l||l|rrrrrrrrrrr}
\hline\hline
&  & \multicolumn{10}{c}{Husbands' quantiles} &  \\
 \hline\hline
&  & d1 & d2 & d3 & d4 & d5 & d6 & d7 & d8 & d9 & d10 &  \\
 \hline
\parbox[t]{2mm}{\multirow{10}{*}{\rotatebox[origin=c]{90}{Wives' quantiles}}}
	&d1	&	\textbf{3.085}	&	1.248	&	0.816	&	0.846	&	0.911	&	1.239	&	0.239	&	0.597	&	0.71	&	0.309	\\	
	&d2	&	1.463	&	\textbf{1.773}	&	1.498	&	1.679	&	1.105	&	0.338	&	0.518	&	0.73	&	0.459	&	0.437	\\	
	&d3	&	0.706	&	1.4	&	\textbf{1.334}	&	1.399	&	1.483	&	0.949	&	0.436	&	0.511	&	0.793	&	0.988	\\	
	&d4	&	0.78	&	1.136	&	1.528	&	\textbf{1.215}	&	0.616	&	1.382	&	1.653	&	0.143	&	0.535	&	1.011	\\	
	&d5	&	1.243	&	0.71	&	0.831	&	1.526	&	\textbf{1.458}	&	1.171	&	0.673	&	1.045	&	0.55	&	0.792	\\	
	&d6	&	1.028	&	1.025	&	0.858	&	0.808	&	0.887	&	\textbf{1.135}	&	0.682	&	1.317	&	1.602	&	0.659	\\	
	&d7	&	0.642	&	0.809	&	0.65	&	0.669	&	1.109	&	0.895	&	\textbf{1.382}	&	1.234	&	0.762	&	1.847	\\	
	&d8	&	0.552	&	0.817	&	0.999	&	0.846	&	0.78	&	1.143	&	1.37	&	\textbf{1.479}	&	1.263	&	0.75	\\	
	&d9	&	0.412	&	0.346	&	0.887	&	0.57	&	0.827	&	0.774	&	1.141	&	1.999	&	\textbf{1.663}	&	1.381	\\	
	&d10	&	0.11	&	0.735	&	0.598	&	0.448	&	0.816	&	0.995	&	1.916	&	0.927	&	1.656	&	\textbf{1.798}	\\	
\hline\hline
\end{tabular}
\caption{Results of estimated sorting measures for individuals 25-35 years old and without children in 2021-2025 at different quantiles when $\protect\sigma_{D^*_w, Y^*_w} = 0, \protect\sigma_{D^*_w, Y^*_h} = 0.$}
\label{tab:results_3_2021_young_nochild}
\end{table}

\begin{table}[tbp]
\begin{tabular}{l||l|rrrrrrrrrrr}
\hline\hline
&  & \multicolumn{10}{c}{Husbands' quantiles} &  \\
 \hline\hline
&  & d1 & d2 & d3 & d4 & d5 & d6 & d7 & d8 & d9 & d10 &  \\
 \hline
\parbox[t]{2mm}{\multirow{10}{*}{\rotatebox[origin=c]{90}{Wives' quantiles}}}
	&d1	&	\textbf{2.806}	&	1.014	&	0.971	&	0.779	&	1.138	&	1.025	&	0.495	&	0.634	&	0.72	&	0.419	\\	
	&d2	&	1.349	&	\textbf{2.114}	&	1.4	&	1.719	&	1.16	&	0.499	&	0.456	&	0.431	&	0.733	&	0.139	\\	
	&d3	&	0.997	&	1.325	&	\textbf{1.021}	&	1.236	&	1.517	&	0.683	&	1.058	&	0.611	&	0.655	&	0.896	\\	
	&d4	&	0.963	&	1.037	&	1.914	&	\textbf{0.971}	&	0.715	&	1.602	&	0.98	&	0.199	&	0.431	&	1.189	\\	
	&d5	&	1.184	&	1.064	&	0.418	&	1.555	&	\textbf{1.335}	&	0.925	&	0.971	&	1.302	&	0.503	&	0.744	\\	
	&d6	&	0.916	&	0.797	&	1.204	&	0.895	&	1.051	&	\textbf{1.628}	&	0.63	&	1.053	&	1.353	&	0.474	\\	
	&d7	&	0.615	&	0.79	&	0.636	&	0.73	&	0.941	&	0.816	&	\textbf{1.274}	&	1.088	&	1.383	&	1.728	\\	
	&d8	&	0.527	&	0.839	&	0.987	&	1.064	&	0.765	&	0.812	&	1.117	&	\textbf{1.868}	&	0.988	&	1.034	\\	
	&d9	&	0.51	&	0.339	&	0.759	&	0.603	&	0.627	&	0.969	&	1.112	&	1.744	&	\textbf{1.918}	&	1.418	\\	
	&d10	&	0.115	&	0.66	&	0.703	&	0.415	&	0.759	&	1.056	&	1.902	&	1.071	&	1.359	&	\textbf{1.96}	\\	
\hline\hline
\end{tabular}
\caption{Results of estimated sorting measures for individuals 25-35 years old and without children in 2021-2025 at different quantiles when $\protect\protect\sigma_{D^*_w, Y^*_w} = 0, \protect\sigma_{D^*_w, Y^*_h} = 0$, and the distribution of $X$ as in total population.}
\label{tab:results_4_2021_young_nochild}
\end{table}

\begin{table}[tbp]
\begin{tabular}{l||l|rrrrrrrrrrr}
\hline\hline
&  & \multicolumn{10}{c}{Husbands' quantiles} &  \\
 \hline\hline
&  & d1 & d2 & d3 & d4 & d5 & d6 & d7 & d8 & d9 & d10 &  \\
 \hline
\parbox[t]{2mm}{\multirow{10}{*}{\rotatebox[origin=c]{90}{Wives' quantiles}}}
	&d1	&	\textbf{1.366}	&	1.35	&	1.19	&	1.152	&	1.171	&	0.824	&	0.872	&	0.807	&	0.689	&	0.577	\\	
	&d2	&	1.291	&	\textbf{1.246}	&	1.123	&	1.145	&	1.121	&	0.971	&	0.872	&	0.867	&	0.698	&	0.665	\\	
	&d3	&	1.218	&	1.269	&	\textbf{1.125}	&	1.043	&	1.108	&	0.968	&	0.947	&	0.847	&	0.802	&	0.672	\\	
	&d4	&	1.058	&	1.17	&	1.084	&	\textbf{0.997}	&	1.018	&	1.025	&	1.007	&	0.903	&	0.923	&	0.815	\\	
	&d5	&	0.999	&	0.968	&	1.001	&	1.025	&	\textbf{1.001}	&	1.001	&	0.957	&	1.053	&	0.98	&	1.014	\\	
	&d6	&	0.914	&	0.945	&	0.917	&	0.949	&	0.957	&	\textbf{1.068}	&	1.024	&	0.998	&	1.113	&	1.116	\\	
	&d7	&	0.894	&	0.901	&	0.96	&	0.95	&	0.916	&	1.037	&	\textbf{1.052}	&	1.061	&	1.088	&	1.141	\\	
	&d8	&	0.803	&	0.787	&	0.962	&	0.973	&	0.973	&	1.014	&	1.077	&	\textbf{1.124}	&	1.104	&	1.183	\\	
	&d9	&	0.69	&	0.658	&	0.788	&	0.839	&	0.842	&	1.053	&	1.104	&	1.178	&	\textbf{1.36}	&	1.488	\\	
	&d10	&	0.763	&	0.708	&	0.848	&	0.923	&	0.893	&	1.039	&	1.091	&	1.155	&	1.25	&	\textbf{1.331}	\\	
\hline\hline
\end{tabular}
\caption{Results of estimated sorting measures for individuals 25-35 years old and without children in 2021-2025 at different quantiles when all correlations are zero and distribution of $X$ as in total population.}
\label{tab:results_5_2021_young_nochild}
\end{table}

Unfortunately, the CPS does not contain information about the year of marriage and hence we cannot select our sample based on whether individuals have been married for a short period. Therefore, we also use data from the IPUMS USA for the years 1980 and 2021. This is the census long-form sample data for the year 1980 and the American Community Survey (ACS) for the year 2021. We use the same restrictions for our data selection as for the CPS. This leaves us with 170,378 couples for the year 1980 and 222,250 couples for the year 2021. The results of our decomposition analysis are reported in Tables \ref{tab:results_1_1980_acs}-\ref{tab:results_4_1980_acs} for the year 1980 and Tables \ref{tab:results_2_2021_acs}-\ref{tab:results_5_2021_acs} for the year 2021. The results are remarkably similar to the results we obtained from the CPS. In addition, we are able to perform a robustness check for the IPUMS USA by only including those observations of couples who are at most two years in their first marriage. This results in 2,177 newlyweds for the year 1980 and 7,194 for the year 2021. The reason for the more than 3 times as many observations for the recent year is related to the fact that newly wed couples in the year 1980 were typically younger than 25 and hence are not included in our analysis. The results of our robustness check are in Tables \ref{tab:results_1_1980_newly_weds}-\ref{tab:results_4_1980_newly_weds} for the year 1980 and Tables \ref{tab:results_2_2021_newly_weds}-\ref{tab:results_5_2021_newly_weds} for the year 2021. 

\begin{table}[tbp]
\begin{tabular}{l||l|rrrrrrrrrrr}
\hline\hline
&  & \multicolumn{10}{c}{Husbands' quantiles} &  \\
 \hline\hline
&  & d1 & d2 & d3 & d4 & d5 & d6 & d7 & d8 & d9 & d10 &  \\
 \hline
\parbox[t]{2mm}{\multirow{10}{*}{\rotatebox[origin=c]{90}{Wives' quantiles}}}
	& d1	&	\textbf{2.451}	&	1.33	&	1.093	&	0.945	&	0.934	&	0.782	&	0.788	&	0.755	&	0.644	&	0.279	\\	
	& d2	&	1.651	&	\textbf{1.46}	&	1.303	&	1.056	&	0.992	&	0.982	&	0.832	&	0.774	&	0.68	&	0.272	\\	
	& d3	&	1.254	&	1.266	&	\textbf{1.233}	&	1.021	&	1.013	&	0.9	&	0.904	&	0.804	&	0.762	&	0.844	\\	
	& d4	&	1.015	&	1.2	&	1.233	&	\textbf{1.136}	&	1.035	&	1.024	&	0.918	&	0.926	&	0.766	&	0.746	\\	
	& d5	&	0.9	&	1.19	&	1.204	&	1.127	&	\textbf{1.131}	&	0.931	&	1.073	&	1.048	&	0.96	&	0.434	\\	
	& d6	&	0.691	&	0.899	&	1.1	&	1.083	&	0.976	&	\textbf{0.991}	&	0.997	&	1.021	&	0.899	&	0.688	\\	
	& d7	&	0.596	&	0.759	&	1.083	&	1.102	&	1.026	&	1.084	&	\textbf{1.03}	&	1.01	&	1.007	&	1.952	\\	
	& d8	&	0.586	&	0.669	&	0.803	&	1.066	&	1.106	&	1.117	&	1.105	&	\textbf{1.12}	&	1.086	&	1.345	\\	
	& d9	&	0.472	&	0.521	&	0.678	&	0.908	&	1.148	&	1.185	&	1.282	&	1.357	&	\textbf{1.365}	&	1.084	\\	
	& d10	&	0.383	&	0.705	&	0.27	&	0.556	&	0.64	&	1.003	&	1.071	&	1.185	&	1.831	&	\textbf{2.355}	\\	
\hline\hline
\end{tabular}
\caption{Results of estimated sorting measures using the ACS in 1980 at different quantiles when $\protect\sigma_{D^*_w, Y^*_w} = 0$.}
\label{tab:results_1_1980_acs}
\end{table}

\begin{table}[tbp]
\begin{tabular}{l||l|rrrrrrrrrrr}
\hline\hline
&  & \multicolumn{10}{c}{Husbands' quantiles} &  \\
 \hline\hline
&  & d1 & d2 & d3 & d4 & d5 & d6 & d7 & d8 & d9 & d10 &  \\
 \hline
\parbox[t]{2mm}{\multirow{10}{*}{\rotatebox[origin=c]{90}{Wives' quantiles}}}
	& d1	&	\textbf{2.289}	&	1.321	&	1.126	&	0.832	&	0.82	&	0.752	&	0.928	&	0.683	&	0.547	&	0.704	\\	
	& d2	&	1.285	&	\textbf{1.568}	&	1.342	&	1.143	&	1.019	&	0.842	&	0.96	&	0.656	&	0.63	&	0.552	\\	
	& d3	&	1.082	&	1.439	&	\textbf{1.274}	&	0.988	&	1.019	&	0.925	&	1.026	&	0.963	&	0.51	&	0.775	\\	
	& d4	&	0.884	&	1.346	&	1.253	&	\textbf{1.12}	&	0.793	&	1.089	&	0.92	&	0.704	&	1.073	&	0.828	\\	
	& d5	&	0.891	&	0.949	&	1.366	&	1.082	&	\textbf{0.985}	&	1.092	&	1.022	&	0.916	&	0.865	&	0.82	\\	
	& d6	&	0.768	&	0.956	&	1.069	&	1.205	&	1.04	&	\textbf{1.065}	&	1.033	&	1.028	&	0.978	&	0.852	\\	
	& d7	&	0.77	&	0.908	&	0.76	&	1.081	&	1.314	&	1.103	&	\textbf{1.029}	&	1.031	&	1.083	&	0.928	\\	
	& d8	&	0.686	&	0.668	&	0.699	&	1.148	&	1.168	&	1.047	&	1.098	&	\textbf{1.319}	&	1.099	&	1.067	\\	
	& d9	&	0.662	&	0.464	&	0.554	&	0.878	&	1.192	&	1.366	&	1.044	&	1.323	&	\textbf{1.536}	&	0.984	\\	
	& d10	&	0.683	&	0.383	&	0.557	&	0.522	&	0.65	&	0.719	&	0.939	&	1.377	&	1.68	&	\textbf{2.489}	\\	
\hline\hline
\end{tabular}
\caption{Results of estimated sorting measures using the ACS in 1980 at different quantiles when $\protect\sigma_{D^*_w, Y^*_w} = 0, \protect\sigma_{D^*_w, Y^*_h} = 0.$}
\label{tab:results_2_1980_acs}
\end{table}

\begin{table}[tbp]
\begin{tabular}{l||l|rrrrrrrrrrr}
\hline\hline
&  & \multicolumn{10}{c}{Husbands' quantiles} &  \\
 \hline\hline
&  & d1 & d2 & d3 & d4 & d5 & d6 & d7 & d8 & d9 & d10 &  \\
 \hline
\parbox[t]{2mm}{\multirow{10}{*}{\rotatebox[origin=c]{90}{Wives' quantiles}}}
	& d1	&	\textbf{2.281}	&	1.328	&	1.137	&	0.836	&	0.794	&	0.776	&	0.956	&	0.624	&	0.561	&	0.702	\\	
	& d2	&	1.286	&	\textbf{1.569}	&	1.348	&	1.133	&	1.039	&	0.839	&	0.956	&	0.678	&	0.611	&	0.546	\\	
	& d3	&	1.075	&	1.43	&	\textbf{1.301}	&	0.993	&	1.004	&	0.941	&	1.028	&	0.717	&	0.722	&	0.786	\\	
	& d4	&	1.01	&	1.231	&	1.278	&	\textbf{1.128}	&	0.734	&	1.09	&	0.969	&	0.897	&	0.824	&	0.84	\\	
	& d5	&	0.791	&	1.218	&	1.196	&	1.059	&	\textbf{1.046}	&	1.104	&	0.977	&	0.933	&	0.871	&	0.802	\\	
	& d6	&	0.772	&	0.813	&	1.212	&	1.077	&	1.183	&	\textbf{1.087}	&	1.018	&	1.019	&	0.986	&	0.841	\\	
	& d7	&	0.745	&	0.827	&	0.804	&	1.228	&	1.124	&	1.104	&	\textbf{1.064}	&	1.048	&	1.094	&	0.96	\\	
	& d8	&	0.704	&	0.743	&	0.655	&	1.104	&	1.267	&	1.073	&	0.982	&	\textbf{1.28}	&	1.167	&	1.021	\\	
	& d9	&	0.661	&	0.458	&	0.524	&	0.932	&	1.156	&	1.278	&	1.103	&	1.412	&	\textbf{1.488}	&	0.987	\\	
	& d10	&	0.673	&	0.384	&	0.544	&	0.509	&	0.652	&	0.709	&	0.948	&	1.392	&	1.677	&	\textbf{2.515}	\\	
\hline\hline
\end{tabular}
\caption{Results of estimated sorting measures using the ACS in 1980 at different quantiles when $\protect\protect\sigma_{D^*_w, Y^*_w} = 0, \protect\sigma_{D^*_w, Y^*_h} = 0$, and the distribution of $X$ as in total population.}
\label{tab:results_3_1980_acs}
\end{table}

\begin{table}[tbp]
\begin{tabular}{l||l|rrrrrrrrrrr}
\hline\hline
&  & \multicolumn{10}{c}{Husbands' quantiles} &  \\
 \hline\hline
&  & d1 & d2 & d3 & d4 & d5 & d6 & d7 & d8 & d9 & d10 &  \\
 \hline
\parbox[t]{2mm}{\multirow{10}{*}{\rotatebox[origin=c]{90}{Wives' quantiles}}}
	& d1	&	\textbf{1.246}	&	1.24	&	1.106	&	1.047	&	1.0	&	0.98	&	0.931	&	0.905	&	0.804	&	0.736	\\	
	& d2	&	1.192	&	\textbf{1.267}	&	1.149	&	1.04	&	1.0	&	0.999	&	0.941	&	0.896	&	0.794	&	0.726	\\	
	& d3	&	1.094	&	1.129	&	\textbf{1.087}	&	1.05	&	1.019	&	1.011	&	0.969	&	0.936	&	0.887	&	0.816	\\	
	& d4	&	1.039	&	1.058	&	1.044	&	\textbf{1.035}	&	1.021	&	0.997	&	0.984	&	0.978	&	0.945	&	0.9	\\	
	& d5	&	1.018	&	1.062	&	1.047	&	1.039	&	\textbf{1.028}	&	1.04	&	1.006	&	0.975	&	0.929	&	0.851	\\	
	& d6	&	0.98	&	0.995	&	1.009	&	1.013	&	1.015	&	\textbf{1.024}	&	1.017	&	1.002	&	0.993	&	0.958	\\	
	& d7	&	0.942	&	0.956	&	0.976	&	0.999	&	1.014	&	1.011	&	\textbf{1.021}	&	1.03	&	1.037	&	1.014	\\	
	& d8	&	0.865	&	0.867	&	0.924	&	0.971	&	1.009	&	1.007	&	1.047	&	\textbf{1.082}	&	1.11	&	1.114	\\	
	& d9	&	0.8	&	0.743	&	0.848	&	0.93	&	0.977	&	0.984	&	1.049	&	1.11	&	\textbf{1.24}	&	1.319	\\	
	& d10	&	0.824	&	0.684	&	0.808	&	0.877	&	0.918	&	0.946	&	1.036	&	1.084	&	1.261	&	\textbf{1.565}	\\	
\hline\hline
\end{tabular}
\caption{Results of estimated sorting measures using the ACS in 1980 at different quantiles when all correlations are zero and distribution of $X$ as in total population.}
\label{tab:results_4_1980_acs}
\end{table}

\begin{table}[tbp]
\begin{tabular}{l||l|rrrrrrrrrrr}
\hline\hline
&  & \multicolumn{10}{c}{Husbands' quantiles} &  \\
 \hline\hline
&  & d1 & d2 & d3 & d4 & d5 & d6 & d7 & d8 & d9 & d10 &  \\
 \hline
\parbox[t]{2mm}{\multirow{10}{*}{\rotatebox[origin=c]{90}{Wives' quantiles}}}
	& d1	&	\textbf{2.642}	&	1.454	&	1.064	&	0.987	&	0.796	&	0.75	&	0.683	&	0.57	&	0.53	&	0.524	\\	
	& d2	&	1.83	&	\textbf{1.732}	&	1.249	&	1.109	&	0.925	&	0.859	&	0.762	&	0.647	&	0.555	&	0.333	\\	
	& d3	&	1.188	&	1.745	&	\textbf{1.463}	&	1.145	&	1.046	&	0.974	&	0.848	&	0.722	&	0.556	&	0.312	\\	
	& d4	&	0.964	&	1.324	&	1.505	&	\textbf{1.201}	&	1.249	&	0.91	&	0.969	&	0.746	&	0.685	&	0.447	\\	
	& d5	&	0.793	&	0.977	&	1.221	&	1.206	&	\textbf{1.406}	&	1.097	&	0.946	&	0.924	&	0.718	&	0.712	\\	
	& d6	&	0.713	&	0.832	&	1.047	&	1.417	&	1.214	&	\textbf{1.415}	&	1.22	&	1.104	&	0.966	&	0.072	\\	
	& d7	&	0.568	&	0.698	&	0.793	&	0.893	&	1.15	&	1.387	&	\textbf{1.203}	&	1.055	&	0.982	&	1.271	\\	
	& d8	&	0.534	&	0.636	&	0.732	&	0.718	&	1.174	&	1.149	&	1.546	&	\textbf{1.404}	&	1.253	&	0.854	\\	
	& d9	&	0.437	&	0.474	&	0.501	&	0.665	&	0.555	&	0.939	&	1.092	&	1.588	&	\textbf{1.638}	&	2.111	\\	
	& d10	&	0.333	&	0.126	&	0.422	&	0.652	&	0.478	&	0.512	&	0.725	&	1.24	&	2.121	&	\textbf{3.39}	\\	
\hline\hline
\end{tabular}
\caption{Results of estimated sorting measures using the ACS and for couples at most two years married in 2021 at different quantiles when $\protect\sigma_{D^*_w, Y^*_w} = 0$.}
\label{tab:results_2_2021_acs}
\end{table}

\begin{table}[tbp]
\begin{tabular}{l||l|rrrrrrrrrrr}
\hline\hline
&  & \multicolumn{10}{c}{Husbands' quantiles} &  \\
 \hline\hline
&  & d1 & d2 & d3 & d4 & d5 & d6 & d7 & d8 & d9 & d10 &  \\
 \hline
\parbox[t]{2mm}{\multirow{10}{*}{\rotatebox[origin=c]{90}{Wives' quantiles}}}
	& d1	&	\textbf{2.778}	&	1.51	&	1.067	&	0.95	&	0.788	&	0.647	&	0.699	&	0.503	&	0.577	&	0.48	\\	
	& d2	&	1.693	&	\textbf{1.784}	&	1.321	&	1.131	&	0.934	&	0.863	&	0.716	&	0.611	&	0.525	&	0.422	\\	
	& d3	&	1.095	&	1.673	&	\textbf{1.47}	&	1.231	&	1.09	&	0.947	&	0.84	&	0.64	&	0.556	&	0.459	\\	
	& d4	&	0.931	&	1.283	&	1.58	&	\textbf{1.06}	&	1.188	&	1.132	&	0.879	&	0.734	&	0.679	&	0.533	\\	
	& d5	&	0.681	&	0.75	&	1.231	&	1.74	&	\textbf{1.187}	&	1.148	&	1.023	&	0.887	&	0.829	&	0.525	\\	
	& d6	&	0.702	&	0.84	&	1.115	&	1.083	&	1.41	&	\textbf{1.15}	&	1.135	&	1.11	&	0.647	&	0.809	\\	
	& d7	&	0.618	&	0.708	&	0.659	&	1.065	&	1.2	&	1.439	&	\textbf{1.285}	&	1.167	&	1.138	&	0.721	\\	
	& d8	&	0.527	&	0.588	&	0.679	&	0.735	&	0.962	&	1.11	&	1.457	&	\textbf{1.493}	&	1.284	&	1.164	\\	
	& d9	&	0.487	&	0.487	&	0.536	&	0.586	&	0.737	&	0.945	&	1.132	&	1.669	&	\textbf{1.788}	&	1.633	\\	
	& d10	&	0.488	&	0.377	&	0.342	&	0.419	&	0.503	&	0.621	&	0.834	&	1.186	&	1.977	&	\textbf{3.253}	\\	
\hline\hline
\end{tabular}
\caption{Results of estimated sorting measures using the ACS and for couples at most two years married in 2021 at different quantiles when $\protect\sigma_{D^*_w, Y^*_w} = 0, \protect\sigma_{D^*_w, Y^*_h} = 0.$}
\label{tab:results_3_2021_acs}
\end{table}

\begin{table}[tbp]
\begin{tabular}{l||l|rrrrrrrrrrr}
\hline\hline
&  & \multicolumn{10}{c}{Husbands' quantiles} &  \\
 \hline\hline
&  & d1 & d2 & d3 & d4 & d5 & d6 & d7 & d8 & d9 & d10 &  \\
 \hline
\parbox[t]{2mm}{\multirow{10}{*}{\rotatebox[origin=c]{90}{Wives' quantiles}}}
	& d1	&	\textbf{2.801}	&	1.501	&	1.072	&	0.963	&	0.779	&	0.646	&	0.694	&	0.496	&	0.561	&	0.488	\\	
	& d2	&	1.7	&	\textbf{1.826}	&	1.312	&	1.119	&	0.948	&	0.828	&	0.7	&	0.595	&	0.549	&	0.424	\\	
	& d3	&	1.072	&	1.636	&	\textbf{1.505}	&	1.237	&	1.099	&	0.979	&	0.815	&	0.679	&	0.533	&	0.445	\\	
	& d4	&	0.942	&	1.308	&	1.547	&	\textbf{1.048}	&	1.287	&	1.011	&	0.915	&	0.707	&	0.693	&	0.541	\\	
	& d5	&	0.685	&	0.777	&	1.324	&	1.647	&	\textbf{1.127}	&	1.19	&	1.018	&	1.012	&	0.699	&	0.522	\\	
	& d6	&	0.682	&	0.809	&	1.026	&	1.163	&	1.365	&	\textbf{1.211}	&	1.143	&	0.861	&	0.926	&	0.814	\\	
	& d7	&	0.6	&	0.687	&	0.677	&	1.073	&	1.207	&	1.452	&	\textbf{1.288}	&	1.301	&	0.99	&	0.726	\\	
	& d8	&	0.534	&	0.604	&	0.647	&	0.739	&	0.959	&	1.113	&	1.46	&	\textbf{1.496}	&	1.284	&	1.162	\\	
	& d9	&	0.485	&	0.48	&	0.515	&	0.611	&	0.732	&	0.941	&	1.141	&	1.672	&	\textbf{1.795}	&	1.628	\\	
	& d10	&	0.503	&	0.373	&	0.375	&	0.4	&	0.496	&	0.628	&	0.826	&	1.181	&	1.971	&	\textbf{3.247}	\\	
\hline\hline
\end{tabular}
\caption{Results of estimated sorting measures using the ACS and for couples at most two years married in 2021 at different quantiles when $\protect\protect\sigma_{D^*_w, Y^*_w} = 0, \protect\sigma_{D^*_w, Y^*_h} = 0$, and the distribution of $X$ as in total population.}
\label{tab:results_4_2021_acs}
\end{table}

\begin{table}[tbp]
\begin{tabular}{l||l|rrrrrrrrrrr}
\hline\hline
&  & \multicolumn{10}{c}{Husbands' quantiles} &  \\
 \hline\hline
&  & d1 & d2 & d3 & d4 & d5 & d6 & d7 & d8 & d9 & d10 &  \\
 \hline
\parbox[t]{2mm}{\multirow{10}{*}{\rotatebox[origin=c]{90}{Wives' quantiles}}}
	& d1	&	\textbf{1.492}	&	1.317	&	1.218	&	1.134	&	1.044	&	0.933	&	0.837	&	0.756	&	0.656	&	0.613	\\	
	& d2	&	1.3	&	\textbf{1.271}	&	1.189	&	1.122	&	1.057	&	0.97	&	0.89	&	0.811	&	0.727	&	0.664	\\	
	& d3	&	1.181	&	1.213	&	\textbf{1.177}	&	1.128	&	1.08	&	1.003	&	0.933	&	0.844	&	0.761	&	0.679	\\	
	& d4	&	1.081	&	1.129	&	1.123	&	\textbf{1.093}	&	1.07	&	1.02	&	0.97	&	0.906	&	0.84	&	0.768	\\	
	& d5	&	0.963	&	1.023	&	1.054	&	1.057	&	\textbf{1.064}	&	1.046	&	1.02	&	0.982	&	0.93	&	0.861	\\	
	& d6	&	0.894	&	0.954	&	0.986	&	1.01	&	1.033	&	\textbf{1.046}	&	1.054	&	1.035	&	1.015	&	0.973	\\	
	& d7	&	0.867	&	0.908	&	0.947	&	0.976	&	1.006	&	1.038	&	\textbf{1.059}	&	1.071	&	1.071	&	1.056	\\	
	& d8	&	0.782	&	0.8	&	0.849	&	0.9	&	0.948	&	1.019	&	1.085	&	\textbf{1.153}	&	1.213	&	1.252	\\	
	& d9	&	0.719	&	0.726	&	0.779	&	0.837	&	0.897	&	0.994	&	1.093	&	1.207	&	\textbf{1.331}	&	1.417	\\	
	& d10	&	0.722	&	0.659	&	0.677	&	0.744	&	0.8	&	0.931	&	1.06	&	1.235	&	1.456	&	\textbf{1.716}	\\	
\hline\hline
\end{tabular}
\caption{Results of estimated sorting measures using the ACS and for couples at most two years married in 2021 at different quantiles when all correlations are zero and distribution of $X$ as in total population.}
\label{tab:results_5_2021_acs}
\end{table}

\begin{table}[tbp]
\begin{tabular}{l||l|rrrrrrrrrrr}
\hline\hline
&  & \multicolumn{10}{c}{Husbands' quantiles} &  \\
 \hline\hline
&  & d1 & d2 & d3 & d4 & d5 & d6 & d7 & d8 & d9 & d10 &  \\
 \hline
\parbox[t]{2mm}{\multirow{10}{*}{\rotatebox[origin=c]{90}{Wives' quantiles}}}
	& d1	&	\textbf{2.963}	&	1.872	&	1.177	&	0.743	&	0.723	&	0.692	&	0.546	&	0.547	&	0.362	&	0.287	\\	
	& d2	&	2.139	&	\textbf{1.304}	&	1.209	&	0.86	&	0.827	&	1.125	&	0.464	&	0.638	&	0.814	&	0.539	\\	
	& d3	&	1.005	&	1.14	&	\textbf{1.544}	&	1.102	&	0.693	&	1.147	&	1.183	&	0.358	&	1.016	&	0.667	\\	
	& d4	&	0.965	&	1.407	&	0.848	&	\textbf{1.316}	&	0.778	&	1.236	&	0.879	&	0.99	&	0.25	&	1.283	\\	
	& d5	&	0.759	&	1.253	&	1.396	&	1.073	&	\textbf{1.174}	&	0.688	&	1.291	&	0.86	&	1.297	&	0.008	\\	
	& d6	&	0.359	&	0.864	&	1.302	&	1.404	&	1.137	&	\textbf{0.825}	&	1.121	&	1.004	&	0.783	&	0.991	\\	
	& d7	&	0.383	&	0.846	&	0.739	&	0.429	&	1.324	&	1.026	&	\textbf{0.944}	&	1.383	&	1.339	&	1.653	\\	
	& d8	&	0.575	&	0.794	&	0.745	&	0.54	&	1.422	&	0.725	&	1.198	&	\textbf{1.268}	&	1.361	&	1.167	\\	
	& d9	&	0.729	&	0.119	&	0.449	&	0.659	&	1.397	&	1.142	&	1.291	&	1.62	&	\textbf{1.766}	&	0.611	\\	
	& d10	&	0.124	&	0.4	&	0.592	&	1.873	&	0.526	&	1.395	&	1.083	&	1.332	&	1.013	&	\textbf{2.795}	\\	
\hline\hline
\end{tabular}
\caption{Results of estimated sorting measures using the ACS and for couples at most two years married in 1980 at different quantiles when $\protect\sigma_{D^*_w, Y^*_w} = 0$.}
\label{tab:results_1_1980_newly_weds}
\end{table}

\begin{table}[tbp]
\begin{tabular}{l||l|rrrrrrrrrrr}
\hline\hline
&  & \multicolumn{10}{c}{Husbands' quantiles} &  \\
 \hline\hline
&  & d1 & d2 & d3 & d4 & d5 & d6 & d7 & d8 & d9 & d10 &  \\
 \hline
\parbox[t]{2mm}{\multirow{10}{*}{\rotatebox[origin=c]{90}{Wives' quantiles}}}
	& d1	&	\textbf{2.726}	&	2.034	&	1.116	&	0.665	&	0.568	&	0.662	&	0.754	&	0.419	&	0.363	&	0.516	\\	
	& d2	&	1.519	&	\textbf{1.11}	&	0.486	&	2.216	&	1.22	&	0.601	&	0.337	&	1.058	&	0.542	&	0.663	\\	
	& d3	&	1.89	&	1.468	&	\textbf{2.213}	&	0.845	&	1.103	&	1.098	&	0.507	&	1.01	&	0.813	&	0.697	\\	
	& d4	&	0.558	&	1.469	&	1.393	&	\textbf{1.507}	&	0.784	&	0.867	&	1.33	&	0.926	&	0.718	&	0.141	\\	
	& d5	&	0.247	&	0.825	&	1.498	&	1.399	&	\textbf{0.93}	&	1.497	&	1.716	&	0.158	&	1.13	&	0.357	\\	
	& d6	&	0.629	&	0.712	&	0.761	&	1.172	&	1.181	&	\textbf{0.837}	&	1.53	&	1.106	&	0.573	&	1.313	\\	
	& d7	&	0.442	&	0.825	&	1.079	&	0.573	&	1.224	&	1.541	&	\textbf{0.153}	&	1.332	&	1.482	&	1.157	\\	
	& d8	&	0.648	&	0.625	&	0.537	&	0.616	&	1.231	&	1.08	&	1.658	&	\textbf{1.248}	&	1.045	&	0.882	\\	
	& d9	&	0.664	&	0.603	&	0.632	&	0.504	&	1.268	&	1.008	&	1.32	&	1.47	&	\textbf{1.464}	&	1.219	\\	
	& d10	&	0.678	&	0.329	&	0.285	&	0.503	&	0.49	&	0.809	&	0.695	&	1.273	&	1.871	&	\textbf{3.055}	\\	
\hline\hline
\end{tabular}
\caption{Results of estimated sorting measures using the ACS and for couples at most two years married in 1980 at different quantiles when $\protect\sigma_{D^*_w, Y^*_w} = 0, \protect\sigma_{D^*_w, Y^*_h} = 0.$}
\label{tab:results_2_1980_newly_weds}
\end{table}

\begin{table}[tbp]
\begin{tabular}{l||l|rrrrrrrrrrr}
\hline\hline
&  & \multicolumn{10}{c}{Husbands' quantiles} &  \\
 \hline\hline
&  & d1 & d2 & d3 & d4 & d5 & d6 & d7 & d8 & d9 & d10 &  \\
 \hline
\parbox[t]{2mm}{\multirow{10}{*}{\rotatebox[origin=c]{90}{Wives' quantiles}}}
	& d1	&	\textbf{2.831}	&	1.775	&	1.157	&	0.806	&	0.562	&	0.676	&	0.766	&	0.417	&	0.417	&	0.533	\\	
	& d2	&	2.21	&	\textbf{1.456}	&	1.563	&	0.952	&	0.826	&	0.844	&	0.356	&	0.861	&	0.255	&	0.815	\\	
	& d3	&	1.243	&	1.412	&	\textbf{1.193}	&	0.731	&	1.574	&	0.896	&	0.415	&	1.131	&	0.961	&	0.364	\\	
	& d4	&	0.515	&	1.006	&	1.356	&	\textbf{1.971}	&	0.65	&	1.126	&	1.169	&	1.122	&	0.885	&	0.357	\\	
	& d5	&	0.386	&	0.885	&	1.744	&	1.586	&	\textbf{0.66}	&	1.328	&	1.943	&	0.067	&	1.01	&	0.362	\\	
	& d6	&	0.547	&	0.716	&	1.113	&	0.994	&	1.236	&	\textbf{1.57}	&	0.889	&	1.102	&	0.722	&	1.062	\\	
	& d7	&	0.344	&	1.036	&	0.766	&	0.829	&	1.283	&	0.989	&	\textbf{0.69}	&	1.338	&	1.364	&	1.366	\\	
	& d8	&	0.561	&	0.743	&	0.443	&	0.714	&	1.365	&	0.754	&	1.847	&	\textbf{1.244}	&	1.016	&	0.853	\\	
	& d9	&	0.656	&	0.611	&	0.426	&	0.965	&	1.263	&	1.006	&	1.393	&	1.53	&	\textbf{1.472}	&	1.245	\\	
	& d10	&	0.708	&	0.359	&	0.24	&	0.452	&	0.582	&	0.81	&	0.532	&	1.189	&	1.899	&	\textbf{3.043}	\\	
\hline\hline
\end{tabular}
\caption{Results of estimated sorting measures using the ACS and for couples at most two years married in 1980 at different quantiles when $\protect\protect\sigma_{D^*_w, Y^*_w} = 0, \protect\sigma_{D^*_w, Y^*_h} = 0$, and the distribution of $X$ as in total population.}
\label{tab:results_3_1980_newly_weds}
\end{table}

\begin{table}[tbp]
\begin{tabular}{l||l|rrrrrrrrrrr}
\hline\hline
&  & \multicolumn{10}{c}{Husbands' quantiles} &  \\
 \hline\hline
&  & d1 & d2 & d3 & d4 & d5 & d6 & d7 & d8 & d9 & d10 &  \\
 \hline
\parbox[t]{2mm}{\multirow{10}{*}{\rotatebox[origin=c]{90}{Wives' quantiles}}}
	& d1	&	\textbf{1.473}	&	1.221	&	1.155	&	0.982	&	0.909	&	0.824	&	0.953	&	0.81	&	0.742	&	0.857	\\	
	& d2	&	1.392	&	\textbf{1.13}	&	1.208	&	1.014	&	1.017	&	0.908	&	0.894	&	0.897	&	0.728	&	0.929	\\	
	& d3	&	1.138	&	1.084	&	\textbf{0.988}	&	0.974	&	0.981	&	1.006	&	0.893	&	0.963	&	0.912	&	0.932	\\	
	& d4	&	1.026	&	1.131	&	0.965	&	\textbf{1.079}	&	1.048	&	0.993	&	1.002	&	1.057	&	0.999	&	0.885	\\	
	& d5	&	0.867	&	1.022	&	1.042	&	1.029	&	\textbf{1.077}	&	1.082	&	1.035	&	0.998	&	0.946	&	0.898	\\	
	& d6	&	0.841	&	0.934	&	1.018	&	1.004	&	1.008	&	\textbf{1.097}	&	1.02	&	1.006	&	1.011	&	1.018	\\	
	& d7	&	0.83	&	0.964	&	0.984	&	0.986	&	1.043	&	1.095	&	\textbf{0.995}	&	1.003	&	1.112	&	0.986	\\	
	& d8	&	0.853	&	0.932	&	0.841	&	0.985	&	1.038	&	1.009	&	1.037	&	\textbf{1.051}	&	0.974	&	0.827	\\	
	& d9	&	0.756	&	0.869	&	0.852	&	1.025	&	1.045	&	1.069	&	1.145	&	1.185	&	\textbf{1.391}	&	1.238	\\	
	& d10	&	0.824	&	0.712	&	0.947	&	0.922	&	0.834	&	0.918	&	1.027	&	1.031	&	1.185	&	\textbf{1.43}	\\	
\hline\hline
\end{tabular}
\caption{Results of estimated sorting measures using the ACS and for couples at most two years married in 1980 at different quantiles when all correlations are zero and distribution of $X$ as in total population.}
\label{tab:results_4_1980_newly_weds}
\end{table}

\begin{table}[tbp]
\begin{tabular}{l||l|rrrrrrrrrrr}
\hline\hline
&  & \multicolumn{10}{c}{Husbands' quantiles} &  \\
 \hline\hline
&  & d1 & d2 & d3 & d4 & d5 & d6 & d7 & d8 & d9 & d10 &  \\
 \hline
\parbox[t]{2mm}{\multirow{10}{*}{\rotatebox[origin=c]{90}{Wives' quantiles}}}
	& d1	&	\textbf{3.111}	&	1.775	&	1.101	&	0.925	&	0.808	&	0.673	&	0.479	&	0.412	&	0.613	&	0.104	\\	
	& d2	&	1.6	&	\textbf{1.67}	&	1.438	&	1.101	&	0.618	&	1.032	&	0.451	&	0.699	&	0.693	&	0.697	\\	
	& d3	&	1.073	&	1.836	&	\textbf{1.377}	&	1.317	&	1.414	&	0.794	&	0.666	&	0.917	&	0.427	&	0.179	\\	
	& d4	&	1.044	&	1.147	&	1.41	&	\textbf{1.584}	&	0.962	&	0.99	&	0.97	&	0.705	&	0.641	&	0.547	\\	
	& d5	&	0.881	&	1.038	&	1.498	&	1.464	&	\textbf{1.173}	&	0.938	&	1.405	&	0.786	&	0.553	&	0.264	\\	
	& d6	&	0.812	&	1.035	&	1.296	&	0.927	&	1.296	&	\textbf{1.279}	&	0.963	&	0.924	&	1.038	&	0.43	\\	
	& d7	&	0.831	&	0.813	&	0.836	&	0.998	&	1.073	&	1.294	&	\textbf{1.441}	&	1.353	&	0.693	&	0.67	\\	
	& d8	&	0.35	&	0.492	&	0.673	&	0.77	&	1.067	&	1.3	&	1.336	&	\textbf{1.995}	&	1.024	&	0.993	\\	
	& d9	&	0.278	&	0.322	&	0.404	&	0.58	&	0.451	&	1.072	&	0.964	&	1.551	&	\textbf{1.96}	&	2.418	\\	
	& d10	&	0.193	&	0.012	&	0.082	&	0.415	&	1.13	&	0.655	&	1.287	&	0.677	&	2.194	&	\textbf{3.355}	\\	
\hline\hline
\end{tabular}
\caption{Results of estimated sorting measures using the ACS and for couples at most two years married in 2021 at different quantiles when $\protect\sigma_{D^*_w, Y^*_w} = 0$.}
\label{tab:results_2_2021_newly_weds}
\end{table}

\begin{table}[tbp]
\begin{tabular}{l||l|rrrrrrrrrrr}
\hline\hline
&  & \multicolumn{10}{c}{Husbands' quantiles} &  \\
 \hline\hline
&  & d1 & d2 & d3 & d4 & d5 & d6 & d7 & d8 & d9 & d10 &  \\
 \hline
\parbox[t]{2mm}{\multirow{10}{*}{\rotatebox[origin=c]{90}{Wives' quantiles}}}
	& d1	&	\textbf{3.063}	&	1.641	&	0.951	&	0.909	&	0.801	&	0.797	&	0.317	&	0.545	&	0.475	&	0.5	\\	
	& d2	&	1.216	&	\textbf{1.974}	&	1.335	&	1.46	&	0.765	&	0.994	&	0.763	&	0.784	&	0.657	&	0.053	\\	
	& d3	&	1.052	&	1.717	&	\textbf{1.44}	&	1.073	&	1.402	&	0.742	&	0.977	&	0.797	&	0.55	&	0.248	\\	
	& d4	&	0.995	&	0.947	&	1.794	&	\textbf{1.789}	&	0.856	&	0.797	&	1.22	&	0.518	&	0.264	&	0.82	\\	
	& d5	&	1.085	&	1.202	&	1.272	&	1.357	&	\textbf{1.309}	&	1.394	&	0.707	&	0.727	&	0.715	&	0.232	\\	
	& d6	&	0.846	&	0.883	&	1.137	&	0.52	&	2.127	&	\textbf{1.281}	&	0.947	&	0.882	&	0.707	&	0.671	\\	
	& d7	&	0.349	&	0.659	&	0.72	&	1.647	&	0.868	&	0.901	&	\textbf{1.51}	&	1.824	&	0.703	&	0.819	\\	
	& d8	&	0.635	&	0.474	&	0.791	&	0.548	&	0.835	&	1.275	&	1.528	&	\textbf{1.5}	&	1.399	&	1.014	\\	
	& d9	&	0.297	&	0.218	&	0.329	&	0.319	&	0.568	&	1.118	&	1.403	&	1.415	&	\textbf{2.091}	&	2.242	\\	
	& d10	&	0.449	&	0.274	&	0.223	&	0.385	&	0.458	&	0.694	&	0.637	&	1.019	&	2.443	&	\textbf{3.419}	\\	
\hline\hline
\end{tabular}
\caption{Results of estimated sorting measures using the ACS and for couples at most two years married in 2021 at different quantiles when $\protect\sigma_{D^*_w, Y^*_w} = 0, \protect\sigma_{D^*_w, Y^*_h} = 0.$}
\label{tab:results_3_2021_newly_weds}
\end{table}

\begin{table}[tbp]
\begin{tabular}{l||l|rrrrrrrrrrr}
\hline\hline
&  & \multicolumn{10}{c}{Husbands' quantiles} &  \\
 \hline\hline
&  & d1 & d2 & d3 & d4 & d5 & d6 & d7 & d8 & d9 & d10 &  \\
 \hline
\parbox[t]{2mm}{\multirow{10}{*}{\rotatebox[origin=c]{90}{Wives' quantiles}}}
	& d1	&	\textbf{3.15}	&	1.55	&	0.911	&	1.08	&	0.63	&	0.87	&	0.22	&	0.594	&	0.534	&	0.462	\\	
	& d2	&	1.229	&	\textbf{1.944}	&	1.572	&	1.17	&	0.774	&	1.031	&	0.706	&	0.869	&	0.595	&	0.11	\\	
	& d3	&	1.011	&	1.649	&	\textbf{1.424}	&	1.358	&	1.276	&	0.727	&	0.918	&	0.881	&	0.482	&	0.273	\\	
	& d4	&	1.155	&	1.082	&	1.558	&	\textbf{1.506}	&	1.225	&	0.941	&	0.877	&	0.584	&	0.388	&	0.683	\\	
	& d5	&	0.777	&	1.073	&	1.401	&	1.473	&	\textbf{1.285}	&	0.902	&	1.372	&	0.737	&	0.596	&	0.385	\\	
	& d6	&	0.918	&	0.945	&	1.153	&	0.55	&	1.806	&	\textbf{1.518}	&	0.785	&	0.781	&	0.83	&	0.714	\\	
	& d7	&	0.366	&	0.721	&	0.656	&	1.559	&	0.994	&	1.008	&	\textbf{1.734}	&	1.732	&	0.593	&	0.637	\\	
	& d8	&	0.58	&	0.532	&	0.685	&	0.532	&	0.967	&	1.343	&	1.437	&	\textbf{1.474}	&	0.995	&	1.455	\\	
	& d9	&	0.355	&	0.206	&	0.423	&	0.394	&	0.61	&	1.063	&	1.249	&	1.375	&	\textbf{2.141}	&	2.185	\\	
	& d10	&	0.449	&	0.297	&	0.224	&	0.388	&	0.436	&	0.591	&	0.708	&	0.968	&	2.844	&	\textbf{3.096}	\\	
\hline\hline
\end{tabular}
\caption{Results of estimated sorting measures using the ACS and for couples at most two years married in 2021 at different quantiles when $\protect\protect\sigma_{D^*_w, Y^*_w} = 0, \protect\sigma_{D^*_w, Y^*_h} = 0$, and the distribution of $X$ as in total population.}
\label{tab:results_4_2021_newly_weds}
\end{table}

\begin{table}[tbp]
\begin{tabular}{l||l|rrrrrrrrrrr}
\hline\hline
&  & \multicolumn{10}{c}{Husbands' quantiles} &  \\
 \hline\hline
&  & d1 & d2 & d3 & d4 & d5 & d6 & d7 & d8 & d9 & d10 &  \\
 \hline
\parbox[t]{2mm}{\multirow{10}{*}{\rotatebox[origin=c]{90}{Wives' quantiles}}}
	& d1	&	\textbf{1.454}	&	1.443	&	1.147	&	1.18	&	0.969	&	0.904	&	0.757	&	0.797	&	0.658	&	0.692	\\	
	& d2	&	1.356	&	\textbf{1.392}	&	1.271	&	1.206	&	0.976	&	0.933	&	0.795	&	0.86	&	0.647	&	0.564	\\	
	& d3	&	1.157	&	1.151	&	\textbf{1.164}	&	1.185	&	1.005	&	0.949	&	0.955	&	0.942	&	0.786	&	0.707	\\	
	& d4	&	1.066	&	1.053	&	1.05	&	\textbf{1.111}	&	1.015	&	1.004	&	1.03	&	0.997	&	0.861	&	0.813	\\	
	& d5	&	1.0	&	1.0	&	1.062	&	1.044	&	\textbf{1.035}	&	1.006	&	1.001	&	0.99	&	0.954	&	0.907	\\	
	& d6	&	0.873	&	0.876	&	0.998	&	0.971	&	1.034	&	\textbf{1.043}	&	1.069	&	1.056	&	1.052	&	1.028	\\	
	& d7	&	0.852	&	0.855	&	0.897	&	0.92	&	1.012	&	1.054	&	\textbf{1.085}	&	1.084	&	1.119	&	1.121	\\	
	& d8	&	0.87	&	0.909	&	0.928	&	0.924	&	1.013	&	1.045	&	1.063	&	\textbf{1.089}	&	1.096	&	1.064	\\	
	& d9	&	0.672	&	0.659	&	0.751	&	0.741	&	0.984	&	1.044	&	1.159	&	1.124	&	\textbf{1.407}	&	1.459	\\	
	& d10	&	0.698	&	0.662	&	0.732	&	0.719	&	0.957	&	1.018	&	1.088	&	1.062	&	1.419	&	\textbf{1.645}	\\	
\hline\hline
\end{tabular}
\caption{Results of estimated sorting measures using the ACS and for couples at most two years married in 2021 at different quantiles when all correlations are zero and distribution of $X$ as in total population.}
\label{tab:results_5_2021_newly_weds}
\end{table}

\section{Robustness checks for the restriction of FTFY married couples}

\label{web_appendix:robust_ftfy}

Our baseline analysis reports results based on FTFY married couples. The reason for this was to obtain a homogenous number of working hours between the couples which makes it possible to expect that our analysis for hourly wages carries over to an analysis of yearly incomes. Nevertheless, we would like to also investigate whether our choice could explain our baseline results. Therefore, we also perform an analysis with couples in which the husband and the wife work any positive amount of hours. The results of this robustness check are in Tables \ref{tab:results_original_1976_parttime}-\ref{tab:results_4_1976_parttime} for the period 1976-1980 and in \ref{tab:results_1_2021_parttime}-\ref{tab:results_5_2021_parttime} for the period 2021-2025.

\begin{table}[tbp]
\begin{tabular}{l||l|rrrrrrrrrrr}
\hline\hline
&  & \multicolumn{10}{c}{Husbands' quantiles} &  \\
 \hline\hline
&  & d1 & d2 & d3 & d4 & d5 & d6 & d7 & d8 & d9 & d10 &  \\
 \hline
\parbox[t]{2mm}{\multirow{10}{*}{\rotatebox[origin=c]{90}{Wives' quantiles}}}
	& d1	&	\textbf{1.677}	&	1.219	&	1.009	&	0.981	&	0.933	&	0.895	&	0.899	&	0.887	&	0.899	&	0.602	\\	
	& d2	&	1.491	&	\textbf{1.19}	&	1.08	&	0.965	&	0.988	&	0.887	&	0.858	&	0.84	&	0.923	&	0.776	\\	
	& d3	&	1.35	&	1.248	&	\textbf{1.117}	&	1.032	&	0.963	&	0.916	&	0.876	&	0.86	&	0.817	&	0.823	\\	
	& d4	&	1.091	&	1.3	&	1.002	&	\textbf{1.188}	&	1.013	&	0.969	&	0.93	&	0.922	&	0.914	&	0.67	\\	
	& d5	&	0.964	&	1.252	&	1.115	&	1.075	&	\textbf{1.038}	&	0.952	&	0.927	&	0.934	&	0.914	&	0.83	\\	
	& d6	&	0.828	&	1.117	&	1.081	&	1.064	&	1.101	&	\textbf{1.042}	&	1.058	&	0.983	&	0.947	&	0.78	\\	
	& d7	&	0.746	&	0.971	&	1.055	&	1.111	&	1.058	&	1.041	&	\textbf{0.961}	&	1.089	&	0.976	&	0.992	\\	
	& d8	&	0.677	&	0.733	&	0.868	&	1.148	&	1.105	&	1.074	&	1.068	&	\textbf{1.025}	&	1.106	&	1.199	\\	
	& d9	&	0.628	&	0.707	&	0.728	&	0.89	&	1.092	&	1.16	&	1.204	&	1.143	&	\textbf{1.202}	&	1.244	\\	
	& d10	&	0.547	&	0.264	&	0.946	&	0.547	&	0.711	&	1.065	&	1.219	&	1.318	&	1.3	&	\textbf{2.083}	\\	
\hline\hline
\end{tabular}
\caption{Results of estimated sorting measures for all two earner couples for the years 1976-1980 at different quantiles}
\label{tab:results_original_1976_parttime}
\end{table}

\begin{table}[tbp]
\begin{tabular}{l||l|rrrrrrrrrrr}
\hline\hline
&  & \multicolumn{10}{c}{Husbands' quantiles} &  \\
 \hline\hline
&  & d1 & d2 & d3 & d4 & d5 & d6 & d7 & d8 & d9 & d10 &  \\
 \hline
\parbox[t]{2mm}{\multirow{10}{*}{\rotatebox[origin=c]{90}{Wives' quantiles}}}
	& d1	&	\textbf{1.679}	&	1.225	&	1.013	&	0.983	&	0.925	&	0.925	&	0.883	&	0.865	&	0.915	&	0.591	\\	
	& d2	&	1.492	&	\textbf{1.19}	&	1.084	&	0.97	&	0.986	&	0.909	&	0.818	&	0.855	&	0.928	&	0.765	\\	
	& d3	&	1.34	&	1.25	&	\textbf{1.105}	&	1.038	&	0.967	&	0.942	&	0.849	&	0.842	&	0.833	&	0.832	\\	
	& d4	&	1.094	&	1.301	&	0.995	&	\textbf{1.175}	&	0.998	&	0.989	&	0.903	&	0.889	&	0.936	&	0.719	\\	
	& d5	&	0.958	&	1.254	&	1.114	&	1.103	&	\textbf{1.027}	&	0.981	&	0.922	&	0.913	&	0.939	&	0.79	\\	
	& d6	&	0.833	&	1.093	&	1.081	&	1.048	&	1.103	&	\textbf{1.042}	&	0.993	&	0.953	&	0.966	&	0.891	\\	
	& d7	&	0.755	&	1.004	&	1.059	&	1.108	&	1.073	&	1.074	&	\textbf{0.964}	&	1.067	&	1.006	&	0.887	\\	
	& d8	&	0.672	&	0.739	&	0.856	&	1.145	&	1.097	&	1.104	&	1.048	&	\textbf{0.996}	&	1.151	&	1.193	\\	
	& d9	&	0.629	&	0.707	&	0.726	&	0.895	&	1.084	&	1.2	&	1.16	&	1.105	&	\textbf{1.235}	&	1.263	\\	
	& d10	&	0.547	&	0.239	&	0.968	&	0.535	&	0.739	&	0.834	&	1.46	&	1.515	&	1.091	&	\textbf{2.067}	\\	
\hline\hline
\end{tabular}
\caption{Results of estimated sorting measures for all two earner couples in 1976-1980 at different quantiles when $\protect\sigma_{D^*_w, Y^*_w} = 0$.}
\label{tab:results_1_1976_parttime}
\end{table}

\begin{table}[tbp]
\begin{tabular}{l||l|rrrrrrrrrrr}
\hline\hline
&  & \multicolumn{10}{c}{Husbands' quantiles} &  \\
 \hline\hline
&  & d1 & d2 & d3 & d4 & d5 & d6 & d7 & d8 & d9 & d10 &  \\
 \hline
\parbox[t]{2mm}{\multirow{10}{*}{\rotatebox[origin=c]{90}{Wives' quantiles}}}
	& d1	&	\textbf{1.662}	&	1.048	&	1.003	&	0.841	&	0.883	&	0.939	&	0.906	&	0.871	&	0.91	&	0.933	\\	
	& d2	&	1.569	&	\textbf{1.27}	&	1.123	&	1.006	&	0.953	&	0.92	&	0.887	&	0.864	&	0.719	&	0.697	\\	
	& d3	&	1.248	&	1.414	&	\textbf{1.142}	&	1.069	&	0.963	&	0.958	&	0.905	&	0.915	&	0.698	&	0.689	\\	
	& d4	&	1.002	&	1.351	&	1.252	&	\textbf{1.158}	&	1.051	&	0.961	&	1.036	&	0.869	&	0.823	&	0.497	\\	
	& d5	&	0.827	&	1.332	&	1.23	&	1.228	&	\textbf{1.045}	&	0.969	&	1.012	&	1.027	&	0.821	&	0.516	\\	
	& d6	&	0.818	&	0.936	&	1.245	&	1.179	&	1.134	&	\textbf{1.008}	&	1.065	&	1.015	&	0.939	&	0.662	\\	
	& d7	&	0.809	&	0.792	&	0.993	&	1.192	&	1.143	&	1.113	&	\textbf{1.062}	&	1.094	&	0.998	&	0.814	\\	
	& d8	&	0.743	&	0.762	&	0.833	&	0.95	&	1.209	&	1.14	&	1.08	&	\textbf{1.062}	&	1.154	&	1.059	\\	
	& d9	&	0.591	&	0.603	&	0.604	&	0.76	&	0.917	&	1.155	&	1.129	&	1.394	&	\textbf{1.441}	&	1.396	\\	
	& d10	&	0.732	&	0.492	&	0.574	&	0.616	&	0.702	&	0.837	&	0.919	&	0.889	&	1.496	&	\textbf{2.737}	\\	
\hline\hline
\end{tabular}
\caption{Results of estimated sorting measures for all two earner couples in 1976-1980 at different quantiles when $\protect\sigma_{D^*_w, Y^*_w} = 0, \protect\sigma_{D^*_w, Y^*_h} = 0.$}
\label{tab:results_2_1976_parttime}
\end{table}

\begin{table}[tbp]
\begin{tabular}{l||l|rrrrrrrrrrr}
\hline\hline
&  & \multicolumn{10}{c}{Husbands' quantiles} &  \\
 \hline\hline
&  & d1 & d2 & d3 & d4 & d5 & d6 & d7 & d8 & d9 & d10 &  \\
 \hline
\parbox[t]{2mm}{\multirow{10}{*}{\rotatebox[origin=c]{90}{Wives' quantiles}}}
	& d1	&	\textbf{1.689}	&	1.047	&	1.015	&	0.844	&	0.871	&	0.939	&	0.886	&	0.867	&	0.918	&	0.925	\\	
	& d2	&	1.58	&	\textbf{1.273}	&	1.112	&	1.028	&	0.934	&	0.916	&	0.901	&	0.883	&	0.674	&	0.694	\\	
	& d3	&	1.2	&	1.401	&	\textbf{1.142}	&	1.096	&	0.954	&	1.0	&	0.915	&	0.882	&	0.702	&	0.714	\\	
	& d4	&	1.047	&	1.405	&	1.213	&	\textbf{1.108}	&	1.082	&	0.92	&	0.983	&	0.869	&	0.835	&	0.541	\\	
	& d5	&	0.841	&	1.34	&	1.227	&	1.222	&	\textbf{1.048}	&	0.987	&	1.018	&	1.007	&	0.822	&	0.49	\\	
	& d6	&	0.791	&	0.949	&	1.245	&	1.186	&	1.129	&	\textbf{1.009}	&	1.079	&	1.013	&	0.956	&	0.647	\\	
	& d7	&	0.807	&	0.816	&	0.965	&	1.197	&	1.137	&	1.089	&	\textbf{1.058}	&	1.047	&	1.045	&	0.837	\\	
	& d8	&	0.743	&	0.711	&	0.863	&	0.957	&	1.194	&	1.168	&	1.083	&	\textbf{1.094}	&	1.147	&	1.048	\\	
	& d9	&	0.563	&	0.599	&	0.616	&	0.746	&	0.963	&	1.146	&	1.142	&	1.38	&	\textbf{1.481}	&	1.359	\\	
	& d10	&	0.738	&	0.46	&	0.602	&	0.616	&	0.689	&	0.827	&	0.937	&	0.957	&	1.42	&	\textbf{2.745}	\\	
\hline\hline
\end{tabular}
\caption{Results of estimated sorting measures for all two earner couples in 1976-1980 at different quantiles when $\protect\protect\sigma_{D^*_w, Y^*_w} = 0, \protect\sigma_{D^*_w, Y^*_h} = 0$, and the distribution of $X$ as in total population.}
\label{tab:results_3_1976}
\end{table}

\begin{table}[tbp]
\begin{tabular}{l||l|rrrrrrrrrrr}
\hline\hline
&  & \multicolumn{10}{c}{Husbands' quantiles} &  \\
 \hline\hline
&  & d1 & d2 & d3 & d4 & d5 & d6 & d7 & d8 & d9 & d10 &  \\
 \hline
\parbox[t]{2mm}{\multirow{10}{*}{\rotatebox[origin=c]{90}{Wives' quantiles}}}
	& d1	&	\textbf{1.231}	&	1.145	&	1.086	&	1.02	&	0.983	&	0.959	&	0.948	&	0.914	&	0.827	&	0.889	\\	
	& d2	&	1.189	&	\textbf{1.134}	&	1.074	&	1.011	&	0.981	&	0.968	&	0.959	&	0.93	&	0.849	&	0.901	\\	
	& d3	&	1.133	&	1.092	&	\textbf{1.064}	&	1.017	&	0.993	&	0.988	&	0.972	&	0.939	&	0.9	&	0.905	\\	
	& d4	&	1.057	&	1.045	&	1.034	&	\textbf{1.012}	&	0.998	&	0.999	&	0.992	&	0.978	&	0.952	&	0.933	\\	
	& d5	&	0.989	&	1.019	&	1.024	&	1.023	&	\textbf{1.016}	&	1.006	&	0.994	&	0.988	&	0.987	&	0.953	\\	
	& d6	&	0.955	&	0.977	&	0.994	&	1.006	&	1.012	&	\textbf{1.013}	&	1.009	&	1.01	&	1.035	&	0.99	\\	
	& d7	&	0.914	&	0.944	&	0.967	&	0.998	&	1.012	&	1.022	&	\textbf{1.021}	&	1.034	&	1.055	&	1.029	\\	
	& d8	&	0.893	&	0.917	&	0.949	&	0.987	&	1.003	&	1.019	&	1.032	&	\textbf{1.068}	&	1.084	&	1.057	\\	
	& d9	&	0.766	&	0.83	&	0.892	&	0.956	&	1.008	&	1.015	&	1.052	&	1.127	&	\textbf{1.21}	&	1.14	\\	
	& d10	&	0.874	&	0.897	&	0.917	&	0.969	&	0.994	&	1.011	&	1.021	&	1.012	&	1.101	&	\textbf{1.203}	\\	
\hline\hline
\end{tabular}
\caption{Results of estimated sorting measures for all two earner couples in 1976-1980 at different quantiles when all correlations are zero and distribution of $X$ as in total population.}
\label{tab:results_4_1976_parttime}
\end{table}

\begin{table}[tbp]
\begin{tabular}{l||l|rrrrrrrrrrr}
\hline\hline
&  & \multicolumn{10}{c}{Husbands' quantiles} &  \\
 \hline\hline
&  & d1 & d2 & d3 & d4 & d5 & d6 & d7 & d8 & d9 & d10 &  \\
 \hline
\parbox[t]{2mm}{\multirow{10}{*}{\rotatebox[origin=c]{90}{Wives' quantiles}}}
	&d1	&	\textbf{2.358}	&	1.474	&	1.179	&	0.994	&	0.93	&	0.738	&	0.782	&	0.683	&	0.555	&	0.306	\\	
	&d2	&	1.872	&	\textbf{1.49}	&	1.286	&	1.125	&	1.068	&	0.828	&	0.814	&	0.687	&	0.569	&	0.261	\\	
	&d3	&	1.318	&	1.58	&	\textbf{1.239}	&	1.131	&	1.06	&	0.844	&	0.777	&	0.662	&	0.674	&	0.716	\\	
	&d4	&	0.85	&	1.184	&	1.439	&	\textbf{1.127}	&	1.166	&	0.963	&	0.932	&	0.834	&	0.814	&	0.691	\\	
	&d5	&	0.902	&	1.118	&	1.219	&	1.247	&	\textbf{1.13}	&	1.095	&	0.949	&	1.036	&	0.758	&	0.548	\\	
	&d6	&	0.803	&	0.742	&	0.906	&	1.344	&	1.39	&	\textbf{1.151}	&	1.126	&	0.843	&	0.972	&	0.723	\\	
	&d7	&	0.615	&	0.773	&	0.859	&	1.01	&	1.143	&	1.367	&	\textbf{1.191}	&	1.129	&	1.001	&	0.913	\\	
	&d8	&	0.687	&	0.693	&	0.684	&	0.757	&	1.01	&	1.216	&	1.451	&	\textbf{1.284}	&	1.318	&	0.901	\\	
	&d9	&	0.434	&	0.523	&	0.72	&	0.66	&	0.853	&	0.916	&	1.15	&	1.32	&	\textbf{1.513}	&	1.911	\\	
	&d10	&	0.161	&	0.423	&	0.468	&	0.605	&	0.25	&	0.882	&	0.828	&	1.523	&	1.827	&	\textbf{3.034}	\\	
\hline\hline
\end{tabular}
\caption{Results of estimated sorting measures for all two earner couples in 2021-2025 at different quantiles}
\label{tab:results_1_2021_parttime}
\end{table}

\begin{table}[tbp]
\begin{tabular}{l||l|rrrrrrrrrrr}
\hline\hline
&  & \multicolumn{10}{c}{Husbands' quantiles} &  \\
 \hline\hline
&  & d1 & d2 & d3 & d4 & d5 & d6 & d7 & d8 & d9 & d10 &  \\
 \hline
\parbox[t]{2mm}{\multirow{10}{*}{\rotatebox[origin=c]{90}{Wives' quantiles}}}
	&d1	&	\textbf{2.418}	&	1.506	&	1.178	&	0.992	&	0.963	&	0.728	&	0.746	&	0.681	&	0.599	&	0.189	\\	
	&d2	&	1.884	&	\textbf{1.529}	&	1.257	&	1.127	&	1.01	&	0.906	&	0.75	&	0.657	&	0.599	&	0.279	\\	
	&d3	&	1.279	&	1.504	&	\textbf{1.203}	&	1.069	&	1.055	&	0.881	&	0.772	&	0.688	&	0.674	&	0.874	\\	
	&d4	&	0.867	&	1.232	&	1.488	&	\textbf{1.11}	&	1.112	&	0.981	&	0.967	&	0.823	&	0.791	&	0.63	\\	
	&d5	&	0.892	&	1.038	&	1.184	&	1.193	&	\textbf{1.01}	&	1.055	&	0.953	&	0.934	&	0.77	&	0.971	\\	
	&d6	&	0.837	&	0.761	&	0.962	&	1.426	&	1.45	&	\textbf{1.213}	&	1.213	&	0.897	&	0.981	&	0.259	\\	
	&d7	&	0.544	&	0.767	&	0.914	&	0.927	&	1.078	&	1.319	&	\textbf{1.175}	&	1.068	&	1.005	&	1.203	\\	
	&d8	&	0.71	&	0.684	&	0.693	&	0.697	&	0.985	&	1.245	&	1.425	&	\textbf{1.304}	&	1.382	&	0.875	\\	
	&d9	&	0.505	&	0.543	&	0.687	&	0.751	&	0.814	&	0.884	&	1.145	&	1.367	&	\textbf{1.456}	&	1.849	\\	
	&d10	&	0.062	&	0.436	&	0.435	&	0.708	&	0.524	&	0.786	&	0.854	&	1.582	&	1.743	&	\textbf{2.87}	\\	
\hline\hline
\end{tabular}
\caption{Results of estimated sorting measures for all two earner couples in2021-2025 at different quantiles when $\protect\sigma_{D^*_w, Y^*_w} = 0$.}
\label{tab:results_2_2021_parttime}
\end{table}

\begin{table}[tbp]
\begin{tabular}{l||l|rrrrrrrrrrr}
\hline\hline
&  & \multicolumn{10}{c}{Husbands' quantiles} &  \\
 \hline\hline
&  & d1 & d2 & d3 & d4 & d5 & d6 & d7 & d8 & d9 & d10 &  \\
 \hline
\parbox[t]{2mm}{\multirow{10}{*}{\rotatebox[origin=c]{90}{Wives' quantiles}}}
	&d1	&	\textbf{2.41}	&	1.513	&	1.069	&	0.994	&	0.789	&	0.793	&	0.751	&	0.615	&	0.493	&	0.572	\\	
	&d2	&	1.542	&	\textbf{1.76}	&	1.36	&	1.172	&	0.815	&	0.91	&	0.74	&	0.765	&	0.542	&	0.394	\\	
	&d3	&	1.033	&	1.631	&	\textbf{1.32}	&	1.299	&	1.04	&	0.864	&	0.825	&	0.789	&	0.63	&	0.569	\\	
	&d4	&	0.734	&	0.991	&	1.393	&	\textbf{1.208}	&	1.179	&	1.081	&	1.164	&	0.84	&	0.759	&	0.652	\\	
	&d5	&	0.869	&	0.962	&	1.191	&	1.319	&	\textbf{1.179}	&	1.12	&	0.868	&	1.034	&	0.79	&	0.668	\\	
	&d6	&	0.771	&	0.822	&	0.965	&	1.168	&	1.488	&	\textbf{0.914}	&	1.22	&	0.938	&	0.981	&	0.733	\\	
	&d7	&	0.74	&	0.565	&	0.64	&	0.962	&	1.322	&	1.371	&	\textbf{0.929}	&	1.396	&	1.109	&	0.966	\\	
	&d8	&	0.798	&	0.807	&	0.98	&	0.768	&	0.864	&	0.97	&	1.396	&	\textbf{0.949}	&	1.268	&	1.201	\\	
	&d9	&	0.619	&	0.426	&	0.608	&	0.651	&	0.666	&	1.135	&	1.201	&	1.631	&	\textbf{1.728}	&	1.335	\\	
	&d10	&	0.485	&	0.522	&	0.474	&	0.459	&	0.656	&	0.842	&	0.907	&	1.043	&	1.702	&	\textbf{2.91}	\\	
\hline\hline
\end{tabular}
\caption{Results of estimated sorting measures for all two earner couples in2021-2025 at different quantiles when $\protect\sigma_{D^*_w, Y^*_w} = 0, \protect\sigma_{D^*_w, Y^*_h} = 0.$}
\label{tab:results_3_2021_parttime}
\end{table}

\begin{table}[tbp]
\begin{tabular}{l||l|rrrrrrrrrrr}
\hline\hline
&  & \multicolumn{10}{c}{Husbands' quantiles} &  \\
 \hline\hline
&  & d1 & d2 & d3 & d4 & d5 & d6 & d7 & d8 & d9 & d10 &  \\
 \hline
\parbox[t]{2mm}{\multirow{10}{*}{\rotatebox[origin=c]{90}{Wives' quantiles}}}
	&d1	&	\textbf{2.466}	&	1.437	&	1.09	&	0.984	&	0.758	&	0.825	&	0.664	&	0.647	&	0.47	&	0.66	\\	
	&d2	&	1.639	&	\textbf{1.795}	&	1.296	&	1.187	&	0.851	&	0.828	&	0.772	&	0.695	&	0.536	&	0.401	\\	
	&d3	&	0.988	&	1.666	&	\textbf{1.373}	&	1.27	&	0.934	&	0.955	&	0.85	&	0.791	&	0.618	&	0.556	\\	
	&d4	&	0.675	&	1.095	&	1.358	&	\textbf{1.187}	&	1.218	&	1.107	&	0.99	&	0.978	&	0.79	&	0.603	\\	
	&d5	&	0.864	&	0.876	&	1.222	&	1.375	&	\textbf{1.132}	&	1.03	&	1.097	&	0.873	&	0.863	&	0.669	\\	
	&d6	&	0.754	&	0.847	&	0.916	&	1.198	&	1.507	&	\textbf{0.948}	&	1.204	&	1.029	&	0.899	&	0.699	\\	
	&d7	&	0.707	&	0.578	&	0.614	&	0.949	&	1.302	&	1.416	&	\textbf{0.977}	&	1.225	&	1.257	&	0.975	\\	
	&d8	&	0.751	&	0.785	&	0.999	&	0.818	&	0.866	&	0.992	&	1.352	&	\textbf{1.099}	&	1.137	&	1.202	\\	
	&d9	&	0.618	&	0.485	&	0.532	&	0.634	&	0.712	&	1.114	&	1.162	&	1.604	&	\textbf{1.771}	&	1.369	\\	
	&d10	&	0.538	&	0.436	&	0.601	&	0.399	&	0.721	&	0.786	&	0.934	&	1.059	&	1.66	&	\textbf{2.867}	\\	
\hline\hline
\end{tabular}
\caption{Results of estimated sorting measures for all two earner couples in 2021-2025 at different quantiles when $\protect\protect\sigma_{D^*_w, Y^*_w} = 0, \protect\sigma_{D^*_w, Y^*_h} = 0$, and the distribution of $X$ as in total population.}
\label{tab:results_4_2021_parttime}
\end{table}

\begin{table}[tbp]
\begin{tabular}{l||l|rrrrrrrrrrr}
\hline\hline
&  & \multicolumn{10}{c}{Husbands' quantiles} &  \\
 \hline\hline
&  & d1 & d2 & d3 & d4 & d5 & d6 & d7 & d8 & d9 & d10 &  \\
 \hline
\parbox[t]{2mm}{\multirow{10}{*}{\rotatebox[origin=c]{90}{Wives' quantiles}}}
	&d1	&	\textbf{1.308}	&	1.184	&	1.079	&	1.032	&	0.953	&	0.973	&	0.878	&	0.868	&	0.858	&	0.866	\\	
	&d2	&	1.145	&	\textbf{1.137}	&	1.089	&	1.05	&	0.996	&	0.986	&	0.932	&	0.915	&	0.89	&	0.861	\\	
	&d3	&	1.093	&	1.094	&	\textbf{1.077}	&	1.048	&	1.008	&	1.006	&	0.959	&	0.941	&	0.906	&	0.868	\\	
	&d4	&	1.033	&	1.06	&	1.067	&	\textbf{1.042}	&	1.022	&	1.009	&	0.98	&	0.971	&	0.933	&	0.883	\\	
	&d5	&	0.967	&	0.997	&	1.035	&	1.022	&	\textbf{1.029}	&	1.024	&	1.021	&	1.004	&	0.963	&	0.938	\\	
	&d6	&	0.953	&	0.973	&	0.999	&	1.005	&	1.02	&	\textbf{1.023}	&	1.021	&	1.016	&	1.0	&	0.989	\\	
	&d7	&	0.913	&	0.935	&	0.95	&	0.977	&	0.997	&	1.007	&	\textbf{1.043}	&	1.053	&	1.058	&	1.067	\\	
	&d8	&	0.885	&	0.91	&	0.952	&	0.973	&	1.01	&	1.014	&	1.059	&	\textbf{1.065}	&	1.07	&	1.06	\\	
	&d9	&	0.804	&	0.831	&	0.884	&	0.939	&	1.0	&	0.991	&	1.075	&	1.11	&	\textbf{1.165}	&	1.203	\\	
	&d10	&	0.897	&	0.878	&	0.868	&	0.912	&	0.965	&	0.966	&	1.032	&	1.057	&	1.159	&	\textbf{1.265}	\\	
\hline\hline
\end{tabular}
\caption{Results of estimated sorting measures for all two earner couples in 2021-2025 at different quantiles when all correlations are zero and distribution of $X$ as in total population.}
\label{tab:results_5_2021_parttime}
\end{table}

\section{Robustness checks for using only married individuals}

\label{web_appendix:robust_cohab}

As discussed in the main text, our analysis only uses married individuals. Since 1995, there are explicit questions to individuals whether they cohabit. Therefore, we perform a robustness check using not only married but also cohabiting couples in our analysis. The results for the period 2021-2025 are reported in Tables \ref{tab:results_1_2021_cohab} to \ref{tab:results_5_2021_cohab}.

\begin{table}[tbp]
\begin{tabular}{l||l|rrrrrrrrrrr}
\hline\hline
&  & \multicolumn{10}{c}{Husbands' quantiles} &  \\
 \hline\hline
&  & d1 & d2 & d3 & d4 & d5 & d6 & d7 & d8 & d9 & d10 &  \\
 \hline
\parbox[t]{2mm}{\multirow{10}{*}{\rotatebox[origin=c]{90}{Wives' quantiles}}}
	&d1	&	\textbf{2.78}	&	1.577	&	1.207	&	0.951	&	0.773	&	0.657	&	0.576	&	0.584	&	0.481	&	0.412	\\	
	&d2	&	1.741	&	\textbf{1.455}	&	1.258	&	1.196	&	0.951	&	0.724	&	0.851	&	0.596	&	0.579	&	0.647	\\	
	&d3	&	1.218	&	1.846	&	\textbf{1.416}	&	1.17	&	1.24	&	0.889	&	0.919	&	0.718	&	0.511	&	0.072	\\	
	&d4	&	0.923	&	1.326	&	1.597	&	\textbf{1.096}	&	1.046	&	0.96	&	0.87	&	0.854	&	0.868	&	0.461	\\	
	&d5	&	0.778	&	0.969	&	1.218	&	1.411	&	\textbf{1.067}	&	0.977	&	0.988	&	0.951	&	0.878	&	0.764	\\	
	&d6	&	0.664	&	0.733	&	1.035	&	1.294	&	1.499	&	\textbf{1.226}	&	1.141	&	0.887	&	1.106	&	0.416	\\	
	&d7	&	0.638	&	0.674	&	0.75	&	0.957	&	1.065	&	1.528	&	\textbf{1.145}	&	1.049	&	1.071	&	1.123	\\	
	&d8	&	0.471	&	0.607	&	0.692	&	0.785	&	0.949	&	1.355	&	1.498	&	\textbf{1.309}	&	1.519	&	0.816	\\	
	&d9	&	0.513	&	0.469	&	0.591	&	0.656	&	0.821	&	0.995	&	1.222	&	1.385	&	\textbf{1.476}	&	1.873	\\	
	&d10	&	0.284	&	0.426	&	0.134	&	0.491	&	0.589	&	0.694	&	0.782	&	1.675	&	1.519	&	\textbf{3.407}	\\	
\hline\hline
\end{tabular}
\caption{Results of estimated sorting measures including couples who cohabit for the years 2021-2025 at different quantiles}
\label{tab:results_1_2021_cohab}
\end{table}

\begin{table}[tbp]
\begin{tabular}{l||l|rrrrrrrrrrr}
\hline\hline
&  & \multicolumn{10}{c}{Husbands' quantiles} &  \\
 \hline\hline
&  & d1 & d2 & d3 & d4 & d5 & d6 & d7 & d8 & d9 & d10 &  \\
 \hline
\parbox[t]{2mm}{\multirow{10}{*}{\rotatebox[origin=c]{90}{Wives' quantiles}}}
	&d1	&	\textbf{2.792}	&	1.567	&	1.143	&	0.971	&	0.765	&	0.719	&	0.574	&	0.546	&	0.515	&	0.406	\\	
	&d2	&	1.769	&	\textbf{1.432}	&	1.259	&	1.219	&	0.91	&	0.829	&	0.786	&	0.574	&	0.546	&	0.677	\\	
	&d3	&	1.227	&	1.844	&	\textbf{1.391}	&	1.213	&	1.195	&	0.864	&	0.944	&	0.807	&	0.513	&	0.0	\\	
	&d4	&	0.878	&	1.286	&	1.584	&	\textbf{1.106}	&	0.972	&	0.975	&	0.833	&	0.842	&	0.839	&	0.685	\\	
	&d5	&	0.761	&	0.989	&	1.177	&	1.363	&	\textbf{1.086}	&	0.924	&	0.983	&	0.978	&	0.902	&	0.836	\\	
	&d6	&	0.711	&	0.744	&	1.018	&	1.328	&	1.506	&	\textbf{1.231}	&	1.142	&	0.904	&	1.019	&	0.396	\\	
	&d7	&	0.595	&	0.713	&	0.735	&	0.957	&	1.017	&	1.472	&	\textbf{1.227}	&	1.078	&	1.055	&	1.151	\\	
	&d8	&	0.506	&	0.602	&	0.691	&	0.778	&	0.974	&	1.271	&	1.475	&	\textbf{1.258}	&	1.532	&	0.913	\\	
	&d9	&	0.527	&	0.509	&	0.606	&	0.687	&	0.763	&	0.978	&	1.164	&	1.469	&	\textbf{1.501}	&	1.796	\\	
	&d10	&	0.233	&	0.332	&	0.444	&	0.314	&	0.807	&	0.734	&	0.863	&	1.557	&	1.548	&	\textbf{3.168}	\\	
\hline\hline
\end{tabular}
\caption{Results of estimated sorting measures including couples who cohabit in 2021-2025 at different quantiles when $\protect\sigma_{D^*_w, Y^*_w} = 0$.}
\label{tab:results_2_2021_cohab}
\end{table}

\begin{table}[tbp]
\begin{tabular}{l||l|rrrrrrrrrrr}
\hline\hline
&  & \multicolumn{10}{c}{Husbands' quantiles} &  \\
 \hline\hline
&  & d1 & d2 & d3 & d4 & d5 & d6 & d7 & d8 & d9 & d10 &  \\
 \hline
\parbox[t]{2mm}{\multirow{10}{*}{\rotatebox[origin=c]{90}{Wives' quantiles}}}
	&d1	&	\textbf{2.795}	&	1.283	&	1.123	&	0.951	&	0.78	&	0.644	&	0.78	&	0.611	&	0.54	&	0.492	\\	
	&d2	&	1.812	&	\textbf{1.633}	&	1.221	&	1.149	&	0.829	&	1.104	&	0.59	&	0.629	&	0.512	&	0.519	\\	
	&d3	&	0.639	&	1.875	&	\textbf{1.057}	&	1.456	&	0.866	&	1.049	&	1.424	&	0.677	&	0.601	&	0.356	\\	
	&d4	&	0.926	&	1.345	&	1.423	&	\textbf{1.027}	&	1.14	&	1.109	&	0.939	&	0.877	&	0.731	&	0.483	\\	
	&d5	&	0.734	&	0.746	&	1.266	&	1.484	&	\textbf{1.172}	&	1.026	&	1.19	&	1.093	&	0.688	&	0.602	\\	
	&d6	&	0.64	&	0.881	&	1.018	&	1.193	&	1.313	&	\textbf{1.036}	&	1.101	&	1.12	&	0.871	&	0.828	\\	
	&d7	&	0.783	&	0.666	&	0.714	&	0.856	&	1.373	&	1.188	&	\textbf{0.993}	&	1.317	&	1.143	&	0.965	\\	
	&d8	&	0.578	&	0.546	&	0.685	&	0.751	&	0.971	&	1.233	&	1.42	&	\textbf{1.203}	&	1.265	&	1.35	\\	
	&d9	&	0.647	&	0.46	&	0.828	&	0.643	&	0.845	&	0.987	&	0.834	&	1.52	&	\textbf{1.852}	&	1.384	\\	
	&d10	&	0.47	&	0.556	&	0.629	&	0.528	&	0.708	&	0.615	&	0.739	&	0.944	&	1.8	&	\textbf{3.012}	\\	
\hline\hline
\end{tabular}
\caption{Results of estimated sorting measures including couples who cohabit in 2021-2025 at different quantiles when $\protect\protect\sigma_{D^*_w, Y^*_w} = 0, \protect\sigma_{D^*_w, Y^*_h} = 0$, and the distribution of $X$ as in total population.}
\label{tab:results_4_2021_cohab}
\end{table}

\begin{table}[tbp]
\begin{tabular}{l||l|rrrrrrrrrrr}
\hline\hline
&  & \multicolumn{10}{c}{Husbands' quantiles} &  \\
 \hline\hline
&  & d1 & d2 & d3 & d4 & d5 & d6 & d7 & d8 & d9 & d10 &  \\
 \hline
\parbox[t]{2mm}{\multirow{10}{*}{\rotatebox[origin=c]{90}{Wives' quantiles}}}
	&d1	&	\textbf{1.604}	&	1.312	&	1.144	&	1.097	&	0.971	&	0.905	&	0.85	&	0.752	&	0.674	&	0.692	\\	
	&d2	&	1.376	&	\textbf{1.259}	&	1.15	&	1.127	&	1.026	&	0.951	&	0.905	&	0.824	&	0.724	&	0.658	\\	
	&d3	&	1.186	&	1.18	&	\textbf{1.111}	&	1.086	&	1.036	&	0.975	&	0.951	&	0.901	&	0.801	&	0.772	\\	
	&d4	&	1.156	&	1.117	&	1.106	&	\textbf{1.083}	&	1.027	&	1.005	&	0.966	&	0.913	&	0.844	&	0.782	\\	
	&d5	&	0.944	&	0.997	&	1.048	&	1.035	&	\textbf{1.024}	&	1.035	&	1.021	&	1.013	&	0.966	&	0.917	\\	
	&d6	&	0.873	&	0.956	&	0.979	&	1.011	&	1.027	&	\textbf{1.033}	&	1.049	&	1.05	&	1.025	&	0.998	\\	
	&d7	&	0.796	&	0.875	&	0.937	&	0.973	&	1.003	&	1.04	&	\textbf{1.065}	&	1.091	&	1.121	&	1.1	\\	
	&d8	&	0.744	&	0.819	&	0.887	&	0.922	&	0.981	&	1.029	&	1.074	&	\textbf{1.127}	&	1.197	&	1.221	\\	
	&d9	&	0.62	&	0.723	&	0.813	&	0.868	&	0.972	&	1.031	&	1.081	&	1.185	&	\textbf{1.321}	&	1.385	\\	
	&d10	&	0.728	&	0.744	&	0.796	&	0.831	&	0.932	&	0.987	&	1.043	&	1.139	&	1.33	&	\textbf{1.469}	\\	
\hline\hline
\end{tabular}
\caption{Results of estimated sorting measures including couples who cohabit in 2021-2025 at different quantiles when all correlations are zero and distribution of $X$ as in total population.}
\label{tab:results_5_2021_cohab}
\end{table}

\section{Investigation of the instruments used for the exclusion restriction}
%%%%%%%%%%%%%%%%%%%%%%%%%%%%%%%%%%%%%%%%%%%%%%%%%%%%%%%%%%%%%%%%%%%%%%%%%%%%

\label{web_appendix:instruments}

We use three instruments in our analysis (1) an indicator for children in the household, (2) the number of household members, and (3) an indicator for children in the household who are under the age of 5. As stated in the main text, although these instruments are often used to correct for sample selection when using the CPS, there are also many reasons to discuss the validity of these instruments. We use two strategies in order to check the validity of these instruments. First, we use only the indicator for children in the household under the age of 5. The results of this exercise are reported in Tables \ref{tab:results_original_1976_instrument} to \ref{tab:results_4_1976_instrument} for the years 1976-1980 and in Tables \ref{tab:results_1_2021_instrument} to \ref{tab:results_5_2021_instrument} for the years 2021-2025. 

As a second strategy we also include an additional instrument that includes the costs of childcare based on county level data from the National Database of Childcare Prices (NDCP). For our exercise, we use the county level childcare prices of infants up to 2 years old and interact these with the indicator variable in the CPS for children below 5 years of age. We were able to merge 656,554 out of the 4,510,115 observations in the CPS from 1976-2025. The main reason for this reduction is that we only have the years 2008-2022 available for the NDCP. Restricting the original dataset to these years gives 1,172,324 observations. The second reason for the loss of observations is that the CPS does not always have the county data available. After our restrictions as described in the main text, we keep 293,015 observations with 233,194 having information about the childcare prices of infants. This leaves us with 114,238 married couples of which 52,292 are working FTFY. 

In order to investigate whether our results change by using this subset, Tables \ref{tab:results_original_merged} to \ref{tab:results_4_merged} report the results employing the same model as in the main text and using the years 2018-2022. Our exercise, which uses the additional instrument, is reported in Tables \ref{tab:results_1_merged1_instrument} to \ref{tab:results_5_merged1_instrument}. These results are not qualitatively different to the results in our main analysis.

\begin{table}[tbp]
\begin{tabular}{l||l|rrrrrrrrrrr}
\hline\hline
&  & \multicolumn{10}{c}{Husbands' quantiles} &  \\
 \hline\hline
&  & d1 & d2 & d3 & d4 & d5 & d6 & d7 & d8 & d9 & d10 &  \\
 \hline
\parbox[t]{2mm}{\multirow{10}{*}{\rotatebox[origin=c]{90}{Wives' quantiles}}}
	&d1	&	\textbf{2.007}	&	1.209	&	1.14	&	0.975	&	0.897	&	0.893	&	0.849	&	0.748	&	0.723	&	0.562	\\	
	&d2	&	1.632	&	\textbf{1.252}	&	1.202	&	1.11	&	0.954	&	0.894	&	0.787	&	0.758	&	0.773	&	0.634	\\	
	&d3	&	1.317	&	1.364	&	\textbf{1.086}	&	1.123	&	1.111	&	0.871	&	0.896	&	0.909	&	0.724	&	0.599	\\	
	&d4	&	1.187	&	1.35	&	1.168	&	\textbf{1.148}	&	0.983	&	0.939	&	0.855	&	0.942	&	0.76	&	0.666	\\	
	&d5	&	0.934	&	1.198	&	1.241	&	1.111	&	\textbf{1.207}	&	0.924	&	0.911	&	0.915	&	0.878	&	0.681	\\	
	&d6	&	0.724	&	1.076	&	1.134	&	1.076	&	1.046	&	\textbf{1.087}	&	1.036	&	0.971	&	0.941	&	0.91	\\	
	&d7	&	0.675	&	0.777	&	1.125	&	1.023	&	1.076	&	1.105	&	\textbf{1.068}	&	1.071	&	1.077	&	1.001	\\	
	&d8	&	0.602	&	0.562	&	0.857	&	1.165	&	1.06	&	1.066	&	1.158	&	\textbf{1.134}	&	1.208	&	1.191	\\	
	&d9	&	0.519	&	0.653	&	0.668	&	0.812	&	1.166	&	1.239	&	1.178	&	1.25	&	\textbf{1.15}	&	1.361	\\	
	&d10	&	0.402	&	0.557	&	0.378	&	0.457	&	0.501	&	0.982	&	1.262	&	1.303	&	1.767	&	\textbf{2.395}	\\	
\hline\hline
\end{tabular}
\caption{Results of estimated sorting measures using only children $<$ 5 y/o as instrument for the years 1976-1980 at different quantiles}
\label{tab:results_original_1976_instrument}
\end{table}

\begin{table}[tbp]
%
\caption{Results of estimated sorting measures using only children $<$ 5 y/o as instrument in 1976-1980 at different quantiles when $\protect\sigma_{D^*_w, Y^*_w} = 0$.}
\label{tab:results_1_1976_instrument}
\end{table}

\begin{table}[tbp]
%
\caption{Results of estimated sorting measures using only children $<$ 5 y/o as instrument in 1976-1980 at different quantiles when $\protect\sigma_{D^*_w, Y^*_w} = 0, \protect\sigma_{D^*_w, Y^*_h} = 0.$}
\label{tab:results_2_1976_instrument}
\end{table}

\begin{table}[tbp]
%
\caption{Results of estimated sorting measures using only children $<$ 5 y/o as instrument in 1976-1980 at different quantiles when $\protect\protect\sigma_{D^*_w, Y^*_w} = 0, \protect\sigma_{D^*_w, Y^*_h} = 0$, and the distribution of $X$ as in total population.}
\label{tab:results_3_1976_instrument}
\end{table}

\begin{table}[tbp]
%
\caption{Results of estimated sorting measures using only children $<$ 5 y/o as instrument in 1976-1980 at different quantiles when all correlations are zero and distribution of $X$ as in total population.}
\label{tab:results_4_1976_instrument}
\end{table}

\begin{table}[tbp]
%
\caption{Results of estimated sorting measures using only children $<$ 5 y/o as instrument for the years 2021-2025 at different quantiles}
\label{tab:results_1_2021_instrument}
\end{table}

\begin{table}[tbp]
%
\caption{Results of estimated sorting measures using only children $<$ 5 y/o as instrument in 2021-2025 at different quantiles when $\protect\sigma_{D^*_w, Y^*_w} = 0$.}
\label{tab:results_2_2021_instrument}
\end{table}

\begin{table}[tbp]
%
\caption{Results of estimated sorting measures using only children $<$ 5 y/o as instrument in 2021-2025 at different quantiles when $\protect\sigma_{D^*_w, Y^*_w} = 0, \protect\sigma_{D^*_w, Y^*_h} = 0.$}
\label{tab:results_3_2021_instrument}
\end{table}

\begin{table}[tbp]
%
\caption{Results of estimated sorting measures using only children $<$ 5 y/o as instrument in 2021-2025 at different quantiles when $\protect\protect\sigma_{D^*_w, Y^*_w} = 0, \protect\sigma_{D^*_w, Y^*_h} = 0$, and the distribution of $X$ as in total population.}
\label{tab:results_4_2021_instrument}
\end{table}

\begin{table}[tbp]
%
\caption{Results of estimated sorting measures using only children $<$ 5 y/o as instrument in 2021-2025 at different quantiles when all correlations are zero and distribution of $X$ as in total population.}
\label{tab:results_5_2021_instrument}
\end{table}

\begin{table}[tbp]
%
\caption{Results of estimated sorting measures using merged data with the NDCP dataset for the years 2018-2022 at different quantiles}
\label{tab:results_original_merged}
\end{table}

\begin{table}[tbp]
%
\caption{Results of estimated sorting measures using merged data with the NDCP dataset in 2018-2022 at different quantiles when $\protect\sigma_{D^*_w, Y^*_w} = 0$.}
\label{tab:results_1_merged}
\end{table}

\begin{table}[tbp]
%
\caption{Results of estimated sorting measures using merged data with the NDCP dataset in 2018-2022 at different quantiles when $\protect\sigma_{D^*_w, Y^*_w} = 0, \protect\sigma_{D^*_w, Y^*_h} = 0.$}
\label{tab:results_2_merged}
\end{table}

\begin{table}[tbp]
%
\caption{Results of estimated sorting measures using merged data with the NDCP dataset in 2018-2022 at different quantiles when $\protect\protect\sigma_{D^*_w, Y^*_w} = 0, \protect\sigma_{D^*_w, Y^*_h} = 0$, and the distribution of $X$ as in total population.}
\label{tab:results_3_merged}
\end{table}

\begin{table}[tbp]
%
\caption{Results of estimated sorting measures using merged data with the NDCP dataset in 2018-2022 at different quantiles when all correlations are zero and distribution of $X$ as in total population.}
\label{tab:results_4_merged}
\end{table}

\begin{table}[tbp]
%
\caption{Results of estimated sorting measures using childcare prices as instrument for the years 2018-2022 at different quantiles}
\label{tab:results_1_merged1_instrument}
\end{table}

\begin{table}[tbp]
%
\caption{Results of estimated sorting measures using childcare prices as instrument in 2018-2022 at different quantiles when $\protect\sigma_{D^*_w, Y^*_w} = 0$.}
\label{tab:results_2_merged1_instrument}
\end{table}

\begin{table}[tbp]
%
\caption{Results of estimated sorting measures using childcare prices as instrument in 2018-2022 at different quantiles when $\protect\sigma_{D^*_w, Y^*_w} = 0, \protect\sigma_{D^*_w, Y^*_h} = 0.$}
\label{tab:results_3_merged1_instrument}
\end{table}

\begin{table}[tbp]
%
\caption{Results of estimated sorting measures using childcare prices as instrument in 2018-2022 at different quantiles when $\protect\protect\sigma_{D^*_w, Y^*_w} = 0, \protect\sigma_{D^*_w, Y^*_h} = 0$, and the distribution of $X$ as in total population.}
\label{tab:results_4_merged1_instrument}
\end{table}

\begin{table}[tbp]
%
\caption{Results of estimated sorting measures using childcare prices as instrument in 2018-2022 at different quantiles when all correlations are zero and distribution of $X$ as in total population.}
\label{tab:results_5_merged1_instrument}
\end{table}